\documentclass[12pt]{article}
\usepackage{amsmath}
\usepackage{graphicx}
\usepackage{caption}
\usepackage{enumerate}
\usepackage[numbers,super,sort&compress]{natbib}
\usepackage{url} 
\usepackage{amsfonts}
\usepackage{hyperref}
\usepackage{float} 
\usepackage{amssymb}
\usepackage[utf8]{inputenc}
\usepackage{multirow}
\usepackage{xcolor}
\usepackage{adjustbox} 
\usepackage{subcaption}
\usepackage{placeins}
\usepackage{lineno}
\usepackage{tikz}
\usetikzlibrary{arrows.meta,positioning}
\usepackage{amsthm}

\theoremstyle{plain}

\newtheorem{proposition}{Proposition}

\theoremstyle{remark}
\newtheorem{remark}{Remark}

\theoremstyle{definition}

\newcommand{\blind}{0}

\usepackage[margin=1in, bottom=1.5in]{geometry}

\begin{document}
\def\spacingset#1{\renewcommand{\baselinestretch}%
{#1}\small\normalsize} \spacingset{1}

\if1\blind
{
  \title{\bf Title}
  \author{Author 1\thanks{
    The authors gratefully acknowledge \textit{please remember to list
    all relevant funding sources in the unblinded version}}\hspace{.2cm}\\
    Department of YYY, University of XXX\\
    and \\
    Author 2 \\
    Department of ZZZ, University of WWW}
  \maketitle
} \fi

\if0\blind
{
  \title{\bf MR-CCC: Bayesian Mendelian Randomization for Causal
    Cell--Cell Communication}
  \author{%
    \large
    Bitan Sarkar$^{1}$ and Yang Ni$^{2}$\thanks{Corresponding author:
    Yang Ni (\href{mailto:yang.ni@austin.utexas.edu}
    {yang.ni@austin.utexas.edu}).}
    \\[0.75em]
    $^{1}$Department of Statistics, Texas A\&M University,
    College Station, TX, USA\\
    $^{2}$Department of Statistics and Data Sciences,\\
    The University of Texas at Austin, Austin, TX, USA
  }
  \date{}
  \maketitle
} \fi

\spacingset{1.45}

\bigskip


\begin{abstract}
Cell--cell communication (CCC) is commonly inferred from
ligand--receptor co-expression, an associational paradigm that
cannot distinguish causal signaling from shared regulation or
confounding. We propose \emph{MR-CCC}, a Bayesian Mendelian
randomization framework that uses cis-eQTLs as instruments for
ligand and receptor expression and explicitly models
receptor-modulated ligand effects through an interaction term, so
the causal effect of a ligand can vary with receptor abundance.
A spike--and--slab prior yields posterior inclusion probabilities
quantifying evidence for causal signaling, and an efficient Gibbs
sampler provides scalable inference. Benchmarked against naive
regression, MVMR, and MR-BMA, MR-CCC controls false discoveries
under confounding while retaining high power, and uniquely
estimates both the ligand main and receptor-modulated interaction
effects. Applied to the OneK1K NK cells\,$\rightarrow$\,monocytes
axis, MR-CCC identifies eight discoveries across GABA, interferon, interleukin, and prostaglandin signaling, including a
stoichiometry-dependent dissociation of the two IL-18 receptor
chains and co-discovery of both obligate IFN-$\gamma$ receptor
subunits.
\end{abstract}


\spacingset{1.45}

\section{Introduction}

Cell--cell communication (CCC) is a fundamental organizing principle
of multicellular systems, governing immune regulation, development,
and disease progression.  At the molecular level, CCC is mediated
through ligand--receptor interactions, whereby ligands expressed in a
sender cell type bind cognate receptors in a receiver cell type to
initiate downstream signaling programs.  Advances in sequencing
technologies have enabled large-scale mapping of these interactions
across diverse biological systems, giving rise to computational
frameworks such as CellPhoneDB \cite{efremova2020cellphonedb} and
CellChat \cite{cellchat}.

These methods have established a widely used paradigm for CCC
inference based on co-expression and enrichment of ligand--receptor
pairs.  However, this paradigm is fundamentally associational: it
identifies coordinated expression patterns but does not distinguish
causal signaling from shared regulation or confounding.  Ligand
expression, receptor expression, and downstream pathway activity may
be jointly driven by latent factors such as donor heterogeneity,
environmental exposures, or cell-state programs, leading to spurious
communication signals.  As a result, current CCC analyses can
overestimate or misattribute intercellular signaling relationships.

This limitation motivates a paradigm shift: rather than
asking \emph{whether} ligand and receptor expression are jointly
elevated across cell types, the key causal question is
instead \emph{whether} variation in ligand expression in the sender
cell type \emph{mechanically drives} downstream pathway activity in
the receiver cell type through receptor engagement, independently of
any shared confounders.  Mendelian randomization (MR) provides a
principled strategy for precisely this causal question by leveraging
genetic variants as instrumental variables
\cite{davey_smith_hemani_2014}.  Because germline variants are fixed
at conception, they are largely independent of environmental and
downstream confounders, enabling identification of causal effects
under standard MR assumptions.  MR has been widely applied in
epidemiology and extended to settings with pleiotropy
\cite{bowden2015mr}, multivariable exposures \cite{Sanderson2021},
and Bayesian formulations \cite{Zuber2020}.  However, existing MR
approaches are not designed for CCC problems: they typically do not
account for interaction effects, which are crucial for modeling
receptor-modulated communication.  In CCC, the effect of a ligand is
inherently context-dependent and often modulated by receptor abundance
in the receiving cell.  Capturing this interaction effect requires
moving beyond additive models toward interaction-based causal
formulations.  At the same time, the large number of candidate
ligand--receptor--pathway combinations necessitates principled
selection procedures that can separate true communication signals
from noise.

To address these challenges, we introduce \emph{MR-CCC}, a Bayesian
Mendelian randomization framework for causal cell--cell communication
inference.  MR-CCC operationalizes this causal paradigm by jointly
modeling ligand expression in the sender cell type, receptor
expression in the receiver cell type, and downstream pathway activity
within a coherent causal framework.  The method uses cis-eQTLs as
instruments for both ligand and receptor expression, replacing
observed molecular quantities with their instrument-based conditional
expectations to mitigate confounding.  It explicitly incorporates
receptor-modulated ligand effects through an interaction term,
allowing communication strength to vary with receptor abundance.  An
identification result establishes that the plug-in working model
recovers the structural causal coefficients under standard MR
assumptions, even in the presence of the endogenous interaction term
(Proposition~\ref{prop:ident}).  A spike--and--slab prior enables
probabilistic selection of communication effects, yielding posterior
inclusion probabilities (PIPs) that quantify evidence for causal
signaling.  Posterior inference is carried out via an efficient Gibbs
sampler with closed-form full conditionals for all model parameters,
making MR-CCC computationally tractable for genome-scale CCC screens
involving thousands of ligand--receptor--pathway triplets. An overview
of the framework, the simulation benchmark, and the lead biological
discovery is shown in Figure~\ref{fig:graphical_abstract}.

\begin{figure}[htb]
  \centering
  \includegraphics[width=0.88\linewidth]{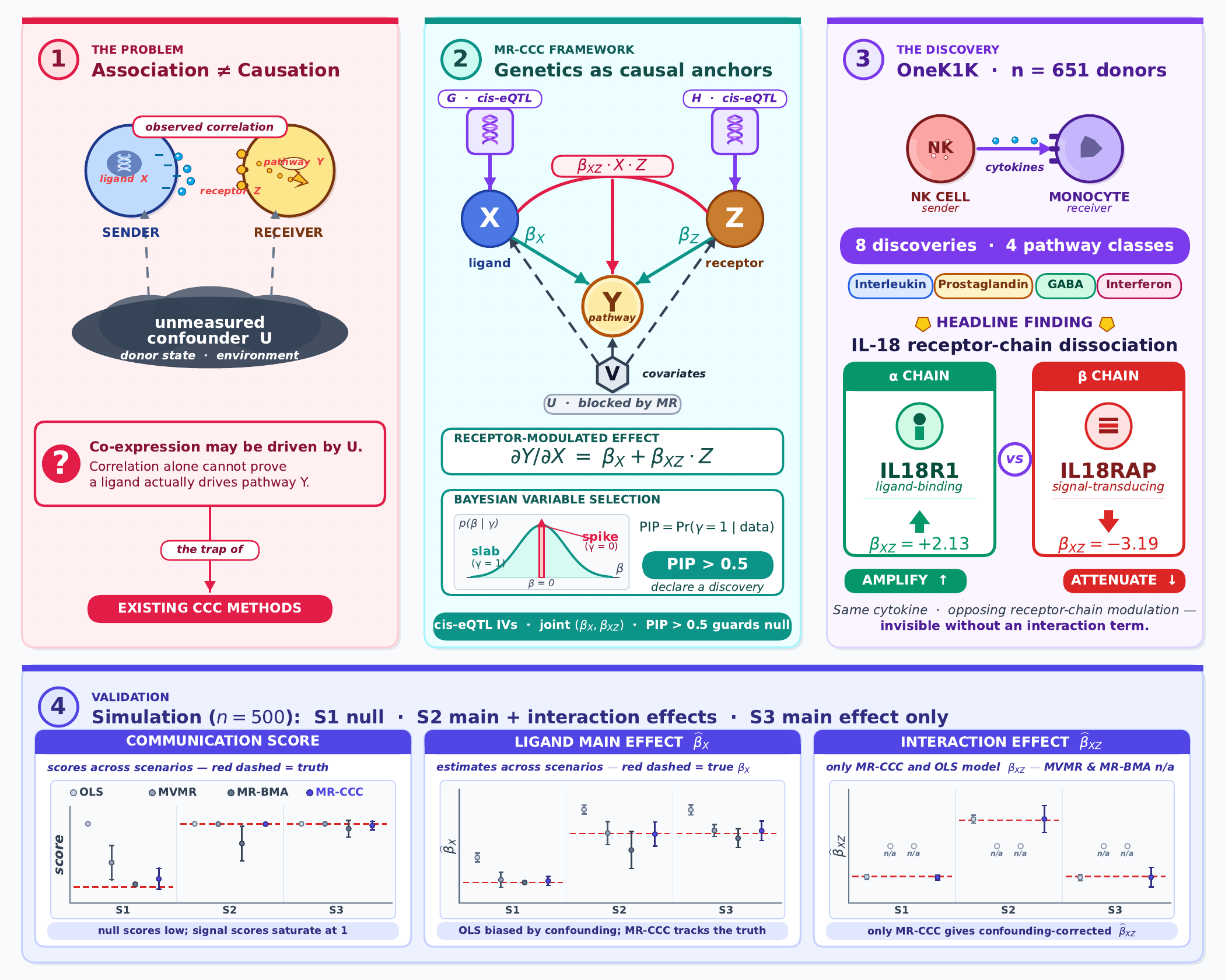}
  \caption{\textbf{Overview of the MR-CCC framework, simulation
  benchmark, and lead biological discovery.}
  \textbf{(1)}~Ligand--receptor co-expression between a sender and a
  receiver cell can be driven by an unmeasured confounder $U$, so
  correlation alone cannot establish causal communication.
  \textbf{(2)}~MR-CCC uses cis-eQTLs $G$ and $H$ as instruments for
  ligand $X$ and receptor $Z$ and adjusts for covariates $V$, blocking
  $U$. The working outcome model identifies the receptor-modulated
  ligand effect $\partial Y/\partial X = \beta_X + \beta_{XZ}Z$; a
  spike--and--slab prior on $(\beta_X,\beta_{XZ})$ yields a posterior
  inclusion probability $\mathrm{PIP} = \Pr(\gamma = 1 \mid
  \mathrm{data})$, with discovery declared at $\mathrm{PIP} > 0.5$.
  \textbf{(3)}~Application to NK cells\,$\rightarrow$\,monocytes in
  the OneK1K cohort ($n = 651$ donors): eight discoveries across
  four pathway classes, led by an IL-18 receptor-chain dissociation
  (\textit{IL18R1}: $\hat\beta_{XZ}^{(s)} = +2.13$;
  \textit{IL18RAP}: $\hat\beta_{XZ}^{(s)} = -3.19$).
  \textbf{(4)}~Simulation at $n = 500$ across S1 (null), S2 (main
  $+$ interaction), and S3 (main only): MR-CCC controls the null
  score, remains unbiased for $\beta_X$ under confounding, and is
  the only method that recovers a confounding-corrected
  $\beta_{XZ}$.}
  \label{fig:graphical_abstract}
\end{figure}

We benchmark MR-CCC against three alternatives --- naive regression
(OLS), multivariable Mendelian randomization (MVMR\cite{Sanderson2021}),
and Bayesian model averaging MR (MR-BMA\cite{Zuber2020}) --- across
a null and two signal scenarios. MR-CCC alone combines strong null
calibration, robustness to unmeasured confounding, and joint
regularized estimation of both the ligand main and receptor-modulated
interaction effects under a biologically motivated structural model.

Applied to the OneK1K cohort \cite{yazar2022onek1k} along the
NK cells\,$\rightarrow$\,monocytes axis, MR-CCC identifies eight
high-confidence causal communication signals across GABA, interferon,
interleukin, and prostaglandin signaling --- including a receptor
stoichiometry-dependent dissociation of the two obligate IL-18
receptor chains --- while correctly assigning near-zero posterior
probability to ICAM adhesion pairs and the reversed IL-15 direction,
demonstrating sender--receiver-specific causal inference not
achievable by co-expression tools.

\section{Results}
\label{sec:results}

We evaluate MR-CCC in two complementary settings: simulation
experiments that benchmark its statistical properties against
three alternatives, and an application to the OneK1K single-cell
eQTL cohort.

\subsection{Simulation Study}
\label{sec:simulation}

We conducted simulation experiments to evaluate MR-CCC. The study was designed to assess three aspects of
performance: (i) discrimination between the absence and presence of
CCC, (ii) estimation accuracy for the ligand main
effect \(\beta_X\) and the ligand--receptor interaction effect
\(\beta_{XZ}\), and (iii) comparative performance relative to a naive
regression baseline and two established MR
alternatives. 
We simulated data from the model in Equation~\eqref{eq:dgm} at four
sample sizes $n \in \{500,\;1000,\;10{,}000,\;30{,}000\}$ across
three scenarios; full simulation parameter values are in Methods
(Section~\ref{subsec:dgm}).

\subsubsection{Simulation Scenarios}
\label{subsec:scenarios}

Three scenarios were defined by the values of \(\beta_X\),
\(\beta_{XZ}\), and
\(\gamma = 1-\mathbb{I}(\beta_X =\beta_{XZ} = 0)\).

\paragraph{Scenario 1 (S1): No communication,
\(\gamma = 0\), \(\beta_X = 0\), \(\beta_{XZ} = 0\).}
A null setting with neither a ligand main effect nor a
ligand-receptor interaction.

\paragraph{Scenario 2 (S2): Communication with receptor modulation,
\(\gamma = 1\), \(\beta_X = 0.3\), \(\beta_{XZ} = 0.3\).}
Communication is driven jointly by a ligand main effect and a
receptor-modulated ligand effect.

\paragraph{Scenario 3 (S3): Communication without receptor modulation,
\(\gamma = 1\), \(\beta_X = 0.3\), \(\beta_{XZ} = 0\).}
Communication is driven solely by the ligand main effect.


\subsubsection{Competing Methods}
We compared MR-CCC against three alternatives: ordinary least
squares (OLS, naive regression), multivariable Mendelian
randomization (MVMR \cite{Sanderson2021}), and Bayesian model
averaging MR (MR-BMA \cite{Zuber2020}). Detailed implementations
and decision rules for each method are in Methods
(Sections~\ref{subsec:methods} and \ref{subsec:metrics}).

\subsubsection{Simulation Results}
\label{subsec:results}

\subsubsection*{Scenario~1: No communication
(\(\gamma=0\), \(\beta_X=0\), \(\beta_{XZ}=0\))}

Scenario~1 evaluates false positive behavior. A well-calibrated method
should assign communication scores near zero, maintain negligible
rejection rates, and recover both \(\beta_X\) and \(\beta_{XZ}\) near
zero.

As shown in Table~\ref{tab:sim_s1}, MR-CCC satisfies all three criteria
across all sample sizes. Its communication score ranges from 0.078--0.130 across sample sizes, and its rejection rate is at most 5\%. Both \(\mathrm{Bias}(\beta_X)\)
and \(\mathrm{MAD}(\beta_X)\) shrink toward zero as \(n\) increases,
reaching 0.001 (0.002) at \(n = 30{,}000\); the same pattern holds for
\(\beta_{XZ}\).

The contrast with OLS is stark. OLS assigns a score of 1.000
and rejects in every single replicate at every sample size. Because
\(X\), \(Z\), and \(Y\) share the unmeasured confounder \(U\),
confounding alone generates an association signal indistinguishable from
true communication under an associational model. Crucially, OLS does
\emph{not} self-correct as \(n\) grows: its bias in \(\beta_X\) remains at or above
0.151 even at \(n = 30{,}000\), confirming that large sample sizes
cannot rescue an associational paradigm under unmeasured confounding.

MVMR substantially corrects the estimation bias by instrumenting both
exposures, demonstrating the value of a causal CCC framework.
Nevertheless, at the detection level MVMR does not fully resolve the
false discovery problem: its null communication scores remain in the
range 0.371--0.564 and its rejection rate reaches 5\% at \(n = 500\).
This reflects residual finite-sample correlation between the
instrumented ligand exposure \(\hat X\) and the outcome even under the
null.

MR-BMA correctly identifies the absence of a ligand effect: its
\(\mathrm{MIP}_X\) scores remain in the narrow range 0.040--0.056
because the ligand instruments \(G\) do not predict \(Y\) after
accounting for the receptor instruments \(H\), which reflect the genuine
receptor main effect \(\beta_Z = 0.5\). However, MR-BMA does not model
\(\beta_{XZ}\), limiting its applicability when receptor-modulated
communication is of scientific interest.

MR-CCC's spike--and--slab prior provides the additional regularization
that closes this gap: by shrinking the communication effect vector toward
zero unless the data supply sufficiently strong evidence of a signal, it
controls false discoveries in a way that neither OLS nor MVMR achieves,
while simultaneously producing a structured joint estimate of
\((\beta_X, \beta_{XZ})\) that neither MR-BMA nor MVMR can provide. As shown in Figure~\ref{fig:sim_score}, MR-CCC
scores cluster tightly near zero, MR-BMA scores remain low, OLS scores
sit uniformly at one, and MVMR occupies an intermediate range.

\begin{table}[htb]
\centering
\caption{\textbf{Simulation results under the null scenario (S1).} Scenario~1 (\(\gamma=0\), \(\beta_X=0\),
\(\beta_{XZ}=0\)). Each entry is reported as mean (sd) across 20 replicates.
Score denotes the communication score produced by each method; lower values and
lower rejection rates are preferred under the null. ``---'' indicates the method does not estimate $\beta_{XZ}$ (MR-BMA and MVMR).}
\begin{adjustbox}{max width=\textwidth}
\begin{tabular}{llcccccc}
\hline
$n$ & Method & Score & Rejection rate & Bias($\beta_X$) & MAD($\beta_X$)
    & Bias($\beta_{XZ}$) & MAD($\beta_{XZ}$) \\
\hline
500
  & OLS    & 1.000 (0.000) & 1.000 (0.000) & 0.154 (0.027) & 0.154 (0.027)
            & $-$0.001 (0.013) & 0.011 (0.006) \\
  & MVMR   & 0.387 (0.272) & 0.050 (0.224) & 0.017 (0.045) & 0.036 (0.031) & --- & --- \\
  & MR-BMA & 0.044 (0.023) & 0.000 (0.000) & 0.001 (0.003) & 0.002 (0.003)
            & --- & --- \\
  & MR-CCC & 0.130 (0.164) & 0.050 (0.224) & 0.010 (0.027) & 0.012 (0.026)
            & $-$0.005 (0.013) & 0.008 (0.012) \\
\hline
1000
  & OLS    & 1.000 (0.000) & 1.000 (0.000) & 0.149 (0.017) & 0.149 (0.017)
            & $-$0.003 (0.009) & 0.007 (0.006) \\
  & MVMR   & 0.371 (0.290) & 0.000 (0.000) & 0.009 (0.031) & 0.024 (0.021) & --- & --- \\
  & MR-BMA & 0.040 (0.011) & 0.000 (0.000) & 0.001 (0.001) & 0.001 (0.001)
            & --- & --- \\
  & MR-CCC & 0.090 (0.082) & 0.000 (0.000) & 0.002 (0.008) & 0.004 (0.007)
            & 0.001 (0.005) & 0.003 (0.004) \\
\hline
10000
  & OLS    & 1.000 (0.000) & 1.000 (0.000) & 0.153 (0.006) & 0.153 (0.006)
            & $-$0.001 (0.005) & 0.004 (0.003) \\
  & MVMR   & 0.428 (0.197) & 0.000 (0.000) & 0.001 (0.009) & 0.008 (0.004) & --- & --- \\
  & MR-BMA & 0.040 (0.009) & 0.000 (0.000) & 0.000 (0.001) & 0.000 (0.000)
            & --- & --- \\
  & MR-CCC & 0.078 (0.048) & 0.000 (0.000) & 0.000 (0.001) & 0.001 (0.001)
            & 0.000 (0.001) & 0.001 (0.001) \\
\hline
30000
  & OLS    & 1.000 (0.000) & 1.000 (0.000) & 0.151 (0.005) & 0.151 (0.005)
            & $-$0.002 (0.002) & 0.002 (0.002) \\
  & MVMR   & 0.564 (0.222) & 0.000 (0.000) & $-$0.002 (0.007) & 0.007 (0.004) & --- & --- \\
  & MR-BMA & 0.056 (0.040) & 0.000 (0.000) & 0.000 (0.001) & 0.001 (0.001)
            & --- & --- \\
  & MR-CCC & 0.101 (0.119) & 0.000 (0.000) & 0.001 (0.002) & 0.001 (0.002)
            & 0.000 (0.001) & 0.001 (0.001) \\
\hline
\end{tabular}
\end{adjustbox}
\label{tab:sim_s1}
\end{table}

\subsubsection*{Scenario~2: Communication with receptor modulation
(\(\gamma=1\), \(\beta_X=0.3\), \(\beta_{XZ}=0.3\))}

Scenario~2 evaluates detection power when both the ligand main effect and
the receptor-modulated interaction are present. A well-performing method
should assign high communication scores and reject in nearly all
replicates.

Table~\ref{tab:sim_s2} shows that MR-CCC achieves this goal throughout.
Its score reaches 0.996 already at \(n = 500\) and is exactly 1.000 at
larger sample sizes; its rejection rate is 1.000 at all \(n\). More
importantly, MR-CCC provides accurate joint estimation:
\(\mathrm{Bias}(\beta_X) = -0.003\) and
\(\mathrm{Bias}(\beta_{XZ}) = 0.006\) at \(n = 500\), converging
toward zero with increasing \(n\). OLS also rejects in every replicate,
but its \(\beta_X\) bias remains above 0.147 uniformly.

MVMR matches MR-CCC in detection power; its score is 0.997 at
\(n = 500\) and exactly 1.000 at all larger sample sizes, a consequence
of the frequentist null being severely violated by both causal signals. Its estimates of \(\beta_X\) are accurate (bias near zero, MAD
comparable to MR-CCC), confirming that MR correction resolves the
confounding component; \(\beta_{XZ}\) is not estimated by MVMR.

MR-BMA shows low power at \(n = 500\) (score 0.690, rejection rate
0.650), improving to near-complete detection only at \(n \geq 1000\).
This underperformance arises because MR-BMA operates on per-instrument
summary statistics averaged across the interaction signal; when both
\(\beta_X\) and \(\beta_{XZ}\) are nonzero, the reduced-form association
of each instrument with \(Y\) conflates the two effects, making it
harder for the BMA step to assign high posterior probability to the
ligand model at small \(n\). Furthermore, even at large sample sizes,
MR-BMA's MACE for \(\beta_X\) incurs a persistent negative bias
(e.g.~\(-0.028\) at \(n = 30{,}000\)); this is a structural artifact
of pooling all instrument information without explicitly modelling the
interaction, so that part of the ligand effect is absorbed into the
receptor-model term. MR-BMA provides no estimate of \(\beta_{XZ}\).

MR-CCC's combination of Bayesian regularization and explicit interaction
modelling allows it to maintain full detection power at small \(n\) while
providing reliable estimation of the full effect vector at all sample
sizes. Figure~\ref{fig:sim_betaX} makes clear that OLS estimates are
shifted uniformly upward, whereas MR-CCC and MVMR track the true value;
Figure~\ref{fig:sim_betaXZ} shows that MR-CCC centers near the true
\(\beta_{XZ} = 0.3\); MVMR and MR-BMA are excluded from this figure as they do not
model the interaction.

\begin{table}[htb]
\centering
\caption{\textbf{Simulation results under the main + interaction
  signal scenario (S2).} Scenario~2 (\(\gamma=1\), \(\beta_X=0.3\),
\(\beta_{XZ}=0.3\)). Each entry is reported as mean (sd) across 20 replicates.
Score denotes the communication score produced by each method; higher values and
higher rejection rates are preferred when signal is present. ``---'' indicates that
\(\beta_{XZ}\) is not estimated by MR-BMA or MVMR.}
\begin{adjustbox}{max width=\textwidth}
\begin{tabular}{llcccccc}
\hline
$n$ & Method & Score & Rejection rate & Bias($\beta_X$) & MAD($\beta_X$)
    & Bias($\beta_{XZ}$) & MAD($\beta_{XZ}$) \\
\hline
500
  & OLS    & 1.000 (0.000) & 1.000 (0.000) & 0.147 (0.027)    & 0.147 (0.027)
            & 0.005 (0.021)    & 0.016 (0.013) \\
  & MVMR   & 0.997 (0.006) & 1.000 (0.000) & 0.003 (0.071)    & 0.053 (0.046)
            & --- & --- \\
  & MR-BMA & 0.690 (0.273) & 0.650 (0.489) & $-$0.101 (0.113) & 0.120 (0.091)
            & --- & --- \\
  & MR-CCC & 0.996 (0.017) & 1.000 (0.000) & $-$0.003 (0.074) & 0.061 (0.039)
            & 0.006 (0.071)    & 0.053 (0.045) \\
\hline
1000
  & OLS    & 1.000 (0.000) & 1.000 (0.000) & 0.147 (0.021)    & 0.147 (0.021)
            & 0.003 (0.012)    & 0.010 (0.007) \\
  & MVMR   & 1.000 (0.000) & 1.000 (0.000) & $-$0.001 (0.045) & 0.038 (0.023)
            & --- & --- \\
  & MR-BMA & 0.984 (0.031) & 1.000 (0.000) & $-$0.028 (0.048) & 0.047 (0.029)
            & --- & --- \\
  & MR-CCC & 1.000 (0.000) & 1.000 (0.000) & 0.005 (0.048)    & 0.041 (0.024)
            & $-$0.005 (0.043) & 0.035 (0.024) \\
\hline
10000
  & OLS    & 1.000 (0.000) & 1.000 (0.000) & 0.153 (0.007)    & 0.153 (0.007)
            & $-$0.001 (0.005) & 0.004 (0.003) \\
  & MVMR   & 1.000 (0.000) & 1.000 (0.000) & 0.003 (0.016)    & 0.012 (0.010)
            & --- & --- \\
  & MR-BMA & 1.000 (0.000) & 1.000 (0.000) & $-$0.022 (0.015) & 0.023 (0.014)
            & --- & --- \\
  & MR-CCC & 1.000 (0.000) & 1.000 (0.000) & 0.006 (0.016)    & 0.013 (0.011)
            & 0.002 (0.015)    & 0.012 (0.009) \\
\hline
30000
  & OLS    & 1.000 (0.000) & 1.000 (0.000) & 0.152 (0.005)    & 0.152 (0.005)
            & 0.000 (0.003)    & 0.002 (0.002) \\
  & MVMR   & 1.000 (0.000) & 1.000 (0.000) & 0.000 (0.008)    & 0.006 (0.005)
            & --- & --- \\
  & MR-BMA & 1.000 (0.000) & 1.000 (0.000) & $-$0.028 (0.008) & 0.028 (0.008)
            & --- & --- \\
  & MR-CCC & 1.000 (0.000) & 1.000 (0.000) & 0.001 (0.009)    & 0.007 (0.005)
            & 0.001 (0.009)    & 0.007 (0.005) \\
\hline
\end{tabular}
\end{adjustbox}
\label{tab:sim_s2}
\end{table}

\subsubsection*{Scenario~3: Communication without receptor modulation
(\(\gamma=1\), \(\beta_X=0.3\), \(\beta_{XZ}=0\))}

In this scenario, communication is present (\(\beta_X \neq 0\)) but the
interaction is absent (\(\beta_{XZ} = 0\)). 

Table~\ref{tab:sim_s3} shows that MR-CCC, MVMR, and MR-BMA all assign
scores close to one and reject in nearly all replicates, correctly
identifying the presence of communication. The more informative
comparison comes from estimation accuracy.
For \(\beta_X\), MR-CCC and MVMR both achieve small bias and decreasing
MAD with \(n\) (MR-CCC: bias $= 0.017$, MAD $= 0.047$ at \(n = 500\);
bias $= 0.000$, MAD $= 0.007$ at \(n = 30{,}000\)). OLS again shows
persistent upward bias exceeding 0.151 at all sample sizes, whereas MR-BMA incurs a small but persistent bias (e.g.~$-0.029$ at
\(n = 30{,}000\)) and
provides no estimate of \(\beta_{XZ}\).

The critical evaluation concerns \(\beta_{XZ}\), which is truly zero.
Figure~\ref{fig:sim_betaXZ} shows that MR-CCC keeps estimates of
\(\beta_{XZ}\) tightly centered near zero across all sample sizes, with
variability decreasing as \(n\) grows (MAD $= 0.036$ at $n=500$;
MAD $= 0.004$ at $n=30{,}000$). In particular,
the MR-CCC spike--and--slab prior does not spuriously activate the
interaction component when only the main effect is present.

 In summary, Scenario~3 confirms that
MR-CCC's joint modelling of \((\beta_X, \beta_{XZ})\) is not only
powerful but also specific: it correctly infers the pattern of
communication without inventing interaction effects that are not present
in the data.

\begin{table}[htb]
\centering
\caption{\textbf{Simulation results under the main-effect-only
  scenario (S3).} Scenario~3 (\(\gamma=1\), \(\beta_X=0.3\),
\(\beta_{XZ}=0\)). Each entry is reported as mean (sd) across 20 replicates.
Score denotes the communication score produced by each method; higher values and
higher rejection rates are preferred when signal is present. ``---'' indicates that
\(\beta_{XZ}\) is not estimated by MR-BMA or MVMR.}
\begin{adjustbox}{max width=\textwidth}
\begin{tabular}{llcccccc}
\hline
$n$ & Method & Score & Rejection rate & Bias($\beta_X$) & MAD($\beta_X$)
    & Bias($\beta_{XZ}$) & MAD($\beta_{XZ}$) \\
\hline
500
  & OLS    & 1.000 (0.000) & 1.000 (0.000) & 0.145 (0.031)    & 0.145 (0.031)
            & $-$0.005 (0.017) & 0.014 (0.010) \\
  & MVMR   & 1.000 (0.000) & 1.000 (0.000) & 0.019 (0.036)    & 0.033 (0.023)
            & --- & --- \\
  & MR-BMA & 0.924 (0.132) & 1.000 (0.000) & $-$0.028 (0.058) & 0.043 (0.048)
            & --- & --- \\
  & MR-CCC & 0.974 (0.068) & 1.000 (0.000) & 0.017 (0.059)    & 0.047 (0.038)
            & $-$0.003 (0.052) & 0.036 (0.037) \\
\hline
1000
  & OLS    & 1.000 (0.000) & 1.000 (0.000) & 0.151 (0.024)    & 0.151 (0.024)
            & 0.001 (0.013)    & 0.010 (0.007) \\
  & MVMR   & 1.000 (0.000) & 1.000 (0.000) & 0.020 (0.036)    & 0.032 (0.024)
            & --- & --- \\
  & MR-BMA & 0.999 (0.001) & 1.000 (0.000) & $-$0.010 (0.030) & 0.026 (0.018)
            & --- & --- \\
  & MR-CCC & 1.000 (0.000) & 1.000 (0.000) & 0.027 (0.045)    & 0.040 (0.034)
            & 0.008 (0.026)    & 0.022 (0.016) \\
\hline
10000
  & OLS    & 1.000 (0.000) & 1.000 (0.000) & 0.153 (0.008)    & 0.153 (0.008)
            & $-$0.001 (0.005) & 0.004 (0.003) \\
  & MVMR   & 1.000 (0.000) & 1.000 (0.000) & 0.002 (0.011)    & 0.008 (0.007)
            & --- & --- \\
  & MR-BMA & 1.000 (0.000) & 1.000 (0.000) & $-$0.025 (0.010) & 0.025 (0.010)
            & --- & --- \\
  & MR-CCC & 1.000 (0.000) & 1.000 (0.000) & 0.004 (0.014)    & 0.012 (0.008)
            & $-$0.006 (0.011) & 0.009 (0.009) \\
\hline
30000
  & OLS    & 1.000 (0.000) & 1.000 (0.000) & 0.151 (0.004)    & 0.151 (0.004)
            & 0.000 (0.002)    & 0.002 (0.002) \\
  & MVMR   & 1.000 (0.000) & 1.000 (0.000) & $-$0.002 (0.007) & 0.006 (0.004)
            & --- & --- \\
  & MR-BMA & 1.000 (0.000) & 1.000 (0.000) & $-$0.029 (0.006) & 0.029 (0.006)
            & --- & --- \\
  & MR-CCC & 1.000 (0.000) & 1.000 (0.000) & 0.000 (0.009)    & 0.007 (0.005)
            & 0.001 (0.005)    & 0.004 (0.003) \\
\hline
\end{tabular}
\end{adjustbox}
\label{tab:sim_s3}
\end{table}

\begin{figure}[htb]
\centering
\includegraphics[width=0.95\linewidth]{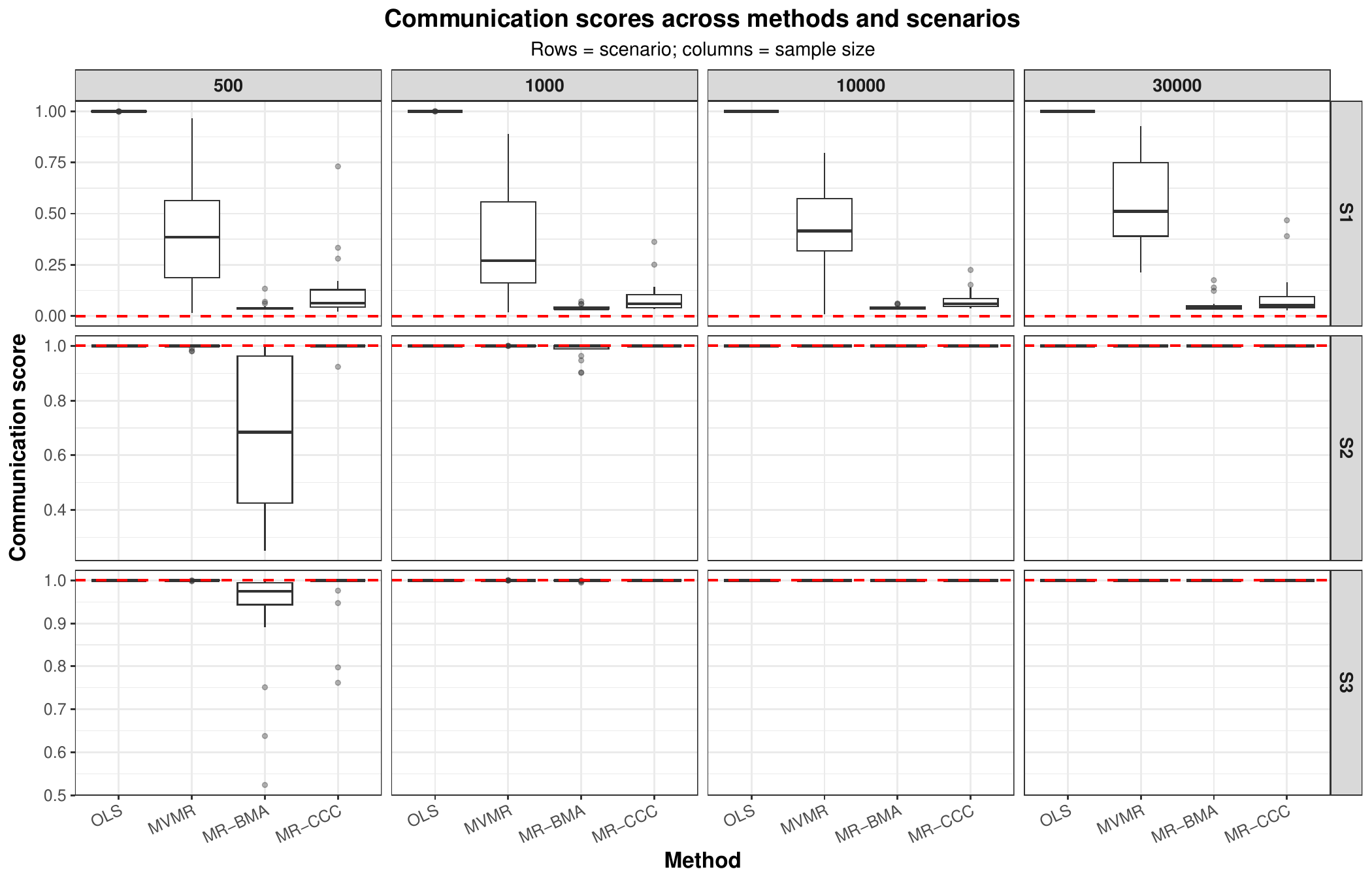}
\caption{\textbf{Communication scores across scenarios, sample sizes,
  and methods.} Boxplots of communication scores across 20 replicates,
  stratified by scenario (rows: S1--S3) and sample size (columns:
  \(n \in \{500, 1000, 10{,}000, 30{,}000\}\)). For MR-CCC the score
  is the posterior inclusion probability
  \(\Pr(\gamma = 1 \mid \text{data})\). For MR-BMA the score
  is \(\mathrm{MIP}_X\). For OLS the score is \(1 - p_F\), where \(p_F\) is the joint \(F\)-test \(p\)-value for \(H_0\colon \beta_X = \beta_{XZ} = 0\).
  For MVMR the score is \(1 - p_t\), where \(p_t\) is the \(t\)-test \(p\)-value for \(H_0\colon \beta_X = 0\). The red dashed line marks
  the true communication state (0 for S1; 1 for S2 and S3). Higher
  scores are preferred when signal is present; lower scores are
  preferred under the null.}
\label{fig:sim_score}
\end{figure}

\begin{figure}[htb]
\centering
\includegraphics[width=0.95\linewidth]{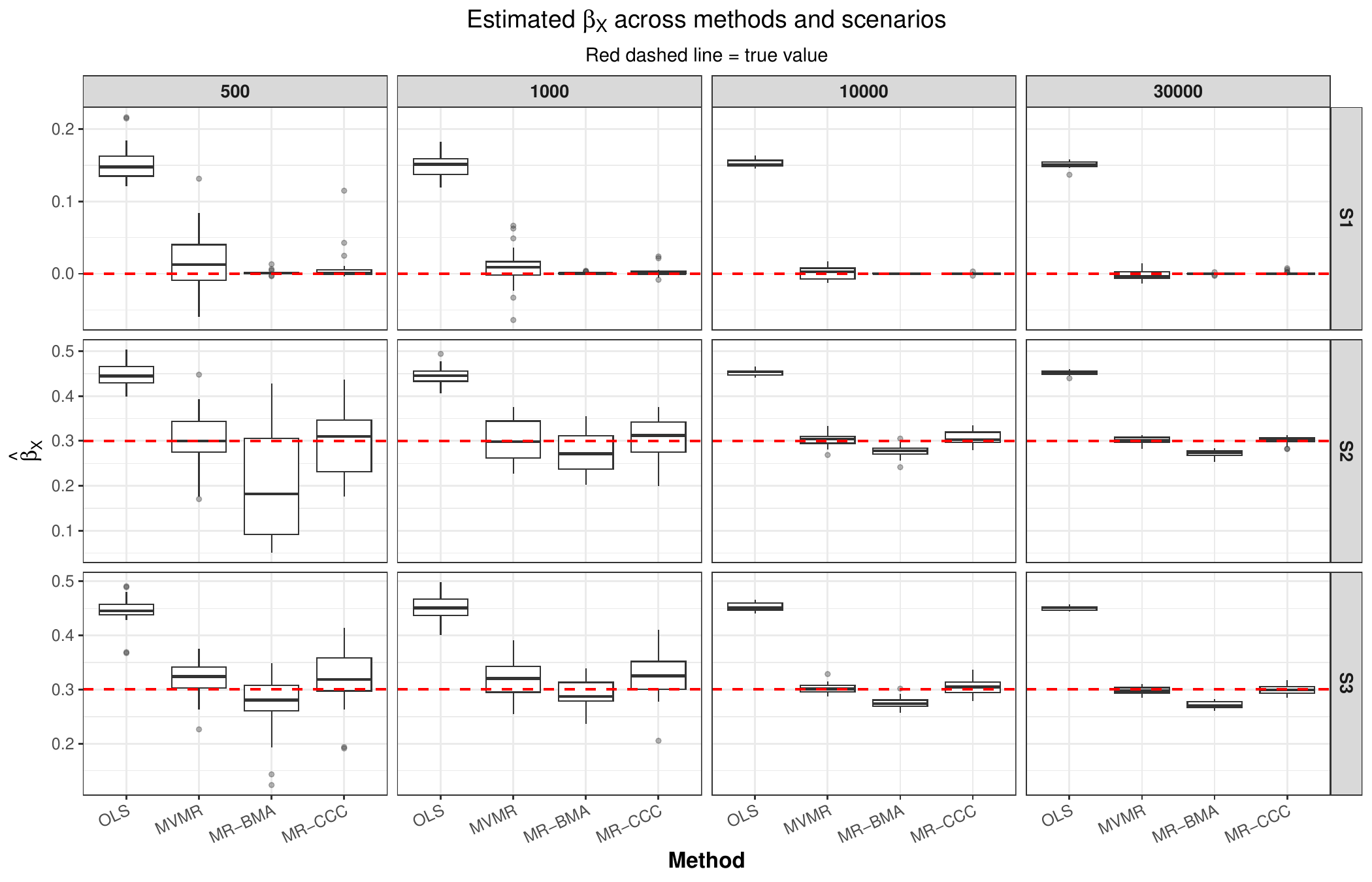}
\caption{\textbf{Estimates of the ligand main effect $\beta_X$ across
  scenarios, sample sizes, and methods.} Boxplots of the estimated ligand main effect \(\hat\beta_X\)
  across 20 replicates, stratified by scenario and sample size. Rows
  correspond to scenarios S1--S3; columns correspond to sample size. For
  MR-BMA the plotted value is the model-averaged causal effect (MACE)
  for the ligand exposure. The red dashed line indicates the true value
  of \(\beta_X\) in each scenario (0 for S1; 0.3 for S2 and S3).}
\label{fig:sim_betaX}
\end{figure}

\begin{figure}[htb]
\centering
\includegraphics[width=0.95\linewidth]{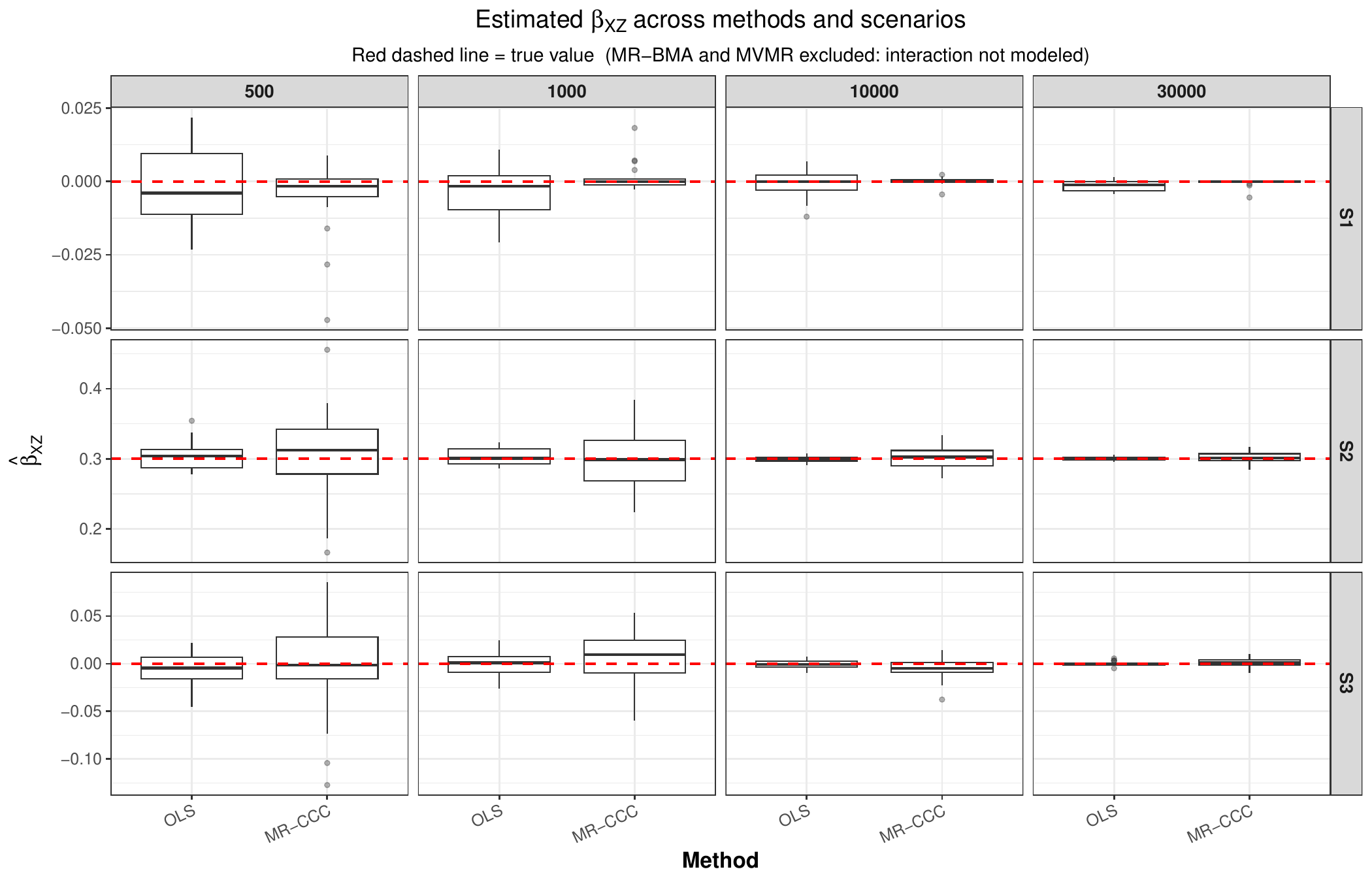}
\caption{\textbf{Estimates of the receptor-modulated interaction
  effect $\beta_{XZ}$ across scenarios, sample sizes, and methods.} Boxplots of the estimated ligand--receptor interaction effect
  \(\hat\beta_{XZ}\) across 20 replicates, stratified by scenario and
  sample size. MR-BMA and MVMR are omitted as neither models the
  interaction term \(\beta_{XZ}\). The red dashed line indicates the
  true value (0 for S1 and S3; 0.3 for S2). Under S1 and S3, where the
  interaction is truly absent, MR-CCC remains well-calibrated near zero.}
\label{fig:sim_betaXZ}
\end{figure}





\paragraph{Summary.}
MR-CCC uniquely combines strong null calibration, robustness to
unmeasured confounding via Bayesian MR, and joint estimation of
$(\beta_X,\beta_{XZ})$ under a biologically motivated
ligand--receptor model. The real-data analysis below uses
$n = 651$, which lies in the range where MR-CCC's advantage over
competing methods is largest.

\subsection{Real Data Analysis}
\label{sec:realdata}

We applied MR-CCC to the OneK1K cohort \cite{yazar2022onek1k}
along the
\[
  \textbf{NK cells} \;\rightarrow\; \textbf{Monocytes}
\]
axis, a central axis of innate immune crosstalk in which NK cells
regulate monocyte function through cytokines, prostaglandins, and
other soluble mediators \cite{vivier2008innate}. Donor-level expression construction, filtering, ligand--receptor
and pathway annotations, pathway-activity construction, and
instrument selection are described in Methods
(Section~\ref{subsec:realdata_methods}); after filtering, the analysis
retained $n = 651$ donors and 41 ligand--receptor--pathway
triplets spanning 11 biological pathways.

\subsubsection{Results for NK Cells
  \texorpdfstring{\(\rightarrow\)}{->}
  Monocytes}
\label{subsec:results_nk_mono}
For each of the 41 triplets, MR-CCC was fitted using the Gibbs sampler with
20{,}000 iterations, a 2{,}000-iteration burn-in, and a thinning factor
of 10.  Posterior means of \(\beta_X\), \(\beta_{XZ}\), and
\(\beta_Z\) were recorded together with the posterior inclusion
probability
\(\mathrm{PIP} = \Pr(\gamma = 1 \mid \text{data})\).
All reported effect sizes are standardized:
\[
\hat{\beta}_X^{(s)} = \hat{\beta}_X \cdot
\frac{\mathrm{sd}(X)}{\mathrm{sd}(Y)}, \qquad
\hat{\beta}_{XZ}^{(s)} = \hat{\beta}_{XZ} \cdot
\frac{\mathrm{sd}(X)\,\mathrm{sd}(Z)}{\mathrm{sd}(Y)},
\]
so that effects are expressed in standard-deviation (SD) units of
pathway activity per SD change in ligand or receptor expression.

\paragraph{Overview.}
Figure~\ref{fig:pip_ranking} ranks all 41 triplets by PIP.  Eight
triplets exceed the PIP \(> 0.5\) discovery threshold:
\textit{IL18}--\textit{IL18RAP} (\(\mathrm{PIP} = 0.976\)),
\textit{IL18}--\textit{IL18R1} (\(\mathrm{PIP} = 0.928\)),
\textit{PTGES3}--\textit{PTGER4} (\(\mathrm{PIP} = 0.906\)),
\textit{SLC6A6}--\textit{GABBR1} (\(\mathrm{PIP} = 0.856\)),
\textit{IL1B}--\textit{IL1RAP} (\(\mathrm{PIP} = 0.716\)),
\textit{IFNG}--\textit{IFNGR2} (\(\mathrm{PIP} = 0.613\)),
\textit{PTGES3}--\textit{PTGER2} (\(\mathrm{PIP} = 0.571\)), and
\textit{IFNG}--\textit{IFNGR1} (\(\mathrm{PIP} = 0.564\)).
These span four biochemically distinct pathway classes --- interleukin,
prostaglandin, GABA, and interferon signaling --- revealing a
cytokine- and lipid-mediator-driven NK--monocyte axis operating
through soluble ligands rather than cell-contact adhesion mechanisms.
A further three triplets occupy the moderate-evidence zone
\(0.20 \leq \mathrm{PIP} < 0.50\):
\textit{AKR1C3}--\textit{PTGDR} (0.377),
\textit{CLEC2B}--\textit{KLRF1} (0.249), and
\textit{IL23A}--\textit{IL12RB1} (0.204).
The remaining 30 triplets fall below PIP \(= 0.20\), with the large
majority at or below PIP \(= 0.10\).  Both the discoveries and the
most confidently absent signals are biologically interpretable, as
discussed below.

\begin{figure}[htb]
\centering
\includegraphics[width=\textwidth]{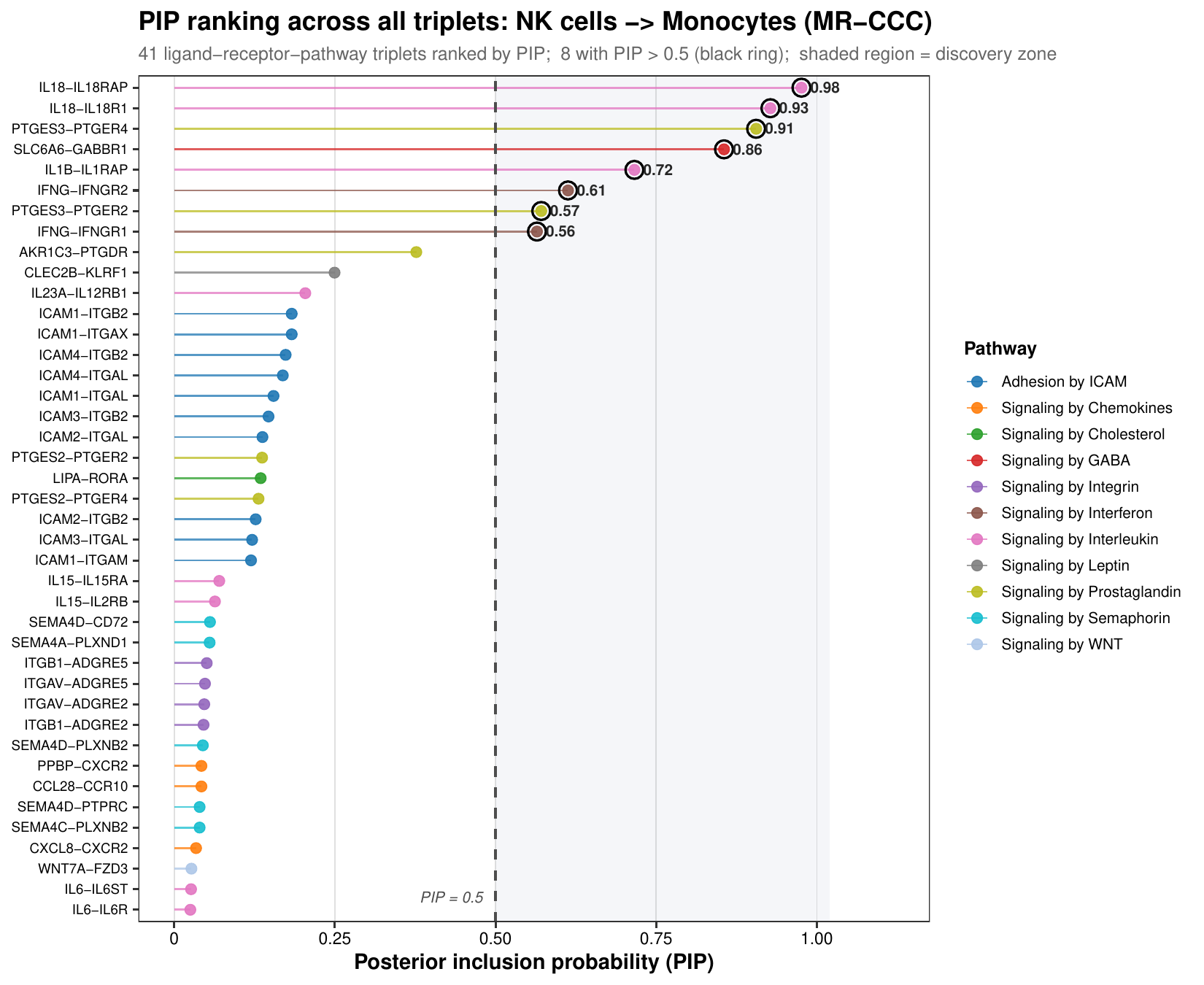}
\caption{\textbf{Posterior inclusion probabilities for the NK
  cells\,$\rightarrow$\,monocytes analysis.} Posterior inclusion probabilities for all 41
  ligand--receptor--pathway triplets in the NK cells \(\rightarrow\)
  Monocytes analysis (\(n = 651\) donors), ranked from highest to
  lowest PIP.  Points are coloured by pathway class.  Black rings
  identify the eight discoveries with PIP \(> 0.5\): the two
  \textit{IL18} receptor-chain pairs lead at PIP\,=\,0.98 and 0.93,
  followed by \textit{PTGES3}--\textit{PTGER4} (0.91),
  \textit{SLC6A6}--\textit{GABBR1} (0.86), \textit{IL1B}--\textit{IL1RAP}
  (0.72), and three additional interferon and prostaglandin pairs
  (0.56--0.61).  The dashed vertical line and shaded region mark the
  PIP \(> 0.5\) discovery zone.}
\label{fig:pip_ranking}
\end{figure}

\paragraph{Interleukin Signaling.}
The interleukin class yielded the three strongest discoveries spanning
two cytokines of the IL-1 superfamily:
\textit{IL18}--\textit{IL18RAP}
(\(\mathrm{PIP} = 0.976\);
\(\hat{\beta}_X^{(s)} = 0.289\);
\(\hat{\beta}_{XZ}^{(s)} = -3.19\)),
\textit{IL18}--\textit{IL18R1}
(\(\mathrm{PIP} = 0.928\);
\(\hat{\beta}_X^{(s)} = 0.333\);
\(\hat{\beta}_{XZ}^{(s)} = +2.13\)), and
\textit{IL1B}--\textit{IL1RAP}
(\(\mathrm{PIP} = 0.716\);
\(\hat{\beta}_X^{(s)} = 0.663\);
\(\hat{\beta}_{XZ}^{(s)} = -1.01\)).
IL-18 is a pleiotropic cytokine produced by NK cells that potently
activates monocytes and co-stimulates interferon production
\cite{okamura1995interleukin18}.  The functional IL-18 receptor
complex is assembled from two chains: the primary binding subunit
\textit{IL18R1} and the accessory signal-transducing subunit
\textit{IL18RAP} (IL-18R$\beta$), which is indispensable for downstream
signaling \cite{born1998il18rap}. MR-CCC reveals a striking
functional dissociation between these two obligate chains.
For \textit{IL18}--\textit{IL18R1}, the interaction coefficient is
\emph{positive} (\(\hat{\beta}_{XZ}^{(s)} = +2.13\)), with a
sign-reversal threshold below the population mean (\(Z^{*} \approx
-0.16\) SD), meaning that the causal effect of NK-derived IL-18
on monocyte pathway activity through \textit{IL18R1} is positive and
reinforced across virtually the entire monocyte population. Higher
\textit{IL18R1} abundance amplifies the activating signal, consistent
with a receptor-availability-dependent reinforcement mechanism. By
contrast, for \textit{IL18}--\textit{IL18RAP}, the interaction is
strongly negative (\(\hat{\beta}_{XZ}^{(s)} = -3.19\)), with a
sign-reversal threshold at \(Z^{*} \approx +0.09\) SD --- effectively
at the population mean --- so that monocytes with even slightly
above-average \textit{IL18RAP} expression exhibit an attenuated or
reversed response to NK-derived IL-18.  This opposing modulation
across the two receptor chains of the same cytokine --- amplification
through the binding chain, attenuation through the accessory chain --- is consistent with a receptor stoichiometry-dependent tuning of IL-18
signaling at the NK--monocyte interface, whereby excess
\textit{IL18RAP} may sequester ligand or promote inhibitory
complexes rather than productive heterodimers.
A third interleukin discovery, \textit{IL1B}--\textit{IL1RAP}
(PIP\,=\,0.716), shows the same receptor-saturation attenuation
pattern at $Z^{*} \approx +0.67$ SD, where IL-1$\beta$
\cite{dinarello2009interleukin} acts through the obligate IL-1
receptor accessory chain \textit{IL1RAP}.

\paragraph{GABA Signaling.}
The GABA class yielded one discovery:
\textit{SLC6A6}--\textit{GABBR1}
(\(\mathrm{PIP} = 0.856\);
\(\hat{\beta}_X^{(s)} = 0.531\);
\(\hat{\beta}_{XZ}^{(s)} = -2.52\)).
\textit{SLC6A6} encodes a sodium- and chloride-dependent GABA
transporter that concentrates and releases GABA into the extracellular
space; \textit{GABBR1} encodes the GABA\textsubscript{B1} subunit of
the metabotropic GABA-B receptor on monocytes
\cite{bjurstrom2008gaba}.  The inferred main effect is positive
(\(\hat{\beta}_X^{(s)} = 0.531\)) and strongly attenuated by a
negative interaction with receptor expression
(\(\hat{\beta}_{XZ}^{(s)} = -2.52\)), placing the sign-reversal
threshold at \(Z^{*} = -\hat{\beta}_X^{(s)} /
\hat{\beta}_{XZ}^{(s)} \approx +0.21\) SD of monocyte
\textit{GABBR1} expression --- very close to the population mean.
Consequently, NK-derived GABA exerts a mild activating influence on
monocytes with low \textit{GABBR1} expression, but an inhibitory
influence on monocytes with above-average receptor abundance,
consistent with the hyperpolarizing, inhibitory nature of GABA-B
signaling \cite{bettler2004gabab} and indicative of a
receptor-saturation-dependent regulatory switch in NK--monocyte
innate immune crosstalk \cite{tian1999gaba}.

\paragraph{Interferon Signaling.}
MR-CCC simultaneously identified both obligate receptor chains of the
functional IFN-$\gamma$ receptor heterodimer:
\textit{IFNG}--\textit{IFNGR1}
(\(\mathrm{PIP} = 0.564\);
\(\hat{\beta}_X^{(s)} = 0.542\);
\(\hat{\beta}_{XZ}^{(s)} = -0.915\)) and
\textit{IFNG}--\textit{IFNGR2}
(\(\mathrm{PIP} = 0.613\);
\(\hat{\beta}_X^{(s)} = 0.584\);
\(\hat{\beta}_{XZ}^{(s)} = -0.714\)).
IFN-$\gamma$ is a hallmark cytokine of cytotoxic lymphocytes that
classically activates monocytes and macrophages to upregulate
antigen-presentation and antimicrobial programs
\cite{schroder2004interferon}.  The co-discovery of both receptor
subunits --- which must heterodimerize to form the signaling-competent
receptor \cite{bach1997ifngr} --- provides internal biological
validation: a spurious, confounding-driven signal would be unlikely to
co-select both obligate chains simultaneously.  Both pairs exhibit a
positive main effect (\(\hat{\beta}_X^{(s)} > 0\)) that is attenuated
and eventually reversed by increasing receptor expression
(\(\hat{\beta}_{XZ}^{(s)} < 0\)), with sign-reversal thresholds at
\(Z^{*} \approx +0.59\) SD for \textit{IFNGR1} and
\(Z^{*} \approx +0.82\) SD for \textit{IFNGR2}, corresponding to
monocytes in the upper tercile of the receptor expression distribution.
This pattern is consistent with receptor-saturation kinetics and
downstream negative-feedback regulation of IFN-$\gamma$ signaling at
high receptor occupancy \cite{bach1997ifngr}.

\paragraph{Prostaglandin Signaling.}
Two prostaglandin E\textsubscript{2} receptor pairs were discovered:
\textit{PTGES3}--\textit{PTGER4}
(\(\mathrm{PIP} = 0.906\);
\(\hat{\beta}_X^{(s)} = 1.13\);
\(\hat{\beta}_{XZ}^{(s)} = -1.61\)) and
\textit{PTGES3}--\textit{PTGER2}
(\(\mathrm{PIP} = 0.571\);
\(\hat{\beta}_X^{(s)} = 0.729\);
\(\hat{\beta}_{XZ}^{(s)} = -0.130\)).
\textit{PTGES3} encodes cytosolic prostaglandin E synthase~3, which
catalyzes the terminal step in prostaglandin E\textsubscript{2}
(PGE\textsubscript{2}) biosynthesis \cite{tanioka2000ptges};
\textit{PTGER2} and \textit{PTGER4} encode the EP2 and EP4
G-protein-coupled receptors for PGE\textsubscript{2} on monocytes
\cite{hata2004pgrecept,kalinski2012pge2}.
These two pairs exhibit structurally distinct receptor-modulation
profiles.  For \textit{PTGES3}--\textit{PTGER4}, the effect is
substantially modulated by receptor expression
(\(\hat{\beta}_{XZ}^{(s)} = -1.61\)) with sign-reversal at \(Z^{*} \approx +0.70\) SD, indicating that PGE\textsubscript{2}-mediated
NK--monocyte signaling through EP4 is activating at low-to-average
receptor levels but suppressive at above-average EP4 abundance.  For
\textit{PTGES3}--\textit{PTGER2}, the interaction term is near zero
(\(\hat{\beta}_{XZ}^{(s)} = -0.130\)), placing the theoretical sign-reversal at \(Z^{*} \approx +5.6\) SD --- far beyond the
empirical distribution of monocyte \textit{PTGER2} expression --- so
that this pair behaves as a near-pure main effect across the entire
population.  The co-discovery of two EP receptor subtypes sharing a
common ligand but with qualitatively different interaction profiles
underscores the receptor-subtype specificity captured by MR-CCC's
interaction model.

\begin{figure}[htb]
\centering
\includegraphics[width=\textwidth]{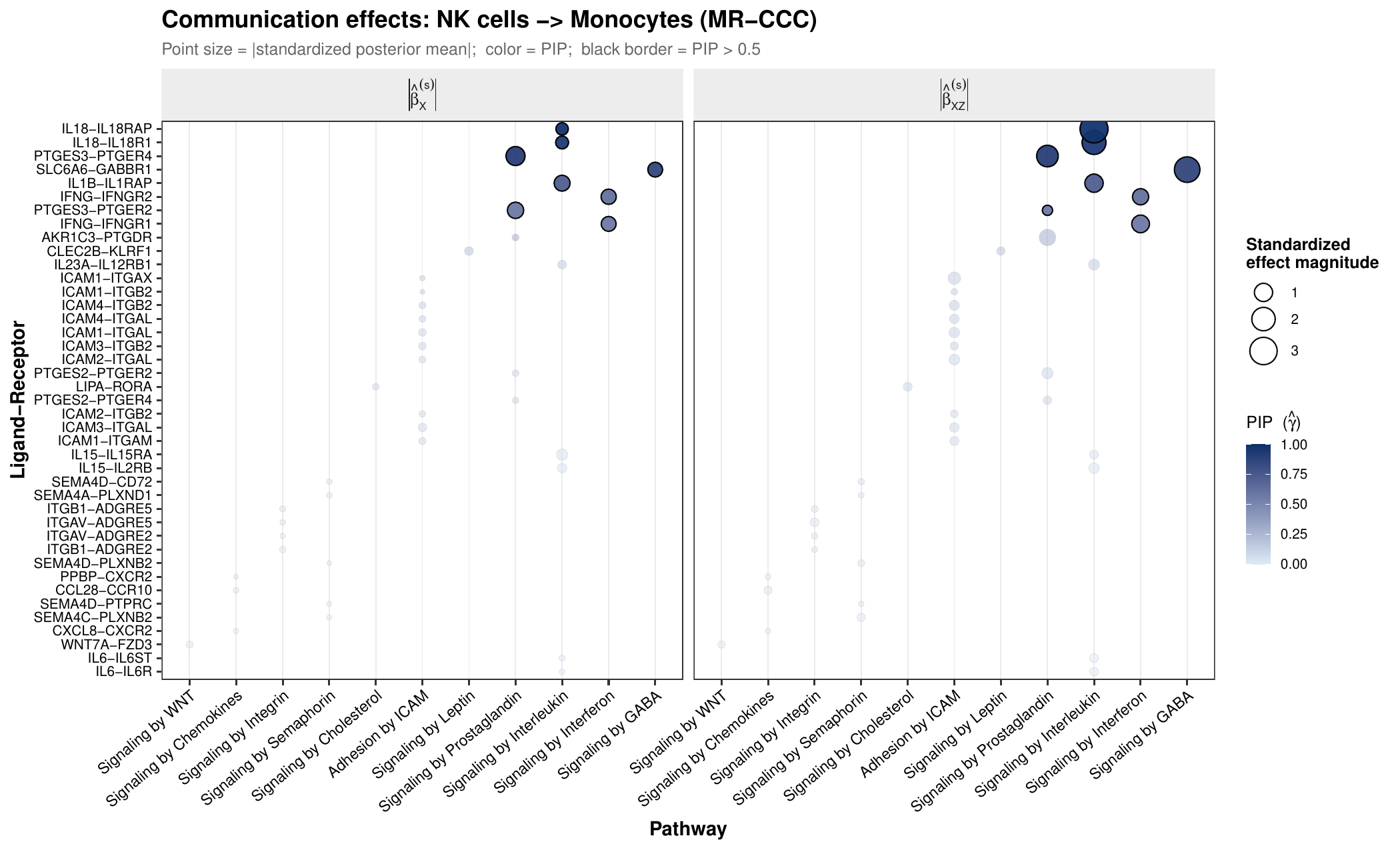}
\caption{\textbf{Pathway-specific posterior mean effects for the NK cells
  \(\rightarrow\) Monocytes analysis.}  The left panel shows
  \(|\hat{\beta}_X^{(s)}|\), the absolute standardized ligand main
  effect, and the right panel shows \(|\hat{\beta}_{XZ}^{(s)}|\), the
  absolute standardized receptor-modulation effect.  Point size encodes
  effect magnitude; color encodes PIP on a 0--1 scale, with black
borders marking the eight discoveries (PIP \(> 0.5\)).  Discoveries
concentrate in the GABA, Interferon, Interleukin, and Prostaglandin pathway classes.  The two \textit{IL18} receptor-chain pairs
(\textit{IL18RAP}, \textit{IL18R1}) are visible as the largest, darkest
points in the Interleukin column of the right panel, reflecting
their opposing and large interaction magnitudes
(\(|\hat{\beta}_{XZ}^{(s)}| = 3.19\) and \(2.13\), respectively). ICAM adhesion pairs show uniformly small, pale points in both panels.}
\label{fig:bubble}
\end{figure}

\paragraph{Moderate-evidence signals.}
Three triplets in the \(0.20\)--\(0.50\) PIP range warrant attention
as candidates for replication in larger cohorts.
\textit{AKR1C3}--\textit{PTGDR} (PIP\,=\,0.377) pairs the
aldo-keto reductase AKR1C3 \cite{penning2006akr1c3}, a
prostaglandin-biosynthetic enzyme, with the PGD\textsubscript{2}
receptor PTGDR \cite{hata2004pgrecept} on monocytes, placing it
squarely within the prostaglandin signaling theme of this axis.
\textit{CLEC2B}--\textit{KLRF1} (PIP\,=\,0.249) nominally pairs an
NK activating ligand with its cognate receptor; however, as
\textit{KLRF1} (NKp80) is predominantly expressed on NK cells rather
than monocytes \cite{welte2006nkp80}, this sub-threshold signal most
plausibly reflects an incidental instrument correlation rather than a
genuine NK\(\rightarrow\)Monocyte communication event.
\textit{IL23A}--\textit{IL12RB1} (PIP\,=\,0.204) shows suggestive
evidence for IL-23/IL-12 receptor cross-signaling, consistent with the
structural homology between these cytokine families and their shared
receptor subunit \cite{oppmann2000il23}.

\paragraph{Absent signals and biological specificity.}
A key strength of MR-CCC is its ability to assign near-zero posterior
probability to biologically implausible signals, providing an
intrinsic specificity calibration that complements the discovery set.
All ten ICAM adhesion pairs annotated in CellPhoneDB --- involving
\textit{ICAM1}, \textit{ICAM2}, \textit{ICAM3}, and \textit{ICAM4}
paired with \textit{ITGAL}, \textit{ITGB2}, \textit{ITGAM}, and
\textit{ITGAX} --- returned PIPs\,$\leq 0.18$, consistent with the
expectation that ICAM-mediated adhesion requires direct cell--cell
contact \cite{springer1990adhesion} and does not operate as a soluble
signaling mechanism propagated through NK-cell ligand expression.
\textit{IL15}--\textit{IL15RA}
(PIP\,=\,0.07) and \textit{IL15}--\textit{IL2RB}
(PIP\,=\,0.06) are near zero: IL-15 is principally produced by
monocytes and dendritic cells to sustain NK cell homeostasis
\cite{fehniger2001il15}, and its absence in the
NK\(\rightarrow\)Monocyte direction correctly reflects this
established biological directionality.
\textit{IL6}--\textit{IL6ST} (PIP\,=\,0.026) and
\textit{WNT7A}--\textit{FZD3} (PIP\,=\,0.027) likewise show
negligible posterior support, appropriate for pathways not associated
with NK cell efferent signaling \cite{clevers2006wnt}.
Directional specificity is further illustrated by cross-axis
comparison with the Monocytes \(\rightarrow\) CD4\(^{+}\) T cell
axis (Supplementary Section~S2.5.2):
\textit{SEMA4D}--\textit{CD72} \cite{kumanogoh2000cd100}, which
achieves PIP\,=\,0.674 in the Monocytes\(\rightarrow\)CD4\(^{+}\)
direction, returns PIP\,=\,0.056 in the NK\(\rightarrow\)Monocytes
direction.  This pair-specific, direction-specific posterior contrast
demonstrates that MR-CCC recovers causal signals tied to
sender--receiver identity rather than simply flagging co-expressed
ligand--receptor pairs irrespective of directionality.

\paragraph{Receptor-modulated effect curves.}
Figure~\ref{fig:curves} plots
\(\hat{\beta}_X^{(s)} + \hat{\beta}_{XZ}^{(s)} \cdot
(Z/\mathrm{sd}(Z))\) across the observed donor range for the eight
discovery triplets, grouped by pathway class. The four panels show
qualitatively distinct modulation profiles: the diverging IL-18
receptor chains (amplification through \textit{IL18R1}, reversal
through \textit{IL18RAP}), the EP4-versus-EP2 prostaglandin contrast,
the GABA-B activation-to-inhibition crossover near the population
mean, and the parallel decline of both \textit{IFNG} receptor chains
in the upper receptor tercile.

\begin{figure}[htb]
\centering
\includegraphics[width=\textwidth]{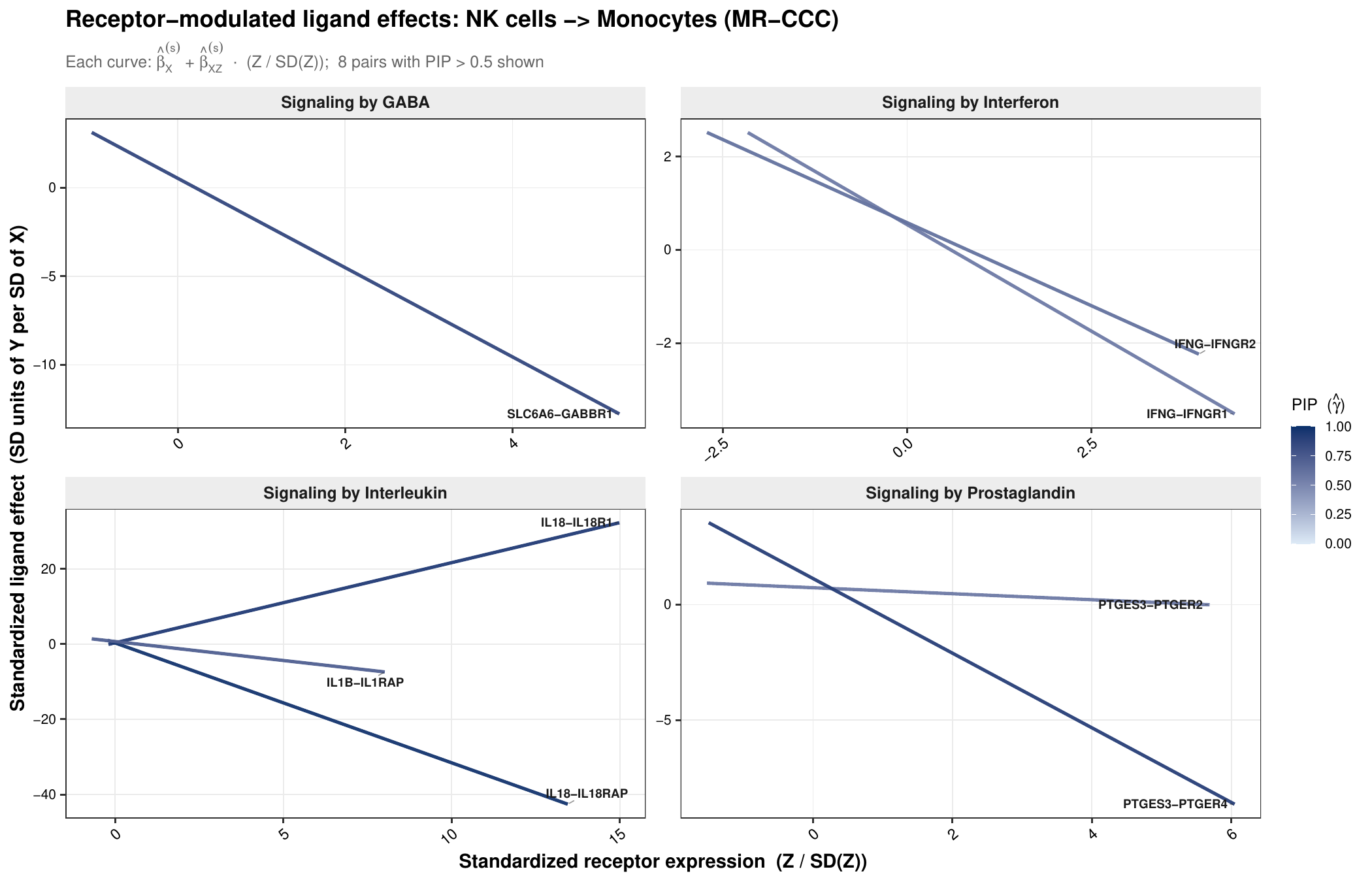}
\caption{\textbf{Receptor-modulated ligand effect curves for the
  eight NK cells\,$\rightarrow$\,monocytes discoveries.} Receptor-modulated ligand effects for the NK cells
  \(\rightarrow\) Monocytes analysis (\(n = 651\) donors).  Each
  curve represents
  \(\hat{\beta}_X^{(s)} + \hat{\beta}_{XZ}^{(s)} \cdot
    (Z / \mathrm{sd}(Z))\),
  plotted over the observed donor range of standardized receptor
  expression.  Only the eight discovery pairs (PIP \(> 0.5\)) are
  displayed, grouped into the four pathway panels with confirmed
  discoveries (GABA, Interferon, Interleukin, and Prostaglandin).
  The Interleukin panel shows the most striking pattern: the
  \textit{IL18}--\textit{IL18R1} curve
  (upward slope; \(\hat{\beta}_{XZ}^{(s)} = +2.13\)) and the
  \textit{IL18}--\textit{IL18RAP} curve (steep downward slope;
  \(\hat{\beta}_{XZ}^{(s)} = -3.19\)) diverge in opposite directions,
  revealing opposing receptor-chain modulation for the same cytokine.
  In the Prostaglandin panel, \textit{PTGES3}--\textit{PTGER4}
  (strong negative slope) and \textit{PTGES3}--\textit{PTGER2}
  (near-horizontal) illustrate receptor-subtype-specific modulation.
  In the GABA panel, \textit{SLC6A6}--\textit{GABBR1} crosses zero
  near the population mean.  In the Interferon panel, both
  \textit{IFNG} receptor-chain curves (\textit{IFNGR1} and
  \textit{IFNGR2}) decline from positive intercepts with sign-reversals
  in the upper tercile of the receptor distribution.}
\label{fig:curves}
\end{figure}

\paragraph{Summary.}
The eight NK cells\,$\rightarrow$\,monocytes discoveries span four
pathway classes (GABA, interferon, interleukin, prostaglandin) with
biologically coherent receptor modulation, while implausible signals
(ICAM adhesion, reversed IL-15) correctly receive near-zero PIPs.
Results for the remaining 19 ordered cell-type pairs are in
Supplementary Materials (Section~S2). Two cross-direction patterns merit note:
the IL-18 opposing-chain pattern is partially mirrored in the
CD8$^{+}$ T cells\,$\rightarrow$\,monocytes direction
(Supplementary Section~S2.3.4), suggesting a receptor stoichiometry-dependent
tuning mechanism general to cytotoxic lymphocyte senders; and the
NK\,$\rightarrow$\,monocytes axis (8 discoveries) is the most
evidence-rich directional axis in the 20-pair panel, far exceeding
the reverse direction, consistent with NK cells as primary cytokine
effectors in innate immune regulation.

\section{Discussion}
\label{sec:discussion}

Inferring biologically meaningful CCC from
transcriptomic data remains challenging because ligand expression,
receptor expression, and downstream pathway activity may all
be influenced by shared unmeasured
confounding.  Existing methods primarily rely on association patterns
and do not distinguish causal signaling from correlated expression
driven by confounders.  We introduced MR-CCC, a Bayesian Mendelian
randomization framework that integrates instrumental variables,
receptor-modulated effect modeling, and probabilistic variable
selection to enable causal inference of intercellular communication.
MR-CCC shifts CCC analysis from an associational
paradigm toward a causal, mechanistically grounded framework in which
the fundamental question is not whether ligand and receptor expression
are jointly elevated, but whether variation in ligand expression
mechanically drives downstream pathway activity through receptor
engagement.

MR-CCC advances current approaches in three key respects.  First, it
leverages cis-eQTL instruments for both ligand and receptor
expression, mitigating bias from unmeasured confounding that would
otherwise propagate into effect estimates and communication scores.
Second, it explicitly models receptor-dependent modulation through an
interaction term, allowing the strength and direction of ligand
effects to vary with receptor expression in the receiving cell; this
structure is identifiable under standard MR assumptions
(Proposition~\ref{prop:ident}).  Third, the spike--and--slab prior
yields posterior inclusion probabilities that serve as an
interpretable, calibrated communication score.

The simulation results demonstrate that these three properties are
jointly necessary for reliable inference under confounding.  In the
null scenario, MR-CCC suppresses spurious signals, maintaining
near-zero PIPs, negligible rejection rates, and accurate parameter
estimates across all sample sizes.  OLS, which is an association-based method, fails entirely: it assigns
the maximum communication score and rejects in every replicate
regardless of sample size, confirming that the associational paradigm
is fundamentally compromised under even moderate unmeasured
confounding.  MVMR corrects confounding bias in estimation but does
not fully resolve the false-discovery problem, as null communication
scores remain elevated due to residual correlation between the
instrumented ligand exposure and the outcome.  MR-BMA achieves low
null scores via Bayesian model selection but incurs a persistent
structural bias in the ligand effect estimate, attributable to the
missing interaction term, and provides no estimate of the
receptor-modulated effect.  MR-CCC is the only evaluated approach
that simultaneously achieves strong null calibration, MR-based
confounding correction, and joint regularized estimation of the full
communication effect vector.  When communication is present, MR-CCC
retains high detection power with stable estimation of both the main
and interaction effects across all sample sizes and signal scenarios.

Application to OneK1K illustrates the biological specificity of
MR-CCC. Along the most evidence-rich axis (NK cells\,$\rightarrow$\,monocytes), MR-CCC's explicit interaction model
uncovered findings inaccessible to additive frameworks: opposing
receptor-modulation signs across the two obligate IL-18 receptor
subunits \cite{okamura1995interleukin18,born1998il18rap}, indicating
a receptor stoichiometry-dependent tuning of IL-18 signaling, and
co-discovery of both obligate IFN-$\gamma$ receptor chains
\cite{bach1997ifngr}, which must heterodimerize to signal --- internal
validation that confounding-driven signals would not produce. The
framework also assigned near-zero posterior probability to
biologically implausible signals: ICAM adhesion pairs
\cite{springer1990adhesion} (no soluble mechanism), the reversed
IL-15 direction \cite{fehniger2001il15} (monocyte-to-NK biology),
and \textit{SEMA4D}--\textit{CD72} \cite{kumanogoh2000cd100}, a
high-confidence discovery in the Monocytes\,$\rightarrow$\,CD4$^{+}$
T direction returning PIP\,=\,0.056 here, demonstrating
sender--receiver-specific inference.

More broadly, modeling pathway activity as the outcome links
ligand--receptor interactions to downstream functional programs in
the receiving cell.  Rather than reporting a single interaction
score per pair, MR-CCC produces pathway-resolved communication
signals that are statistically supported and directly interpretable
in terms of the biological programs they engage, supporting the
construction of causally annotated communication networks that can
guide experimental prioritization.

Several extensions may further enhance the framework.  Future work
could develop cell-state-resolved formulations that preserve
within-donor variation, enabling detection of communication signals
that differ across discrete cell states or continuous transcriptional
trajectories rather than capturing only cell population-average effects.
A second direction is to incorporate hierarchical structure across
pathways or cell-type pairs, sharing information across related
triplets while preserving context-specific effects.  A third is to
jointly model multiple ligands or receptors to study coordinated
or higher-order signaling patterns.  Finally, integrating alternative
outcome representations for nonlinear or non-Gaussian pathway
behavior may improve robustness in complex settings, and the
demonstrated specificity of MR-CCC suggests its use as a principled
screening filter to complement, rather than replace, existing
co-expression pipelines.

In summary, MR-CCC provides a principled, instrument-based framework
for causal inference of receptor-modulated CCC.
By integrating MR with interaction modeling and
Bayesian variable selection, it establishes a causal, pathway-aware
approach to intercellular signaling that moves beyond correlational
analyses toward statistically rigorous and mechanistically
interpretable causal CCC.

\section{Methods}\label{sec:methods}

Consider an ordered pair of cell types, with one acting as the putative
sender and the other as the receiver. The goal is to determine whether
variation in ligand expression in the sender cell type causes the change of
pathway activity in the receiver cell type, while allowing this effect to
be moderated by receptor expression. For donor \(i=1,\dots,n\), let
\(X_i\) denote ligand expression in the sender cell, \(Z_i\)
receptor expression in the receiver cell, \(Y_i\) pathway activity
in the receiver cell, \(G_i\) cis-eQTL instruments for \(X_i\),
\(H_i\) cis-eQTL instruments for \(Z_i\), \(V_i\) observed donor-level
covariates, and \(U_i\) any unmeasured confounders. Because germline genetic variants are fixed at conception, they are largely independent of
environmental confounding, and provide a basis for causal identification
via MR. 

\subsection{Structural Model and Identification}
\label{subsec:structural}

Without loss of generality, we assume all observed variables have been centered (zero mean). We posit the following
data-generating structural model:
\begin{equation}
\label{eq:dgm}
\begin{aligned}
  X_i &= \pi_X^\top G_i + \alpha_X^\top V_i + \lambda_X U_i +
          \varepsilon_{Xi},\\
  Z_i &= \pi_Z^\top H_i + \alpha_Z^\top V_i + \lambda_Z U_i +
          \varepsilon_{Zi},\\
  Y_i &= \beta_X X_i + \beta_Z Z_i + \beta_{XZ} X_i Z_i +
          \alpha_Y^\top V_i + \lambda_Y U_i + \varepsilon_{Yi},
\end{aligned}
\end{equation}
with \(E[\varepsilon_{Xi}\varepsilon_{Zi}\mid U_i,G_i,H_i,V_i]=0\),
where \(\pi_X,\pi_Z\) are instrument effects; \(\alpha_X,\alpha_Z,
\alpha_Y\) are covariate effects; \(\lambda_X,\lambda_Z,\lambda_Y\)
quantify the influence of the unmeasured confounder; \(\beta_X,\beta_Z\)
are the ligand and receptor main effects; \(\beta_{XZ}\) is the
ligand--receptor interaction; and \(\varepsilon_{Xi},\varepsilon_{Zi},
\varepsilon_{Yi}\) are independent mean-zero errors. The
receptor-modulated effect of ligand on pathway activity is
\(\beta_X + \beta_{XZ}Z_i\), which is the main inferential target. The causal diagram corresponding to \eqref{eq:dgm} is depicted in Figure~\ref{fig:mrccc_dag}.

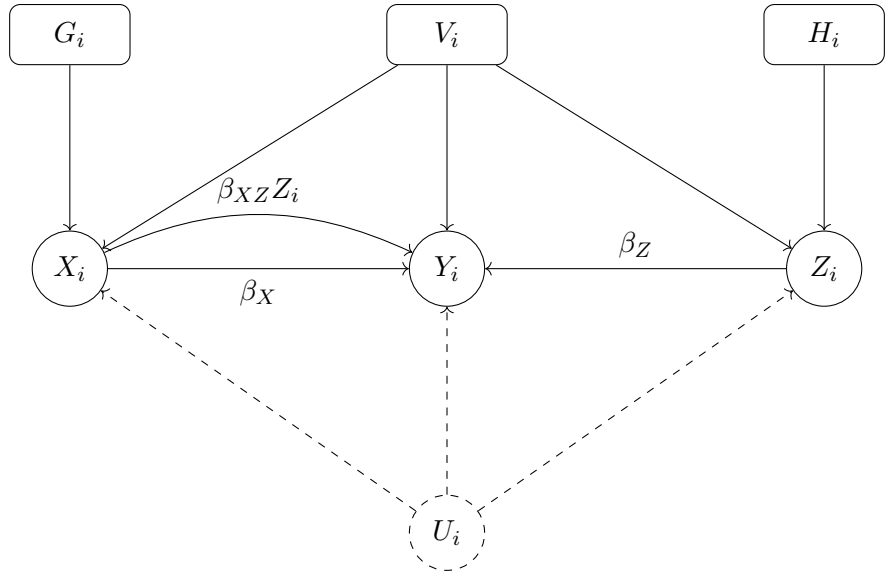
\begin{figure}[htb]
\centering
\begin{tikzpicture}[
    node distance=2.8cm,
    every node/.style={font=\small},
    box/.style={draw, rectangle, rounded corners, minimum width=1.6cm,
                minimum height=0.8cm},
    latent/.style={draw, circle, dashed, minimum size=1cm},
    obs/.style={draw, circle, minimum size=1cm}
]
\node[obs] (X) {$X_i$};
\node[obs, right=4cm of X] (Y) {$Y_i$};
\node[obs, right=4cm of Y] (Z) {$Z_i$};
\node[box, above=2.2cm of X] (G) {$G_i$};
\node[box, above=2.2cm of Y] (V) {$V_i$};
\node[box, above=2.2cm of Z] (H) {$H_i$};
\node[latent, below=2.5cm of Y] (U) {$U_i$};
\draw[->] (G) -- (X);
\draw[->] (H) -- (Z);
\draw[->] (V) -- (X);
\draw[->] (V) -- (Y);
\draw[->] (V) -- (Z);
\draw[->] (X) -- node[midway, below] {$\beta_X$} (Y);
\draw[->] (Z) -- node[midway, above] {$\beta_Z$} (Y);
\draw[->, bend left=25] (X) to node[midway, above] {$\beta_{XZ}Z_i$} (Y);
\draw[->, dashed] (U) -- (X);
\draw[->, dashed] (U) -- (Y);
\draw[->, dashed] (U) -- (Z);
\end{tikzpicture}
\caption{\textbf{Causal diagram of the MR-CCC structural model.}
  The causal effect of ligand expression on pathway activity is
  \(\beta_X + \beta_{XZ} Z_i\), allowing receptor-modulated communication.
  Genetic instruments \(G_i\) and \(H_i\) identify ligand and receptor
  expression respectively; observed covariates \(V_i\) and unmeasured
  confounders \(U_i\) may affect all three observed biological
  variables ($X_i$, $Z_i$, $Y_i$). Solid arrows indicate causal effects;
  dashed arrows indicate the unobserved confounder pathways that MR
  blocks.}
\label{fig:mrccc_dag}
\end{figure}

Because \(U_i\) is unobserved, regressing \(Y_i\) directly on \((X_i,
Z_i)\) yields biased estimates. Define the instrument-based conditional
means
\begin{align*}
X_i^\ast := E[X_i \mid G_i, V_i] = \pi_X^\top G_i + \alpha_X^\top V_i,
\qquad
Z_i^\ast := E[Z_i \mid H_i, V_i] = \pi_Z^\top H_i + \alpha_Z^\top V_i.
\end{align*}
Replacing \((X_i, Z_i)\) with \((X_i^\ast, Z_i^\ast)\) in the outcome
equation gives the working model
\begin{align}\label{eq:workm}
Y_i = \mu + \beta_X X_i^\ast + \beta_Z Z_i^\ast +
      \beta_{XZ} X_i^\ast Z_i^\ast + \alpha_Y^\top V_i + \varepsilon_{Yi}.
\end{align}
The intercept \(\mu\) is necessary even under centering because
\(E[X_i^\ast Z_i^\ast] = \mathrm{Cov}(X_i^\ast, Z_i^\ast) \neq 0\) in
general. The following proposition shows the plug-in regression
\eqref{eq:workm} identifies the structural coefficients.

\begin{proposition}[Identification]\label{prop:ident}
Consider the data-generating model in \eqref{eq:dgm}. Assume relevance (\(\pi_X \neq 0\),
\(\pi_Z \neq 0\)), exclusion restriction (\(G_i, H_i\) do not directly
affect \(Y_i\)), and independence (\(U_i \perp G_i, H_i, V_i\)). Then
\begin{align*}
E[Y_i \mid X_i^\ast, Z_i^\ast, V_i]
= \mu + \beta_X X_i^\ast + \beta_Z Z_i^\ast + \beta_{XZ} X_i^\ast Z_i^\ast
  + \alpha_Y^\top V_i ,
\end{align*}
where \(\mu := \lambda_Y E[U_i] + \beta_{XZ}\lambda_X\lambda_Z E[U_i^2]\)
is a constant. Consequently, \eqref{eq:workm} identifies
\((\beta_X, \beta_Z, \beta_{XZ})\).
\end{proposition}

\begin{proof}
Write \(X_i = X_i^\ast + \eta_{Xi}\) and \(Z_i = Z_i^\ast + \eta_{Zi}\)
where \(\eta_{Xi} = \lambda_X U_i + \varepsilon_{Xi}\) and
\(\eta_{Zi} = \lambda_Z U_i + \varepsilon_{Zi}\). Substituting into the
outcome equation and expanding \(X_i Z_i = X_i^\ast Z_i^\ast +
X_i^\ast \eta_{Zi} + Z_i^\ast \eta_{Xi} + \eta_{Xi}\eta_{Zi}\), we
take conditional expectation given \((X_i^\ast, Z_i^\ast, V_i)\). Since
\(X_i^\ast\) and \(Z_i^\ast\) are measurable functions of
\((G_i, H_i, V_i)\), iterated expectation gives
\(E[\eta_{Xi} \mid X_i^\ast, Z_i^\ast, V_i] = 0\) and likewise for
\(\eta_{Zi}\), eliminating all cross-terms. The product
\(\eta_{Xi}\eta_{Zi}\) contributes
\[
E[\eta_{Xi}\eta_{Zi} \mid X_i^\ast, Z_i^\ast, V_i]
= \lambda_X\lambda_Z E[U_i^2],
\]
a constant by independence of \(U_i\) from \((G_i, H_i, V_i)\) and
mutual orthogonality of the idiosyncratic errors. The confounding term
contributes \(\lambda_Y E[U_i]\), also a constant. Collecting these
constants into \(\mu\) and using \(E[\varepsilon_{Yi} \mid X_i^\ast,
Z_i^\ast, V_i] = 0\) completes the proof.
\end{proof}

\begin{remark}
Under the centering assumption \(E[U_i] = 0\), the intercept reduces to
\(\mu = \beta_{XZ}\lambda_X\lambda_Z E[U_i^2]\), which vanishes when
\(\beta_{XZ} = 0\) (no receptor modulation) or when confounding does not
propagate through both exposures (\(\lambda_X\lambda_Z = 0\)). In all
other cases \(\mu \neq 0\), confirming the necessity of the intercept in
\eqref{eq:workm}.
\end{remark}


\subsection{Bayesian Formulation of MR-CCC}\label{subsec:mrccc}

We embed the working model \eqref{eq:workm} in a Bayesian
framework by placing a spike--and--slab prior on the communication effect
vector \(\beta = (\beta_X, \beta_{XZ})^\top\), which enables joint variable
selection over the presence of communication. 
Let \(X_\beta = [X^\ast \;\;
X^\ast\!\circ Z^\ast] \in \mathbb{R}^{n\times 2}\) denote the design
matrix for the ligands and the interaction of ligands and receptors, so that
\(X_\beta\beta = \beta_X X^\ast + \beta_{XZ}\,X^\ast \circ Z^\ast\) where $\circ$ is the element-wise product.


\subsubsection{Prior Specification}

\paragraph{Exposure model parameters.}
Zellner's \(g\)-priors \cite{zellner1986} and inverse gamma priors are placed on the exposure model parameters:
\begin{align*}
\pi_X \mid \sigma_X^2 &\sim \mathcal{N}(0,\,\sigma_X^2 g_G(G^\top G)^{-1}),
&\alpha_X \mid \sigma_X^2 &\sim \mathcal{N}(0,\,\sigma_X^2 g_V(V^\top V)^{-1}),
&\sigma_X^2 &\sim \mathrm{IG}(a_\sigma, b_\sigma),\\
\pi_Z \mid \sigma_Z^2 &\sim \mathcal{N}(0,\,\sigma_Z^2 g_H(H^\top H)^{-1}),
&\alpha_Z \mid \sigma_Z^2 &\sim \mathcal{N}(0,\,\sigma_Z^2 g_V(V^\top V)^{-1}),
&\sigma_Z^2 &\sim \mathrm{IG}(a_\sigma, b_\sigma).
\end{align*}

\paragraph{Outcome model parameters.} To enable the selection of the communication effect $\beta$, we impose a spike--and--slab prior: 
\begin{align*}
\beta \mid \gamma, \sigma_Y^2 &\sim
  \mathcal{N}_2\!\left(0,\; s_\gamma\,\sigma_Y^2\, g_\beta
  (X_\beta^\top X_\beta)^{-1}\right),
\qquad
s_\gamma =
\begin{cases}
1, & \gamma = 1,\\
\nu_1, & \gamma = 0,
\end{cases}\\
\gamma &\sim \mathrm{Bernoulli}(\rho),\qquad \rho \sim \mathrm{Beta}(a_\rho, b_\rho).
\end{align*}
Here \(\nu_1 \ll 1\) concentrates the slab near zero under
\(\gamma = 0\). In all analyses we set
\(g_G = g_H = g_V = g_Z = g_\beta = \min(n, 100)\),
\(a_{\sigma} = 3\), \(b_{\sigma} = 2\),
\(a_{\rho} = 3\), \(b_{\rho} = 1\), and \(\nu_1 = 10^{-4}\). For the nuisance parameters, we assume the Zellner's \(g\) and inverse gamma priors as before: 
\begin{align*}
\alpha_Y \mid \sigma_Y^2 &\sim \mathcal{N}(0,\,\sigma_Y^2 g_V(V^\top V)^{-1}),
&\beta_Z \mid \sigma_Y^2 &\sim \mathcal{N}(0,\,\sigma_Y^2 g_Z(Z^{\ast\top}Z^\ast)^{-1}),
&\sigma_Y^2 &\sim \mathrm{IG}(a_\sigma, b_\sigma),
\end{align*}
with a diffuse normal prior on the intercept,
\(\mu \mid \sigma_Y^2 \sim \mathcal{N}(0,\, n\,\sigma_Y^2)\).


\subsubsection{Posterior Inference}
Posterior inference proceeds via a Gibbs sampler with closed-form full
conditionals for every parameter. At each iteration, the sampler updates
the first-stage exposure parameters $(\pi_X,\alpha_X,\sigma_X^2)$ and
$(\pi_Z,\alpha_Z,\sigma_Z^2)$; then recomputes the plug-in conditional
means $X^\ast=G\pi_X+V\alpha_X$ and $Z^\ast=H\pi_Z+V\alpha_Z$ and the
communication design $X_\beta=[X^\ast\;|\;X^\ast\!\circ Z^\ast]$; then
updates the outcome-model nuisance parameters $(\mu,\alpha_Y,\beta_Z,
\sigma_Y^2)$; and finally updates $\beta\mid\gamma$, the inclusion
indicator $\gamma$, and the inclusion probability $\rho$. All matrix
inverses are regularized by a small ridge term $\lambda>0$ for numerical
stability. The full conditional distributions and the complete sampling
scheme are given in Supplementary Materials (Section~S1); the sampler is implemented
in \texttt{Rcpp} for computational efficiency. Posterior summaries are
the posterior means of $\beta_X,\beta_{XZ}$ and the posterior inclusion
probability $\mathrm{PIP}=\Pr(\gamma=1\mid\text{data})$.

\subsection{Simulation Design}
\label{subsec:dgm}

We generated data
from the data-generating model in \eqref{eq:dgm} with four sample sizes $n \in \{500,\;1000,\;10{,}000,\;30{,}000\}$.
We fixed \(p_G = p_H = 5\) and \(p_V = 3\). Instruments, covariates, and
the confounder were drawn independently from standard normal
$G_i,\, H_i,\, V_i,\, U_i \;\overset{\mathrm{iid}}{\sim}\;
\mathcal{N}(0, I)$. Instrument effects were set to moderate strength
\(\pi_X = \pi_Z = (0.5, \ldots, 0.5)^\top \in \mathbb{R}^5\). Covariate effects were
\(\alpha_X = \alpha_Z = \alpha_Y = (0.3, \ldots, 0.3)^\top\). The
receptor main effect was fixed at \(\beta_Z = 0.5\), confounding
effects at \(\lambda_X = \lambda_Z = \lambda_Y = 0.7\), and residual
variances at \(\sigma_X^2 = \sigma_Z^2 = \sigma_Y^2 = 1\). For each
scenario and sample size, 20 independent replicates were generated.

\subsection{Simulation: Competing Methods}
\label{subsec:methods}

We compared MR-CCC against three alternative approaches.

\paragraph{(1) OLS (naive regression).}
We fitted
\[
Y_i = \mu + \beta_X X_i + \beta_Z Z_i + \beta_{XZ} X_i Z_i +
       V_i^\top \alpha_Y + \varepsilon_i
\]
by ordinary least squares, ignoring the instrumental variables and
therefore not correcting for unmeasured confounding. The communication
score is \(1 - p_F\), where \(p_F\) is the \(p\)-value from a joint \(F\)-test of \(H_0 \colon \beta_X = \beta_{XZ} = 0\):
\[
F = \tfrac{1}{2}\,(C\hat\beta)^\top
    \bigl[C\,\widehat{\operatorname{Var}}(\hat\beta)\,C^\top\bigr]^{-1}
    (C\hat\beta),
\qquad
p_F = P\bigl(F(2,\, n - p) > F_{\mathrm{obs}}\bigr),
\]
where \(C\) is the \(2 \times p\) contrast matrix selecting
\(\beta_X\) and \(\beta_{XZ}\) from the OLS coefficient vector.
Communication is declared when \(p_F \leq 0.05\).

\paragraph{(2) MVMR (multivariable Mendelian randomization).}
We applied MVMR \cite{Sanderson2021}. First,
ligand and receptor expressions were each regressed on the full
instrument set \([G, H]\) and covariates \(V\) by OLS to obtain fitted
values \(\hat X\) and \(\hat Z\). Then, the outcome was
regressed on \(\hat X\), \(\hat Z\), and covariates by OLS.
The communication score is
\(1 - p_t\), where \(p_t\) is the \(p\)-value of the \(t\)-test
for \(H_0 \colon \beta_X = 0\); communication is declared when
\(p_t \leq 0.05\). MVMR corrects for unmeasured confounding through MR
but relies on a fixed frequentist significance threshold, does not model
the ligand--receptor interaction, and does not impose a structured prior
on the communication effect.

\paragraph{(3) MR-BMA (Bayesian model averaging MR).}
We applied the Bayesian model averaging approach of \cite{Zuber2020}.
For each instrument \(g_j \in [G, H]\), one-sample summary statistics
\((\hat\beta_{Xj}, \hat\beta_{Zj}, \hat\beta_{Yj})\) and their standard
errors were obtained by simple OLS. With ligand \(X\) and receptor \(Z\)
as the two candidate exposures, MR-BMA enumerates all \(2^2 = 4\)
exposure inclusion models, assigns a \(g\)-prior Bayes factor to each
model, normalizes to posterior model probabilities, and extracts the
marginal inclusion probability (MIP) and model-averaged causal effect
(MACE) for each exposure. The communication score is
\(\mathrm{MIP}_X\), the posterior probability that the ligand has a
nonzero causal effect on the pathway; communication is declared when
\(\mathrm{MIP}_X > 0.5\). Crucially, MR-BMA operates on per-instrument
summary statistics and models only additive effects of the two exposures;
it does not model the ligand--receptor interaction \(\beta_{XZ}\), and
estimates of this parameter are unavailable from this method.

\paragraph{(4) MR-CCC.}
We applied the Gibbs sampler described in Section~\ref{subsec:mrccc}.
The posterior inclusion probability
\[
\mathrm{PIP} = \Pr(\gamma = 1 \mid \text{data})
\]
served as the natural communication score; communication is declared when
\(\mathrm{PIP} > 0.5\). Posterior means of \(\beta_X\) and
\(\beta_{XZ}\) were used as point estimates. The sampler was run for
20{,}000 iterations with a 2{,}000-iteration burn-in and a thinning
factor of 5. MR-CCC is the only method that simultaneously (i) corrects
for unmeasured confounding via MR, (ii) performs Bayesian
model selection over the presence of communication, and (iii) provides
regularized joint estimates of both \(\beta_X\) and \(\beta_{XZ}\).

\subsection{Simulation: Performance Metrics}
\label{subsec:metrics}

\paragraph{Communication score and rejection rule.}
Each method produces a scalar communication score on \([0, 1]\). For
MR-CCC the score is \(\mathrm{PIP} = \Pr(\gamma = 1 \mid
\text{data})\), a Bayesian posterior probability. For MR-BMA the score
is \(\mathrm{MIP}_X\), likewise a Bayesian quantity. For OLS the score is \(1 - p_F\), the complement of the joint \(F\)-test \(p\)-value for \(H_0\colon\beta_X = \beta_{XZ} = 0\). For
MVMR the score is \(1 - p_t\), the complement of the \(t\)-test \(p\)-value for \(H_0\colon\beta_X = 0\). Neither is a probability, but
both place all four methods on a common \([0, 1]\) scale for
visualization. Bayesian methods declare communication when the respective
score exceeds 0.5; frequentist methods declare communication when the
corresponding \(p\)-value is \(\leq 0.05\).

For each scenario, sample size, and method we report: (i) mean
communication score with standard deviation and (ii) rejection rate
across the 20 replicates.

\paragraph{Parameter estimation.}
For \(\beta_X\) and \(\beta_{XZ}\), we recorded the point estimate from
each method in every replicate. For MR-BMA, \(\beta_X\) is the MACE;
\(\beta_{XZ}\) is not estimable (denoted ``---''). For MVMR,
\(\beta_{XZ}\) is likewise not estimated as it is not included in the
standard formulation. We computed bias
(signed deviation from the true value) and mean absolute deviation (MAD)
across the 20 replicates. Results are reported in
Tables~\ref{tab:sim_s1}--\ref{tab:sim_s3}.

\subsection{Real-Data Analysis: Data Preparation and Instrument Selection}
\label{subsec:realdata_methods}

\paragraph{Dataset.}
The OneK1K cohort \cite{yazar2022onek1k} comprises donor-matched
germline genotype data and single-cell RNA sequencing profiles for
982 donors across five major peripheral immune cell types: B cells,
CD4$^{+}$ T cells, CD8$^{+}$ T cells, natural killer (NK) cells,
and monocytes.  For each ordered sender$\rightarrow$receiver pair we
evaluated ligand--receptor--pathway triplets in which ligand
expression was measured in the sender cell type and receptor
expression and pathway activity were obtained in the receiver cell
type.

\paragraph{Preprocessing and donor-level expression construction.}
Raw single-cell RNA counts were aggregated to donor-level gene
expression matrices by summing counts across all cells of a given
type per donor. Donor-specific library sizes were computed within
each cell type, normalized to the cell-type median, and used to
scale donor-level expression values. Donors were excluded if their
library size exceeded the median plus three median absolute
deviations, or fell below 25\% of the median, in either the sender
or receiver cell type. After filtering, the NK cells $\rightarrow$
Monocytes analysis retained 651 donors. For each triplet, ligand
expression $X$ was the normalized donor-level expression of the
ligand gene in NK cells; receptor expression $Z$ the corresponding
gene in monocytes; and pathway activity $Y$ a donor-level scalar
derived from the pathway gene set in monocytes.

\paragraph{Ligand--receptor and pathway definitions.}
Candidate ligand--receptor pairs were obtained from CellPhoneDB
\cite{efremova2020cellphonedb}. Pathway annotations were taken from
MSigDB Reactome gene sets \cite{liberzon2015msigdb}. After
restricting to genes present in the filtered expression matrices
and triplets with valid cis-eQTL instruments, the analysis yielded
41 ligand--receptor--pathway triplets spanning 11 biological
pathways.

\paragraph{Pathway activity construction and instrument selection.}
For each Reactome gene set in monocytes, donor-level pathway
activity was summarized using six complementary representations:
first principal component, mean expression, AUCell
\cite{aibar2017scenic}, UCell \cite{andreatta2021ucell}, ssGSEA
\cite{barbie2009ssgsea}, and GSVA \cite{hanzelmann2013gsva}. The
representation with the strongest absolute association with ligand
expression was retained for each triplet. Candidate cis-SNPs were
identified within a $\pm 200$~kb promoter-centred window for each
ligand and receptor gene, screened by regressing donor-level
expression on genotype dosage, adjusting for age, sex, and the
first three genotype principal components (these variables are also
used as covariates in MR-CCC), and ranked by association strength.
Up to ten top cis-SNPs per gene were used as instruments; triplets
lacking valid instruments for either gene were excluded.

\section*{Data availability}
Preprocessed donor-level single-cell RNA-seq expression matrices and
genotype data from the OneK1K cohort used in this study are publicly
available on Zenodo at
\url{https://doi.org/10.5281/zenodo.19675075}. The original OneK1K
genotype and scRNA-seq data are available from the study authors
(\url{https://onek1k.org}) under the original data access terms
\cite{yazar2022onek1k}. Ligand--receptor annotations were obtained
from CellPhoneDB v5 \cite{efremova2020cellphonedb}; pathway gene
sets from MSigDB \cite{liberzon2015msigdb}.

\section*{Code availability}
All analysis code, including the Gibbs sampler
(\texttt{mr\_ccc\_gibbs.cpp}),\\ simulation scripts
(\texttt{simulation\_mrccc.R}), database construction
(\texttt{build\_lr\_database.R}) and real-data analysis
(\texttt{real\_data\_analysis.R}) is publicly available at
\url{https://github.com/bitansa/MR-CCC} under the MIT License.

\section*{Author contributions}
Y.N.\ conceived the study and proposed the MR-CCC framework.
B.S.\ developed the methodology, implemented the Gibbs sampler
and the competing methods (OLS, MVMR, MR-BMA), performed the
simulation experiments and the OneK1K real-data analysis,
interpreted the biological findings together with Y.N., conducted
the supporting biological literature review, and wrote the
manuscript under Y.N.'s supervision. Both authors approved the
final version.

\section*{Competing interests}
The authors declare no competing interests.

\bibliographystyle{unsrtnat}   
\bibliography{Bibliography-MM-MC}


\end{document}


\def\spacingset#1{\renewcommand{\baselinestretch}%
{#1}\small\normalsize} \spacingset{1}


\if1\blind
{
  \title{\bf Title}
  \author{Author 1\thanks{
    The authors gratefully acknowledge \textit{please remember to list all relevant funding sources in the unblinded version}}\hspace{.2cm}\\
    Department of YYY, University of XXX\\
    and \\
    Author 2 \\
    Department of ZZZ, University of WWW}
  \maketitle
} \fi

\if0\blind
{
  \title{\bf Supplementary Materials\\ MR-CCC: Bayesian Mendelian Randomization for Causal Cell--Cell Communication}

  \author{%
    \large
    Bitan Sarkar$^{1}$ and Yang Ni$^{2}$\thanks{Corresponding author: Yang Ni (\href{mailto:yang.ni@austin.utexas.edu}{yang.ni@austin.utexas.edu}).}
    \\[0.75em]
    $^{1}$Department of Statistics, Texas A\&M University, College Station, TX, USA\\
    $^{2}$Department of Statistics and Data Sciences,\\ The University of Texas at Austin, Austin, TX, USA
  }

  \date{} 
  \maketitle
} \fi

\spacingset{1.45}

\section{Posterior Inference}

Posterior inference is performed via a Gibbs sampler with closed-form
updates. All matrix inverses are regularized by a small ridge term
\(\lambda > 0\) for numerical stability. For a matrix
\(A \in \mathbb{R}^{n \times p}\) and vector \(b \in \mathbb{R}^{n}\),
we write \(A \odot b := \mathrm{diag}(b)\,A\), i.e., the matrix
whose \(i\)-th row is \(b_i\) times the \(i\)-th row of \(A\).

\paragraph{1)}
Define
\(w_X := \beta_X\mathbf{1}_n + \beta_{XZ}Z^\ast\),
\(GW := G \odot w_X\), and
\(r_{\pi_X} := y - \mu\mathbf{1}_n - \beta_Z Z^\ast - V\alpha_Y
- (V\alpha_X)\circ w_X\). Then
\begin{align*}
\Sigma_{\pi_X} &= \Bigl[\tfrac{1}{\sigma_X^2}(1+g_G^{-1})G^\top G
  + \tfrac{1}{\sigma_Y^2}(GW)^\top(GW) + \lambda I\Bigr]^{-1},\\
m_{\pi_X} &= \Sigma_{\pi_X}\Bigl[\tfrac{1}{\sigma_X^2}G^\top(x - V\alpha_X)
  + \tfrac{1}{\sigma_Y^2}(GW)^\top r_{\pi_X}\Bigr],
\qquad \pi_X \mid \cdot \sim \mathcal{N}(m_{\pi_X},\Sigma_{\pi_X}).
\end{align*}

\paragraph{2)}
Define
\(VW := V \odot w_X\) and
\(r_{\alpha_X} := y - \mu\mathbf{1}_n - \beta_Z Z^\ast - V\alpha_Y
- (G\pi_X)\circ w_X\). Then
\begin{align*}
\Sigma_{\alpha_X} &= \Bigl[\tfrac{1}{\sigma_X^2}(1+g_V^{-1})V^\top V
  + \tfrac{1}{\sigma_Y^2}(VW)^\top(VW) + \lambda I\Bigr]^{-1},\\
m_{\alpha_X} &= \Sigma_{\alpha_X}\Bigl[\tfrac{1}{\sigma_X^2}V^\top(x-G\pi_X)
  + \tfrac{1}{\sigma_Y^2}(VW)^\top r_{\alpha_X}\Bigr],
\qquad \alpha_X \mid \cdot \sim \mathcal{N}(m_{\alpha_X},\Sigma_{\alpha_X}).
\end{align*}

\paragraph{3)}
Let \(x_\mathrm{res} := x - G\pi_X - V\alpha_X\). Then
\begin{align*}
\sigma_X^2 \mid \cdot \sim \mathrm{IG}\!\Biggl(
  a_\sigma + \tfrac{n+p_G+p_V}{2},\;\;
  b_\sigma + \tfrac{1}{2}\Bigl[x_\mathrm{res}^\top x_\mathrm{res}
  + g_G^{-1}\pi_X^\top G^\top G\pi_X
  + g_V^{-1}\alpha_X^\top V^\top V\alpha_X\Bigr]\Biggr).
\end{align*}

\paragraph{4)}
Define
\(w_Z := \beta_Z\mathbf{1}_n + \beta_{XZ}X^\ast\),
\(HW := H \odot w_Z\), and
\(r_{\pi_Z} := y - \mu\mathbf{1}_n - \beta_X X^\ast - V\alpha_Y
- (V\alpha_Z)\circ w_Z\). Then
\begin{align*}
\Sigma_{\pi_Z} &= \Bigl[\tfrac{1}{\sigma_Z^2}(1+g_H^{-1})H^\top H
  + \tfrac{1}{\sigma_Y^2}(HW)^\top(HW) + \lambda I\Bigr]^{-1},\\
m_{\pi_Z} &= \Sigma_{\pi_Z}\Bigl[\tfrac{1}{\sigma_Z^2}H^\top(z - V\alpha_Z)
  + \tfrac{1}{\sigma_Y^2}(HW)^\top r_{\pi_Z}\Bigr],
\qquad \pi_Z \mid \cdot \sim \mathcal{N}(m_{\pi_Z},\Sigma_{\pi_Z}).
\end{align*}

\paragraph{5)}
Define
\(VW_Z := V \odot w_Z\) and
\(r_{\alpha_Z} := y - \mu\mathbf{1}_n - \beta_X X^\ast - V\alpha_Y
- (H\pi_Z)\circ w_Z\). Then
\begin{align*}
\Sigma_{\alpha_Z} &= \Bigl[\tfrac{1}{\sigma_Z^2}(1+g_V^{-1})V^\top V
  + \tfrac{1}{\sigma_Y^2}(VW_Z)^\top(VW_Z) + \lambda I\Bigr]^{-1},\\
m_{\alpha_Z} &= \Sigma_{\alpha_Z}\Bigl[\tfrac{1}{\sigma_Z^2}V^\top(z - H\pi_Z)
  + \tfrac{1}{\sigma_Y^2}(VW_Z)^\top r_{\alpha_Z}\Bigr],
\qquad \alpha_Z \mid \cdot \sim \mathcal{N}(m_{\alpha_Z},\Sigma_{\alpha_Z}).
\end{align*}

\paragraph{6)}
Let \(z_\mathrm{res} := z - H\pi_Z - V\alpha_Z\). Then
\begin{align*}
\sigma_Z^2 \mid \cdot \sim \mathrm{IG}\!\Biggl(
  a_\sigma + \tfrac{n+p_H+p_V}{2},\;\;
  b_\sigma + \tfrac{1}{2}\Bigl[z_\mathrm{res}^\top z_\mathrm{res}
  + g_H^{-1}\pi_Z^\top H^\top H\pi_Z
  + g_V^{-1}\alpha_Z^\top V^\top V\alpha_Z\Bigr]\Biggr).
\end{align*}

\paragraph{7)}
Let \(r_\mu := y - X_\beta\beta - \beta_Z Z^\ast - V\alpha_Y\). Then
\begin{align*}
v_\mu = \frac{\sigma_Y^2}{n + n^{-1}}, \qquad
m_\mu = \frac{\mathbf{1}_n^\top r_\mu}{n + n^{-1}},
\qquad \mu \mid \cdot \sim \mathcal{N}(m_\mu,\, v_\mu).
\end{align*}

\paragraph{8)}
Let \(r_Y := y - \mu\mathbf{1}_n - X_\beta\beta - \beta_Z Z^\ast\) and
\(c_V := g_V/(1+g_V)\). Then
\begin{align*}
\alpha_Y \mid \cdot \sim \mathcal{N}\!\bigl(c_V(V^\top V)^{-1}V^\top r_Y,\;\;
c_V\,\sigma_Y^2\,(V^\top V)^{-1}\bigr).
\end{align*}

\paragraph{9)}
Let \(r_Z := y - \mu\mathbf{1}_n - X_\beta\beta - V\alpha_Y\) and
\(c_Z := g_Z/(1+g_Z)\). Then
\begin{align*}
\beta_Z \mid \cdot \sim \mathcal{N}\!\Bigl(
  c_Z\,\frac{Z^{\ast\top}r_Z}{Z^{\ast\top}Z^\ast},\;\;
  c_Z\,\sigma_Y^2\,(Z^{\ast\top}Z^\ast)^{-1}\Bigr),
\end{align*}
with \(Z^{\ast\top}Z^\ast\) regularized as \(Z^{\ast\top}Z^\ast + \lambda\).

\paragraph{10)}
Let \(e_Y := y - \mu\mathbf{1}_n - X_\beta\beta - \beta_Z Z^\ast - V\alpha_Y\). Then
\begin{align*}
\sigma_Y^2 \mid \cdot \sim \mathrm{IG}\!\Biggl(
  a_\sigma + \tfrac{n + p_V + 4}{2},\;\;
  b_\sigma + \tfrac{1}{2}\Bigl[e_Y^\top e_Y
  + n^{-1}\mu^2
  + (g_\beta s_\gamma)^{-1}\beta^\top X_\beta^\top X_\beta\beta\\
  + g_Z^{-1}\beta_Z^2 Z^{\ast\top}Z^\ast
  + g_V^{-1}\alpha_Y^\top V^\top V\alpha_Y\Bigr]\Biggr).
\end{align*}

\paragraph{11)}
Let \(r_\beta := y - \mu\mathbf{1}_n - \beta_Z Z^\ast - V\alpha_Y\) and
\(c_\gamma := g_\beta s_\gamma/(1 + g_\beta s_\gamma)\). Then
\begin{align*}
\beta \mid \cdot \sim \mathcal{N}_2\!\bigl(
  c_\gamma(X_\beta^\top X_\beta)^{-1}X_\beta^\top r_\beta,\;\;
  c_\gamma\,\sigma_Y^2\,(X_\beta^\top X_\beta)^{-1}\bigr),
\end{align*}
with \(X_\beta^\top X_\beta\) regularized by \(\lambda I_2\).

\paragraph{12)}
Let \(Q := \beta^\top(X_\beta^\top X_\beta)\beta\). Define
\begin{align*}
\log A &= -\tfrac{1}{2}(g_\beta\sigma_Y^2)^{-1}Q + \log\rho,\\
\log B &= -\tfrac{1}{2}(g_\beta\nu_1\sigma_Y^2)^{-1}Q + \log(1-\rho) - \log\nu_1,
\end{align*}
and set \(p = \exp(\log A)/[\exp(\log A)+\exp(\log B)]\). Then
\(\gamma \mid \cdot \sim \mathrm{Bernoulli}(p)\).

\paragraph{13)}
\[
\rho \mid \gamma \sim \mathrm{Beta}(a_\rho + \gamma,\; b_\rho + 1 - \gamma).
\]

\section{Real Data Results for Additional Ordered Cell-Type Pairs}
\label{sec:supp_realdata}

The main manuscript presents the results for the NK cells
\(\rightarrow\) Monocytes axis. Here we report the results
for the remaining 19 ordered sender\,$\rightarrow$\,receiver
pairs among the five peripheral immune cell types profiled in
the OneK1K cohort \cite{yazar2022onek1k}. Data
preprocessing, donor-level expression construction, cis-eQTL
instrument selection, and model fitting follow exactly the same
pipeline described in main paper. For each ordered pair we display three figures:
(i) the PIP ranking across all evaluated ligand--receptor--pathway
triplets; (ii) the bubble plot of standardized posterior main and
interaction effects; and (iii) the receptor-modulated effect curves
for the high-PIP triplets.

\subsection{B Cells as Sender}

\subsubsection{B Cells \(\rightarrow\) CD4\(^+\) T Cells}
\label{supp:B_CD4}

Across 868 donors, MR-CCC evaluated 42 ligand--receptor--pathway
triplets for the B cell (sender) to CD4\(^+\) T cell (receiver)
direction and identified three high-confidence causal communication
signals with posterior inclusion probability (PIP) exceeding the
discovery threshold of 0.5: \textit{IL15}--\textit{IL15RA} within the Interleukin signaling pathway (PIP \(= 0.867\)),
\textit{LIPA}--\textit{RORA} within the Cholesterol signaling pathway (PIP \(= 0.686\)),
and \textit{SLC6A6}--\textit{GABBR1} within the GABA signaling pathway (PIP \(= 0.54\))
(Figures~\ref{fig:supp_B_CD4_pip}--\ref{fig:supp_B_CD4_curves}).

The strongest signal, \textit{IL15}--\textit{IL15RA} (PIP \(= 0.867\)),
reflects the well-characterized mechanism of IL-15 transpresentation,
in which antigen-presenting cells, including B cells, display IL-15 on
their surface via IL-15R\(\alpha\) for recognition by neighboring
CD4\(^+\) T cells~\cite{dubois2002il15ra,fehniger2001il15}.
IL-15 is a pleiotropic cytokine that promotes T cell survival,
proliferation, and memory formation, with its bioavailability tightly
controlled by IL-15R\(\alpha\)-mediated
transpresentation~\cite{waldmann2006il15}. The estimated main effect
is negative (\(\hat{\beta}_X = -1.13\)), while the receptor-modulated
interaction is large and positive (\(\hat{\beta}_{XZ} = 2.28\)),
producing an effect curve that crosses zero near the population-mean
\textit{IL15RA} expression level. CD4\(^+\) T cells expressing
\textit{IL15RA} above this mean therefore exhibit a progressively
positive causal response to B cell IL-15 production, consistent with
IL-15R\(\alpha\)-dependent co-stimulation that amplifies with receptor
abundance.

The second discovered triplet, \textit{LIPA}--\textit{RORA} (PIP \(= 0.686\)), connects B cell lysosomal acid lipase A
(\textit{LIPA}) activity---which releases free cholesterol and fatty
acids from intracellular lipid
esters~\cite{tall2015cholesterol}---to \textit{RORA} signaling in
CD4\(^+\) T cells. \textit{RORA} is a nuclear receptor activated by
cholesterol metabolites and oxysterols that negatively regulates
inflammatory NF-\(\kappa\)B signaling and shapes helper T cell
differentiation~\cite{jetten2009ror,yang2008rora}. Both the main
effect (\(\hat{\beta}_X = -0.682\)) and the interaction
(\(\hat{\beta}_{XZ} = -2.40\)) are negative, yielding a monotonically
decreasing effect curve: B cell lipid-catabolic activity exerts a
progressively stronger suppressive influence on CD4\(^+\) T cell
downstream pathway activity as \textit{RORA} expression rises. This
pattern suggests that B cell-derived cholesterol metabolites engage
\textit{RORA}-mediated anti-inflammatory programs in CD4\(^+\) T
cells in a receptor dose-dependent manner.

The third discovered triplet, \textit{SLC6A6}--\textit{GABBR1} (PIP \(= 0.54\)), exhibits a near-zero main effect
(\(\hat{\beta}_X = 0.014\)) alongside a large positive interaction
(\(\hat{\beta}_{XZ} = 2.01\)), yielding an almost purely
receptor-dependent effect curve that passes through zero at the mean
and rises steeply. \textit{SLC6A6} encodes the taurine and
\(\beta\)-alanine transporter; taurine is a known partial agonist of
GABA receptors and provides a mechanistic link between B cell
amino-acid transport and GABA-B receptor activation on CD4\(^+\) T
cells~\cite{tian1999gaba}. GABA-B receptor signaling has been shown
to suppress T cell activation and proliferation in a receptor
density-dependent manner~\cite{bjurstrom2008gaba}, consistent with
the purely interaction-driven pattern observed here: the communication
signal is essentially absent when \textit{GABBR1} expression is at or
below the population mean and emerges only in CD4\(^+\) T cells with
elevated receptor levels.

Notably, several biologically implausible triplets received near-zero
PIPs, providing implicit validation of the model's specificity.
Interferon-\(\gamma\) (\textit{IFNG}) paired with either of its
receptors (\textit{IFNGR1}, PIP \(= 0.043\); \textit{IFNGR2},
PIP \(= 0.033\)) ranked near the bottom of the PIP distribution,
consistent with the well-established fact that IFN-\(\gamma\) is
produced primarily by Th1 CD4\(^+\) and CD8\(^+\) T cells rather
than B cells, making B cells an implausible sender in this
axis~\cite{schroder2004interferon}. The \textit{CCL28}--\textit{CCR10}
pair (PIP \(= 0.029\)) similarly received one of the lowest PIPs:
\textit{CCL28}--\textit{CCR10} mediates homing of T cells to skin and
mucosal epithelia rather than peripheral lymphoid communication, and
this tissue-specific axis is not expected to be active in peripheral
blood donors~\cite{pan2000ccl28,homey2002ccr10}. Finally, the
\textit{WNT7A}--\textit{FZD3} pair (PIP \(= 0.028\)) received the
lowest PIP of all 42 triplets; WNT signaling between B cells and
CD4\(^+\) T cells has no established biological basis in peripheral
blood~\cite{clevers2006wnt}. 

\begin{figure}[htbp]
\centering
\includegraphics[width=\textwidth]{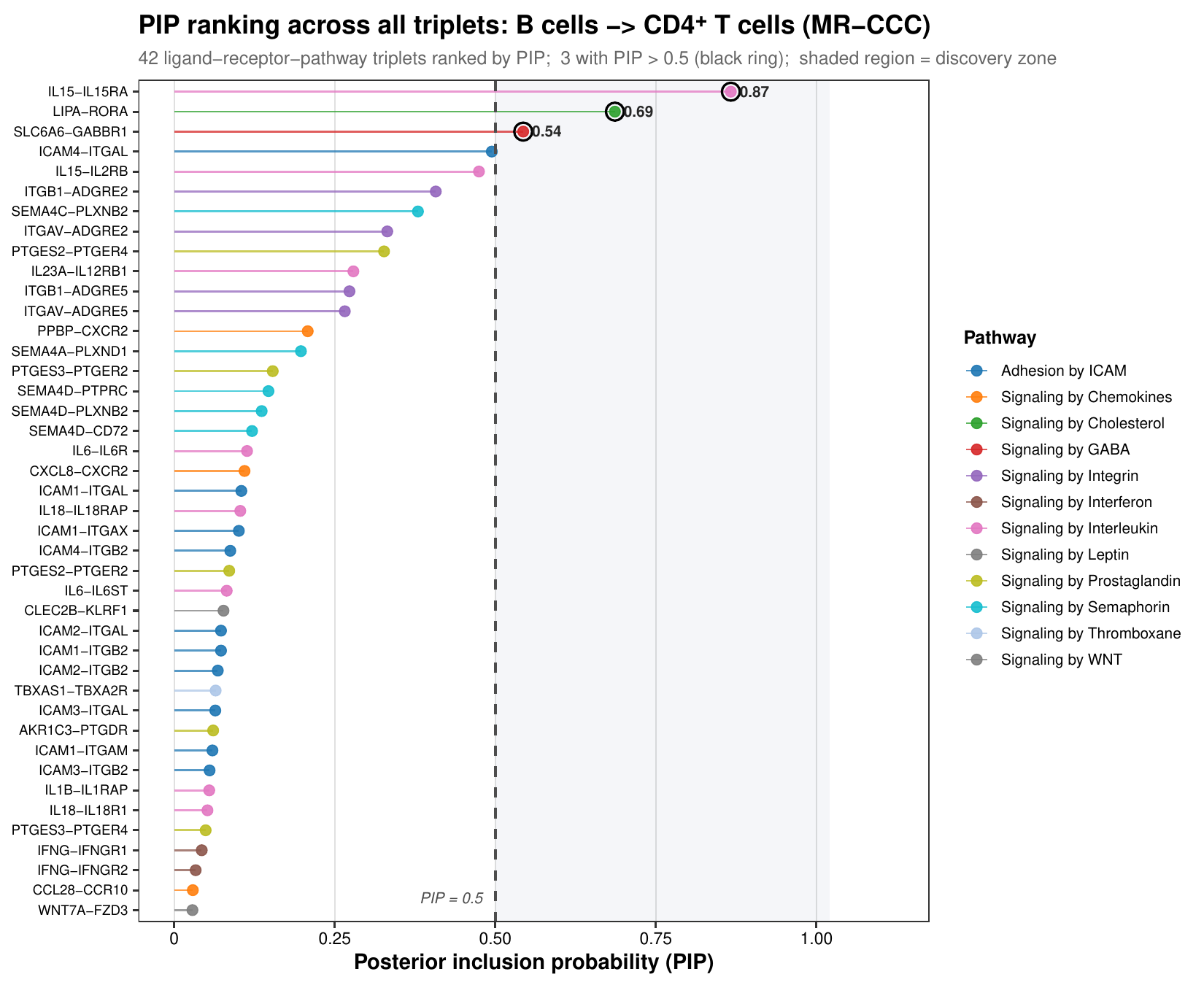}
\caption{\textbf{PIP ranking for the B cells\,$\rightarrow$\,CD4$^+$
  T cells analysis.} All 42 ligand--receptor--pathway triplets across 868
  donors. Each point represents one triplet, colored by pathway;
  the dashed vertical line marks the discovery threshold of
  PIP \(= 0.5\); the shaded region to the right is the discovery zone.
  Black rings identify the three discovered triplets:
  \textit{IL15}--\textit{IL15RA} (PIP \(= 0.867\)),
  \textit{LIPA}--\textit{RORA} (PIP \(= 0.686\)), and
  \textit{SLC6A6}--\textit{GABBR1} (PIP \(= 0.54\)).
  Biologically implausible pairs such as
  \textit{IFNG}--\textit{IFNGR1}/\textit{IFNGR2},
  \textit{CCL28}--\textit{CCR10}, and \textit{WNT7A}--\textit{FZD3}
  rank at the bottom of the distribution with PIPs below 0.05.}
\label{fig:supp_B_CD4_pip}
\end{figure}

\begin{figure}[htbp]
\centering
\includegraphics[width=\textwidth]{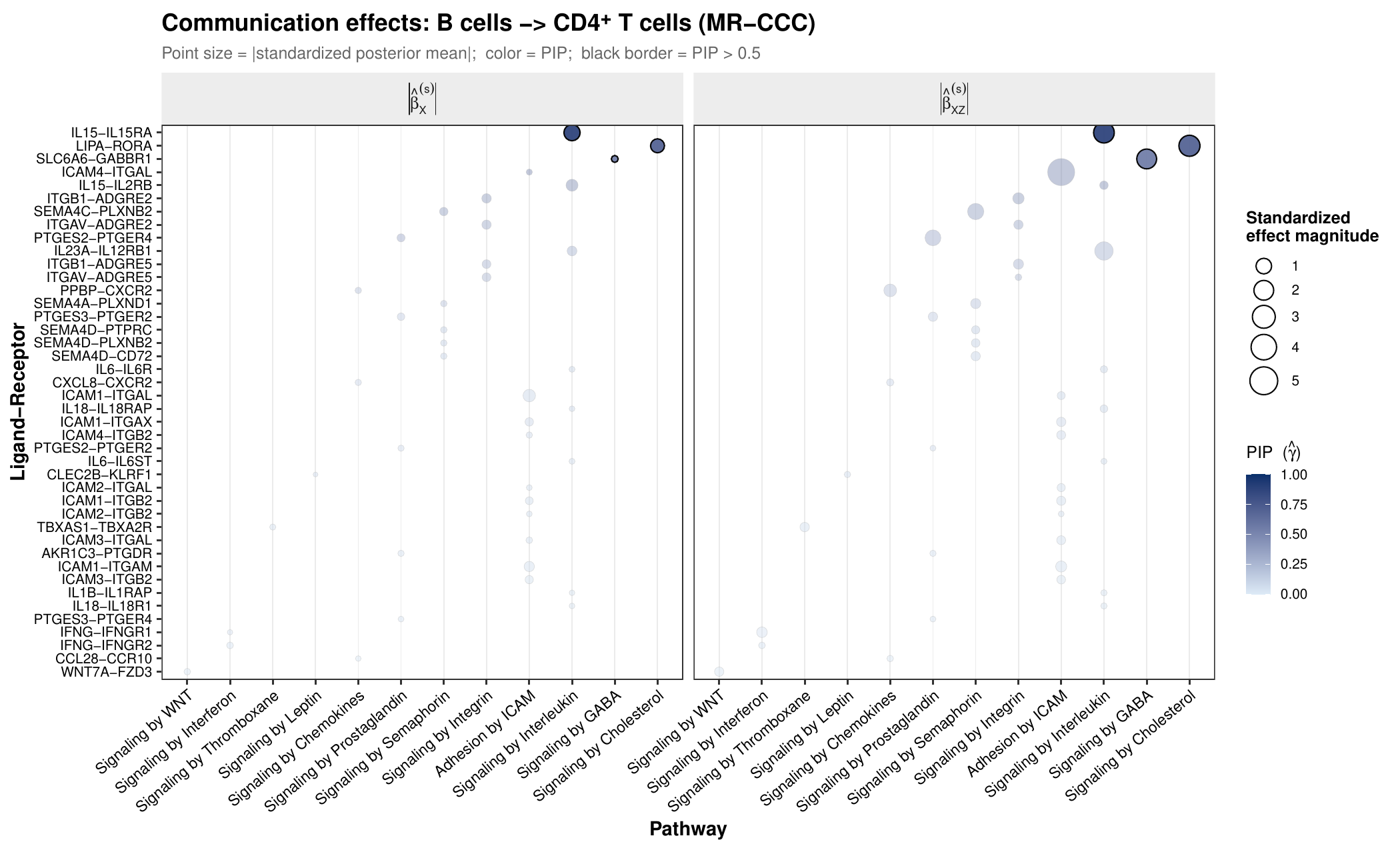}
\caption{\textbf{Standardized posterior effects for the B cells
  \(\rightarrow\) CD4\(^+\) T cells analysis.} Left panel: absolute
  main ligand effect \(|\hat{\beta}_X^{(s)}|\); right panel:
  absolute receptor-modulated interaction effect
  \(|\hat{\beta}_{XZ}^{(s)}|\). Point size encodes effect magnitude;
  fill color encodes PIP (dark blue \(=\) high PIP); black borders
  identify the three triplets with PIP \(> 0.5\). All three
  discovered triplets show large interaction components in the right
  panel, indicating that the causal communication effects are strongly
  modulated by receptor expression in the CD4\(^+\) T cell. Low-PIP
  triplets (faint, no black border) show small effect magnitudes and
  low color saturation throughout.}
\label{fig:supp_B_CD4_bubble}
\end{figure}

\begin{figure}[htbp]
\centering
\includegraphics[width=\textwidth]{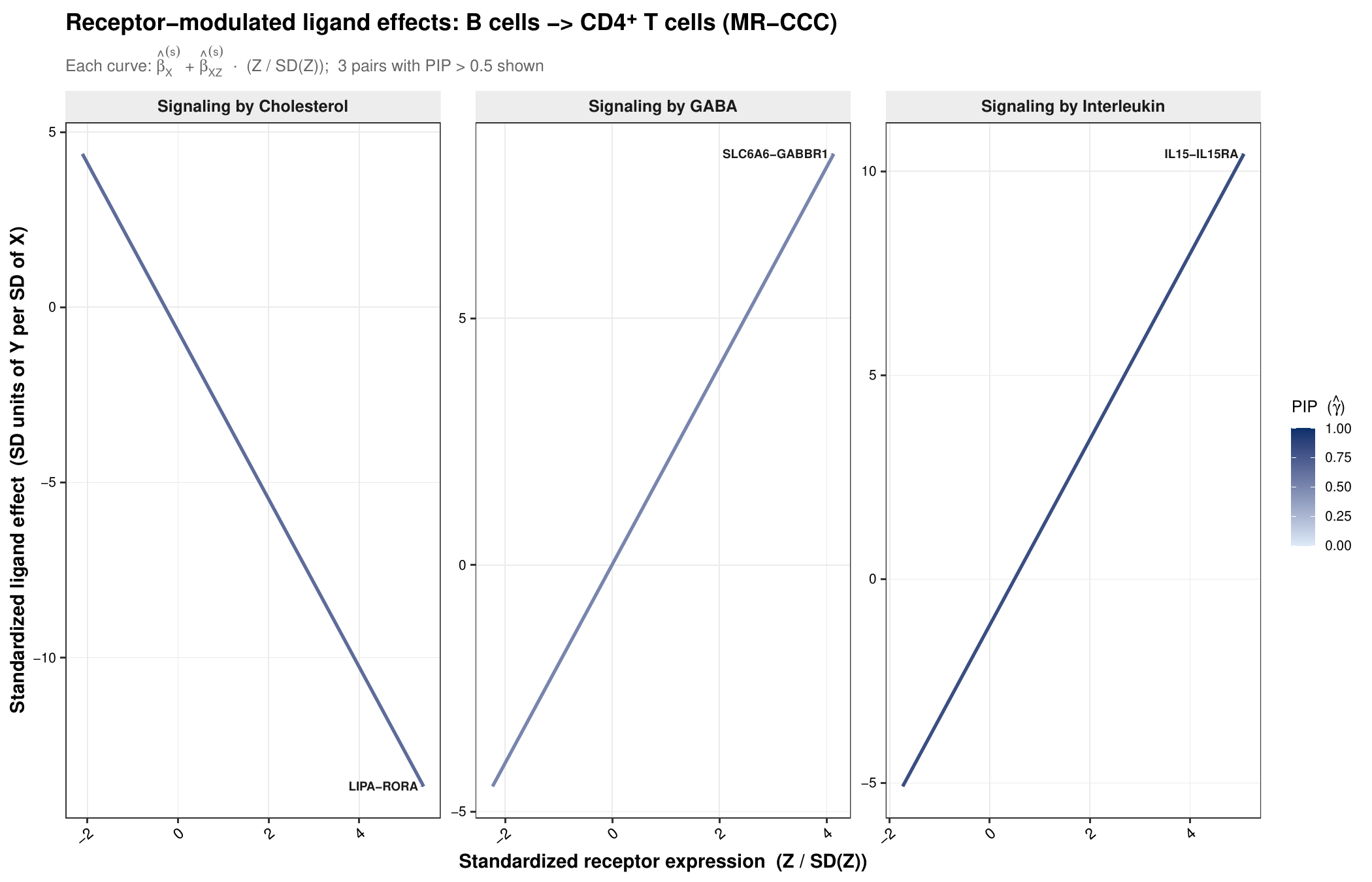}
\caption{\textbf{Receptor-modulated effect curves
  (\(\hat{\beta}_X + \hat{\beta}_{XZ} \cdot Z/\mathrm{SD}(Z)\))
  for the B cells \(\rightarrow\) CD4\(^+\) T cells analysis.}
  Only the three discovery pairs (PIP \(> 0.5\)) are displayed,
  grouped into the three pathway panels with confirmed discoveries
  (Signaling by Cholesterol, Signaling by GABA, and Signaling by
  Interleukin). In the Signaling by Interleukin panel, the
  \textit{IL15}--\textit{IL15RA} curve is steeply increasing and
  crosses zero near the mean receptor level, reflecting a
  receptor-dependent sign switch from suppression to activation.
  In the Signaling by Cholesterol panel, the
  \textit{LIPA}--\textit{RORA} curve is monotonically decreasing,
  indicating receptor-amplified suppression that strengthens as
  \textit{RORA} expression rises. In the Signaling by GABA panel,
  the \textit{SLC6A6}--\textit{GABBR1} curve passes near zero at
  the mean and rises steeply, consistent with a purely
  interaction-driven signal gated by \textit{GABBR1} receptor
  density.}
\label{fig:supp_B_CD4_curves}
\end{figure}

\FloatBarrier

\subsubsection{B Cells \(\rightarrow\) CD8\(^+\) T Cells}
\label{supp:B_CD8}

Across 845 donors, MR-CCC evaluated 42 ligand--receptor--pathway
triplets for the B cell (sender) to CD8\(^+\) T cell (receiver)
direction and identified two high-confidence causal communication
signals with PIP exceeding 0.5: \textit{IL15}--\textit{IL15RA}
(PIP \(= 0.554\)) and \textit{IL15}--\textit{IL2RB} (PIP \(= 0.518\)),
both within the Interleukin signaling pathway
(Figures~\ref{fig:supp_B_CD8_pip}--\ref{fig:supp_B_CD8_curves}).
Both discoveries implicate B cell-derived IL-15 as a causal
communicator to CD8\(^+\) T cells, engaging two distinct receptor
chains of the IL-15/IL-2 receptor complex.

The \textit{IL15}--\textit{IL15RA} triplet (PIP \(= 0.554\))
recapitulates the IL-15 transpresentation axis identified in the B
cells \(\rightarrow\) CD4\(^+\) T cells direction, here operating on
CD8\(^+\) T cells.
IL-15 displayed on B cell surfaces via IL-15R\(\alpha\) delivers
homeostatic survival and proliferation signals to CD8\(^+\) T cells,
a mechanism well established to be indispensable for the maintenance
of virus-specific memory CD8\(^+\) T
cells~\cite{waldmann2006il15,kennedy2000il15}.
Both the estimated main effect (\(\hat{\beta}_X = 0.662\)) and the
receptor-modulated interaction (\(\hat{\beta}_{XZ} = 0.589\)) are
positive, yielding a monotonically increasing effect curve:
CD8\(^+\) T cells expressing higher levels of \textit{IL15RA} show a
uniformly stronger positive causal response to B cell IL-15
production, with no sign reversal across the observed receptor
expression range.

The second discovered triplet, \textit{IL15}--\textit{IL2RB}
(PIP \(= 0.518\)), engages the IL-2R\(\beta\) chain (encoded by
\textit{IL2RB}; also known as CD122), which is shared between the
IL-15 and IL-2 receptor complexes~\cite{waldmann2006il15}.
On CD8\(^+\) T cells, signaling through the IL-2R\(\beta\)/\(\gamma_c\)
heterodimer supports homeostatic and antigen-driven
proliferation~\cite{schluns2003review}.
The positive main effect (\(\hat{\beta}_X = 0.796\)) indicates that B
cell IL-15 production generally up-regulates downstream Interleukin
pathway activity in CD8\(^+\) T cells; the negative receptor-modulated
interaction (\(\hat{\beta}_{XZ} = -0.647\)) reveals that this
stimulatory effect is progressively attenuated in cells with high
\textit{IL2RB} expression.
The effect curve crosses zero near \({\sim}1.2\) standard deviations
above the population-mean \textit{IL2RB} level and becomes mildly
suppressive at the highest receptor densities, consistent with
signaling saturation or competitive displacement by IL-2 at the shared
IL-2R\(\beta\) chain in CD8\(^+\) T cells that strongly up-regulate
this subunit.

Notably, several triplets received near-zero PIPs.
\textit{IFNG}--\textit{IFNGR1} (PIP \(= 0.032\)) and
\textit{IFNG}--\textit{IFNGR2} (PIP \(= 0.032\)) remained at the
bottom of the distribution, consistent with B cells not being a
physiological source of IFN-\(\gamma\) in peripheral
blood~\cite{schroder2004interferon}.
The \textit{LIPA}--\textit{RORA} triplet, which reached
PIP \(= 0.686\) in the B cells \(\rightarrow\) CD4\(^+\) T cells
direction, fell to PIP \(= 0.191\) here, suggesting that B
cell-derived cholesterol-catabolic signals preferentially engage
\textit{RORA}-mediated immunomodulatory programs in helper T cells
rather than cytotoxic T cells, consistent with \textit{RORA}'s
established role in Th17/Treg balance~\cite{yang2008rora}.
Similarly, \textit{SLC6A6}--\textit{GABBR1}, which was discovered in
the B cells \(\rightarrow\) CD4\(^+\) T cells direction
(PIP \(= 0.54\)), received a PIP of only 0.035 here, indicating that
the GABA-B receptor-mediated inhibitory axis from B cell taurine
transport operates in a CD4\(^+\)-selective manner.

\begin{figure}[htbp]
\centering
\includegraphics[width=\textwidth]{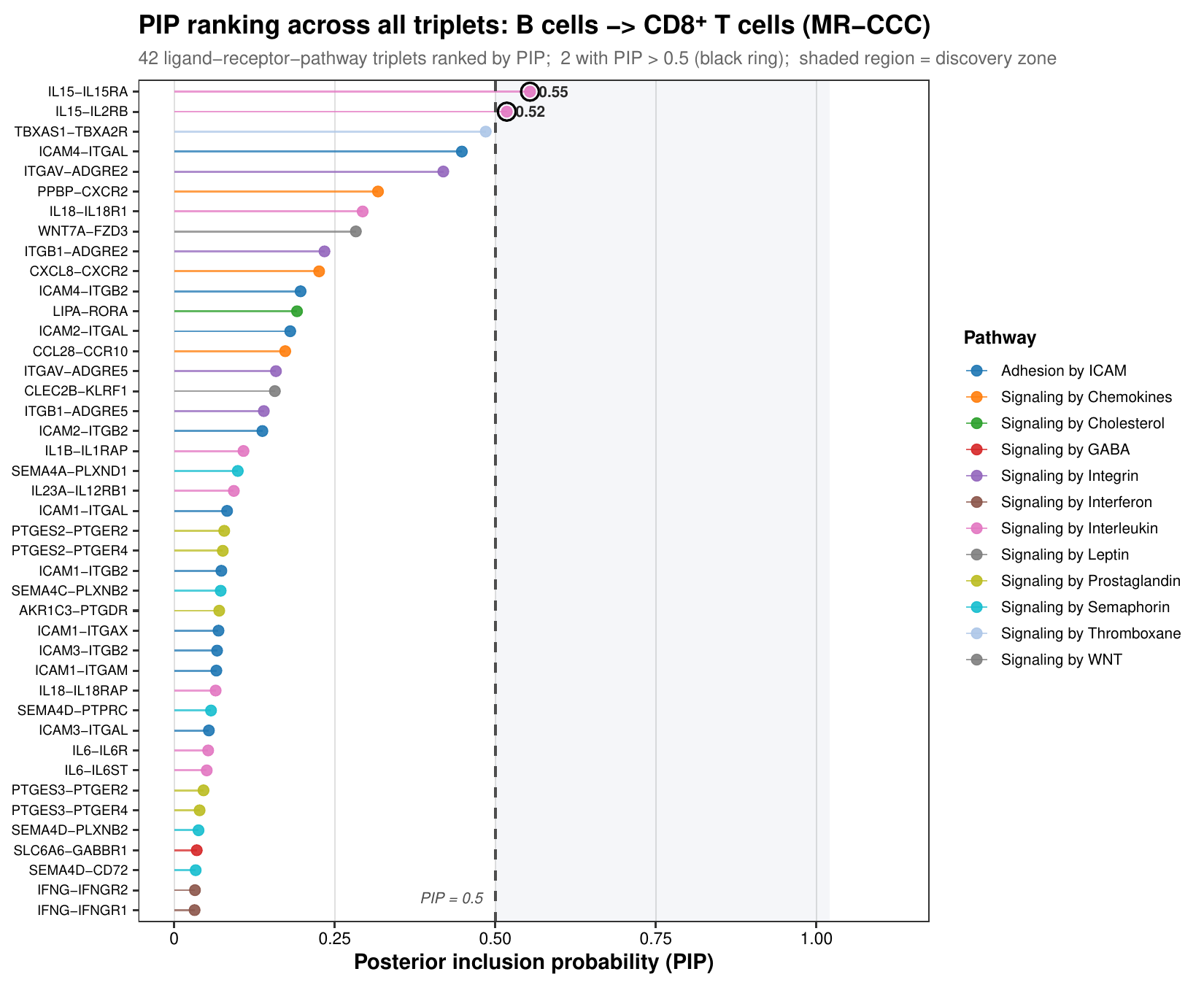}
\caption{\textbf{PIP ranking for the B cells\,$\rightarrow$\,CD8$^+$
  T cells analysis.} All 42 ligand--receptor--pathway triplets across 845
  donors. Points are colored by pathway; the dashed vertical line
  marks the discovery threshold of PIP \(= 0.5\); the shaded region
  to the right is the discovery zone. Black rings identify the two
  discovered triplets: \textit{IL15}--\textit{IL15RA}
  (PIP \(= 0.55\)) and \textit{IL15}--\textit{IL2RB}
  (PIP \(= 0.52\)), both within Signaling by Interleukin.}
\label{fig:supp_B_CD8_pip}
\end{figure}

\begin{figure}[htbp]
\centering
\includegraphics[width=\textwidth]{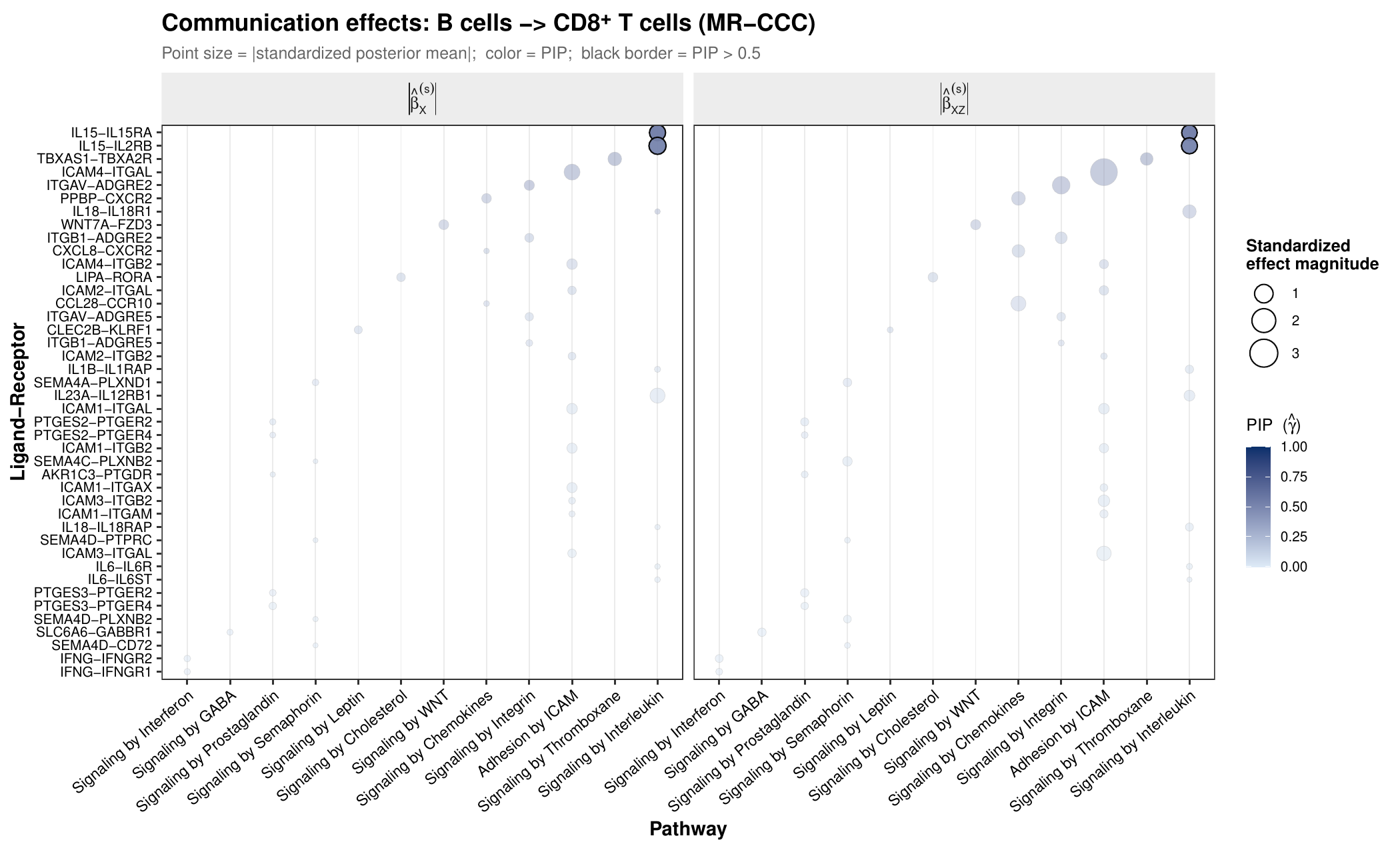}
\caption{\textbf{Standardized posterior effects for the B cells
  \(\rightarrow\) CD8\(^+\) T cells analysis.} Left panel: absolute
  main ligand effect \(|\hat{\beta}_X^{(s)}|\); right panel: absolute
  receptor-modulated interaction effect \(|\hat{\beta}_{XZ}^{(s)}|\).
  Point size encodes effect magnitude; fill color encodes PIP; black
  borders identify the two discovered triplets. Both
  \textit{IL15}--\textit{IL15RA} and \textit{IL15}--\textit{IL2RB}
  show appreciable effects in both panels, indicating a combination of
  a direct ligand effect and receptor-level modulation.}
\label{fig:supp_B_CD8_bubble}
\end{figure}

\begin{figure}[htbp]
\centering
\includegraphics[width=\textwidth]{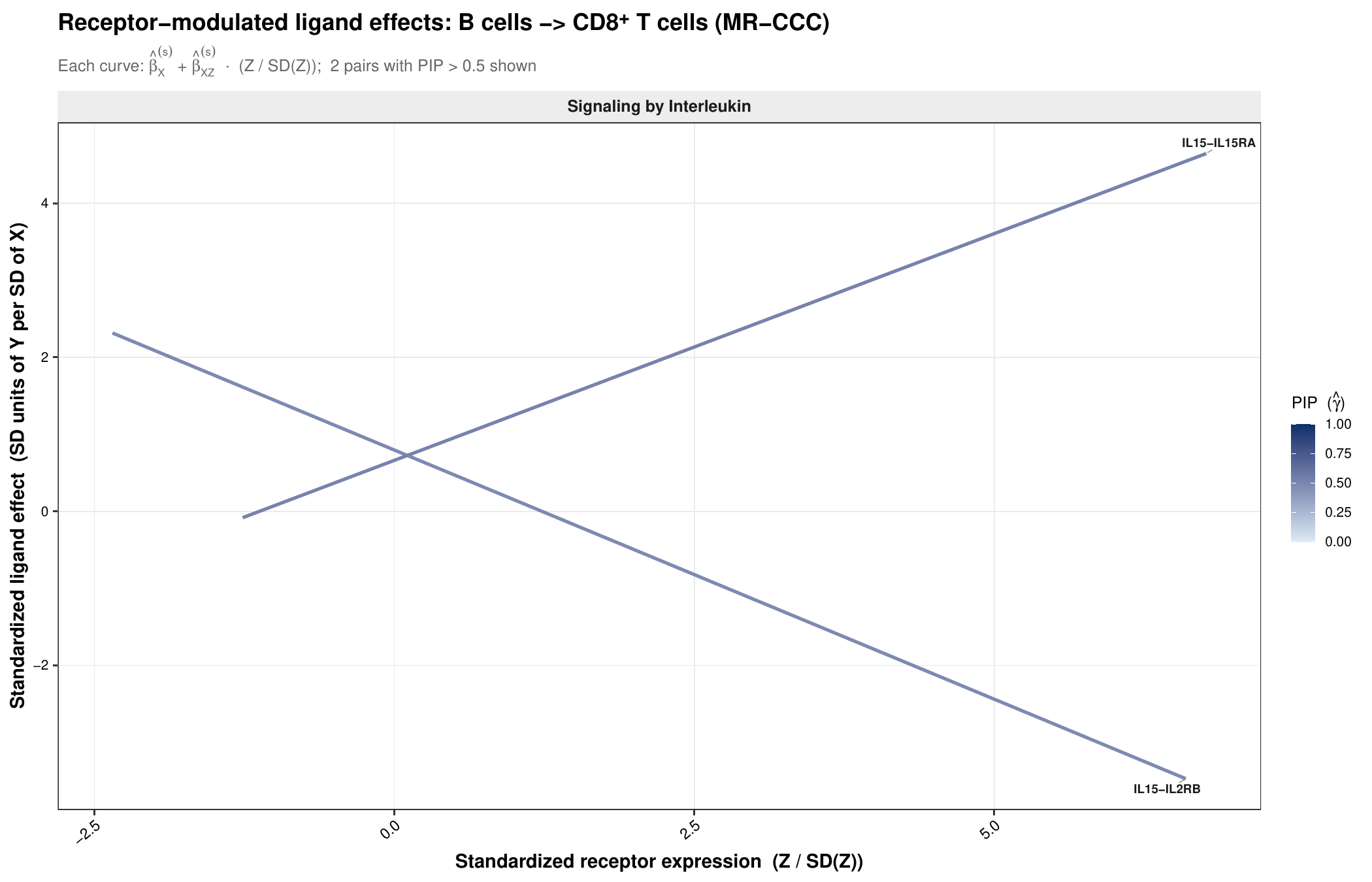}
\caption{\textbf{Receptor-modulated effect curves
  (\(\hat{\beta}_X + \hat{\beta}_{XZ} \cdot Z/\mathrm{SD}(Z)\))
  for the B cells \(\rightarrow\) CD8\(^+\) T cells analysis.}
  Only the two discovery pairs (PIP \(> 0.5\)) are displayed in the
  single pathway panel with confirmed discoveries (Signaling by
  Interleukin). The \textit{IL15}--\textit{IL15RA} curve is
  monotonically increasing (both \(\hat{\beta}_X\) and
  \(\hat{\beta}_{XZ}\) positive), while the
  \textit{IL15}--\textit{IL2RB} curve has a positive intercept but a
  negative slope, crossing zero near \({\sim}1.2\) standard deviations
  above the mean \textit{IL2RB} expression level.}
\label{fig:supp_B_CD8_curves}
\end{figure}

\FloatBarrier

\subsubsection{B Cells \(\rightarrow\) NK Cells}
\label{supp:B_NK}
Across 818 donors, MR-CCC evaluated 42 ligand--receptor--pathway
triplets for the B cell (sender) to NK cell (receiver) direction and
identified three high-confidence causal communication signals with PIP
exceeding 0.5: \textit{IL6}--\textit{IL6ST} within the Interleukin
signaling pathway (PIP \(= 0.972\)),
\textit{PTGES3}--\textit{PTGER2} within the Prostaglandin signaling
pathway (PIP \(= 0.556\)), and \textit{ITGAV}--\textit{ADGRE2}
within the Integrin signaling pathway (PIP \(= 0.506\))
(Figures~\ref{fig:supp_B_NK_pip}--\ref{fig:supp_B_NK_curves}).
All three discoveries are characterized by large receptor-modulated
interaction terms whose sign reversal occurs near the population-mean
receptor expression level, indicating that the direction as well as
the magnitude of each B cell communication effect is gated by
recipient NK cell receptor density.

The dominant signal, \textit{IL6}--\textit{IL6ST}
(PIP \(= 0.972\)), identifies B cell-derived IL-6 as a causal
regulator of NK cell Interleukin pathway activity mediated through
the signal-transducing subunit gp130 (encoded by \textit{IL6ST}).
B cells are well-established producers of IL-6, particularly
following TLR or BCR activation, and IL-6 engages gp130-containing
receptor complexes to activate the JAK1/JAK2--STAT1/STAT3 cascade in
target cells~\cite{heinrich2003il6}.
The modest positive main effect (\(\hat{\beta}_X = 0.394\)) combined
with a large negative interaction (\(\hat{\beta}_{XZ} = -2.09\))
yields an effect curve that crosses zero approximately \(0.19\)
standard deviations above the mean \textit{IL6ST} level and then
becomes steeply negative.
NK cells expressing \textit{IL6ST} at or above the population mean
therefore experience a progressively stronger suppressive effect of B
cell IL-6 on NK cell Interleukin pathway gene expression, consistent
with IL-6's documented ability to inhibit NK cell cytotoxicity and
effector programs when gp130-mediated signaling is
sustained~\cite{cifaldi2015nk}.

The second discovered triplet, \textit{PTGES3}--\textit{PTGER2}
(PIP \(= 0.556\)), links B cell prostaglandin E synthase 3
(\textit{PTGES3})---a cytosolic enzyme that converts prostaglandin
H\(_2\) to prostaglandin E\(_2\) (PGE\(_2\))---to the EP2
prostaglandin receptor (encoded by \textit{PTGER2}) on NK cells.
PGE\(_2\) signaling through the EP2 receptor activates the adenylyl
cyclase--cAMP--PKA axis, a well-characterized suppressor of NK cell
cytotoxicity and IFN-\(\gamma\) production~\cite{kalinski2012pge2}.
The negative main effect (\(\hat{\beta}_X = -0.416\)) and positive
interaction (\(\hat{\beta}_{XZ} = 1.57\)) produce a curve that crosses
zero approximately \(0.26\) standard deviations above the mean
\textit{PTGER2} level.
NK cells with low EP2 expression show a net suppressive response to B
cell PGE\(_2\) output, while NK cells expressing \textit{PTGER2}
above this threshold show progressive up-regulation of their own
Prostaglandin pathway activity, consistent with a PGE\(_2\)--EP2
feed-forward mechanism in which receptor-replete NK cells amplify
prostaglandin signaling in response to B cell-derived PGE\(_2\).

The third discovered triplet, \textit{ITGAV}--\textit{ADGRE2}
(PIP \(= 0.506\)), implicates a direct adhesion contact between B cell
integrin \(\alpha\)V (\textit{ITGAV}) and the adhesion G
protein-coupled receptor EMR2 (encoded by \textit{ADGRE2}) on NK
cells.
\textit{ADGRE2} (EMR2) carries extracellular EGF-like domains that
serve as integrin counter-ligands mediating cell--cell contact, and is
expressed on NK cells and myeloid cells in peripheral
blood~\cite{stacey2003emr2}.
The small positive main effect (\(\hat{\beta}_X = 0.136\)) and large
negative interaction (\(\hat{\beta}_{XZ} = -1.47\)) yield a curve
that crosses zero within \(0.1\) standard deviations of the mean
\textit{ADGRE2} expression level and becomes strongly negative at high
receptor density.
NK cells expressing high levels of \textit{ADGRE2} therefore
experience a strongly suppressive effect of B cell \(\alpha\)V integrin
engagement, suggesting that \(\alpha\)V--EMR2 contact at the B--NK
cell interface inhibits NK cell Integrin pathway activity in an
\textit{ADGRE2} dose-dependent manner.

Notably, several triplets received low PIPs.
\textit{IL15}--\textit{IL15RA} (PIP \(= 0.373\)) and
\textit{IL15}--\textit{IL2RB} (PIP \(= 0.353\)) were the
highest-ranking non-discovered signals but fell below the discovery
threshold, suggesting that B cells are not a dominant cellular source
of transpresented IL-15 for NK cells in peripheral blood; the primary
providers of NK cell-directed IL-15 are dendritic cells and
macrophages~\cite{fehniger2001il15,kennedy2000il15}.
\textit{IFNG}--\textit{IFNGR1} (PIP \(= 0.171\)) and
\textit{IFNG}--\textit{IFNGR2} (PIP \(= 0.060\)) again received
near-zero support, consistent with B cells not being a physiological
source of IFN-\(\gamma\) in peripheral
blood~\cite{schroder2004interferon}.

\begin{figure}[htbp]
\centering
\includegraphics[width=\textwidth]{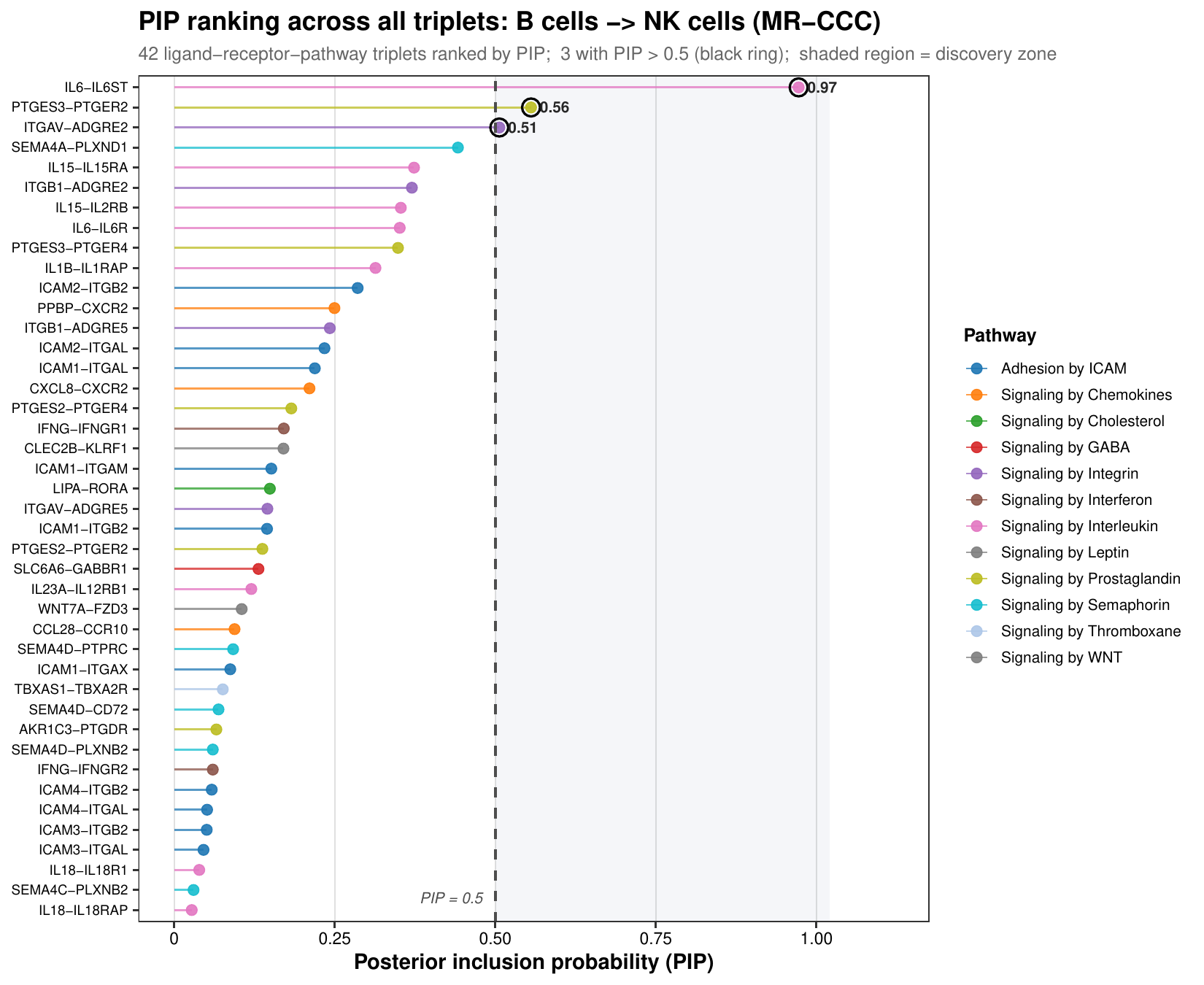}
\caption{\textbf{PIP ranking for the B cells\,$\rightarrow$\,NK cells analysis.} All 42 ligand--receptor--pathway triplets across 818 donors.
  Points are colored by pathway; the dashed vertical line marks the
  discovery threshold of PIP \(= 0.5\); the shaded region to the
  right is the discovery zone. Black rings identify the three
  discovered triplets: \textit{IL6}--\textit{IL6ST}
  (PIP \(= 0.97\)), \textit{PTGES3}--\textit{PTGER2}
  (PIP \(= 0.56\)), and \textit{ITGAV}--\textit{ADGRE2}
  (PIP \(= 0.51\)).}
\label{fig:supp_B_NK_pip}
\end{figure}

\begin{figure}[htbp]
\centering
\includegraphics[width=\textwidth]{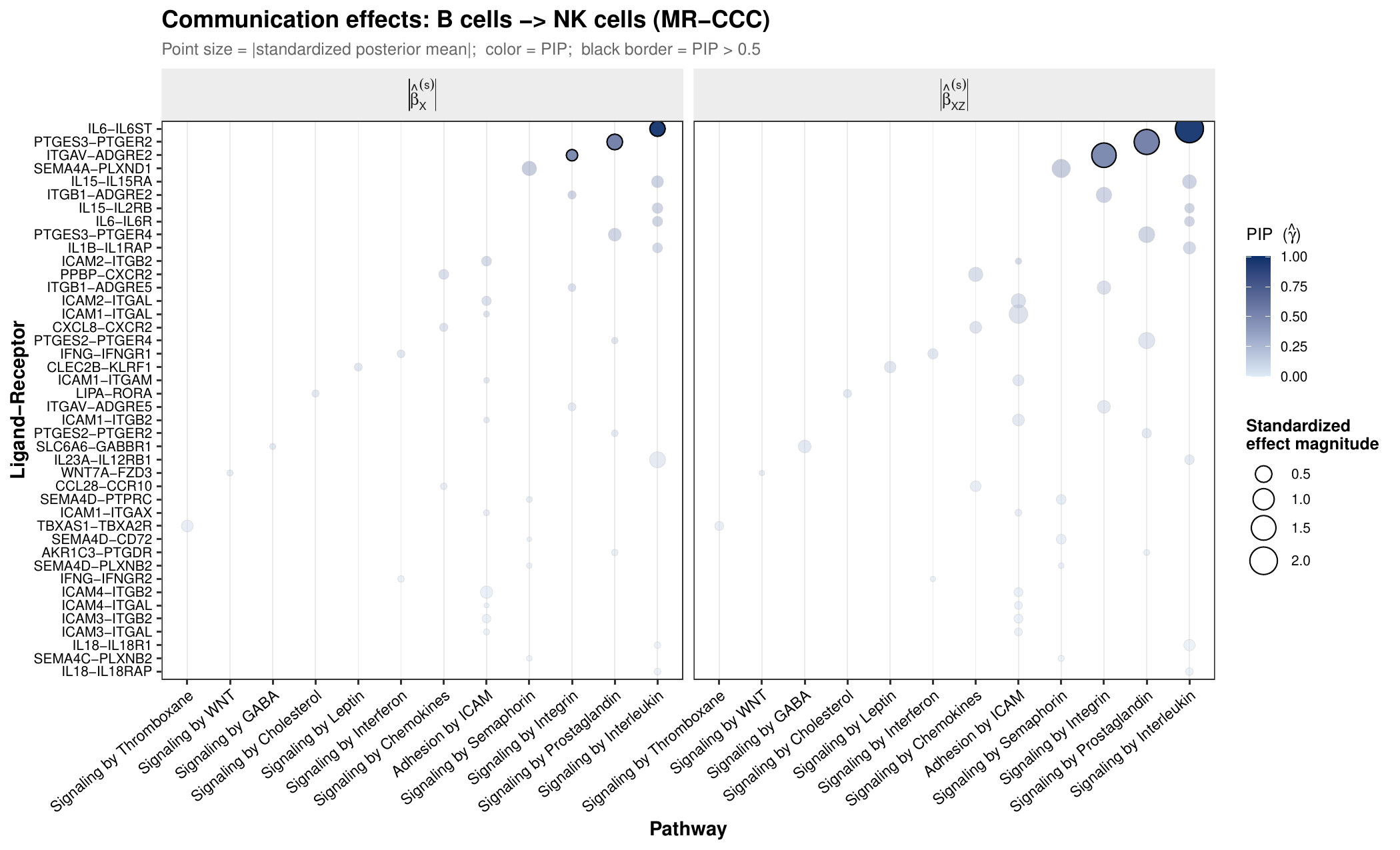}
\caption{\textbf{Standardized posterior effects for the B cells
  \(\rightarrow\) NK cells analysis.} Left panel: absolute main ligand
  effect \(|\hat{\beta}_X^{(s)}|\); right panel: absolute
  receptor-modulated interaction effect
  \(|\hat{\beta}_{XZ}^{(s)}|\). Point size encodes effect magnitude;
  fill color encodes PIP; black borders identify the three discovered
  triplets. All three show dominant interaction components in the
  right panel relative to their main effects in the left panel,
  indicating receptor expression-gated communication effects.}
\label{fig:supp_B_NK_bubble}
\end{figure}

\begin{figure}[htbp]
\centering
\includegraphics[width=\textwidth]{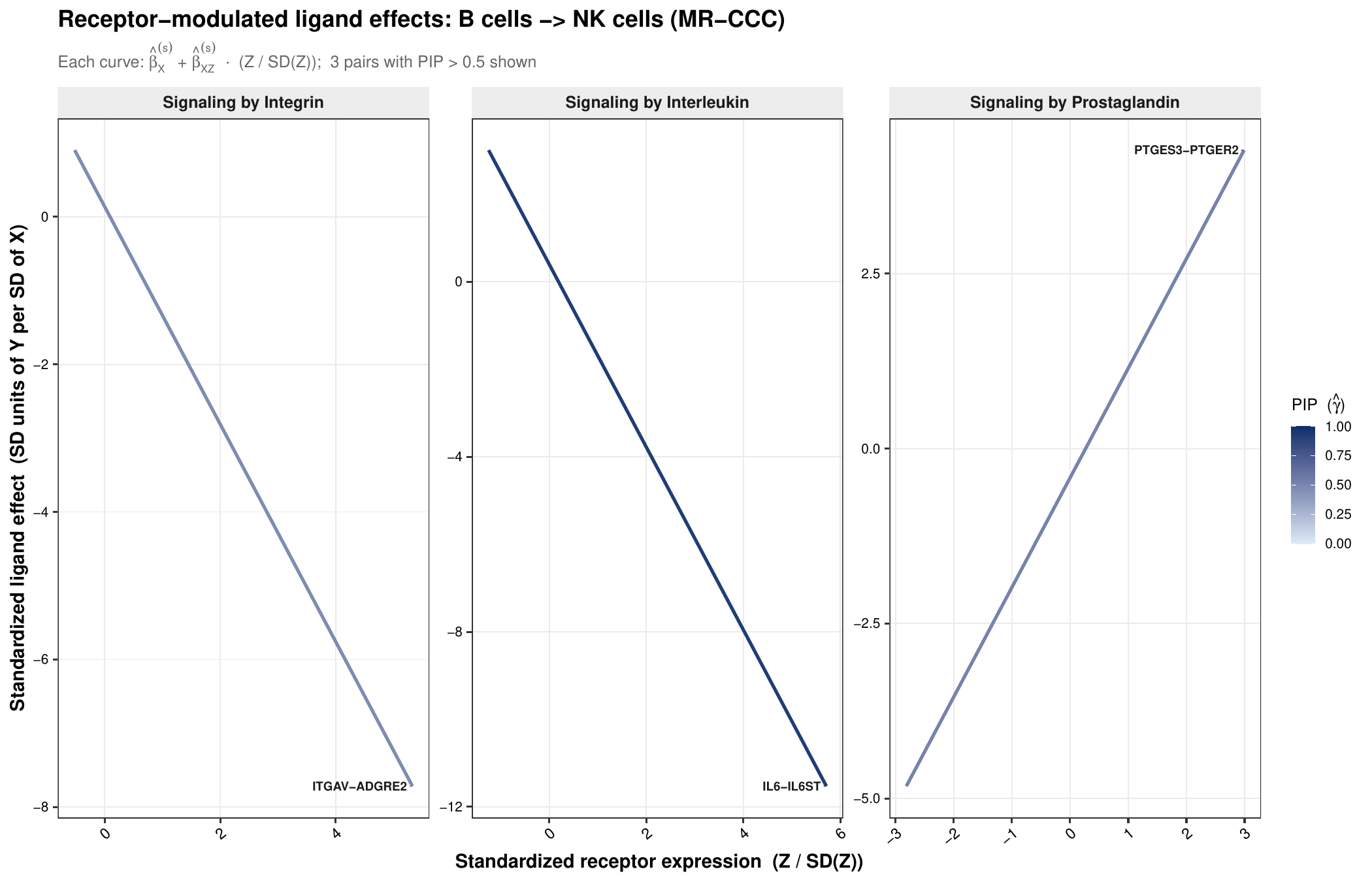}
\caption{\textbf{Receptor-modulated effect curves
  (\(\hat{\beta}_X + \hat{\beta}_{XZ} \cdot Z/\mathrm{SD}(Z)\))
  for the B cells \(\rightarrow\) NK cells analysis.} Only the three
  discovery pairs (PIP \(> 0.5\)) are displayed, grouped into the
  three pathway panels with confirmed discoveries (Signaling by
  Integrin, Signaling by Interleukin, and Signaling by Prostaglandin).
  In the Signaling by Interleukin panel, the
  \textit{IL6}--\textit{IL6ST} curve is steeply decreasing, crossing
  zero near the mean receptor level and becoming strongly negative at
  high \textit{IL6ST} expression. In the Signaling by Prostaglandin
  panel, the \textit{PTGES3}--\textit{PTGER2} curve has a positive
  slope that crosses zero just above the mean \textit{PTGER2} level,
  transitioning from suppression to amplification. In the Signaling
  by Integrin panel, the \textit{ITGAV}--\textit{ADGRE2} curve
  crosses zero near the mean and becomes strongly negative at high
  \textit{ADGRE2} expression.}
\label{fig:supp_B_NK_curves}
\end{figure}

\FloatBarrier

\subsubsection{B Cells \(\rightarrow\) Monocytes}
\label{supp:B_Mono}

Across 665 donors, MR-CCC evaluated 41 ligand--receptor--pathway triplets for the B cell (sender) to Monocyte (receiver) direction and
identified two high-confidence causal communication signals with PIP
exceeding 0.5: \textit{ICAM1}--\textit{ITGAX} within the Adhesion by
ICAM pathway (PIP \(= 0.562\)) and \textit{IL6}--\textit{IL6ST}
within the Interleukin signaling pathway (PIP \(= 0.554\))
(Figures~\ref{fig:supp_B_Mono_pip}--\ref{fig:supp_B_Mono_curves}).
Both discoveries exhibit a strong receptor-modulated interaction
relative to the main ligand effect, and both yield sign-reversing
effect curves near the population-mean receptor expression level.

The top-ranked triplet, \textit{ICAM1}--\textit{ITGAX}
(PIP \(= 0.562\)), implicates B cell intercellular adhesion molecule 1
(ICAM-1, encoded by \textit{ICAM1}) as a causal regulator of monocyte
Adhesion pathway activity through integrin \(\alpha\)X (CD11c, encoded
by \textit{ITGAX}).
ICAM-1 is constitutively expressed on B cells and is further
up-regulated upon activation; it serves as the canonical counter-
receptor for \(\beta_2\) integrins on myeloid
cells~\cite{dustin1986icam}.
CD11c (\(\alpha\)X chain) pairs with \(\beta_2\) (CD18) to form
p150,95 (CR4), a complement receptor expressed at high levels on
classical monocytes and conventional dendritic
cells~\cite{myones1988cd11c}.
The positive main effect (\(\hat{\beta}_X = 0.291\)) and large
negative interaction (\(\hat{\beta}_{XZ} = -2.04\)) yield a curve that crosses zero approximately \(0.14\) standard deviations above the
mean \textit{ITGAX} expression level and becomes strongly negative at
high receptor density.
Monocytes expressing \textit{ITGAX} below this threshold show a
modestly positive Adhesion pathway response to B cell ICAM-1, whereas
monocytes with high CD11c expression---a phenotype associated with a
more differentiated, dendritic cell-like state---experience a
progressively stronger suppressive effect of B cell ICAM-1 engagement
on their Adhesion pathway activity, consistent with contact-mediated
restraint of DC-like monocyte activation through sustained
\(\alpha\)X\(\beta_2\)--ICAM-1 ligation.

The second discovered triplet, \textit{IL6}--\textit{IL6ST}
(PIP \(= 0.554\)), again identifies B cell-derived IL-6 signaling
through gp130 (\textit{IL6ST}) as a causal communication axis,
paralleling the discovery in the B cells \(\rightarrow\) NK cells
direction.
However, the receptor-modulated interaction term here is large and
\emph{positive} (\(\hat{\beta}_{XZ} = 1.35\)), in contrast to the
large \emph{negative} interaction observed for NK cells
(\(\hat{\beta}_{XZ} = -2.09\)).
The modest positive main effect (\(\hat{\beta}_X = 0.153\)) combined
with the positive interaction produces a curve that crosses zero
approximately \(0.11\) standard deviations below the mean
\textit{IL6ST} level, making the effect positive across essentially
the full observed receptor expression range and increasingly so at
high gp130 density.
Monocytes expressing high levels of \textit{IL6ST} therefore show a
strongly amplified stimulatory response to B cell IL-6, consistent
with the well-established role of IL-6/gp130 signaling in promoting
monocyte and macrophage pro-inflammatory activation, acute-phase
responses, and cytokine production~\cite{tanaka2014il6,heinrich2003il6}.
The direction-specific sign of the \textit{IL6ST} interaction
term---activating for monocytes, suppressive for NK cells---reflects
the fundamentally different downstream consequences of sustained gp130
signaling in myeloid versus innate lymphoid cell types.

Notably, several triplets received near-zero PIPs.
\textit{IFNG}--\textit{IFNGR1} (PIP \(= 0.047\)) and
\textit{IFNG}--\textit{IFNGR2} (PIP \(= 0.032\)) remained at the
bottom of the distribution, again consistent with B cells not being a
physiological source of IFN-\(\gamma\) in peripheral
blood~\cite{schroder2004interferon}.
\textit{IL15}--\textit{IL15RA} (PIP \(= 0.143\)) and
\textit{IL15}--\textit{IL2RB} (PIP \(= 0.167\)) both fell well below
the discovery threshold, consistent with monocytes not being primary
targets of IL-15-mediated homeostatic signaling; IL-15 acts
principally on T and NK cells rather than on circulating
monocytes~\cite{waldmann2006il15}.

\begin{figure}[htbp]
\centering
\includegraphics[width=\textwidth]{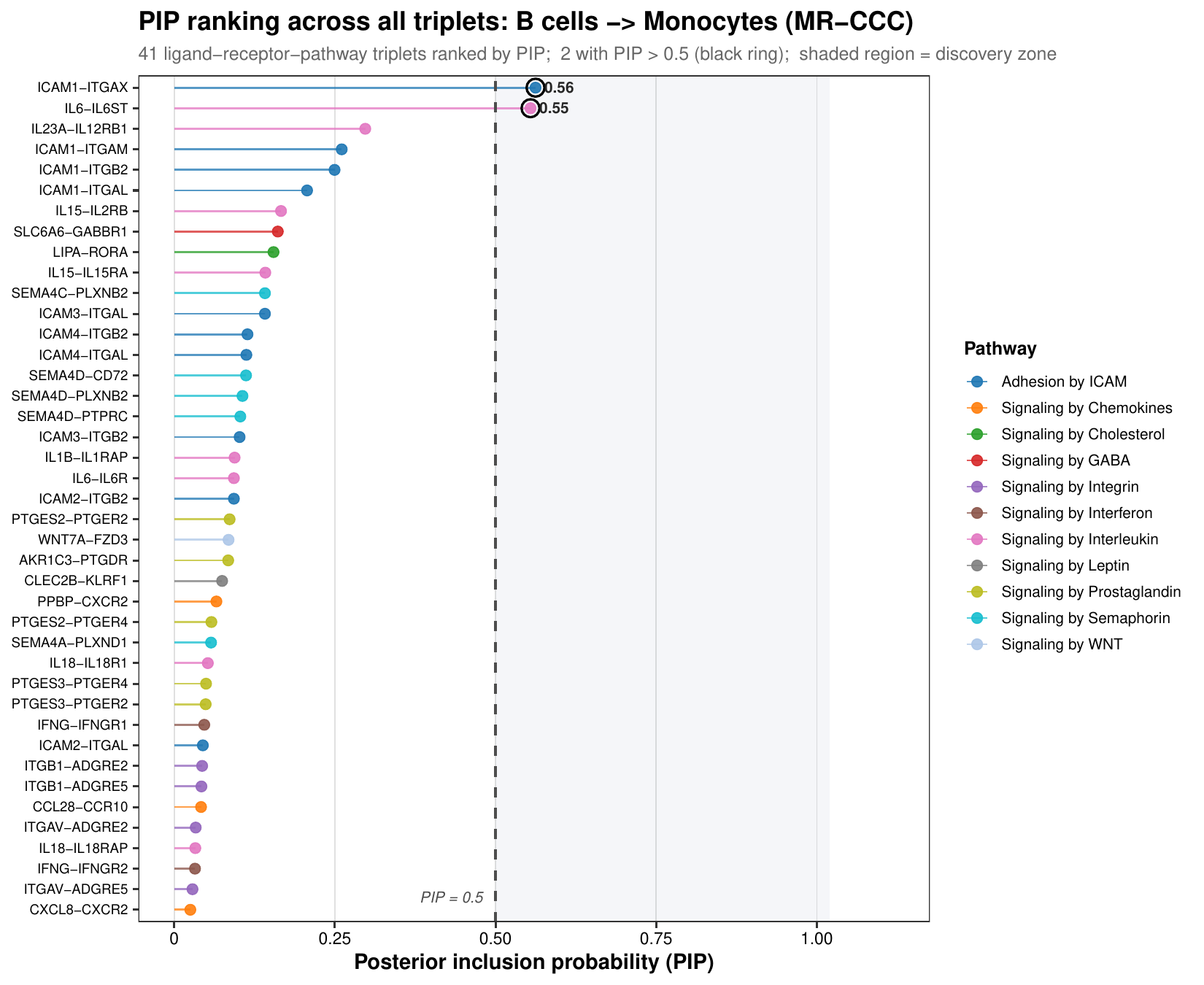}
\caption{\textbf{PIP ranking for the B cells\,$\rightarrow$\,Monocytes analysis.} All 41 ligand--receptor--pathway triplets across 665 donors.
  Points are colored by pathway; the dashed vertical line marks the
  discovery threshold of PIP \(= 0.5\); the shaded region to the
  right is the discovery zone. Black rings identify the two
  discovered triplets: \textit{ICAM1}--\textit{ITGAX}
  (PIP \(= 0.56\)) and \textit{IL6}--\textit{IL6ST}
  (PIP \(= 0.55\)).}
\label{fig:supp_B_Mono_pip}
\end{figure}

\begin{figure}[htbp]
\centering
\includegraphics[width=\textwidth]{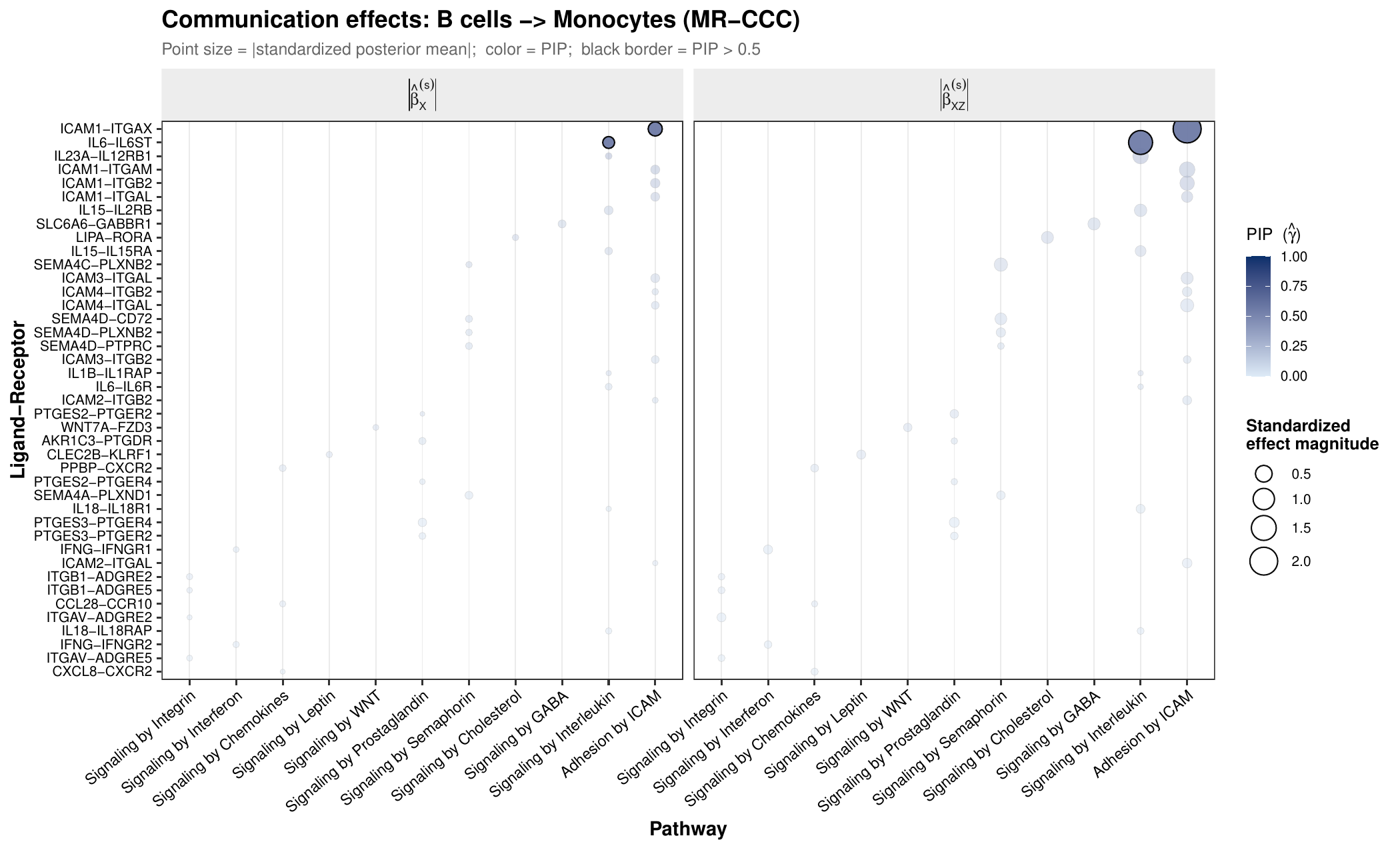}
\caption{\textbf{Standardized posterior effects for the B cells
  \(\rightarrow\) Monocytes analysis.} Left panel: absolute main
  ligand effect \(|\hat{\beta}_X^{(s)}|\); right panel: absolute
  receptor-modulated interaction effect
  \(|\hat{\beta}_{XZ}^{(s)}|\). Point size encodes effect magnitude;
  fill color encodes PIP; black borders identify the two discovered
  triplets. Both \textit{ICAM1}--\textit{ITGAX} and
  \textit{IL6}--\textit{IL6ST} show substantially larger interaction
  components in the right panel than main effects in the left panel,
  indicating interaction-dominated communication effects.}
\label{fig:supp_B_Mono_bubble}
\end{figure}

\begin{figure}[htbp]
\centering
\includegraphics[width=\textwidth]{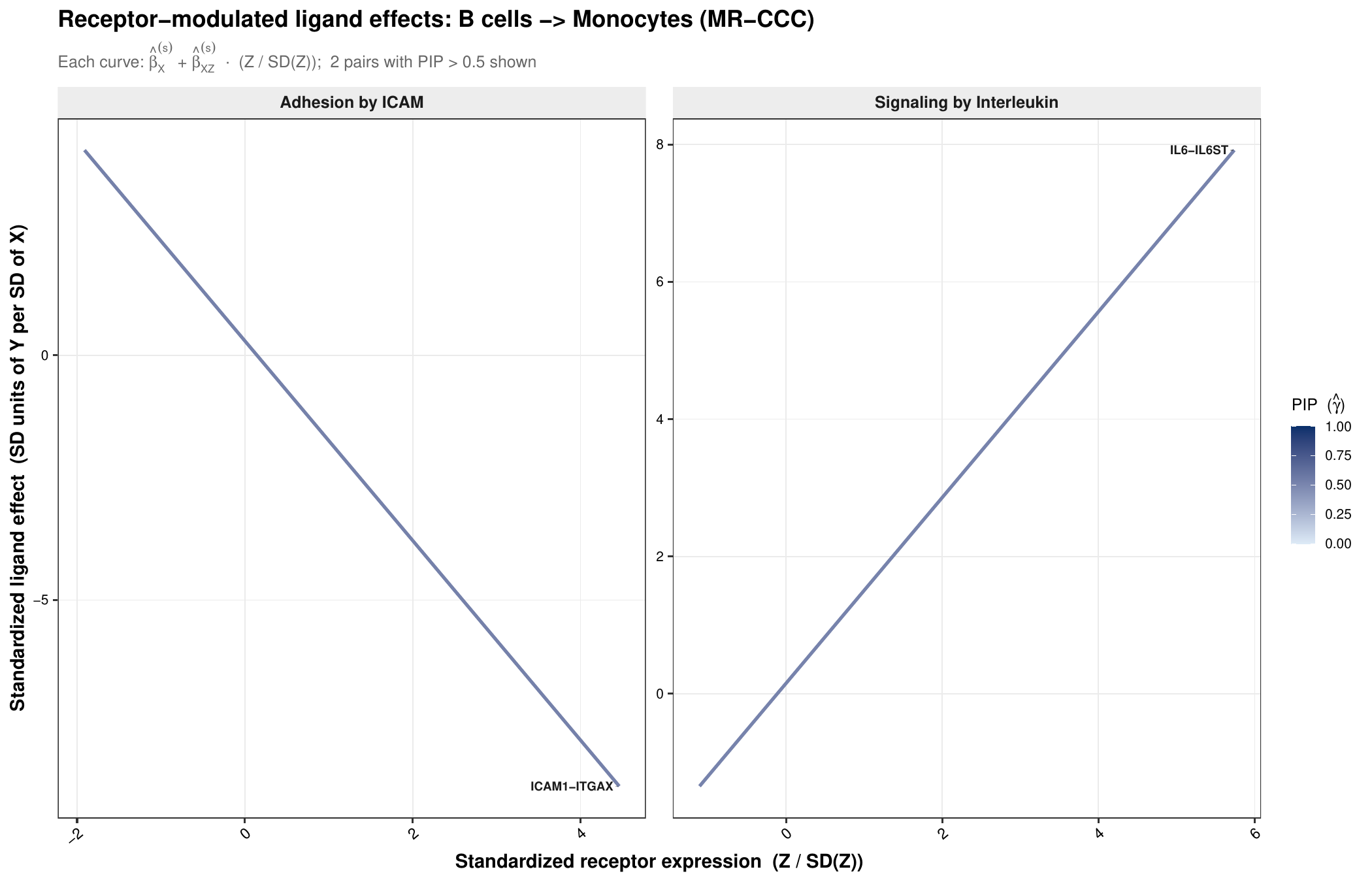}
\caption{\textbf{Receptor-modulated effect curves
  (\(\hat{\beta}_X + \hat{\beta}_{XZ} \cdot Z/\mathrm{SD}(Z)\))
  for the B cells \(\rightarrow\) Monocytes analysis.} Only the two
  discovery pairs (PIP \(> 0.5\)) are displayed, grouped into the two
  pathway panels with confirmed discoveries (Adhesion by ICAM and
  Signaling by Interleukin). In the Adhesion by ICAM panel, the
  \textit{ICAM1}--\textit{ITGAX} curve is steeply decreasing,
  crossing zero just above the mean \textit{ITGAX} level and becoming
  strongly negative at high CD11c expression. In the Signaling by
  Interleukin panel, the \textit{IL6}--\textit{IL6ST} curve has a
  positive slope crossing zero just below the mean receptor level,
  producing a monotonically increasing effect that is strongly
  positive at high \textit{IL6ST} expression.}
\label{fig:supp_B_Mono_curves}
\end{figure}

\FloatBarrier

\clearpage
\subsection{CD4\texorpdfstring{\(^+\)}{+} T Cells as Sender}

\subsubsection{CD4\(^+\) T Cells \(\rightarrow\) B Cells}
\label{supp:CD4_B}
Across 868 donors, MR-CCC evaluated 42 ligand--receptor--pathway
triplets for the CD4\(^+\) T cell (sender) to B cell (receiver)
direction and identified ten high-confidence causal communication
signals with PIP exceeding 0.5
(Figures~\ref{fig:supp_CD4_B_pip}--\ref{fig:supp_CD4_B_curves}),
the largest number of discoveries in any sender--receiver direction analyzed here.
The ten triplets span four pathway classes: six Adhesion by ICAM
pairs (\textit{ICAM1}--\textit{ITGAL}, PIP \(= 0.658\);
\textit{ICAM1}--\textit{ITGAM}, PIP \(= 0.639\);
\textit{ICAM2}--\textit{ITGAL}, PIP \(= 0.947\);
\textit{ICAM2}--\textit{ITGB2}, PIP \(= 0.749\);
\textit{ICAM3}--\textit{ITGAL}, PIP \(= 0.504\);
\textit{ICAM4}--\textit{ITGAL}, PIP \(= 0.593\)), two Signaling by
Integrin pairs (\textit{ITGB1}--\textit{ADGRE5}, PIP \(= 0.992\);
\textit{ITGB1}--\textit{ADGRE2}, PIP \(= 0.606\)), one Signaling by
Interleukin triplet (\textit{IL23A}--\textit{IL12RB1},
PIP \(= 0.798\)), and one Signaling by Thromboxane triplet
(\textit{TBXAS1}--\textit{TBXA2R}, PIP \(= 0.595\)).

\noindent\textbf{Adhesion by ICAM signals.}
The six discovered ICAM--integrin triplets collectively reflect the
central role of adhesion molecule contacts at the CD4\(^+\) T
cell--B cell immunological synapse, where ICAM family members on T
cells engage \(\beta_2\) integrins on B cells to stabilize
antigen-recognition and co-stimulatory
interactions~\cite{springer1990adhesion,dustin1986icam}.
Despite sharing the same physical interface, the six pairs exhibit
strikingly heterogeneous receptor-modulated interaction terms,
producing a characteristic fan of effect curves in the Adhesion by
ICAM panel (Figure~\ref{fig:supp_CD4_B_curves}).
\textit{ICAM2}--\textit{ITGAL} (PIP \(= 0.947\)) and
\textit{ICAM1}--\textit{ITGAM} (PIP \(= 0.639\)) both show
sign-reversing curves: negative main effects
(\(\hat{\beta}_X = -0.368\) and \(-0.614\), respectively) combined
with large positive interactions
(\(\hat{\beta}_{XZ} = 2.59\) and \(2.25\)) mean that B cells
expressing LFA-1 (\textit{ITGAL}) or Mac-1 (\textit{ITGAM}) above
approximately 0.14 and 0.27 standard deviations above the population
mean transition from a suppressed to a progressively stimulated
Adhesion pathway response to T cell ICAM engagement.
\textit{ICAM1}--\textit{ITGAL} (PIP \(= 0.658\),
\(\hat{\beta}_X = -0.740\), \(\hat{\beta}_{XZ} = 0.321\)) shows a
predominantly suppressive curve with a modest positive interaction
that partially attenuates this suppression at the highest LFA-1
densities.
\textit{ICAM2}--\textit{ITGB2} (PIP \(= 0.749\),
\(\hat{\beta}_X = -0.346\), \(\hat{\beta}_{XZ} = -0.241\)) is
monotonically suppressive, with both a negative main effect and a
negative interaction delivering a consistent down-regulation of B
cell Adhesion pathway activity regardless of \(\beta_2\) expression
level.
In contrast, \textit{ICAM3}--\textit{ITGAL} (PIP \(= 0.504\),
\(\hat{\beta}_X = 0.059\), \(\hat{\beta}_{XZ} = -4.79\)) and
\textit{ICAM4}--\textit{ITGAL} (PIP \(= 0.593\),
\(\hat{\beta}_X = -0.017\), \(\hat{\beta}_{XZ} = -5.65\)) carry the
steepest negative interaction terms in the entire analysis, with
near-zero main effects and dominant negative interactions that
produce strongly receptor-amplified suppression: B cells with high
LFA-1 expression experience a steeply increasing suppression of
Adhesion pathway activity upon contact with T cell ICAM3 or ICAM4.
The co-discovery of stimulatory and suppressive ICAM--LFA-1 signals
through different ICAM ligands is consistent with the architectural organization of the immunological synapse, where different ICAM
family members occupy distinct supramolecular activation cluster
domains and fulfill complementary activation and negative-feedback roles~\cite{springer1990adhesion}.

\noindent\textbf{Integrin--ADGRE signals.}
The highest-PIP discovery in this direction,
\textit{ITGB1}--\textit{ADGRE5} (PIP \(= 0.992\)), implicates CD4\(^+\)
T cell \(\beta_1\) integrin (\textit{ITGB1}) as a causal activator of
B cell Integrin pathway activity through the adhesion G
protein-coupled receptor CD97 (encoded by \textit{ADGRE5}).
CD97 is expressed on resting and activated B cells and serves as a
multifunctional adhesion receptor whose extracellular EGF-like domains
mediate cell--cell contact~\cite{hamann1996cd97}.
Both the main effect (\(\hat{\beta}_X = 0.498\)) and the interaction
(\(\hat{\beta}_{XZ} = 3.63\)) are positive, yielding a monotonically
increasing curve that crosses zero approximately \(0.14\) standard
deviations below the mean \textit{ADGRE5} level: across essentially
the full observed receptor expression range, T cell \(\beta_1\)
integrin engagement potentiates B cell Integrin pathway activity, with
this potentiation growing steeply in B cells expressing high CD97.
The companion triplet \textit{ITGB1}--\textit{ADGRE2}
(PIP \(= 0.606\), \(\hat{\beta}_X = 0.309\), \(\hat{\beta}_{XZ} = 1.48\))
shows a parallel, though weaker, stimulatory and interaction-amplified
pattern through the related adhesion receptor EMR2 (encoded by
\textit{ADGRE2}), which also carries EGF-like counter-receptor
domains~\cite{stacey2003emr2}.

\noindent\textbf{Interleukin signal.}
The \textit{IL23A}--\textit{IL12RB1} triplet (PIP \(= 0.798\))
identifies CD4\(^+\) T cell-derived IL-23 subunit \(\alpha\)
(\textit{IL23A}, the p19 subunit) as a causal regulator of B cell
Interleukin pathway activity through the IL-12R\(\beta_1\) chain
(encoded by \textit{IL12RB1}).
IL-23 is assembled from IL-23A and the shared p40 subunit (IL-12B),
and its cognate receptor complex on B cells incorporates
IL-12R\(\beta_1\) as the signal-transducing shared
chain~\cite{oppmann2000il23}.
The modest positive main effect (\(\hat{\beta}_X = 0.186\)) and large
positive interaction (\(\hat{\beta}_{XZ} = 3.16\)) produce a curve
that crosses zero approximately \(0.06\) standard deviations below the
mean \textit{IL12RB1} level, making the stimulatory effect positive
across virtually the entire observed receptor expression range and
strongly amplified in B cells with high IL-12R\(\beta_1\) expression.

\noindent\textbf{Thromboxane signal.}
\textit{TBXAS1}--\textit{TBXA2R} (PIP \(= 0.595\)) links thromboxane
A\(_2\) synthase 1 (\textit{TBXAS1}) expressed in CD4\(^+\) T cells to
the thromboxane A\(_2\) receptor (TP receptor, encoded by
\textit{TBXA2R}) on B cells.
Activated CD4\(^+\) T cells produce thromboxane A\(_2\) (TXA\(_2\)) via
TBXAS1, and TXA\(_2\) signaling through TBXA2R on lymphocytes regulates
cell activation, migration, and cytokine
production~\cite{nakahata2008tbxa2r}.
Both the main effect (\(\hat{\beta}_X = -0.620\)) and the
receptor-modulated interaction (\(\hat{\beta}_{XZ} = -5.47\)) are
negative, yielding a monotonically and steeply decreasing curve that
crosses zero approximately \(0.11\) standard deviations below the mean
\textit{TBXA2R} level.
B cells expressing \textit{TBXA2R} at or above the population mean
therefore experience an increasingly strong suppression of their own
Thromboxane pathway activity as CD4\(^+\) T cell TXA\(_2\) production
rises, consistent with a T cell-mediated negative-feedback loop that
dampens B cell autologous thromboxane signaling in proportion to TP
receptor density.

\noindent\textbf{Low-PIP pairs.}
\textit{IFNG}--\textit{IFNGR1} (PIP \(= 0.144\)) and
\textit{IFNG}--\textit{IFNGR2} (PIP \(= 0.123\)) both fell below the
discovery threshold.
Unlike the B cell sender direction---where the low PIPs reflect the
fact that B cells are not a physiological IFN-\(\gamma\) source---the
low PIPs here reflect the cellular heterogeneity of the CD4\(^+\)
T cell pool: only the Th1 subset produces IFN-\(\gamma\), and the
population-averaged causal effect across mixed CD4\(^+\) T cell states
is insufficient to clear the discovery
threshold~\cite{schroder2004interferon}.
\textit{IL6}--\textit{IL6ST} (PIP \(= 0.079\)) and
\textit{IL6}--\textit{IL6R} (PIP \(= 0.091\)) received near-zero
support, indicating that CD4\(^+\) T cell IL-6 production does not
constitute a strong causal signal to B cells via gp130 in this dataset.

\begin{figure}[htbp]
\centering
\includegraphics[width=\textwidth]{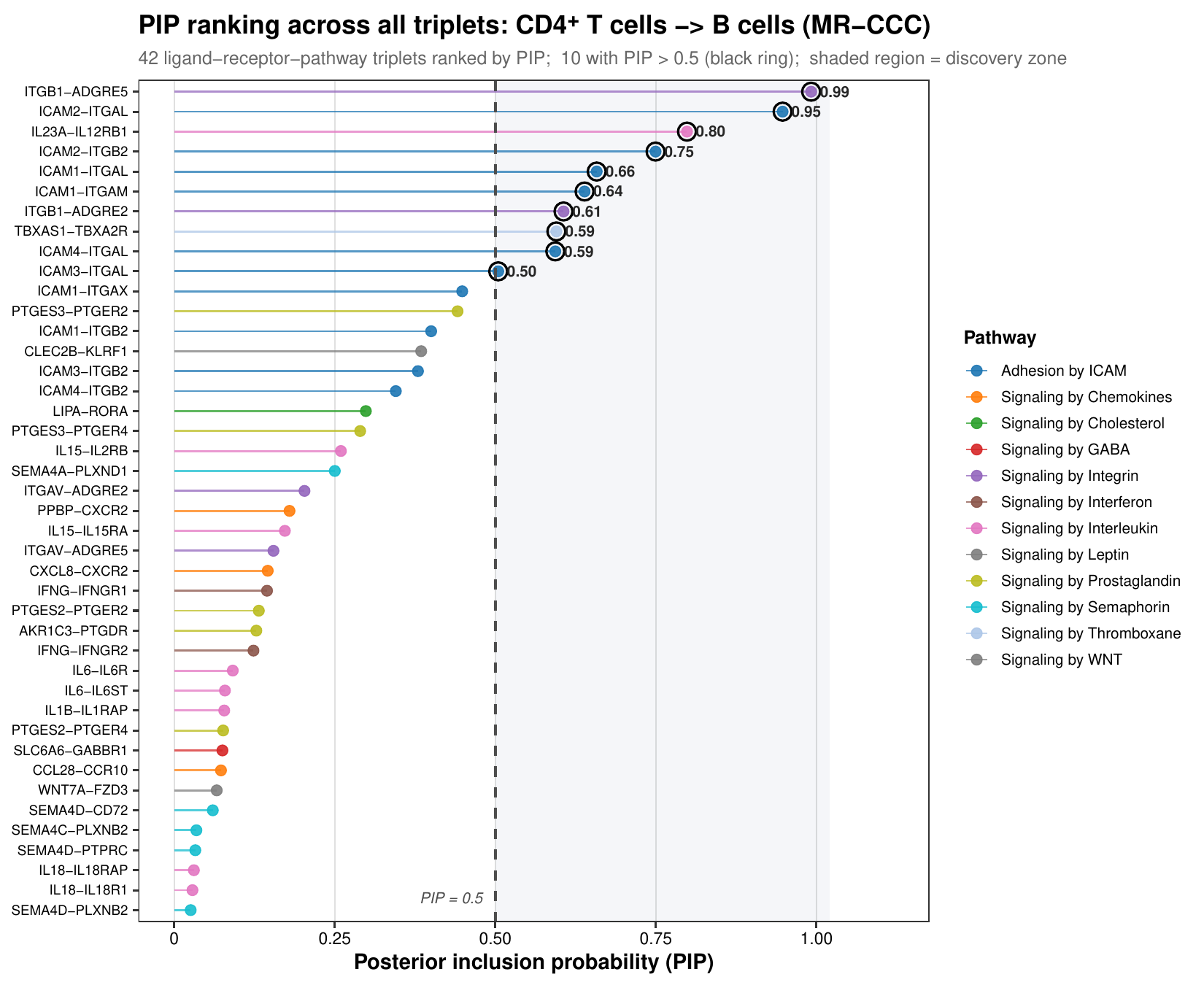}
\caption{\textbf{PIP ranking for the CD4$^+$ T cells\,$\rightarrow$\,B cells analysis.} All 42 ligand--receptor--pathway triplets across 868
  donors. Points are colored by pathway; the dashed vertical line
  marks the discovery threshold of PIP \(= 0.5\); the shaded region
  to the right is the discovery zone. Black rings identify the ten
  discovered triplets spanning Adhesion by ICAM
  (\textit{ICAM1}--\textit{ITGAL}, \textit{ICAM1}--\textit{ITGAM},
  \textit{ICAM2}--\textit{ITGAL}, \textit{ICAM2}--\textit{ITGB2},
  \textit{ICAM3}--\textit{ITGAL}, \textit{ICAM4}--\textit{ITGAL}),
  Signaling by Integrin (\textit{ITGB1}--\textit{ADGRE5},
  \textit{ITGB1}--\textit{ADGRE2}), Signaling by Interleukin
  (\textit{IL23A}--\textit{IL12RB1}), and Signaling by Thromboxane
  (\textit{TBXAS1}--\textit{TBXA2R}) pathways.}
\label{fig:supp_CD4_B_pip}
\end{figure}

\begin{figure}[htbp]
\centering
\includegraphics[width=\textwidth]{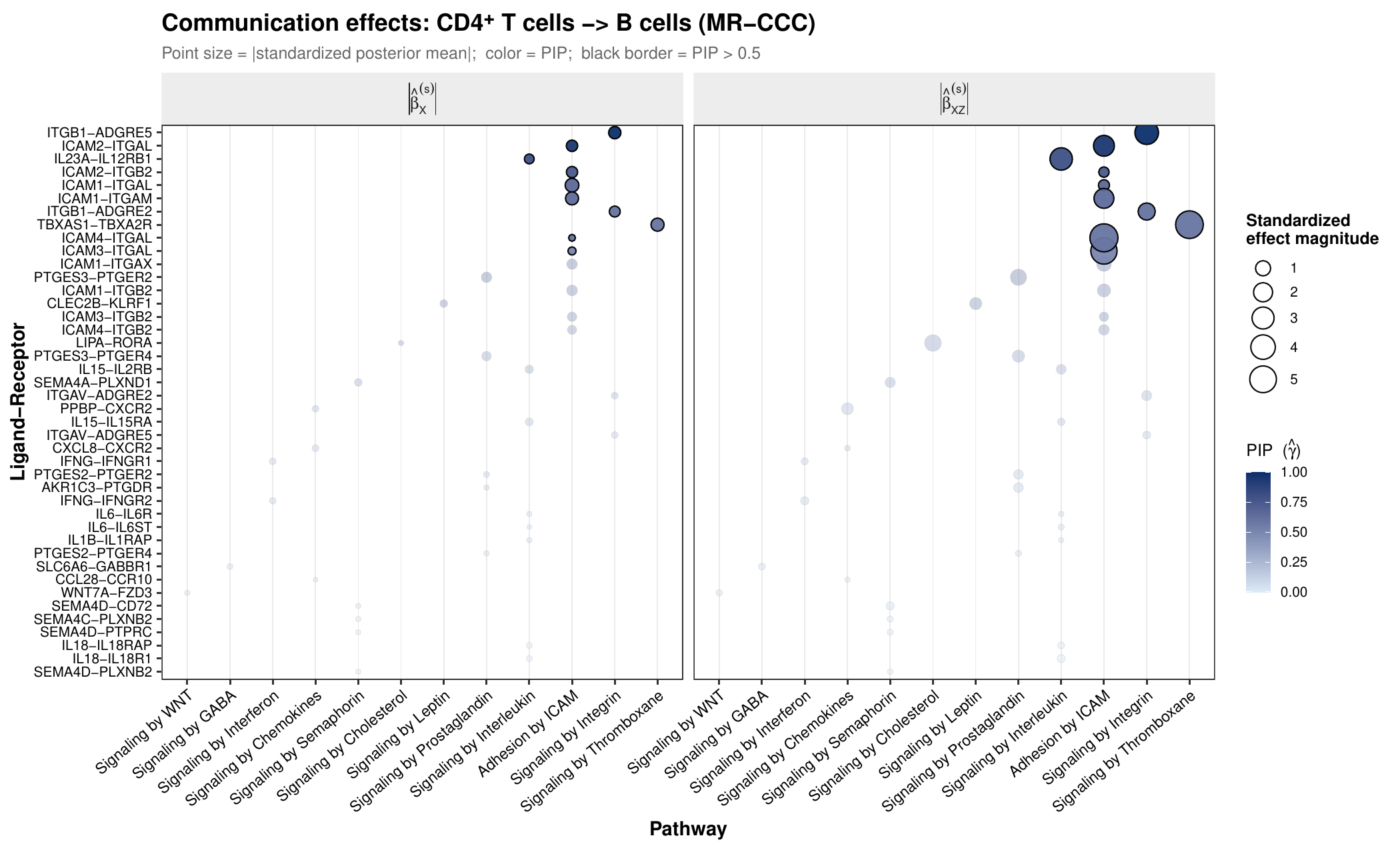}
\caption{\textbf{Standardized posterior effects for the CD4\(^+\) T cells
  \(\rightarrow\) B cells analysis.} Left panel: absolute main ligand
  effect \(|\hat{\beta}_X^{(s)}|\); right panel: absolute
  receptor-modulated interaction effect \(|\hat{\beta}_{XZ}^{(s)}|\).
  Point size encodes effect magnitude; fill color encodes PIP; black
  borders identify the ten discovered triplets. The interaction panel
  (right) reveals large \(\hat{\beta}_{XZ}\) magnitudes for both the
  ICAM--integrin cluster and the \textit{ITGB1}--\textit{ADGRE5},
  \textit{TBXAS1}--\textit{TBXA2R}, and
  \textit{IL23A}--\textit{IL12RB1} triplets, indicating that receptor
  expression on B cells dominates the modulation of each effect.}
\label{fig:supp_CD4_B_bubble}
\end{figure}

\begin{figure}[htbp]
\centering
\includegraphics[width=\textwidth]{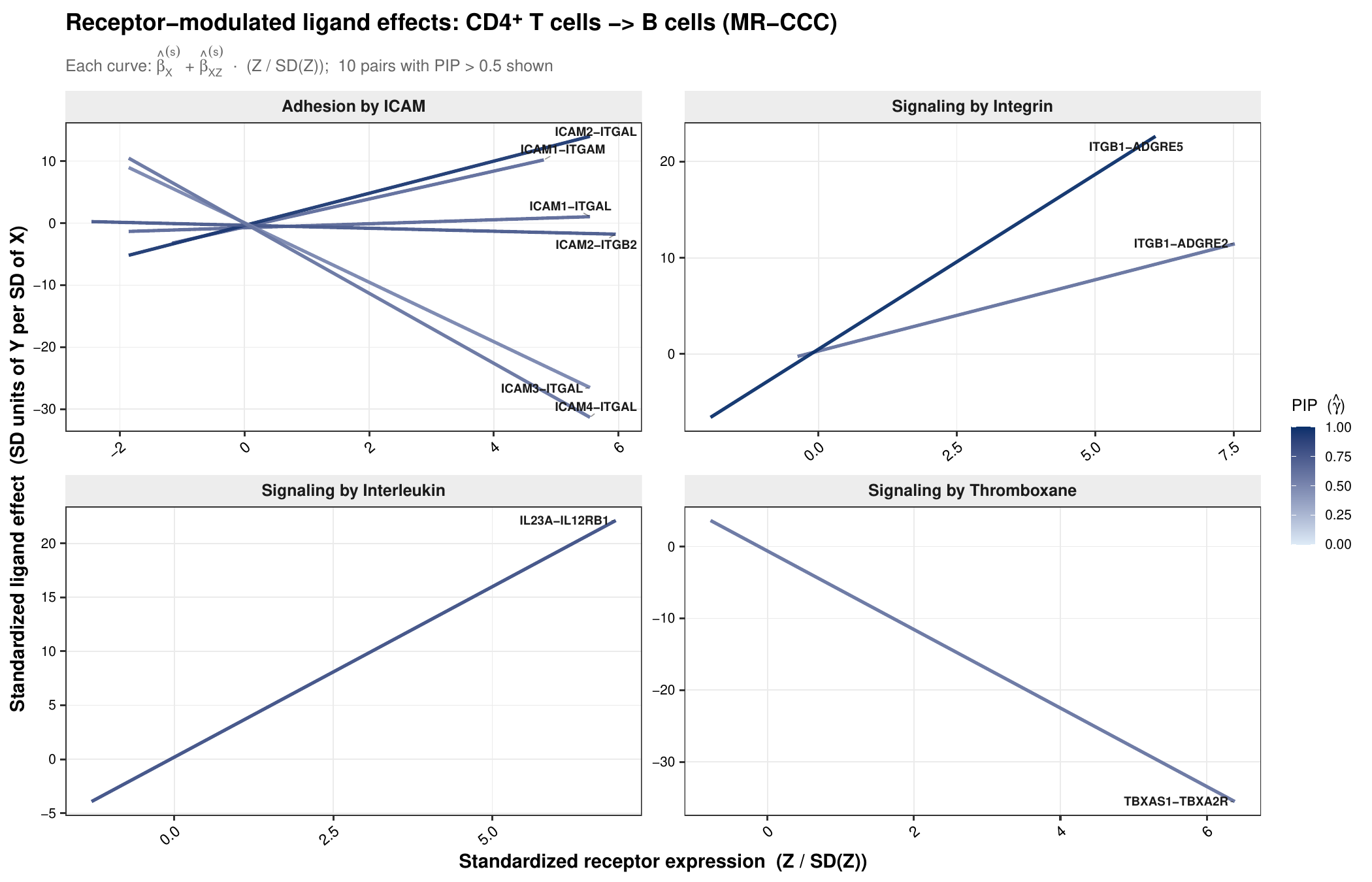}
\caption{\textbf{Receptor-modulated effect curves
  (\(\hat{\beta}_X + \hat{\beta}_{XZ} \cdot Z/\mathrm{SD}(Z)\)) for
  the CD4\(^+\) T cells \(\rightarrow\) B cells analysis.} Only the
  ten discovery pairs (PIP \(> 0.5\)) are displayed, grouped into the
  four pathway panels with confirmed discoveries (Adhesion by ICAM,
  Signaling by Integrin, Signaling by Interleukin, and Signaling by
  Thromboxane). The Adhesion by ICAM panel contains six curves forming
  a divergent fan: \textit{ICAM2}--\textit{ITGAL} and
  \textit{ICAM1}--\textit{ITGAM} cross from negative to positive near
  the mean receptor level, while \textit{ICAM3}--\textit{ITGAL} and
  \textit{ICAM4}--\textit{ITGAL} show the steepest negative slopes in
  the analysis. In the Signaling by Integrin panel,
  \textit{ITGB1}--\textit{ADGRE5} and \textit{ITGB1}--\textit{ADGRE2}
  both show steep positive-slope curves. In the Signaling by
  Interleukin panel, \textit{IL23A}--\textit{IL12RB1} rises steeply
  from near zero at the mean. In the Signaling by Thromboxane panel,
  \textit{TBXAS1}--\textit{TBXA2R} is the steepest declining curve,
  becoming strongly negative at high \textit{TBXA2R} expression.}
\label{fig:supp_CD4_B_curves}
\end{figure}
\FloatBarrier

\subsubsection{CD4\(^+\) T Cells \(\rightarrow\) CD8\(^+\) T Cells}
\label{supp:CD4_CD8}
Across 882 donors, MR-CCC evaluated 42 ligand--receptor--pathway
triplets for the CD4\(^+\) T cell (sender) to CD8\(^+\) T cell (receiver)
direction and identified three high-confidence causal communication
signals with PIP exceeding 0.5: \textit{IL15}--\textit{IL2RB} within
the Interleukin signaling pathway (PIP \(= 0.668\)),
\textit{ICAM2}--\textit{ITGAL} within the Adhesion by ICAM pathway
(PIP \(= 0.547\)), and \textit{SEMA4D}--\textit{PTPRC} within the
Semaphorin signaling pathway (PIP \(= 0.538\))
(Figures~\ref{fig:supp_CD4_CD8_pip}--\ref{fig:supp_CD4_CD8_curves}).

The top-ranked discovery, \textit{IL15}--\textit{IL2RB}
(PIP \(= 0.668\)), implicates CD4\(^+\) T cell-derived IL-15 as a
causal activator of CD8\(^+\) T cell Interleukin pathway activity
through the IL-2R\(\beta\) chain (encoded by \textit{IL2RB}; also known
as CD122), which is shared between the IL-15 and IL-2 receptor
complexes and supports homeostatic proliferation and memory maintenance
in CD8\(^+\) T cells~\cite{waldmann2006il15,kennedy2000il15,schluns2003review}.
Both the main effect (\(\hat{\beta}_X = 0.770\)) and the
receptor-modulated interaction (\(\hat{\beta}_{XZ} = 0.614\)) are
positive, yielding a monotonically increasing effect curve: the
estimated causal effect is positive across the full observed range of
\textit{IL2RB} expression, with the zero crossing located at
\(Z^* \approx -1.25\) standard deviations below the mean---well outside
the bulk of the observed distribution.
CD8\(^+\) T cells expressing higher levels of the IL-2R\(\beta\) chain
therefore show a uniformly stronger and progressively amplified positive
response to CD4\(^+\) T cell IL-15 production.

The second discovered triplet, \textit{ICAM2}--\textit{ITGAL}
(PIP \(= 0.547\)), links CD4\(^+\) T cell intercellular adhesion
molecule 2 (\textit{ICAM2}) to LFA-1 (\textit{ITGAL}) on CD8\(^+\) T
cells. ICAM--LFA-1 contacts are a central stabilizing force at the T
cell--T cell immunological interface, coordinating antigen-dependent and
helper-driven cytotoxic T cell
activation~\cite{springer1990adhesion,dustin1986icam}.
The negative main effect (\(\hat{\beta}_X = -0.403\)) combined with a
large positive interaction (\(\hat{\beta}_{XZ} = 2.17\)) produces a
sign-reversing effect curve that crosses zero approximately \(0.19\)
standard deviations above the mean \textit{ITGAL} level.
CD8\(^+\) T cells expressing \textit{ITGAL} below this threshold
experience a net suppressive effect of CD4\(^+\) T cell ICAM2
engagement on their Adhesion pathway activity; above it, the direction
flips to increasingly stimulatory, consistent with a context-dependent
adhesion signal in which LFA-1 density gates the transition between
inhibitory and activating CD4--CD8 contact.

The third discovery, \textit{SEMA4D}--\textit{PTPRC}
(PIP \(= 0.538\)), implicates CD100 (semaphorin 4D, encoded by
\textit{SEMA4D}), a class~IV transmembrane semaphorin expressed on
activated CD4\(^+\) T cells, as a causal regulator of CD8\(^+\) T cell
Semaphorin pathway activity through the receptor tyrosine phosphatase
CD45 (encoded by \textit{PTPRC}).
SEMA4D is known to modulate immune cell activation through engagement
of plexin and co-receptor signaling complexes on lymphocytes, with
established roles in T cell cross-talk and leukocyte
recruitment~\cite{suzuki2008sema4d}.
The near-zero main effect (\(\hat{\beta}_X = -0.041\)) and strongly
positive interaction (\(\hat{\beta}_{XZ} = 0.907\)) produce a curve
that crosses zero approximately \(0.05\) standard deviations above the
mean \textit{PTPRC} level---essentially at the population mean.
The causal signal is therefore almost entirely interaction-driven: only
in CD8\(^+\) T cells expressing \textit{PTPRC} above the population
mean does CD4\(^+\) T cell SEMA4D exert a meaningful positive effect,
suggesting that CD45 expression level gates the permissiveness of
CD8\(^+\) T cells to semaphorin-mediated signals from helper T cells.

Notably, several triplets did not reach the discovery threshold.
\textit{IL15}--\textit{IL15RA} (PIP \(= 0.473\)) was the highest-ranking
non-discovered signal and a close near-miss; while the shared
IL-2R\(\beta\) chain was discovered (PIP \(= 0.668\)), the
transpresentation axis through IL-15R\(\alpha\) appears less dominant in
this CD4\(^+\) sender direction, consistent with IL-15R\(\alpha\)-mediated
transpresentation being a function more characteristic of antigen-presenting
cells such as B cells and dendritic
cells~\cite{fehniger2001il15,kennedy2000il15}.
\textit{LIPA}--\textit{RORA} (PIP \(= 0.419\)) was also a notable
near-miss; this triplet was discovered in the B cells \(\rightarrow\)
CD4\(^+\) T cells direction (PIP \(= 0.686\)) but does not clear the
threshold here, consistent with CD4\(^+\) T cells expressing lower levels
of lysosomal lipid-catabolic activity than B cells and thus transmitting a
weaker cholesterol-metabolite signal~\cite{yang2008rora}.
\textit{IFNG}--\textit{IFNGR2} (PIP \(= 0.348\)) and
\textit{IFNG}--\textit{IFNGR1} (PIP \(= 0.299\)) received moderate but
sub-threshold support; unlike the B cell sender direction---where
IFN-\(\gamma\) production is biologically implausible---CD4\(^+\) Th1 cells
are a bona fide source of IFN-\(\gamma\), and the moderate PIPs here likely
reflect attenuation of the population-level signal by the cellular
heterogeneity of the CD4\(^+\) T cell pool across
donors~\cite{schroder2004interferon}.

\begin{figure}[htbp]
\centering
\includegraphics[width=\textwidth]{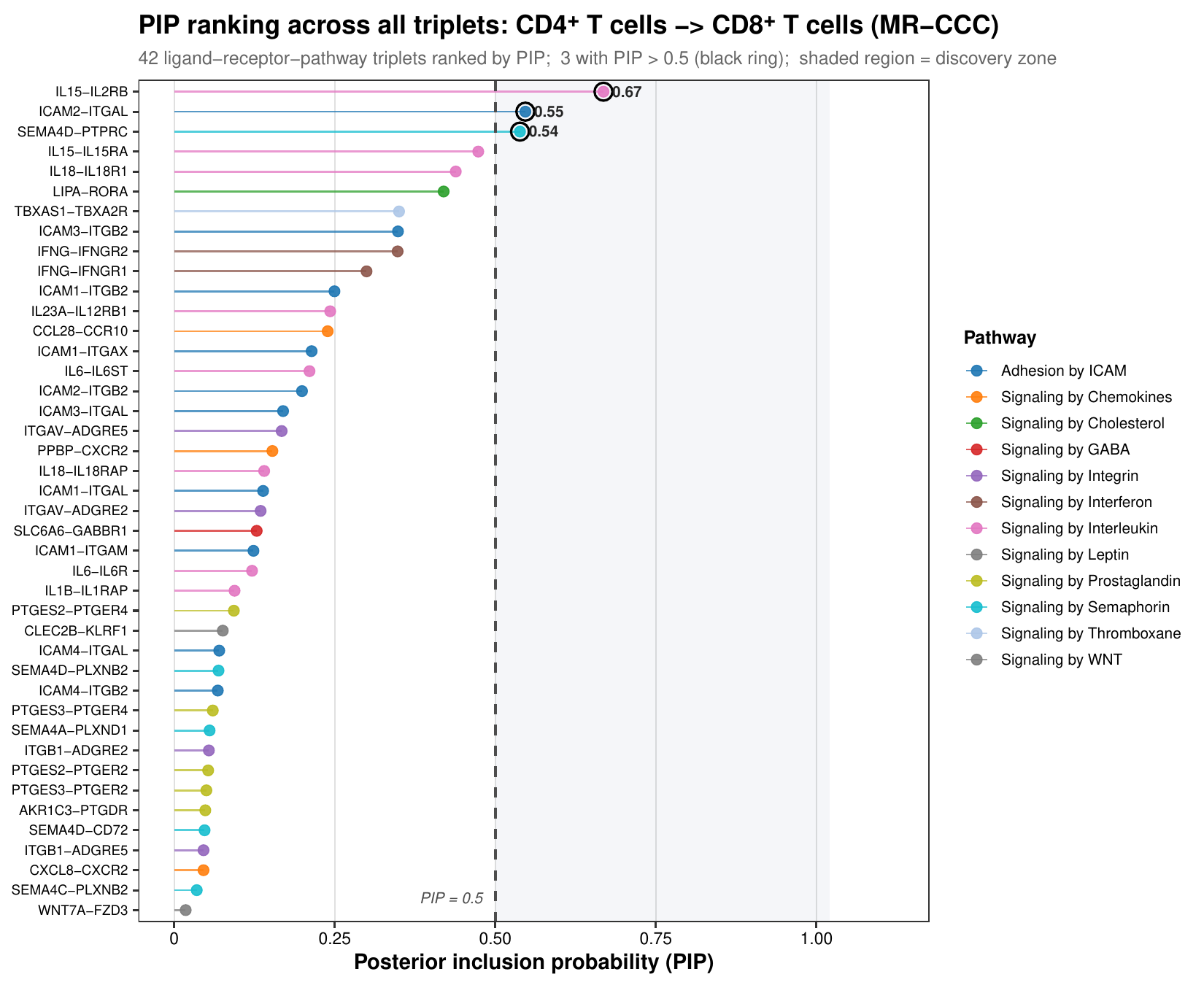}
\caption{\textbf{PIP ranking for the CD4$^+$ T cells\,$\rightarrow$\,CD8\(^+\) T cells analysis.} All 42 ligand--receptor--pathway triplets across 882 donors. Points are colored by pathway; the dashed vertical
  line marks the discovery threshold of PIP \(= 0.5\); the shaded region
  to the right is the discovery zone. Black rings identify the three
  discovered triplets: \textit{IL15}--\textit{IL2RB}
  (PIP \(= 0.67\)), \textit{ICAM2}--\textit{ITGAL}
  (PIP \(= 0.55\)), and \textit{SEMA4D}--\textit{PTPRC}
  (PIP \(= 0.54\)).}
\label{fig:supp_CD4_CD8_pip}
\end{figure}

\begin{figure}[htbp]
\centering
\includegraphics[width=\textwidth]{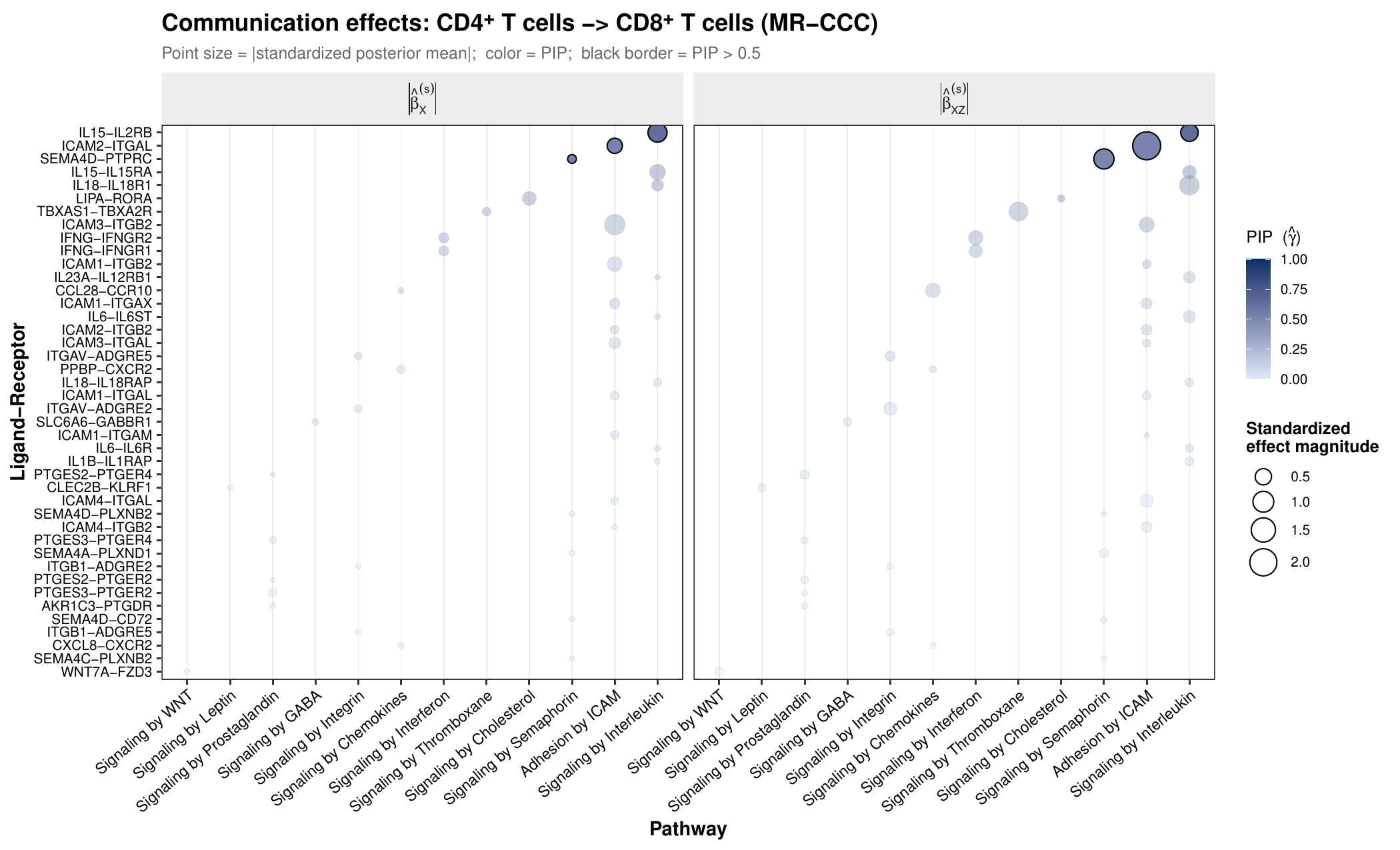}
\caption{\textbf{Standardized posterior effects for the CD4\(^+\) T cells
  \(\rightarrow\) CD8\(^+\) T cells analysis.} Left panel: absolute
  main ligand effect \(|\hat{\beta}_X^{(s)}|\); right panel: absolute
  receptor-modulated interaction effect \(|\hat{\beta}_{XZ}^{(s)}|\).
  Point size encodes effect magnitude; fill color encodes PIP; black
  borders identify the three discovered triplets.
  \textit{IL15}--\textit{IL2RB} shows appreciable effects in both
  panels; \textit{ICAM2}--\textit{ITGAL} and
  \textit{SEMA4D}--\textit{PTPRC} display dominant interaction
  components relative to their main effects, indicating
  receptor-expression-gated communication signals.}
\label{fig:supp_CD4_CD8_bubble}
\end{figure}

\begin{figure}[htbp]
\centering
\includegraphics[width=\textwidth]{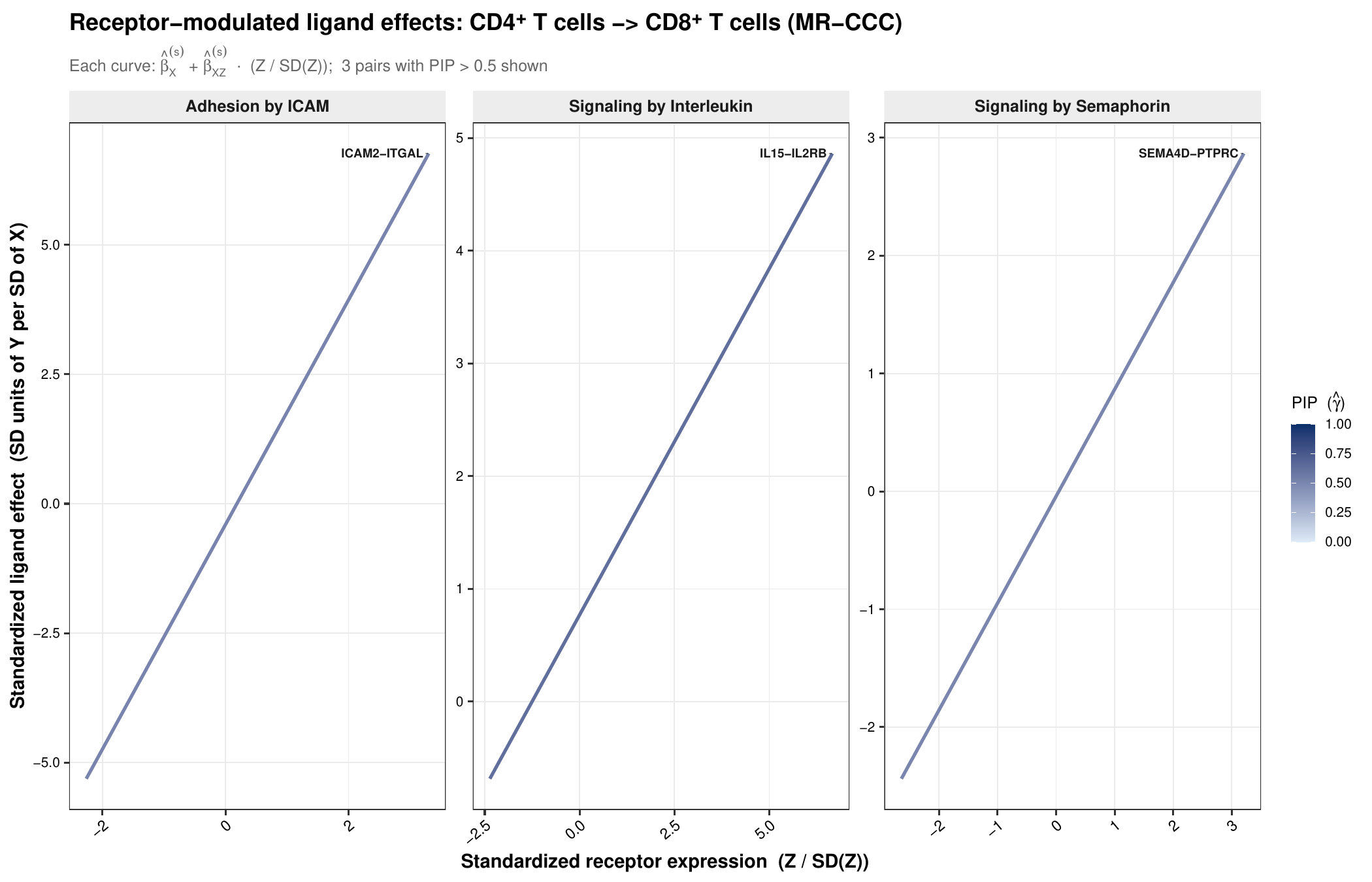}
\caption{\textbf{Receptor-modulated effect curves
  (\(\hat{\beta}_X + \hat{\beta}_{XZ} \cdot Z/\mathrm{SD}(Z)\))
  for the CD4\(^+\) T cells \(\rightarrow\) CD8\(^+\) T cells analysis.}
  Only the three discovery pairs (PIP \(> 0.5\)) are displayed, grouped
  into the three pathway panels with confirmed discoveries (Adhesion by
  ICAM, Signaling by Interleukin, and Signaling by Semaphorin). In the
  Signaling by Interleukin panel, the \textit{IL15}--\textit{IL2RB}
  curve is monotonically increasing, remaining positive throughout the
  observed receptor range with a zero crossing near \(-1.25\) standard
  deviations below the mean. In the Adhesion by ICAM panel, the
  \textit{ICAM2}--\textit{ITGAL} curve crosses zero near \(+0.19\)
  standard deviations above the mean, transitioning from suppressive to
  increasingly activating at high LFA-1 densities. In the Signaling by
  Semaphorin panel, the \textit{SEMA4D}--\textit{PTPRC} curve passes
  near zero at the mean and rises steeply, consistent with an almost
  purely interaction-driven signal gated by CD45 expression level.}
\label{fig:supp_CD4_CD8_curves}
\end{figure}
\FloatBarrier

\subsubsection{CD4\(^+\) T Cells \(\rightarrow\) NK Cells}
\label{supp:CD4_NK}
Across 850 donors, MR-CCC evaluated 42 ligand--receptor--pathway
triplets for the CD4\(^+\) T cell (sender) to NK cell (receiver)
direction and identified four high-confidence causal communication
signals with PIP exceeding 0.5: \textit{ITGB1}--\textit{ADGRE2}
within the Integrin signaling pathway (PIP \(= 0.910\)),
\textit{ITGAV}--\textit{ADGRE5} within the Integrin signaling pathway
(PIP \(= 0.858\)), \textit{PTGES3}--\textit{PTGER2} within the
Prostaglandin signaling pathway (PIP \(= 0.625\)), and
\textit{IFNG}--\textit{IFNGR2} within the Interferon signaling pathway
(PIP \(= 0.542\))
(Figures~\ref{fig:supp_CD4_NK_pip}--\ref{fig:supp_CD4_NK_curves}).

The top-ranked discovery, \textit{ITGB1}--\textit{ADGRE2}
(PIP \(= 0.910\)), implicates CD4\(^+\) T cell \(\beta_1\) integrin
(\textit{ITGB1}) as a causal regulator of NK cell Integrin pathway
activity through the adhesion G protein-coupled receptor EMR2
(encoded by \textit{ADGRE2}), which is expressed on NK cells and
myeloid cells and carries extracellular EGF-like counter-receptor
domains~\cite{stacey2003emr2}.
The positive main effect (\(\hat{\beta}_X = 0.404\)) and large
negative interaction (\(\hat{\beta}_{XZ} = -2.67\)) produce a
sign-reversing effect curve that crosses zero approximately \(0.15\)
standard deviations above the mean \textit{ADGRE2} expression level.
NK cells with low EMR2 expression therefore experience a net
stimulatory response to CD4\(^+\) T cell \(\beta_1\) integrin
engagement, while those above this threshold transition to a
progressively suppressive effect that grows steeply with EMR2 density.
Notably, the \textit{ITGB1}--\textit{ADGRE2} triplet was also
discovered in the CD4\(^+\) T cells \(\rightarrow\) B cells direction
(PIP \(= 0.606\)), but there the interaction term was positive
(\(\hat{\beta}_{XZ} = 1.48\)); the sign reversal of the
interaction between receiver cell types illustrates the cell-type
specificity of receptor-modulated causal communication detected by
MR-CCC.

The second discovered triplet, \textit{ITGAV}--\textit{ADGRE5}
(PIP \(= 0.858\)), implicates CD4\(^+\) T cell \(\alpha\)V integrin
(\textit{ITGAV}) as a causal activator of NK cell Integrin pathway
activity through the adhesion G protein-coupled receptor CD97
(encoded by \textit{ADGRE5}), whose EGF-like ectodomain serves as a
counter-receptor for multiple integrin species~\cite{hamann1996cd97}.
Both the main effect (\(\hat{\beta}_X = 0.342\)) and the interaction
(\(\hat{\beta}_{XZ} = 3.44\)) are positive, yielding a curve that
crosses zero approximately \(0.10\) standard deviations below the
mean \textit{ADGRE5} level and is therefore positive across
essentially the full observed receptor expression range, growing
steeply with CD97 density.
Together, \textit{ITGB1}--\textit{ADGRE2} and
\textit{ITGAV}--\textit{ADGRE5} form a characteristic crossing pair
in the Signaling by Integrin panel of the effect-curve figure
(Figure~\ref{fig:supp_CD4_NK_curves}): the former descends steeply
from positive to negative as receptor expression rises, while the
latter rises steeply from near zero to strongly positive, producing
an X-shaped intersection that reflects divergent receptor-modulated
outcomes of two distinct integrin--adhesion-GPCR contacts on NK cell
Integrin pathway activity.

The third discovered triplet, \textit{PTGES3}--\textit{PTGER2}
(PIP \(= 0.625\)), links CD4\(^+\) T cell prostaglandin E
synthase 3 (\textit{PTGES3}) to the EP2 prostaglandin receptor
(encoded by \textit{PTGER2}) on NK cells, recapitulating the same
PGE\(_2\)--EP2 axis identified in the B cells \(\rightarrow\) NK
cells direction (PIP \(= 0.556\)).
PGE\(_2\) signaling through EP2 activates the adenylyl
cyclase--cAMP--PKA cascade, a well-characterized suppressor of NK
cell cytotoxicity and effector cytokine
production~\cite{kalinski2012pge2}.
The negative main effect (\(\hat{\beta}_X = -0.532\)) and large
positive interaction (\(\hat{\beta}_{XZ} = 2.24\)) produce a curve
that crosses zero approximately \(0.24\) standard deviations above
the mean \textit{PTGER2} level.
NK cells with EP2 expression below this threshold experience net
suppression of their Prostaglandin pathway activity in response to
CD4\(^+\) T cell PGE\(_2\), while those with high EP2 expression
show progressively stronger activation, consistent with a
PGE\(_2\)--EP2 feed-forward mechanism in which receptor-replete NK
cells amplify prostaglandin signaling independently of the
lymphocyte source.

The fourth discovery, \textit{IFNG}--\textit{IFNGR2}
(PIP \(= 0.542\)), identifies CD4\(^+\) T cell-derived
IFN-\(\gamma\) as a causal communicator to NK cells through the
signal-transducing \(\beta\) chain of the IFN-\(\gamma\) receptor
complex (encoded by \textit{IFNGR2})~\cite{bach1997ifngr,schroder2004interferon}.
This is the first direction in the present analysis in which an
IFN-\(\gamma\) triplet clears the discovery threshold, consistent
with CD4\(^+\) Th1 cells being a bona fide source of IFN-\(\gamma\)
in peripheral blood---unlike B cells, for which all IFN-\(\gamma\)
pairs received near-zero PIPs.
The positive main effect (\(\hat{\beta}_X = 0.274\)) reflects a
baseline stimulatory IFN-\(\gamma\) signal on NK cell Interferon
pathway activity; the negative interaction
(\(\hat{\beta}_{XZ} = -1.05\)) indicates that this stimulatory
effect diminishes with increasing \textit{IFNGR2} expression and
crosses zero approximately \(0.26\) standard deviations above the
mean, becoming suppressive at the highest receptor densities,
consistent with receptor-mediated negative feedback on sustained
IFN-\(\gamma\) signaling~\cite{schroder2004interferon}.
The receptor-chain specificity of the discovery---\textit{IFNGR2}
discovered at PIP \(= 0.542\) while the ligand-binding \(\alpha\)
chain \textit{IFNGR1} remained sub-threshold at PIP \(= 0.224\)---points
to the expression level of the signal-transducing \(\beta\) chain as
the primary determinant of the receptor-modulated interaction.

Notably, several triplets did not reach the discovery threshold.
\textit{IL1B}--\textit{IL1RAP} (PIP \(= 0.448\)) was the
highest-ranking non-discovered triplet; IL-1\(\beta\) produced by
activated CD4\(^+\) T cells is a known activator of NK cell
cytotoxicity and IFN-\(\gamma\) production through the IL-1 receptor
accessory protein, but the population-level signal did not achieve
consistent discovery strength~\cite{dinarello2009interleukin}.
\textit{ICAM1}--\textit{ITGAX} (PIP \(= 0.385\)) and
\textit{SEMA4A}--\textit{PLXND1} (PIP \(= 0.376\)) were notable
near-misses within the Adhesion by ICAM and Semaphorin signaling
pathways, respectively.
\textit{CLEC2B}--\textit{KLRF1} (PIP \(= 0.342\)) implicated AICL
(activation-induced C-type lectin, encoded by \textit{CLEC2B}) as a
ligand for the activating NK receptor NKp80 (encoded by
\textit{KLRF1}), an axis established to mediate mutual activation
between NK cells and myeloid cells~\cite{welte2006nkp80}, but fell
below the discovery threshold.
\textit{IFNG}--\textit{IFNGR1} (PIP \(= 0.222\)) fell well short of
discovery, confirming that the causal IFN-\(\gamma\) signal from
CD4\(^+\) T cells to NK cells is primarily modulated by the
expression level of the signal-transducing \(\beta\) chain rather
than the ligand-binding \(\alpha\) chain of the receptor complex.

\begin{figure}[htbp]
\centering
\includegraphics[width=\textwidth]{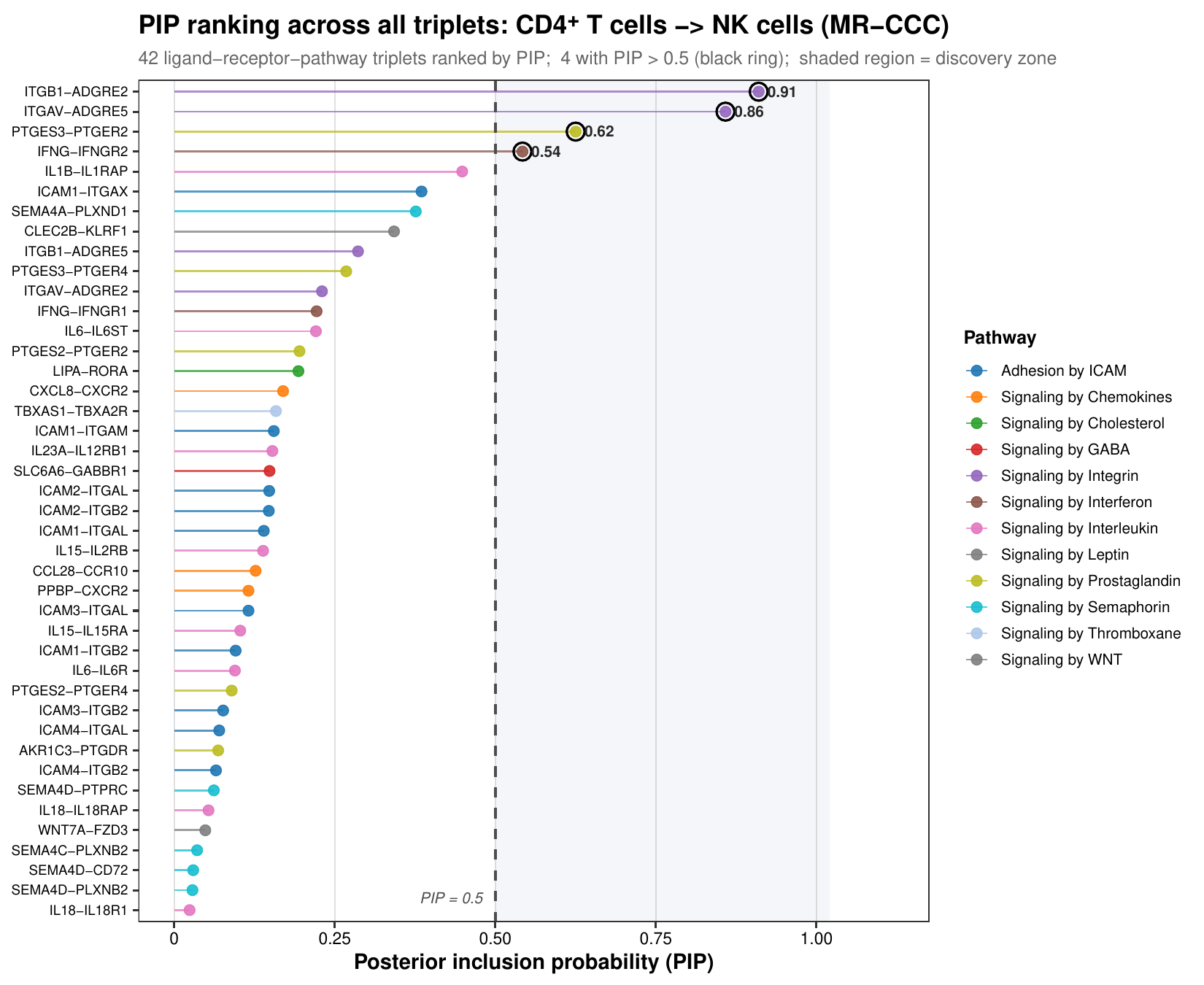}
\caption{\textbf{PIP ranking for the CD4$^+$ T cells\,$\rightarrow$\,NK cells analysis.} All 42 ligand--receptor--pathway triplets across 850
  donors. Points are colored by pathway; the dashed vertical line
  marks the discovery threshold of PIP \(= 0.5\); the shaded region
  to the right is the discovery zone. Black rings identify the four
  discovered triplets: \textit{ITGB1}--\textit{ADGRE2}
  (PIP \(= 0.91\)), \textit{ITGAV}--\textit{ADGRE5}
  (PIP \(= 0.86\)), \textit{PTGES3}--\textit{PTGER2}
  (PIP \(= 0.63\)), and \textit{IFNG}--\textit{IFNGR2}
  (PIP \(= 0.54\)).}
\label{fig:supp_CD4_NK_pip}
\end{figure}

\begin{figure}[htbp]
\centering
\includegraphics[width=\textwidth]{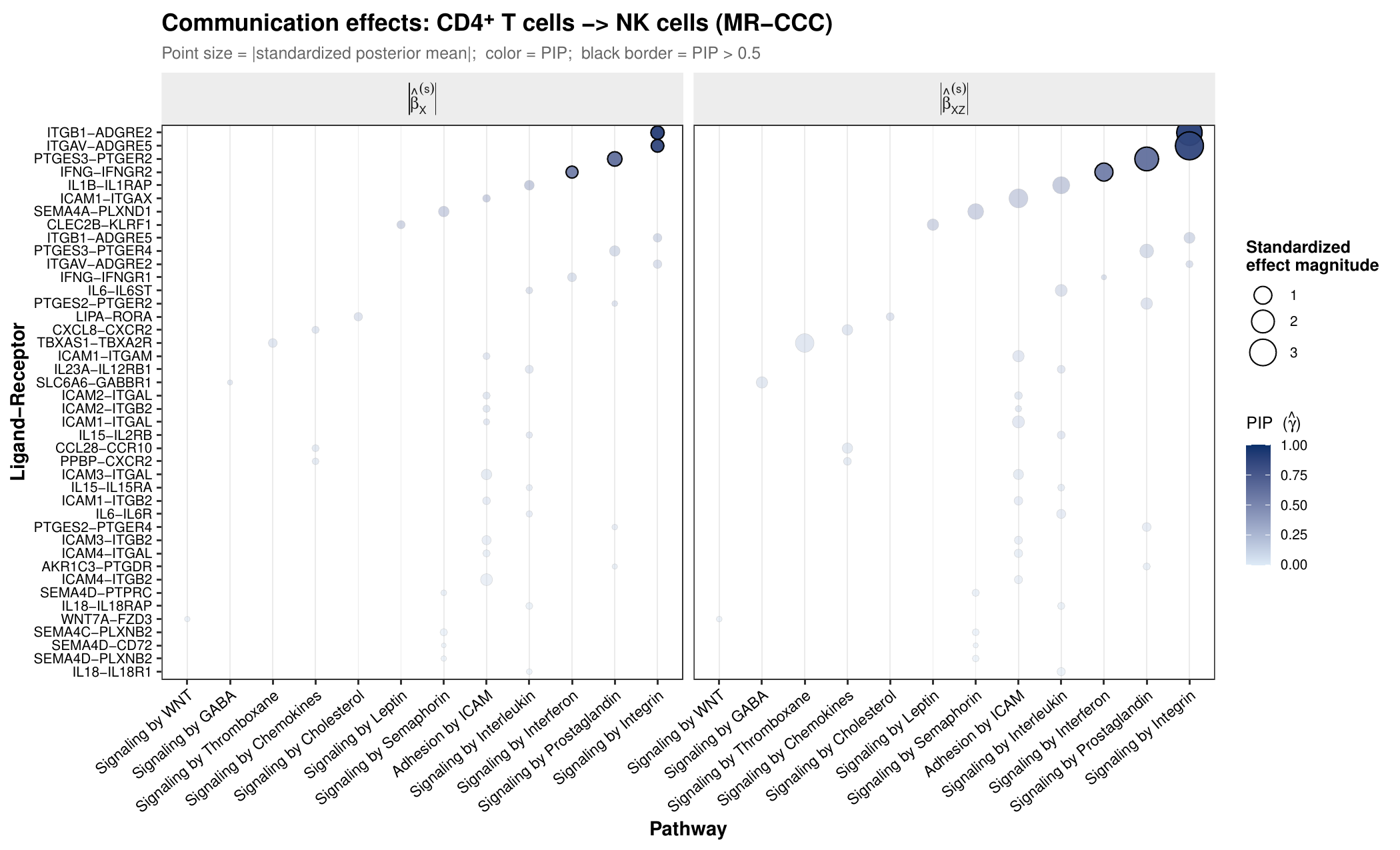}
\caption{\textbf{Standardized posterior effects for the CD4\(^+\) T cells
  \(\rightarrow\) NK cells analysis.} Left panel: absolute main ligand
  effect \(|\hat{\beta}_X^{(s)}|\); right panel: absolute
  receptor-modulated interaction effect \(|\hat{\beta}_{XZ}^{(s)}|\).
  Point size encodes effect magnitude; fill color encodes PIP; black
  borders identify the four discovered triplets. Both Integrin
  discoveries (\textit{ITGB1}--\textit{ADGRE2} and
  \textit{ITGAV}--\textit{ADGRE5}) show large effects in both
  panels; \textit{PTGES3}--\textit{PTGER2} and
  \textit{IFNG}--\textit{IFNGR2} display dominant interaction
  components relative to their main effects, indicating
  receptor-expression-gated communication signals.}
\label{fig:supp_CD4_NK_bubble}
\end{figure}

\begin{figure}[htbp]
\centering
\includegraphics[width=\textwidth]{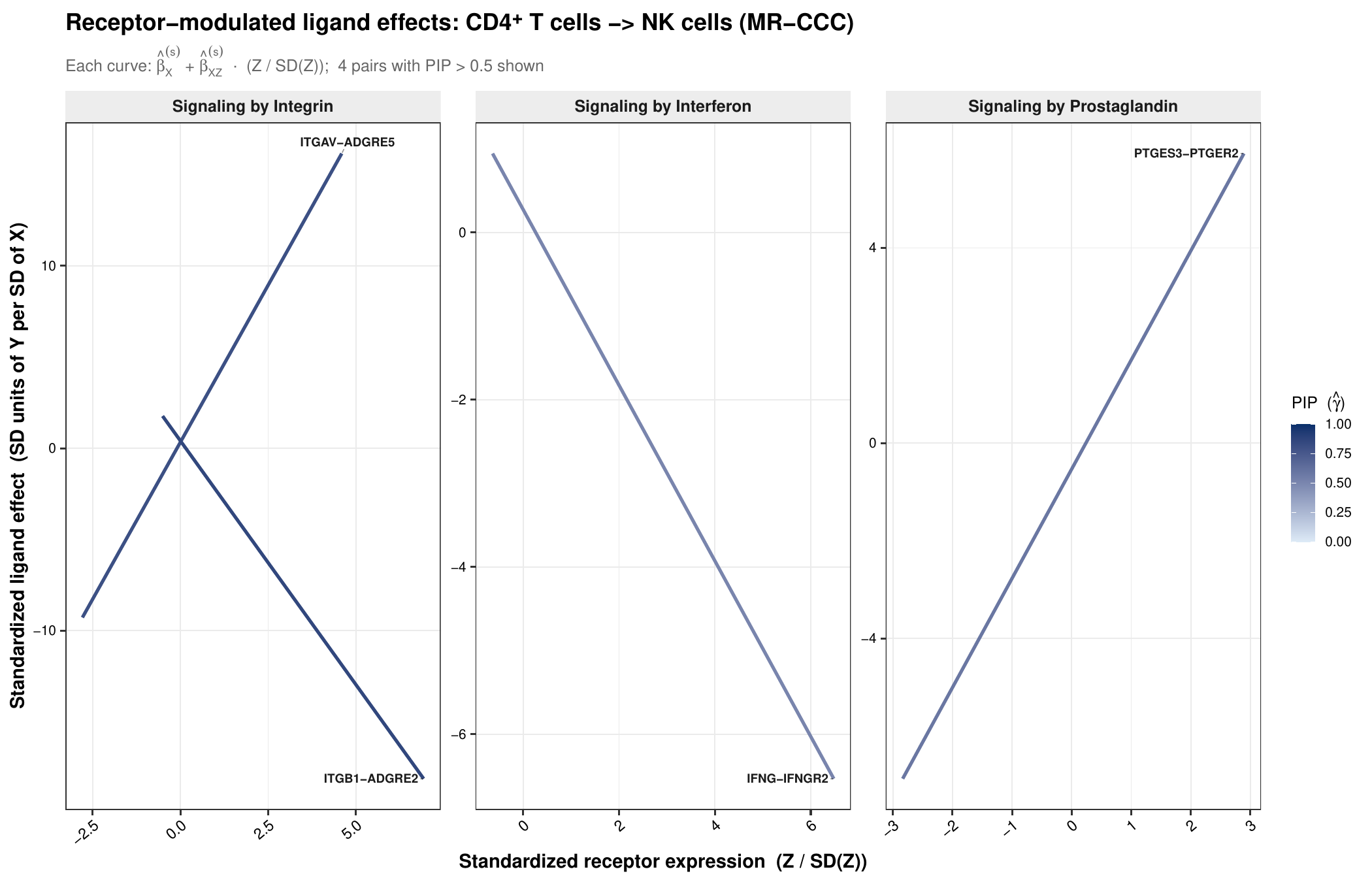}
\caption{\textbf{Receptor-modulated effect curves
  (\(\hat{\beta}_X + \hat{\beta}_{XZ} \cdot Z/\mathrm{SD}(Z)\))
  for the CD4\(^+\) T cells \(\rightarrow\) NK cells analysis.}
  Only the four discovery pairs (PIP \(> 0.5\)) are displayed, grouped
  into the three pathway panels with confirmed discoveries (Signaling
  by Integrin, Signaling by Interferon, and Signaling by Prostaglandin).
  In the Signaling by Integrin panel,
  \textit{ITGB1}--\textit{ADGRE2} and \textit{ITGAV}--\textit{ADGRE5}
  form an X-shaped crossing pair: \textit{ITGB1}--\textit{ADGRE2}
  descends steeply from positive to negative above the mean receptor
  level, while \textit{ITGAV}--\textit{ADGRE5} rises steeply from
  near zero to strongly positive. In the Signaling by Interferon
  panel, the \textit{IFNG}--\textit{IFNGR2} curve starts positive at
  low \textit{IFNGR2} expression and becomes increasingly negative at
  high receptor levels, crossing zero near \(+0.26\) standard
  deviations above the mean. In the Signaling by Prostaglandin panel,
  the \textit{PTGES3}--\textit{PTGER2} curve crosses zero near
  \(+0.24\) standard deviations above the mean \textit{PTGER2} level,
  transitioning from suppressive to strongly activating at high EP2
  densities.}
\label{fig:supp_CD4_NK_curves}
\end{figure}
\FloatBarrier

\subsubsection{CD4\(^+\) T Cells \(\rightarrow\) Monocytes}
\label{supp:CD4_Mono}
Across 691 donors, MR-CCC evaluated 42 ligand--receptor--pathway
triplets for the CD4\(^+\) T cell (sender) to Monocyte (receiver)
direction and identified two high-confidence causal communication
signals with PIP exceeding 0.5: \textit{PPBP}--\textit{CXCR2} within
the Chemokine signaling pathway (PIP \(= 0.988\)) and
\textit{ITGB1}--\textit{ADGRE5} within the Integrin signaling pathway
(PIP \(= 0.597\))
(Figures~\ref{fig:supp_CD4_Mono_pip}--\ref{fig:supp_CD4_Mono_curves}).

The dominant discovery, \textit{PPBP}--\textit{CXCR2}
(PIP \(= 0.988\)), identifies the platelet basic protein precursor
(\textit{PPBP}, which encodes the CXCL7/NAP-2 family of CXC
chemokines upon processing) expressed in CD4\(^+\) T cells as a
causal regulator of monocyte Chemokine pathway activity through the
CXC chemokine receptor 2 (CXCR2).
CXCL7/NAP-2, the processed product of PPBP, is a potent agonist for
CXCR2 on myeloid cells~\cite{ahuja1996cxcr2}, and CXCR2 is the
dominant chemokine receptor governing monocyte recruitment and
activation in response to ELR\(^+\) CXC chemokines.
The positive main effect (\(\hat{\beta}_X = 0.882\)) indicates that
higher CD4\(^+\) T cell PPBP production exerts a baseline stimulatory
effect on monocyte Chemokine pathway activity; however, the
interaction term (\(\hat{\beta}_{XZ} = -12.6\)) is the largest in
magnitude of any triplet in the present analysis.
The effect curve crosses zero at only \(Z^* \approx +0.07\) standard
deviations above the mean \textit{CXCR2} level---essentially at the
population mean---and then descends steeply to strongly negative
values.
Monocytes expressing \textit{CXCR2} at or above the population mean
therefore experience a rapidly escalating suppression of Chemokine
pathway activity as CD4\(^+\) T cell PPBP production rises, consistent
with receptor saturation and homologous desensitization: CXCR2 is a
G protein-coupled receptor subject to rapid phosphorylation,
arrestin recruitment, and internalization upon sustained ligand
exposure~\cite{ahuja1996cxcr2}, and monocytes already expressing high
CXCR2 levels may be in a prior state of chemokine exposure that
amplifies this counter-regulatory response.

The second discovered triplet, \textit{ITGB1}--\textit{ADGRE5}
(PIP \(= 0.597\)), implicates CD4\(^+\) T cell \(\beta_1\) integrin
(\textit{ITGB1}) as a causal regulator of monocyte Integrin pathway
activity through the adhesion G protein-coupled receptor CD97
(encoded by \textit{ADGRE5})~\cite{hamann1996cd97}.
The positive main effect (\(\hat{\beta}_X = 0.360\)) and negative
interaction (\(\hat{\beta}_{XZ} = -0.618\)) produce a sign-reversing
curve that crosses zero approximately \(0.58\) standard deviations
above the mean \textit{ADGRE5} level.
Monocytes with low CD97 expression show a modest stimulatory response
to CD4\(^+\) T cell \(\beta_1\) integrin engagement; those expressing
\textit{ADGRE5} above this threshold transition to a progressively
suppressive effect.
This pattern contrasts with the same triplet in the CD4\(^+\) T cells
\(\rightarrow\) B cells direction, where the interaction term was
large and positive (\(\hat{\beta}_{XZ} = 3.63\)), yielding a
monotonically and steeply activating curve.
The cell-type-specific sign reversal of the \textit{ADGRE5}
interaction term between monocytes and B cells reflects the
fundamentally different downstream signaling environments of
CD97-expressing cells in these two lineages, and illustrates the
receiver-cell specificity of MR-CCC discoveries.

Notably, several triplets did not reach the discovery threshold.
\textit{ITGAV}--\textit{ADGRE2} (PIP \(= 0.484\)) was the
highest-ranking non-discovered triplet and a close near-miss; the same receptor \textit{ADGRE2} was part of the top-ranked discovery in
the CD4\(^+\) T cells \(\rightarrow\) NK cells direction
(\textit{ITGB1}--\textit{ADGRE2}, PIP \(= 0.910\)), but here it
falls just below threshold, suggesting the EMR2--integrin axis on
monocytes may be more heterogeneous at the population level.
\textit{IL1B}--\textit{IL1RAP} (PIP \(= 0.349\)) was also a notable
near-miss, consistent with the sub-threshold IL-1\(\beta\) signal
observed in the CD4\(^+\) T cells \(\rightarrow\) NK cells direction
as well~\cite{dinarello2009interleukin}.
\textit{IFNG}--\textit{IFNGR2} (PIP \(= 0.309\)) and
\textit{IFNG}--\textit{IFNGR1} (PIP \(= 0.159\)) both fell well
below the discovery threshold; while CD4\(^+\) Th1 cells are a
physiological source of IFN-\(\gamma\), the IFN-\(\gamma\) signal to
monocytes does not achieve discovery strength here, in contrast to the
CD4\(^+\) T cells \(\rightarrow\) NK cells direction where
\textit{IFNG}--\textit{IFNGR2} was discovered (PIP \(= 0.542\)),
suggesting NK cells are a more consistent causal target of Th1-derived
IFN-\(\gamma\) in this cohort~\cite{schroder2004interferon}.

\begin{figure}[htbp]
\centering
\includegraphics[width=\textwidth]{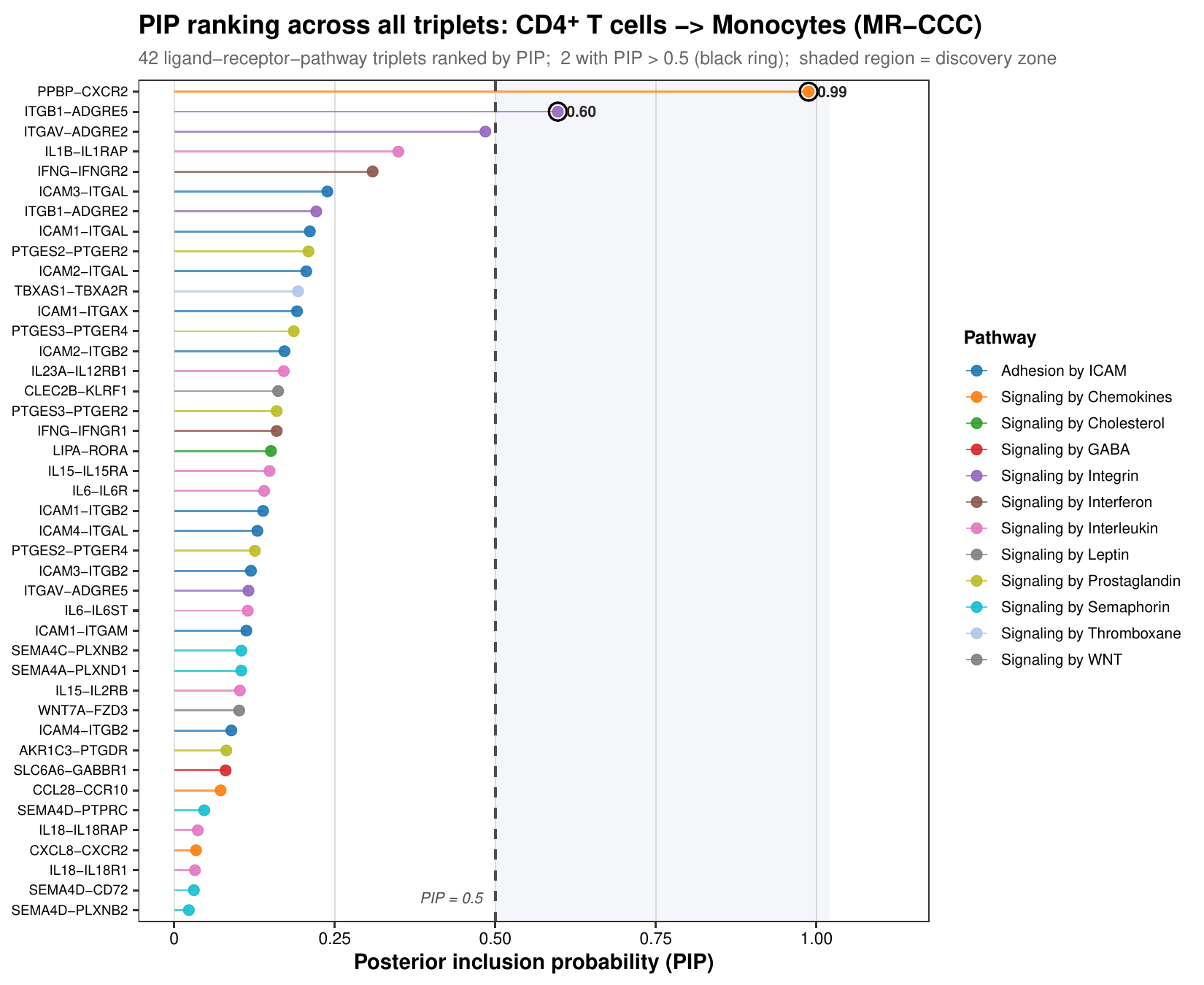}
\caption{\textbf{PIP ranking for the CD4$^+$ T cells\,$\rightarrow$\,Monocytes analysis.} All 42 ligand--receptor--pathway triplets across 691
  donors. Points are colored by pathway; the dashed vertical line
  marks the discovery threshold of PIP \(= 0.5\); the shaded region
  to the right is the discovery zone. Black rings identify the two
  discovered triplets: \textit{PPBP}--\textit{CXCR2}
  (PIP \(= 0.99\)) and \textit{ITGB1}--\textit{ADGRE5}
  (PIP \(= 0.60\)).}
\label{fig:supp_CD4_Mono_pip}
\end{figure}

\begin{figure}[htbp]
\centering
\includegraphics[width=\textwidth]{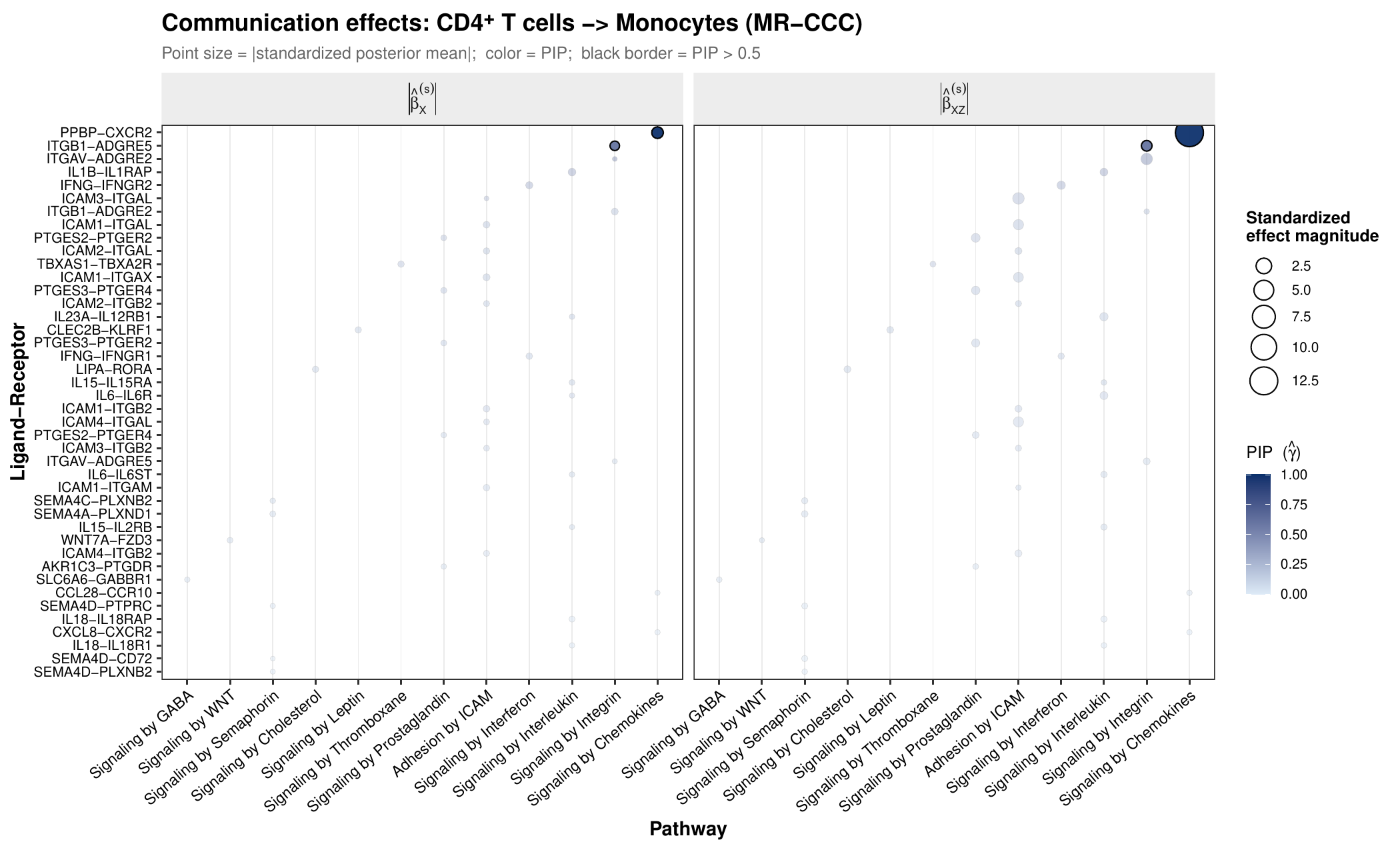}
\caption{\textbf{Standardized posterior effects for the CD4\(^+\) T cells
  \(\rightarrow\) Monocytes analysis.} Left panel: absolute main ligand
  effect \(|\hat{\beta}_X^{(s)}|\); right panel: absolute
  receptor-modulated interaction effect \(|\hat{\beta}_{XZ}^{(s)}|\).
  Point size encodes effect magnitude; fill color encodes PIP; black
  borders identify the two discovered triplets.
  \textit{PPBP}--\textit{CXCR2} dominates the right panel with by
  far the largest interaction magnitude in the analysis, while
  \textit{ITGB1}--\textit{ADGRE5} shows moderate effects in both
  panels. All remaining triplets show small effect magnitudes and low
  PIP.}
\label{fig:supp_CD4_Mono_bubble}
\end{figure}

\begin{figure}[htbp]
\centering
\includegraphics[width=\textwidth]{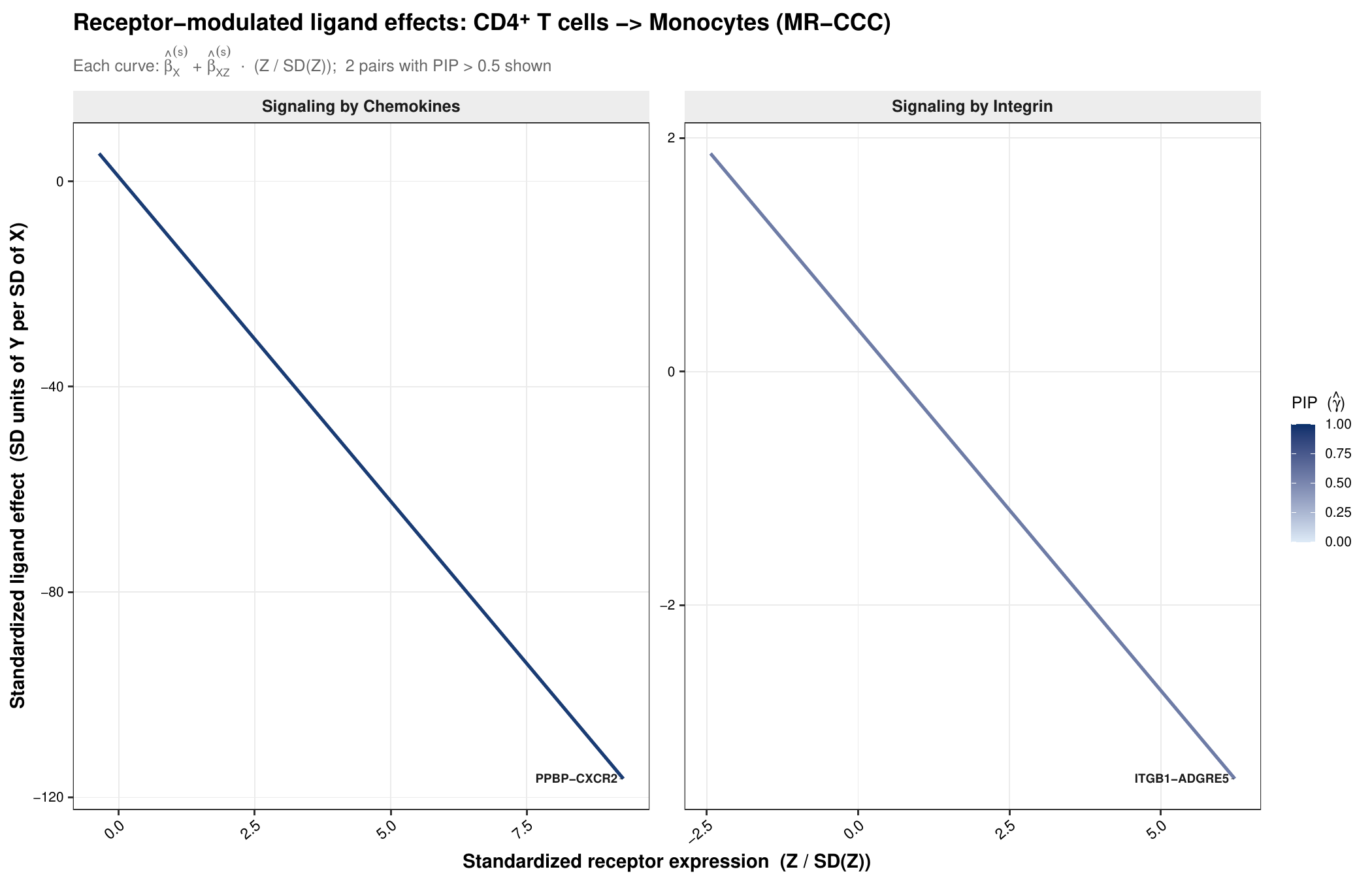}
\caption{\textbf{Receptor-modulated effect curves
  (\(\hat{\beta}_X + \hat{\beta}_{XZ} \cdot Z/\mathrm{SD}(Z)\))
  for the CD4\(^+\) T cells \(\rightarrow\) Monocytes analysis.}
  Only the two discovery pairs (PIP \(> 0.5\)) are displayed, grouped
  into the two pathway panels with confirmed discoveries (Signaling by
  Chemokines and Signaling by Integrin). In the Signaling by
  Chemokines panel, the \textit{PPBP}--\textit{CXCR2} curve descends
  steeply from a small positive value near zero receptor expression to
  strongly negative values at high \textit{CXCR2} levels, with a zero
  crossing essentially at the population mean; the y-axis range
  (reaching below \(-100\) in standardized units) reflects the largest
  interaction magnitude in the entire analysis. In the Signaling by
  Integrin panel, the \textit{ITGB1}--\textit{ADGRE5} curve
  decreases gently from positive to negative, crossing zero
  approximately \(0.58\) standard deviations above the mean
  \textit{ADGRE5} expression level.}
\label{fig:supp_CD4_Mono_curves}
\end{figure}
\FloatBarrier

\clearpage
\subsection{CD8\texorpdfstring{\(^+\)}{+} T Cells as Sender}

\subsubsection{CD8\(^+\) T Cells \(\rightarrow\) B Cells}
\label{supp:CD8_B}
Across 845 donors, MR-CCC evaluated 42 ligand--receptor--pathway
triplets for the CD8\(^+\) T cell (sender) to B cell (receiver)
direction and identified seven high-confidence causal communication
signals with PIP exceeding 0.5
(Figures~\ref{fig:supp_CD8_B_pip}--\ref{fig:supp_CD8_B_curves}).
The seven triplets span two pathway classes: six Adhesion by ICAM
pairs (\textit{ICAM2}--\textit{ITGAL}, PIP \(= 0.861\);
\textit{ICAM1}--\textit{ITGAX}, PIP \(= 0.755\);
\textit{ICAM1}--\textit{ITGAL}, PIP \(= 0.738\);
\textit{ICAM2}--\textit{ITGB2}, PIP \(= 0.712\);
\textit{ICAM1}--\textit{ITGAM}, PIP \(= 0.676\);
\textit{ICAM1}--\textit{ITGB2}, PIP \(= 0.556\)) and one Signaling
by Interleukin triplet (\textit{IL23A}--\textit{IL12RB1},
PIP \(= 0.800\)).

\noindent\textbf{Adhesion by ICAM signals.}
The six discovered ICAM--integrin triplets again implicate the CD8\(^+\)
T cell--B cell adhesive interface, where ICAM family members on T cells
engage \(\beta_2\) integrins on B cells~\cite{springer1990adhesion,dustin1986icam}.
In contrast to the CD4\(^+\) T cells \(\rightarrow\) B cells direction
---where the six ICAM pairs produced a heterogeneous fan including both
strongly stimulatory and strongly suppressive curves---the CD8\(^+\)
T cell ICAM contacts are predominantly suppressive at physiological
receptor expression levels, consistent with the cytotoxic rather than
helper function of CD8\(^+\) T cells at the T cell--B cell interface.
All six curves are visible in the Adhesion by ICAM panel
(Figure~\ref{fig:supp_CD8_B_curves}).

The top-ranked ICAM pair, \textit{ICAM2}--\textit{ITGAL}
(PIP \(= 0.861\), \(\hat{\beta}_X = 0.165\), \(\hat{\beta}_{XZ} = -3.98\)),
has a near-zero main effect and a large negative interaction, producing
a curve that crosses zero at only \(Z^* \approx +0.04\) standard
deviations above the mean \textit{ITGAL} level---essentially at the
population mean.
For B cells expressing LFA-1 above this threshold, which encompasses
the vast majority of the observed population, CD8\(^+\) T cell ICAM2
engagement exerts a steeply increasing suppression of Adhesion pathway
activity, with the effect growing strongly negative at high LFA-1
densities.
This contrasts markedly with the CD4\(^+\) T cells \(\rightarrow\) B
cells direction, where \textit{ICAM2}--\textit{ITGAL} had a large
\emph{positive} interaction (\(\hat{\beta}_{XZ} = 2.59\)), yielding
the opposite receptor-modulated outcome despite the same ligand--receptor
pair.

Of the four discovered \textit{ICAM1}--integrin pairs,
\textit{ICAM1}--\textit{ITGAX} (PIP \(= 0.755\),
\(\hat{\beta}_X = -0.758\), \(\hat{\beta}_{XZ} = -2.50\)) and
\textit{ICAM1}--\textit{ITGAL} (PIP \(= 0.738\),
\(\hat{\beta}_X = -0.986\), \(\hat{\beta}_{XZ} = 0.927\)) show the
largest coefficient magnitudes.
\textit{ICAM1}--\textit{ITGAX} crosses zero at approximately
\(0.30\) standard deviations \emph{below} the mean \textit{ITGAX}
level, placing essentially the full observed CD11c-expressing B cell
population in the suppressive regime, with the suppression intensifying
steeply at high receptor density.
\textit{ICAM1}--\textit{ITGAL} has a negative main effect and a
partially positive interaction, crossing zero at approximately
\(+1.06\) standard deviations above the mean \textit{ITGAL} level;
because this threshold is well above the population mean, the majority
of B cells experience a net suppressive effect of CD8\(^+\) T cell
ICAM1 engagement through LFA-1.
\textit{ICAM1}--\textit{ITGAM} (PIP \(= 0.676\),
\(\hat{\beta}_X = -0.902\), \(\hat{\beta}_{XZ} = -0.782\)) and
\textit{ICAM1}--\textit{ITGB2} (PIP \(= 0.556\),
\(\hat{\beta}_X = -0.690\), \(\hat{\beta}_{XZ} = -0.893\)) are both
monotonically suppressive with zero crossings at approximately \(1.15\)
and \(0.77\) standard deviations \emph{below} the mean receptor
expression level, respectively, placing the entire observed population
of B cells in the suppressive regime for both pairs.

The exception within the ICAM discoveries is
\textit{ICAM2}--\textit{ITGB2} (PIP \(= 0.712\),
\(\hat{\beta}_X = 0.195\), \(\hat{\beta}_{XZ} = 1.61\)): both the main
effect and the interaction are positive, yielding a monotonically
increasing curve that crosses zero approximately \(0.12\) standard
deviations below the mean \textit{ITGB2} level and rises steeply with
\(\beta_2\) density.
The divergence between \textit{ICAM2}--\textit{ITGAL} (strongly
suppressive via LFA-1) and \textit{ICAM2}--\textit{ITGB2}
(monotonically activating via the \(\beta_2\) chain alone) indicates
that the identity of the integrin chain engaged by CD8\(^+\) T cell
ICAM2 governs the sign of the downstream B cell Adhesion pathway
response.

\noindent\textbf{Interleukin signal.}
The \textit{IL23A}--\textit{IL12RB1} triplet (PIP \(= 0.800\),
\(\hat{\beta}_X = 0.904\), \(\hat{\beta}_{XZ} = -3.17\)) recapitulates
the same ligand--receptor--pathway axis discovered in the CD4\(^+\)
T cells \(\rightarrow\) B cells direction (PIP \(= 0.798\)), but with
a strikingly opposite interaction sign.
In the CD4\(^+\) sender direction, the interaction was large and
positive (\(\hat{\beta}_{XZ} = 3.16\)), producing a monotonically and
steeply stimulatory curve through IL-12R\(\beta_1\); here, the
interaction is equally large but negative (\(\hat{\beta}_{XZ} = -3.17\)),
producing a sign-reversing curve that crosses zero at approximately
\(0.29\) standard deviations above the mean \textit{IL12RB1} level
and descends steeply thereafter.
B cells with below-average IL-12R\(\beta_1\) expression therefore
experience a strong stimulatory response to CD8\(^+\) T cell IL-23A
production, while those with above-average receptor expression undergo
progressively stronger suppression.
The exact inversion of the interaction term between the CD4\(^+\) and
CD8\(^+\) sender directions for the same triplet---\(\hat{\beta}_{XZ} \approx +3.16\)
versus \(\hat{\beta}_{XZ} \approx -3.17\)---is a striking example of
sender-cell-specific modulation of the same ligand--receptor--pathway
axis, underscoring the capacity of MR-CCC to detect qualitatively
different causal communication patterns from different sender cell
populations~\cite{oppmann2000il23}.

\noindent\textbf{Low-PIP pairs.}
\textit{WNT7A}--\textit{FZD3} (PIP \(= 0.458\)) was the
highest-ranking non-discovered triplet, a notable near-miss for which
no established biological basis exists for CD8\(^+\) T cell WNT
signaling to B cells~\cite{clevers2006wnt}.
\textit{IFNG}--\textit{IFNGR1} (PIP \(= 0.087\)) and
\textit{IFNG}--\textit{IFNGR2} (PIP \(= 0.074\)) both received
near-zero PIPs despite CD8\(^+\) cytotoxic T cells being a major
physiological source of IFN-\(\gamma\)~\cite{schroder2004interferon}.
The low PIPs here likely reflect that IFN-\(\gamma\)-driven B cell
Interferon pathway activity is not well-explained by direct CD8\(^+\)
T cell cis-eQTL variation for \textit{IFNG} at the population level in
this cohort, or that the target pathway genes in B cells are
insufficiently variable to detect this signal.
\textit{IL15}--\textit{IL15RA} (PIP \(= 0.094\)) and
\textit{IL15}--\textit{IL2RB} (PIP \(= 0.101\)) received near-zero
support, consistent with B cells not being primary targets of
IL-15-mediated homeostatic signaling~\cite{waldmann2006il15}.

\begin{figure}[htbp]
\centering
\includegraphics[width=\textwidth]{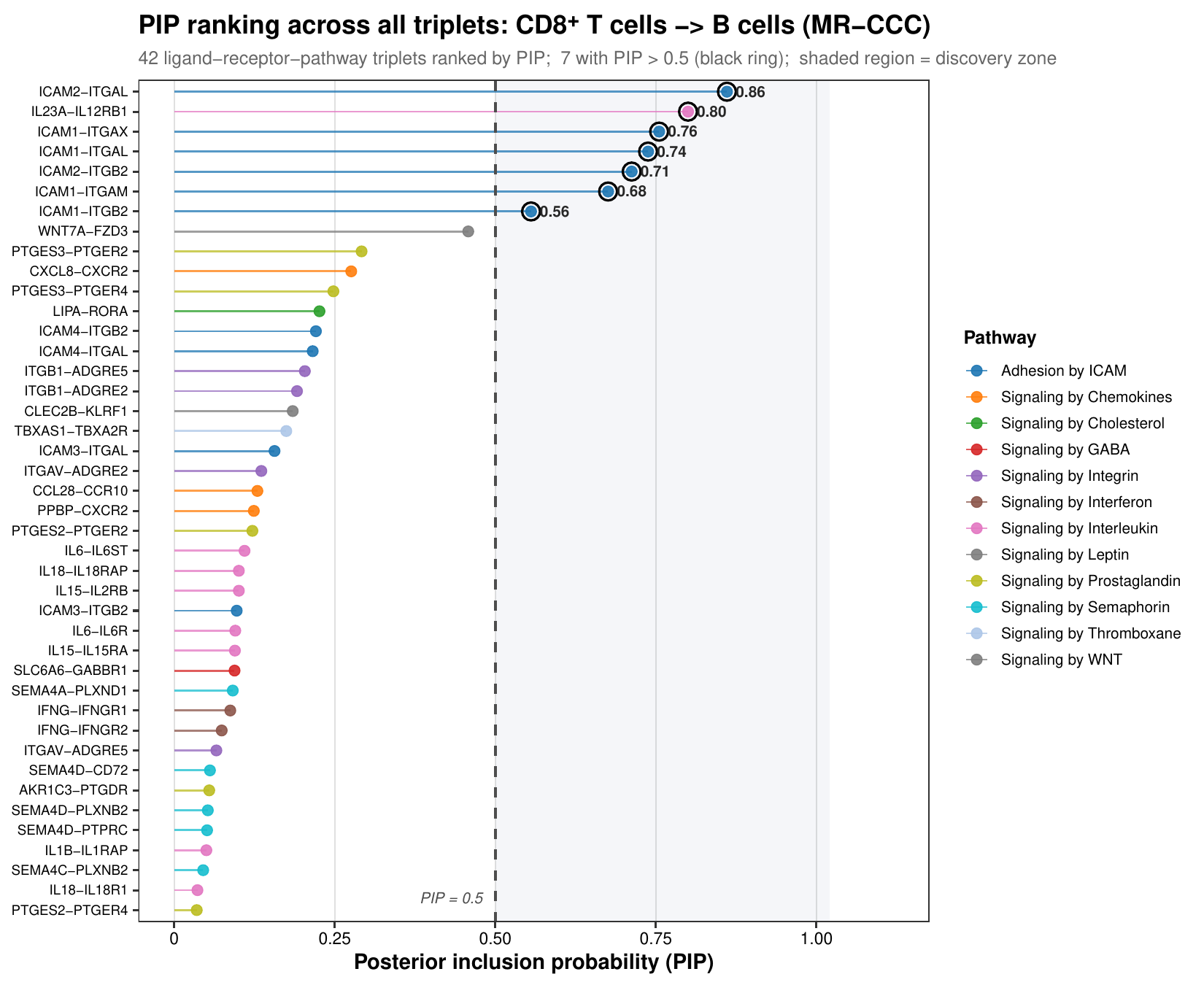}
\caption{\textbf{PIP ranking for the CD8$^+$ T cells\,$\rightarrow$\,B cells analysis.} All 42 ligand--receptor--pathway triplets across 845
  donors. Points are colored by pathway; the dashed vertical line
  marks the discovery threshold of PIP \(= 0.5\); the shaded region
  to the right is the discovery zone. Black rings identify the seven
  discovered triplets spanning Adhesion by ICAM
  (\textit{ICAM2}--\textit{ITGAL}, \textit{ICAM1}--\textit{ITGAX},
  \textit{ICAM1}--\textit{ITGAL}, \textit{ICAM2}--\textit{ITGB2},
  \textit{ICAM1}--\textit{ITGAM}, \textit{ICAM1}--\textit{ITGB2})
  and Signaling by Interleukin (\textit{IL23A}--\textit{IL12RB1}).}
\label{fig:supp_CD8_B_pip}
\end{figure}

\begin{figure}[htbp]
\centering
\includegraphics[width=\textwidth]{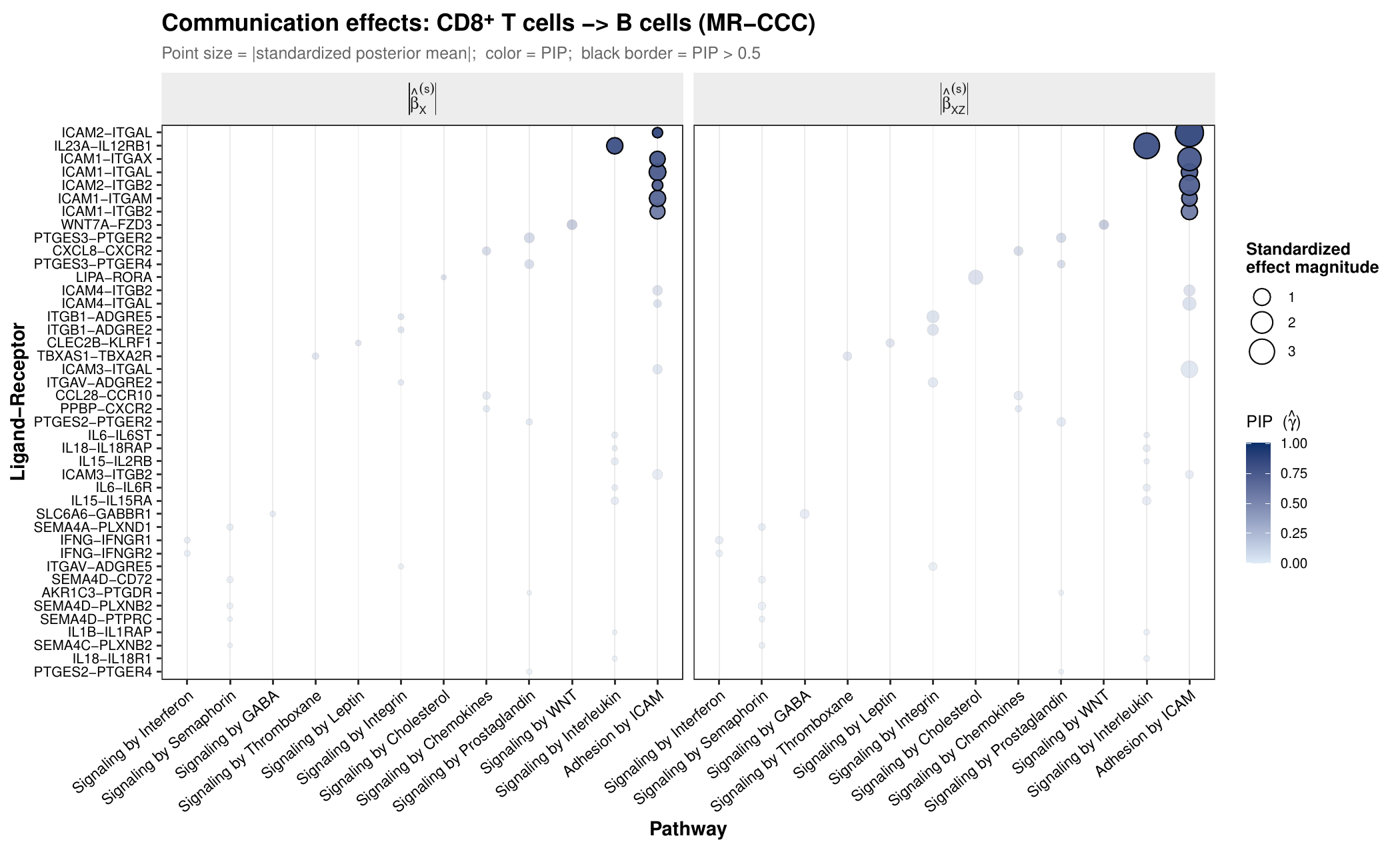}
\caption{\textbf{Standardized posterior effects for the CD8\(^+\) T cells
  \(\rightarrow\) B cells analysis.} Left panel: absolute main ligand
  effect \(|\hat{\beta}_X^{(s)}|\); right panel: absolute
  receptor-modulated interaction effect \(|\hat{\beta}_{XZ}^{(s)}|\).
  Point size encodes effect magnitude; fill color encodes PIP; black
  borders identify the seven discovered triplets. The ICAM--integrin
  cluster dominates the Adhesion by ICAM column in both panels;
  \textit{IL23A}--\textit{IL12RB1} shows the largest main effect
  and a large interaction magnitude in the Signaling by Interleukin
  column.}
\label{fig:supp_CD8_B_bubble}
\end{figure}

\begin{figure}[htbp]
\centering
\includegraphics[width=\textwidth]{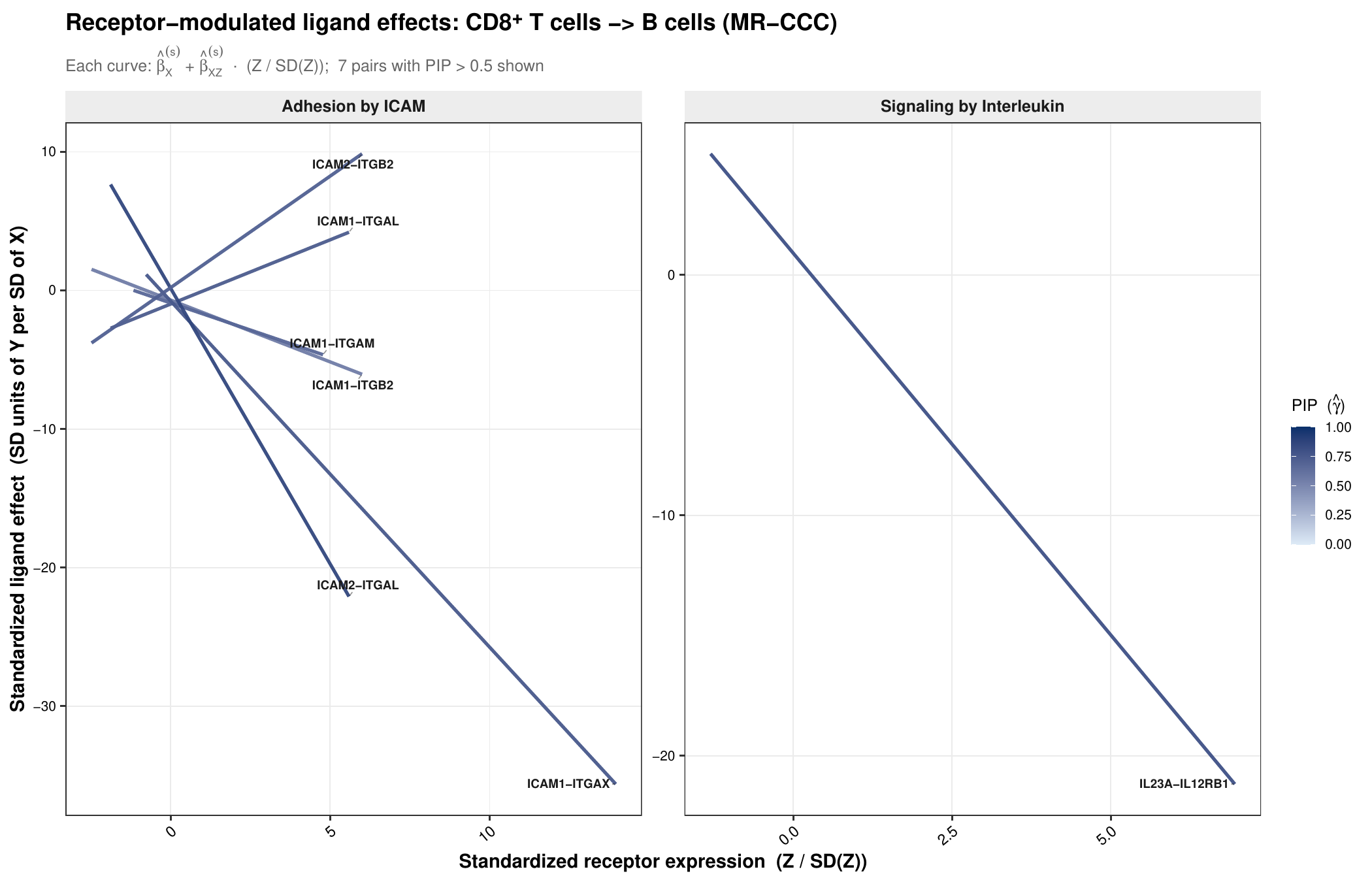}
\caption{\textbf{Receptor-modulated effect curves
  (\(\hat{\beta}_X + \hat{\beta}_{XZ} \cdot Z/\mathrm{SD}(Z)\))
  for the CD8\(^+\) T cells \(\rightarrow\) B cells analysis.}
  Only the seven discovery pairs (PIP \(> 0.5\)) are displayed,
  grouped into the two pathway panels with confirmed discoveries
  (Adhesion by ICAM and Signaling by Interleukin). In the Adhesion
  by ICAM panel, six curves form a predominantly downward-sloping
  fan: \textit{ICAM2}--\textit{ITGAL} and
  \textit{ICAM1}--\textit{ITGAX} descend most steeply to strongly
  negative values at high receptor expression, while
  \textit{ICAM1}--\textit{ITGAM} and \textit{ICAM1}--\textit{ITGB2}
  are monotonically suppressive across the observed range;
  \textit{ICAM1}--\textit{ITGAL} shows a partially positive
  interaction that moderates suppression at the highest LFA-1 densities;
  and \textit{ICAM2}--\textit{ITGB2} is the sole rising curve,
  monotonically stimulatory across the observed \(\beta_2\) expression
  range. In the Signaling by Interleukin panel, the
  \textit{IL23A}--\textit{IL12RB1} curve starts at a strongly
  positive value at low receptor expression and descends steeply,
  crossing zero near \(+0.29\) standard deviations above the mean
  and reaching strongly negative values at high IL-12R\(\beta_1\)
  density.}
\label{fig:supp_CD8_B_curves}
\end{figure}
\FloatBarrier

\subsubsection{CD8\(^+\) T Cells \(\rightarrow\) CD4\(^+\) T Cells}
\label{supp:CD8_CD4}
Across 882 donors, MR-CCC evaluated 42 ligand--receptor--pathway
triplets for the CD8\(^+\) T cell (sender) to CD4\(^+\) T cell
(receiver) direction.
No triplet reached the PIP \(> 0.5\) discovery threshold: the
highest-ranking pair, \textit{CXCL8}--\textit{CXCR2}, achieved a
PIP of only 0.367, and the distribution of PIPs was compressed below
0.4 across all 42 triplets
(Figures~\ref{fig:supp_CD8_CD4_pip}--\ref{fig:supp_CD8_CD4_bubble}).
Because no triplet exceeded the PIP \(> 0.5\) threshold, no
receptor-modulated effect curves are displayed for this direction.
This null result is consistent with the biological asymmetry of the
CD8\(^+\)--CD4\(^+\) T cell axis: CD8\(^+\) cytotoxic T cells are
primarily effectors of target-cell killing and cytokine-mediated
immune regulation, and do not function as major ligand-expressing
senders of pathway-activating signals to CD4\(^+\) T cells in the
context of peripheral blood homeostasis.
The absence of high-PIP discoveries here, alongside the substantial
number of discoveries identified for the reverse direction
(CD4\(^+\) T cells \(\rightarrow\) CD8\(^+\) T cells; three
triplets, PIP up to 0.668), underscores the directional specificity
of MR-CCC and the biological asymmetry of helper versus cytotoxic T
cell communication in peripheral blood.

Among the highest-ranking sub-threshold triplets,
\textit{CXCL8}--\textit{CXCR2} (PIP \(= 0.367\),
\(\hat{\beta}_X = -0.145\), \(\hat{\beta}_{XZ} = 0.795\)) reflects
the capacity of activated CD8\(^+\) T cells to produce CXCL8
(IL-8), which can engage CXCR2 on CD4\(^+\) T cells; however,
CXCR2 expression on CD4\(^+\) T cells is limited and variable in
peripheral blood, likely attenuating the population-level causal
signal below the discovery threshold.
\textit{IL15}--\textit{IL2RB} (PIP \(= 0.334\),
\(\hat{\beta}_X = 0.196\), \(\hat{\beta}_{XZ} = 0.557\)) and
\textit{IL15}--\textit{IL15RA} (PIP \(= 0.235\),
\(\hat{\beta}_X = 0.163\), \(\hat{\beta}_{XZ} = 0.054\)) both
represent the potential for CD8\(^+\) T cells to contribute to
IL-15-mediated signaling to CD4\(^+\) T cells through the
shared IL-2R\(\beta\) chain and the high-affinity
IL-15R\(\alpha\) chain~\cite{waldmann2006il15,kennedy2000il15};
the sub-threshold PIPs suggest this IL-15 axis from CD8\(^+\) T
cells to CD4\(^+\) T cells is insufficiently consistent at the
population level compared with the B cell sender direction, where
both chains were discovered.
\textit{SEMA4A}--\textit{PLXND1} (PIP \(= 0.320\)) and
\textit{SEMA4C}--\textit{PLXNB2} (PIP \(= 0.258\)) represent the
two highest-ranking Semaphorin pathway near-misses, consistent with
semaphorin expression on activated CD8\(^+\) T cells but
insufficient to drive a detectable causal communication effect on
CD4\(^+\) T cell Semaphorin pathway activity~\cite{suzuki2008sema4d}.
\textit{IFNG}--\textit{IFNGR1} (PIP \(= 0.032\)) and
\textit{IFNG}--\textit{IFNGR2} (PIP \(= 0.030\)) received
near-zero support despite CD8\(^+\) cytotoxic T cells being a
major physiological source of IFN-\(\gamma\); the very low PIPs
indicate that CD8\(^+\) T cell-derived IFN-\(\gamma\) does not
exert a detectable causal effect on CD4\(^+\) T cell Interferon
pathway activity at the population level in this cohort, in
contrast to the CD4\(^+\) T cells \(\rightarrow\) NK cells
direction where \textit{IFNG}--\textit{IFNGR2} was
discovered~\cite{schroder2004interferon,bach1997ifngr}.

\begin{figure}[htbp]
\centering
\includegraphics[width=\textwidth]{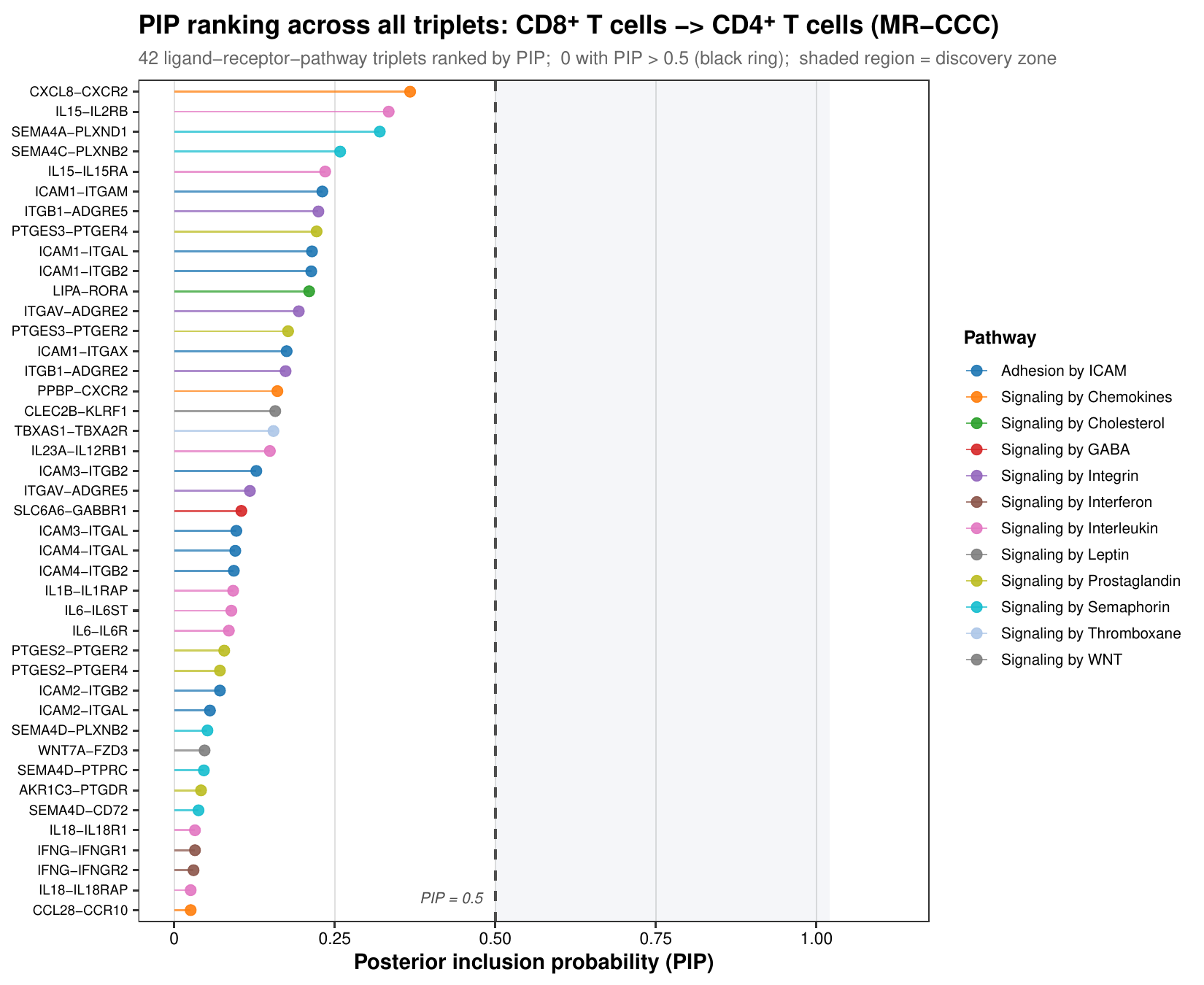}
\caption{\textbf{PIP ranking for the CD8$^+$ T cells\,$\rightarrow$\,CD4\(^+\) T cells analysis.} All 42 ligand--receptor--pathway triplets across 882 donors. Points are colored by pathway; the dashed
  vertical line marks the discovery threshold of PIP \(= 0.5\); the
  shaded region to the right is the discovery zone. No triplet
  reached the discovery threshold; the highest-ranking pair is
  \textit{CXCL8}--\textit{CXCR2} (PIP \(= 0.37\)), and all 42
  triplets have PIP \(< 0.4\).}
\label{fig:supp_CD8_CD4_pip}
\end{figure}

\begin{figure}[htbp]
\centering
\includegraphics[width=\textwidth]{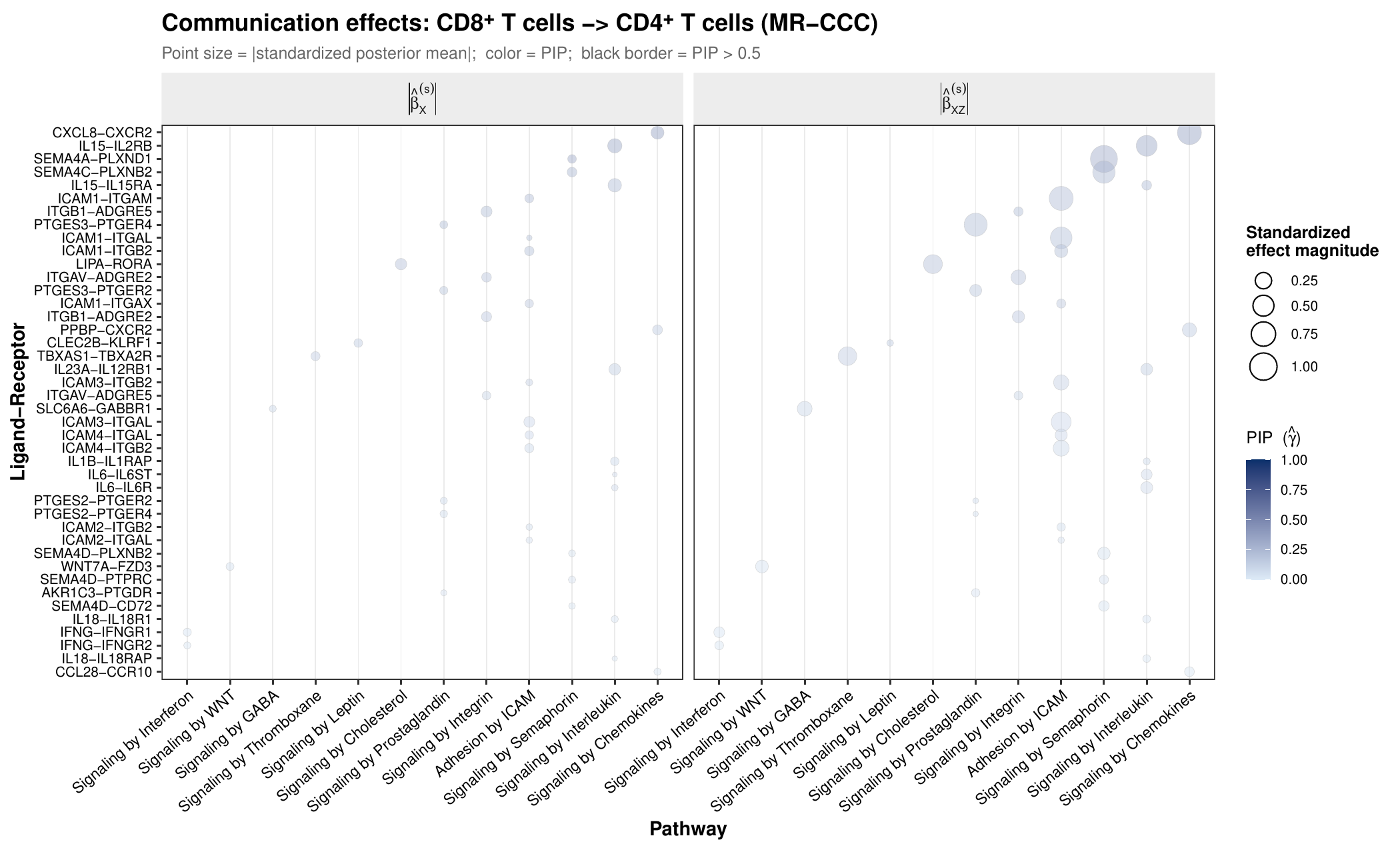}
\caption{\textbf{Standardized posterior effects for the CD8\(^+\) T cells
  \(\rightarrow\) CD4\(^+\) T cells analysis.} Left panel: absolute
  main ligand effect \(|\hat{\beta}_X^{(s)}|\); right panel: absolute
  receptor-modulated interaction effect \(|\hat{\beta}_{XZ}^{(s)}|\).
  Point size encodes effect magnitude; fill color encodes PIP. No
  triplet has a black border (PIP \(> 0.5\)). Effect magnitudes are
  small throughout both panels and no pathway cluster shows a
  concentration of large high-PIP effects, consistent with the
  absence of discoveries in this direction.}
\label{fig:supp_CD8_CD4_bubble}
\end{figure}

\FloatBarrier

\subsubsection{CD8\(^+\) T Cells \(\rightarrow\) NK Cells}
\label{supp:CD8_NK}
Across 830 donors, MR-CCC evaluated 42 ligand--receptor--pathway
triplets for the CD8\(^+\) T cell (sender) to NK cell (receiver)
direction and identified one high-confidence causal communication
signal with PIP exceeding 0.5: \textit{ICAM1}--\textit{ITGAX}
within the Adhesion by ICAM pathway (PIP \(= 0.642\))
(Figures~\ref{fig:supp_CD8_NK_pip}--\ref{fig:supp_CD8_NK_curves}).

The sole discovery, \textit{ICAM1}--\textit{ITGAX}
(PIP \(= 0.642\)), implicates CD8\(^+\) T cell intercellular
adhesion molecule 1 (\textit{ICAM1}) as a causal regulator of NK
cell Adhesion pathway activity through the integrin \(\alpha\)X chain
(CD11c, encoded by \textit{ITGAX}), which pairs with \(\beta_2\)
(CD18) to form complement receptor 4 (CR4, p150,95) expressed on
mature NK cells and myeloid
cells~\cite{myones1988cd11c,dustin1986icam}.
The near-zero main effect (\(\hat{\beta}_X = 0.024\)) and large
negative interaction (\(\hat{\beta}_{XZ} = -2.25\)) produce a
curve that crosses zero approximately \(0.01\) standard deviations
above the mean \textit{ITGAX} level---essentially at the population
mean.
The causal signal is therefore almost entirely interaction-driven:
NK cells expressing \textit{ITGAX} at or above the population mean
experience a steeply increasing suppression of Adhesion pathway
activity upon CD8\(^+\) T cell ICAM1 contact, while those with
below-mean CD11c expression are effectively unaffected.
This receptor-gated suppressive pattern mirrors the
\textit{ICAM1}--\textit{ITGAX} discovery in the B cells
\(\rightarrow\) Monocytes direction (PIP \(= 0.562\),
\(\hat{\beta}_{XZ} = -2.04\)), where CD11c-high monocytes
similarly experienced contact-mediated suppression via the same
integrin \(\alpha\)X axis, suggesting that CD11c expression level
is a general gate for ICAM1-mediated suppressive adhesion contacts
across multiple receiver cell types.

Notably, several triplets narrowly missed the discovery threshold.
\textit{CLEC2B}--\textit{KLRF1} (PIP \(= 0.498\)) was the
highest-ranking sub-threshold triplet, with a PIP essentially at
the discovery boundary.
This pair implicates AICL (activation-induced C-type lectin, encoded
by \textit{CLEC2B}) on CD8\(^+\) T cells as a ligand for the
activating NK receptor NKp80 (encoded by \textit{KLRF1}), an axis
established to mediate mutual cytotoxic activation between NK cells
and myeloid cells, and expressed on activated T
cells~\cite{welte2006nkp80}.
The moderately negative main effect (\(\hat{\beta}_X = -0.097\)) and
negative interaction (\(\hat{\beta}_{XZ} = -0.495\)) produce a
predominantly suppressive curve for NK cells expressing
\textit{KLRF1} above approximately \(0.20\) standard deviations
below the mean, encompassing the majority of the observed
population.
The \textit{PTGES3}--\textit{PTGER2} triplet (PIP \(= 0.423\))
again represents the PGE\(_2\)--EP2 axis on NK cells, which was
discovered from B cell (PIP \(= 0.556\)) and CD4\(^+\) T cell
(PIP \(= 0.625\)) senders but remains sub-threshold from the CD8\(^+\)
sender, suggesting that CD8\(^+\) T cells are less dominant
prostaglandin producers than B cells or CD4\(^+\) T cells in this
cohort~\cite{kalinski2012pge2}.
\textit{ICAM1}--\textit{ITGAM} (PIP \(= 0.375\)) and
\textit{ICAM2}--\textit{ITGAL} (PIP \(= 0.369\)) were notable
Adhesion by ICAM near-misses, both showing near-zero main effects
with moderate negative interactions (\(\hat{\beta}_{XZ} = -1.14\)
and \(-0.359\), respectively), indicating receptor-density-gated
suppressive contacts that fell below the discovery threshold.
\textit{PTGES3}--\textit{PTGER4} (PIP \(= 0.367\)) provided
parallel evidence for a sub-threshold PGE\(_2\)--EP4 axis
through a distinct prostaglandin receptor.
\textit{IFNG}--\textit{IFNGR2} (PIP \(= 0.209\)) and
\textit{IFNG}--\textit{IFNGR1} (PIP \(= 0.137\)) were both
sub-threshold; the contrast with the CD4\(^+\) T cells
\(\rightarrow\) NK cells direction, where
\textit{IFNG}--\textit{IFNGR2} was discovered (PIP \(= 0.542\)),
likely reflects greater heterogeneity in IFN-\(\gamma\) production
across CD8\(^+\) T cell states at the population level compared with
the more coordinated Th1-driven IFN-\(\gamma\) output of CD4\(^+\)
T cells~\cite{schroder2004interferon,bach1997ifngr}.

\begin{figure}[htbp]
\centering
\includegraphics[width=\textwidth]{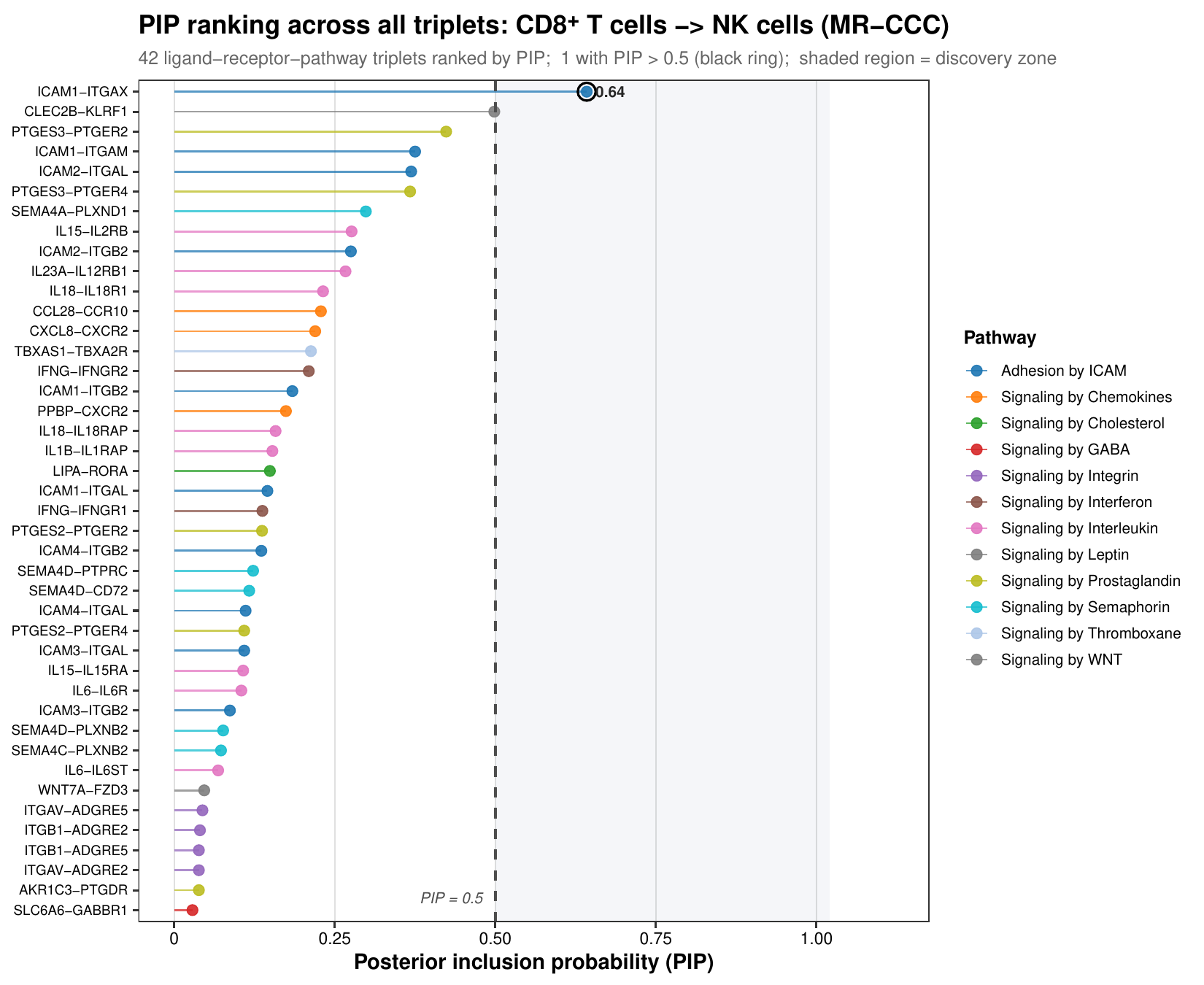}
\caption{\textbf{PIP ranking for the CD8$^+$ T cells\,$\rightarrow$\,NK cells analysis.} All 42 ligand--receptor--pathway triplets across 830
  donors. Points are colored by pathway; the dashed vertical line
  marks the discovery threshold of PIP \(= 0.5\); the shaded region
  to the right is the discovery zone. The single discovered triplet,
  \textit{ICAM1}--\textit{ITGAX} (PIP \(= 0.64\)), is identified
  with a black ring. The second-ranked triplet,
  \textit{CLEC2B}--\textit{KLRF1} (PIP \(= 0.498\)), falls
  marginally below the threshold.}
\label{fig:supp_CD8_NK_pip}
\end{figure}

\begin{figure}[htbp]
\centering
\includegraphics[width=\textwidth]{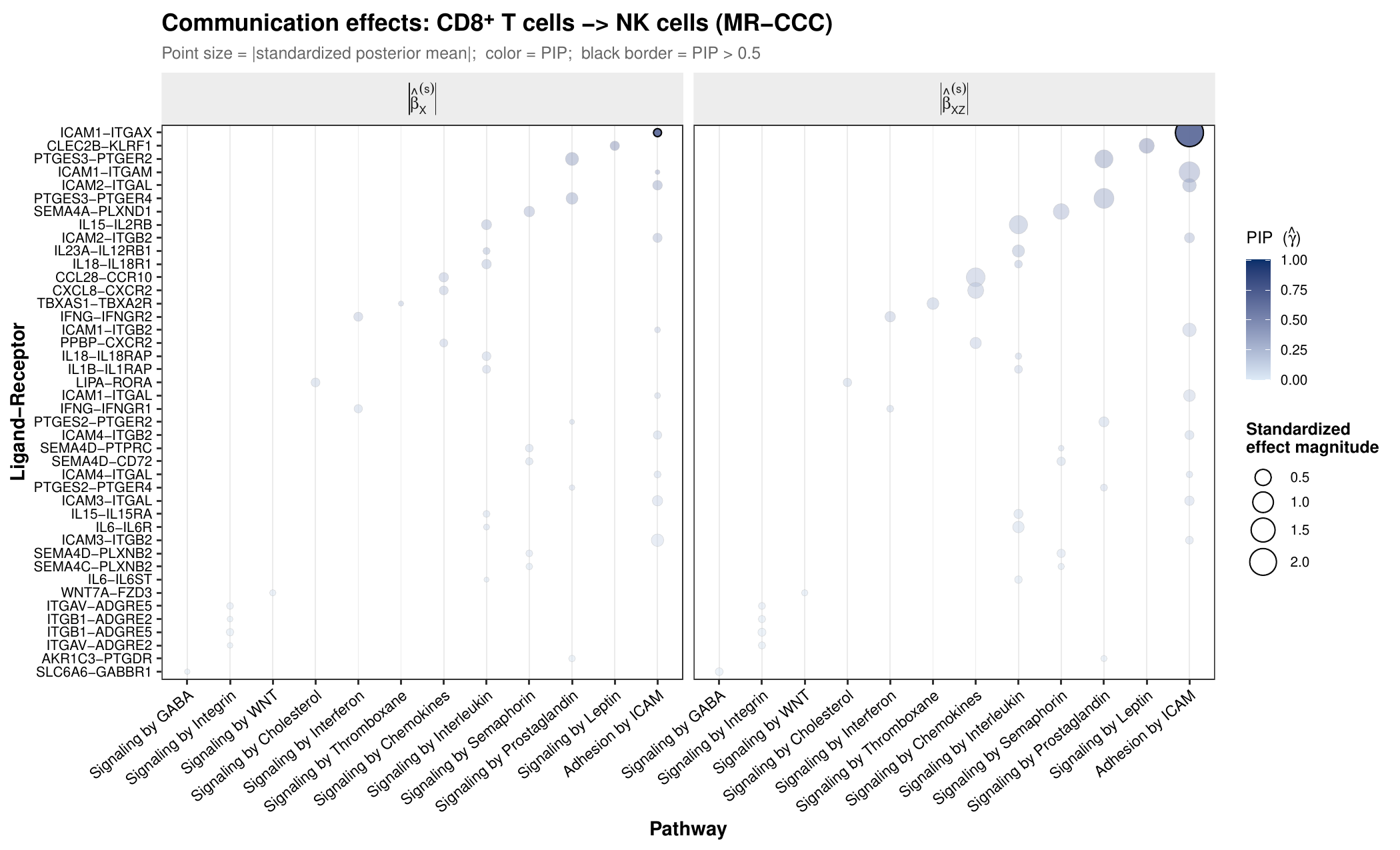}
\caption{\textbf{Standardized posterior effects for the CD8\(^+\) T cells
  \(\rightarrow\) NK cells analysis.} Left panel: absolute main
  ligand effect \(|\hat{\beta}_X^{(s)}|\); right panel: absolute
  receptor-modulated interaction effect \(|\hat{\beta}_{XZ}^{(s)}|\).
  Point size encodes effect magnitude; fill color encodes PIP; the
  single black-bordered point identifies the discovered triplet
  \textit{ICAM1}--\textit{ITGAX}. The right panel shows a larger
  interaction magnitude for \textit{ICAM1}--\textit{ITGAX} relative
  to its main effect in the left panel, consistent with the
  near-purely interaction-driven causal signal.}
\label{fig:supp_CD8_NK_bubble}
\end{figure}

\begin{figure}[htbp]
\centering
\includegraphics[width=\textwidth]{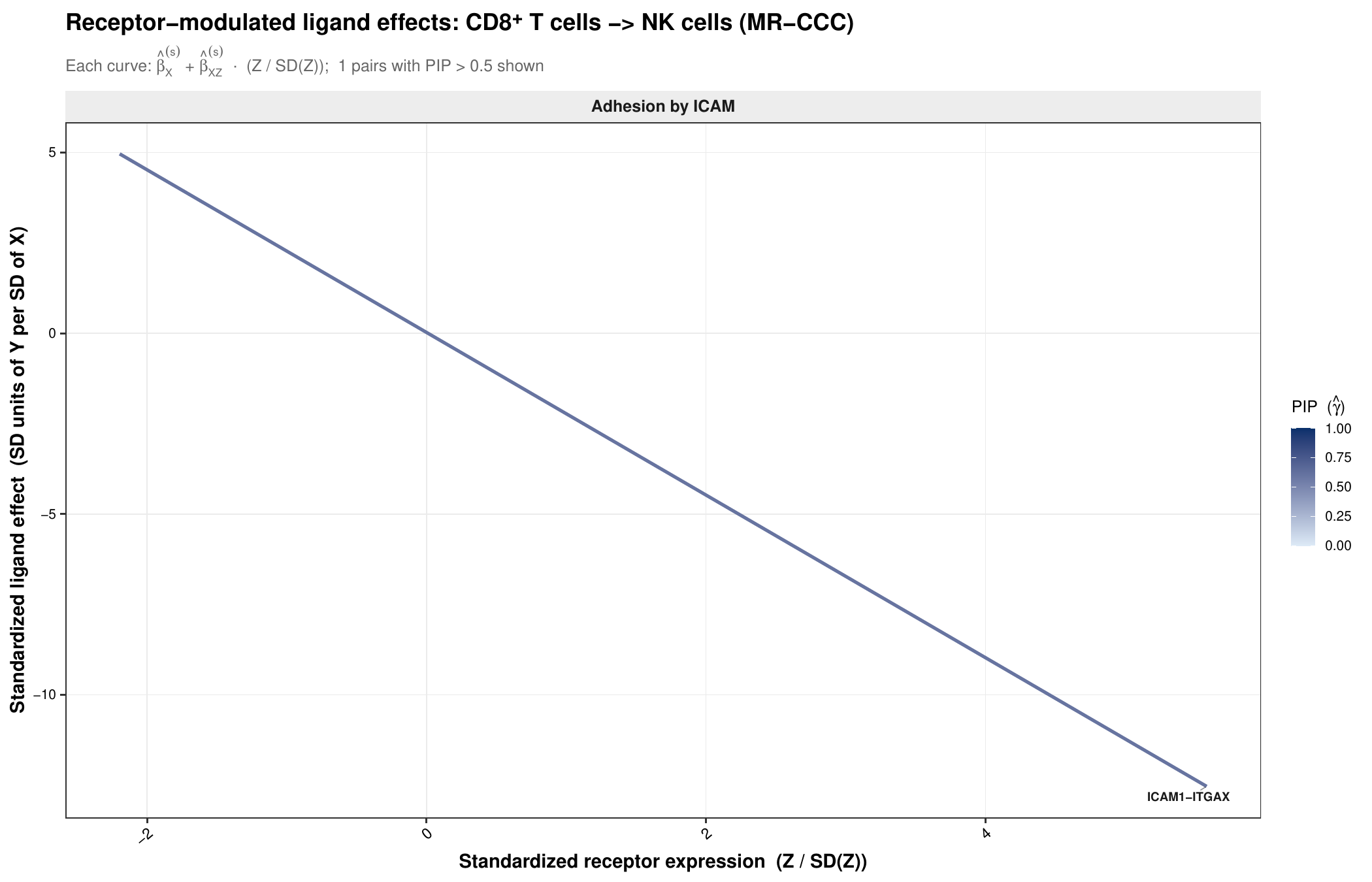}
\caption{\textbf{Receptor-modulated effect curves
  (\(\hat{\beta}_X + \hat{\beta}_{XZ} \cdot Z/\mathrm{SD}(Z)\))
  for the CD8\(^+\) T cells \(\rightarrow\) NK cells analysis.}
  Only the single discovery pair (PIP \(> 0.5\)) is displayed in
  the Adhesion by ICAM panel. The
  \textit{ICAM1}--\textit{ITGAX} curve crosses zero essentially at
  the population-mean \textit{ITGAX} level and descends steeply
  to strongly negative values at high CD11c expression, consistent
  with a near-purely interaction-driven suppressive signal.}
\label{fig:supp_CD8_NK_curves}
\end{figure}
\FloatBarrier

\subsubsection{CD8\(^+\) T Cells \(\rightarrow\) Monocytes}
\label{supp:CD8_Mono}

Across 678 donors, MR-CCC evaluated 42 ligand--receptor--pathway triplets for the
CD8\(^+\) T cell (sender) to Monocyte (receiver) direction and identified two
high-confidence causal communication signals with PIP exceeding 0.5, both involving
the cytokine \textit{IL18} paired with components of its heterodimeric receptor complex
(Figures~\ref{fig:supp_CD8_Mono_pip}--\ref{fig:supp_CD8_Mono_curves}).

\noindent\textbf{Signaling by Interleukin.}
\textit{IL18}--\textit{IL18R1} attained PIP $= 0.82$ with
$\hat{\beta}_X = 0.515$ and $\hat{\beta}_{XZ} = 1.63$, indicating that the positive main
effect of \textit{IL18} on monocyte Interleukin pathway activity is \emph{amplified} as
monocyte \textit{IL18R1} levels rise. The sign-reversal threshold
$Z^{*} = -\hat{\beta}_X / \hat{\beta}_{XZ} \approx -0.32$ SD lies well below the
population mean, so the causal effect is positive across essentially the entire monocyte
population and increases monotonically with receptor abundance---consistent with a
feed-forward architecture in which IL-18-responsive monocytes upregulate the very receptor
subunit (\textit{IL18R1}, the ligand-binding $\alpha$ chain) that mediates the
response~\cite{okamura1995interleukin18}.

\textit{IL18}--\textit{IL18RAP} attained PIP $= 0.85$ with
$\hat{\beta}_X = 0.483$ and $\hat{\beta}_{XZ} = -1.87$. Here the receptor covariate is
\textit{IL18RAP}, the signal-transducing $\beta$ (accessory-protein) chain of the same
heterodimer~\cite{born1998il18rap}. The negative $\hat{\beta}_{XZ}$ means that monocytes
expressing high levels of the accessory chain show an \emph{attenuated}---and, above the
sign-reversal threshold $Z^{*} \approx +0.26$ SD, \emph{reversed}---causal effect of
\textit{IL18}. Together with the \textit{IL18}--\textit{IL18R1} result, this suggests
that the two receptor chains are not simply co-regulated: monocytes with relatively
abundant $\alpha$ chain are sensitized to IL-18, whereas monocytes with relatively
abundant $\beta$ accessory chain may exhibit saturated or downregulated signaling, and
can even show a negative response---potentially through receptor desensitization or
negative-feedback pathways downstream of IL-18R~\cite{okamura1995interleukin18,born1998il18rap}.

The most notable sub-threshold triplets were \textit{ICAM3}--\textit{ITGB2}
(PIP $= 0.43$), \textit{IFNG}--\textit{IFNGR2} (PIP $= 0.40$), and
\textit{IFNG}--\textit{IFNGR1} (PIP $= 0.39$). The IFN-$\gamma$ pairs are
biologically noteworthy: CD8\(^+\) T cells are a principal source of IFN-$\gamma$ and
monocytes express both receptor chains at high levels, yet neither triplet crossed the
discovery threshold here, in contrast to the CD4\(^+\) T cells \(\rightarrow\) NK cells
direction where \textit{IFNG}--\textit{IFNGR2} was robustly discovered
(PIP $= 0.542$). This likely reflects greater donor-to-donor variability in CD8\(^+\)
IFN-$\gamma$ secretion relative to CD4\(^+\) Th1 cells in this cohort, or a more
complex multivariate architecture for CD8-mediated monocyte activation that is not fully
captured by a single ligand--receptor pair~\cite{bach1997ifngr}.

\begin{figure}[htbp]
\centering
\includegraphics[width=\textwidth]{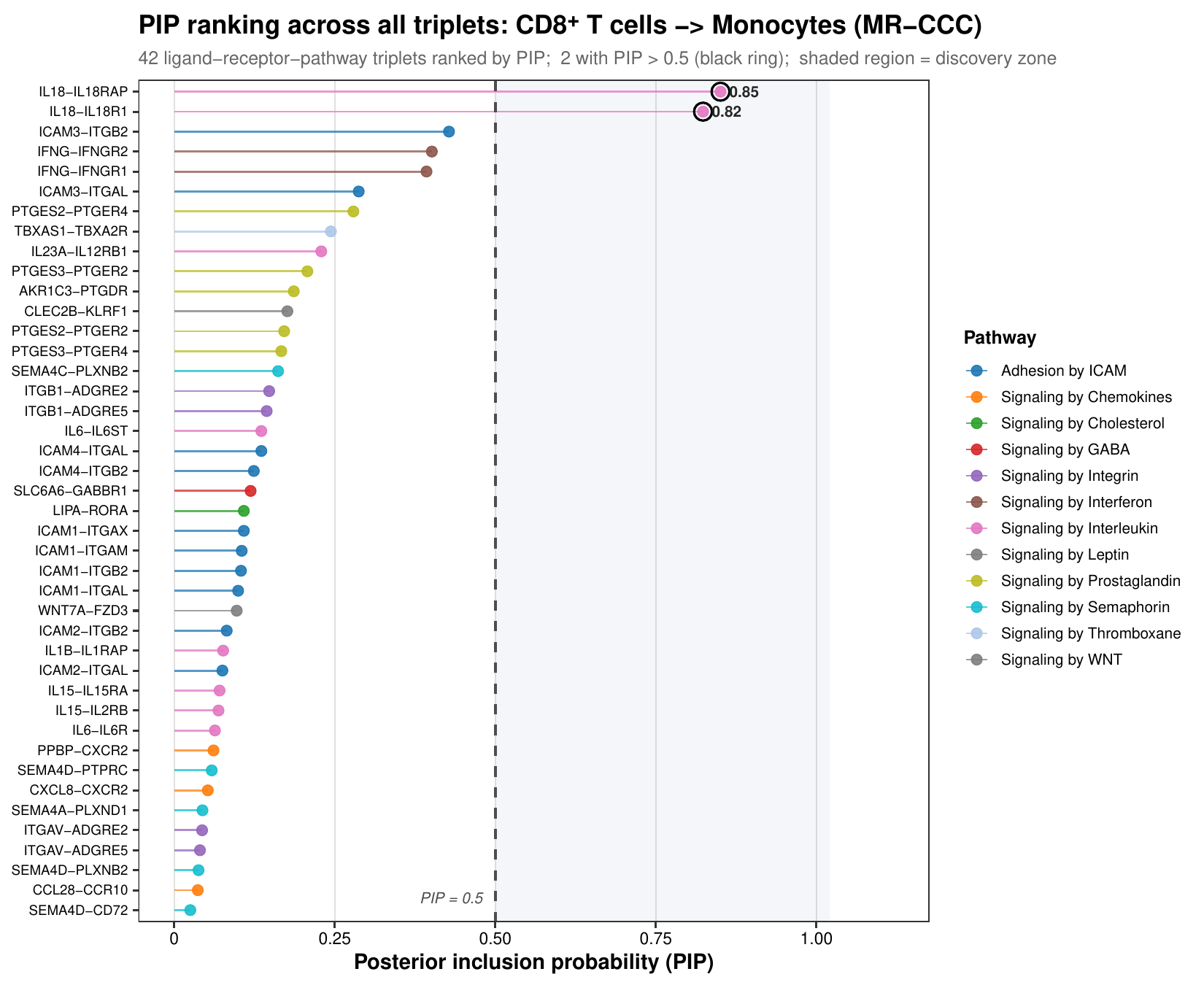}
\caption{\textbf{PIP ranking for the CD8$^+$ T cells\,$\rightarrow$\,Monocytes analysis.} All 42 ligand--receptor--pathway triplets across 678
  donors. Points are colored by pathway; the dashed vertical line
  marks the discovery threshold of PIP \(= 0.5\); the shaded region
  to the right is the discovery zone. Black rings identify the two
  discovered triplets: \textit{IL18}--\textit{IL18RAP}
  (PIP \(= 0.85\)) and \textit{IL18}--\textit{IL18R1}
  (PIP \(= 0.82\)), both within Signaling by Interleukin.}
\label{fig:supp_CD8_Mono_pip}
\end{figure}

\begin{figure}[htbp]
\centering
\includegraphics[width=\textwidth]{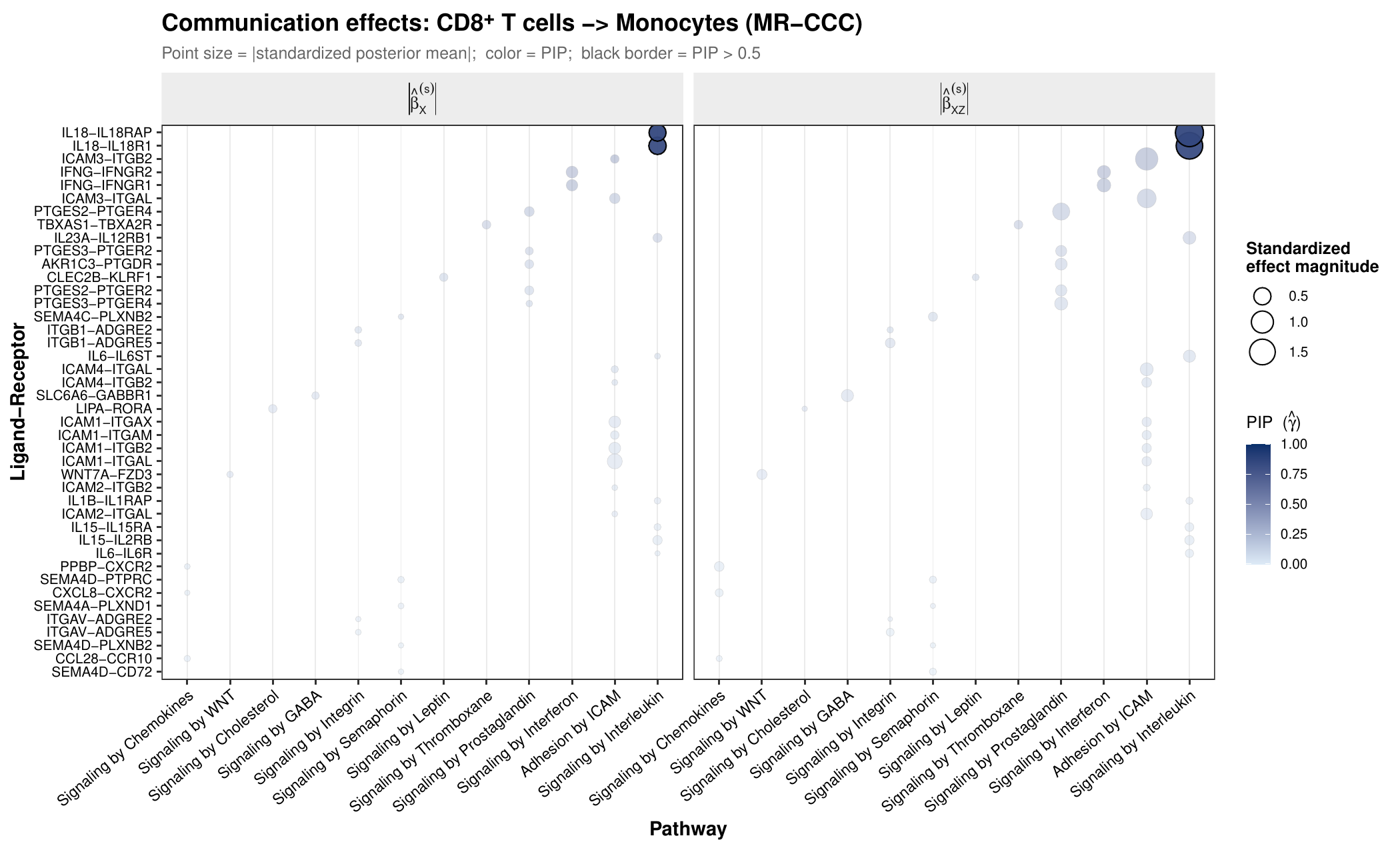}
\caption{\textbf{Standardized posterior effects for the CD8\(^+\) T cells
  \(\rightarrow\) Monocytes analysis.} Left panel: absolute main ligand
  effect \(|\hat{\beta}_X^{(s)}|\); right panel: absolute
  receptor-modulated interaction effect \(|\hat{\beta}_{XZ}^{(s)}|\).
  Point size encodes effect magnitude; fill color encodes PIP; black
  borders identify the two discovered triplets. Both
  \textit{IL18}--\textit{IL18RAP} and \textit{IL18}--\textit{IL18R1}
  show appreciable effects in both panels, with dominant interaction
  components relative to their main effects.}
\label{fig:supp_CD8_Mono_bubble}
\end{figure}

\begin{figure}[htbp]
\centering
\includegraphics[width=\textwidth]{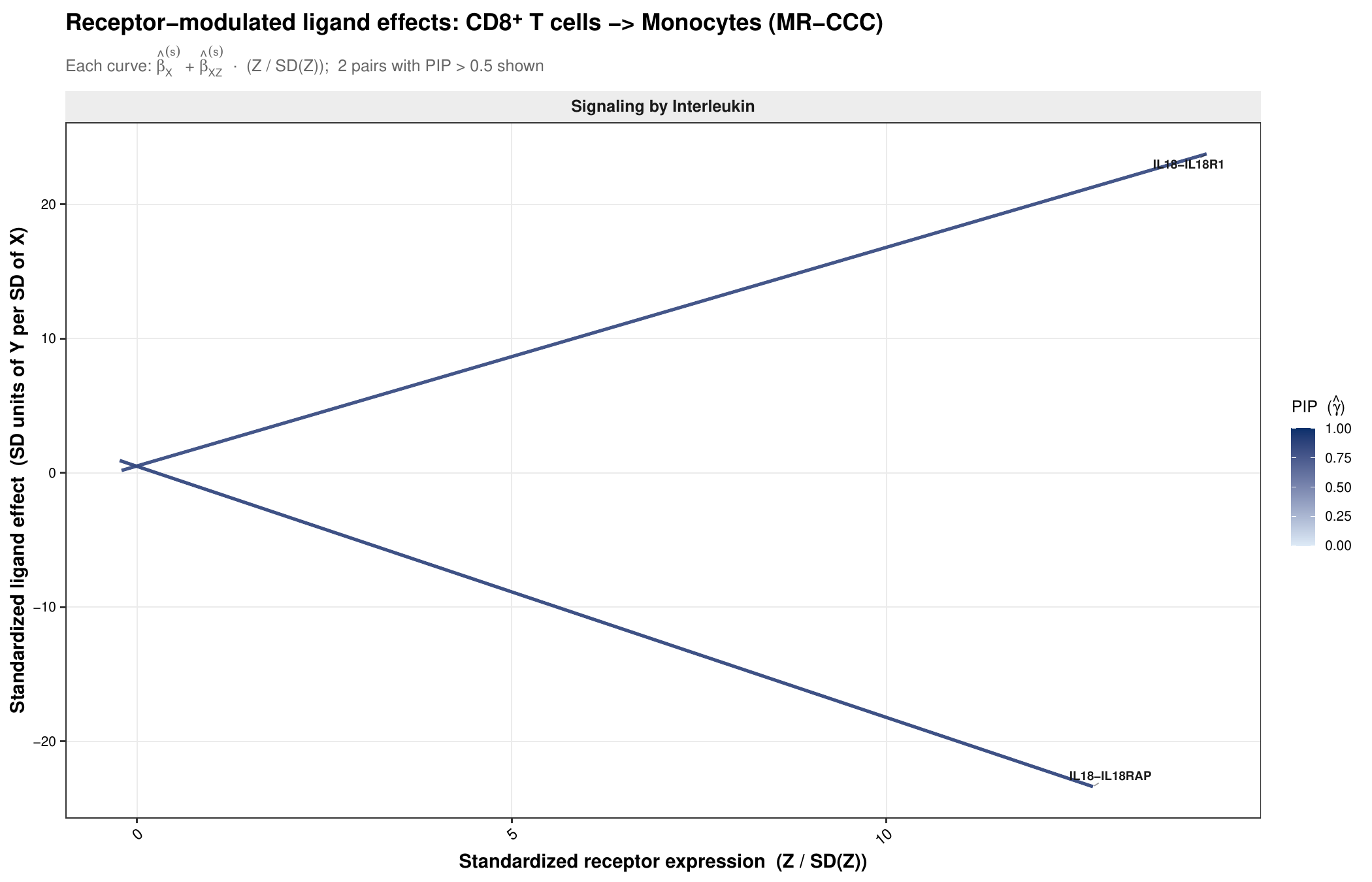}
\caption{\textbf{Receptor-modulated effect curves
  (\(\hat{\beta}_X + \hat{\beta}_{XZ} \cdot Z/\mathrm{SD}(Z)\))
  for the CD8\(^+\) T cells \(\rightarrow\) Monocytes analysis.}
  Only the two discovery pairs (PIP \(> 0.5\)) are displayed in the
  single pathway panel with confirmed discoveries (Signaling by
  Interleukin). The two curves slope in opposite directions:
  \textit{IL18}--\textit{IL18R1} rises monotonically, remaining
  positive across essentially the full observed receptor range with a
  zero crossing near \(-0.32\) SD below the mean;
  \textit{IL18}--\textit{IL18RAP} descends from positive to negative,
  crossing zero near \(+0.26\) SD above the mean. The opposing slopes
  reflect the divergent interaction signs for the $\alpha$ and $\beta$
  chains of the IL-18 receptor heterodimer.}
\label{fig:supp_CD8_Mono_curves}
\end{figure}
\FloatBarrier

\clearpage
\subsection{NK Cells as Sender}

\subsubsection{NK Cells \(\rightarrow\) B Cells}
\label{supp:NK_B}

Across 818 donors, MR-CCC evaluated 42 ligand--receptor--pathway
triplets for the NK cell (sender) to B cell (receiver) direction and
identified six high-confidence causal communication signals with PIP
exceeding 0.5: four Adhesion by ICAM pairs
(\textit{ICAM1}--\textit{ITGAL}, PIP \(= 0.782\);
\textit{ICAM2}--\textit{ITGAL}, PIP \(= 0.738\);
\textit{ICAM1}--\textit{ITGAM}, PIP \(= 0.603\);
\textit{ICAM1}--\textit{ITGAX}, PIP \(= 0.507\)),
one Signaling by GABA triplet (\textit{SLC6A6}--\textit{GABBR1},
PIP \(= 0.618\)), and one Signaling by Prostaglandin triplet
(\textit{PTGES3}--\textit{PTGER2}, PIP \(= 0.544\))
(Figures~\ref{fig:supp_NK_B_pip}--\ref{fig:supp_NK_B_curves}).

\noindent\textbf{Adhesion by ICAM signals.}
The four discovered ICAM--integrin triplets collectively implicate NK cell
adhesion molecule contacts as causal regulators of B cell Adhesion pathway
activity, paralleling the multi-ICAM discovery pattern observed from CD4\(^+\)
and CD8\(^+\) T cell senders~\cite{springer1990adhesion,dustin1986icam}.
The dominant pair, \textit{ICAM1}--\textit{ITGAL} (PIP \(= 0.782\),
\(\hat{\beta}_X = 0.961\), \(\hat{\beta}_{XZ} = -2.13\)), has a large positive
main effect and a negative interaction, yielding a sign-reversing curve that
crosses zero approximately \(0.45\) standard deviations above the mean
\textit{ITGAL} level. B cells expressing LFA-1 below this threshold experience a
net stimulatory response to NK cell ICAM-1 engagement; those with above-average
LFA-1 density transition to a progressively suppressive effect. This receptor-gated
switch mirrors the same \textit{ICAM1}--\textit{ITGAL} pattern observed from CD4\(^+\)
and CD8\(^+\) T cell senders, indicating that LFA-1 density on B cells is a
general determinant of the direction of ICAM-1-mediated adhesion
signals regardless of the lymphocyte sender type.

In contrast, \textit{ICAM2}--\textit{ITGAL} (PIP \(= 0.738\),
\(\hat{\beta}_X = -0.046\), \(\hat{\beta}_{XZ} = 5.29\)) has a near-zero main
effect and an unusually large positive interaction, placing the zero crossing at
only \(\approx 0.01\) standard deviations above the mean---essentially at the
population mean. The causal signal is therefore almost entirely interaction-driven:
B cells expressing \textit{ITGAL} at or above the mean experience a steeply
amplified positive response to NK cell ICAM2, while those below the mean are
effectively unaffected. The opposing interaction signs for \textit{ICAM1}--\textit{ITGAL}
(\(\hat{\beta}_{XZ} = -2.13\)) and \textit{ICAM2}--\textit{ITGAL}
(\(\hat{\beta}_{XZ} = 5.29\)) through the same LFA-1 receptor recapitulate the
divergent ICAM1 versus ICAM2 pattern seen from CD4\(^+\) and CD8\(^+\) senders:
at high LFA-1 densities, NK cell ICAM1 contact is suppressive while NK cell
ICAM2 contact is strongly stimulatory.

\textit{ICAM1}--\textit{ITGAM} (PIP \(= 0.603\), \(\hat{\beta}_X = 0.786\),
\(\hat{\beta}_{XZ} = -3.08\)) follows the same receptor-gated suppressive pattern
as \textit{ICAM1}--\textit{ITGAL}, with the zero crossing at approximately
\(0.26\) standard deviations above the mean \textit{ITGAM} level; B cells
expressing the Mac-1 \(\alpha\) chain (CD11b) above this threshold experience
progressively stronger suppression of Adhesion pathway activity upon NK cell
ICAM-1 contact~\cite{myones1988cd11c}.

The fourth ICAM discovery, \textit{ICAM1}--\textit{ITGAX} (PIP \(= 0.507\),
\(\hat{\beta}_X = 0.530\), \(\hat{\beta}_{XZ} = 0.056\)), is qualitatively
distinct from the other three: the interaction term is negligibly small and
positive, placing the zero crossing at approximately \(-9.4\) standard deviations
below the mean---far outside the observed data range. This triplet therefore
represents a near-purely main-effect discovery: NK cell ICAM-1 exerts a
uniformly positive causal effect on the Adhesion pathway activity of B cells
expressing \textit{ITGAX} (CD11c), independent of receptor expression level.
CD11c\textsuperscript{+} B cells are a recognized subset (including age-associated B cells
and activated B cell subpopulations) with an altered adhesion receptor
profile~\cite{dustin1986icam}, and the flat, uniformly positive curve suggests
that NK cell ICAM-1 constitutively engages this subset irrespective of CD11c
density.

\noindent\textbf{Signaling by GABA.}
\textit{SLC6A6}--\textit{GABBR1} (PIP \(= 0.618\), \(\hat{\beta}_X = -0.042\),
\(\hat{\beta}_{XZ} = -3.32\)) implicates NK cell taurine and \(\beta\)-alanine
transporter activity (\textit{SLC6A6}) as a causal suppressor of B cell GABA-B
receptor-mediated pathway activity. With a near-zero main effect and a large
negative interaction, the zero crossing lies essentially at the population-mean
\textit{GABBR1} expression level (\(Z^* \approx -0.013\) SD): B cells expressing
\textit{GABBR1} at or above the mean experience a steeply increasing suppression
of GABA signaling pathway activity as NK cell taurine-transport activity rises.
Taurine, the primary substrate of SLC6A6, is a known partial agonist of
GABA\(_B\) receptors and provides a mechanistic link between NK cell amino-acid
transport and GABA-B receptor signaling in neighboring
cells~\cite{tian1999gaba,bjurstrom2008gaba}.
This triplet was also discovered in the B cells \(\rightarrow\) CD4\(^+\) T cells
direction (PIP \(= 0.54\)), but there the interaction was large and positive
(\(\hat{\beta}_{XZ} = 2.01\)), yielding a stimulatory rather than suppressive
receptor-amplified effect. The sign reversal of the \textit{GABBR1} interaction
between receiver cell types---stimulatory for CD4\(^+\) T cells, suppressive for
B cells---illustrates how the same ligand--receptor axis can mediate qualitatively
different outcomes depending on the downstream signaling environment of the
recipient cell.

\noindent\textbf{Signaling by Prostaglandin.}
\textit{PTGES3}--\textit{PTGER2} (PIP \(= 0.544\), \(\hat{\beta}_X = 0.438\),
\(\hat{\beta}_{XZ} = -1.97\)) links NK cell prostaglandin E synthase 3
(\textit{PTGES3}) to the EP2 prostaglandin receptor (encoded by \textit{PTGER2})
on B cells. The positive main effect and negative interaction produce a
sign-reversing curve that crosses zero approximately \(0.22\) standard deviations
above the mean \textit{PTGER2} level: B cells with low EP2 expression experience
a net stimulatory response to NK cell PGE\(_2\) output, while those above the
threshold undergo progressively stronger suppression of Prostaglandin pathway
activity~\cite{kalinski2012pge2}. This PGE\(_2\)--EP2 axis has now been
discovered across multiple sender types targeting NK cells (B cell sender,
PIP \(= 0.556\); CD4\(^+\) T cell sender, PIP \(= 0.625\)), and here NK cells
themselves appear as a PGE\(_2\) sender, now acting on B cells as receiver.
The recurrence of \textit{PTGES3}--\textit{PTGER2} across directions involving
NK cells---both as receiver and as sender---suggests that PGE\(_2\)--EP2
signaling is a constitutively active immunomodulatory axis in the NK cell niche
of peripheral blood.

Several triplets narrowly missed the discovery threshold.
\textit{ICAM1}--\textit{ITGB2} (PIP \(= 0.480\)) was the highest-ranking
sub-threshold triplet, with a PIP at the discovery boundary; the modest positive
main effect (\(\hat{\beta}_X = 0.568\)) and positive interaction
(\(\hat{\beta}_{XZ} = 1.65\)) produce a monotonically increasing curve, and its
near-discovery suggests the \(\beta_2\) chain axis from NK cells is borderline
consistent at the population level.
\textit{IL6}--\textit{IL6ST} (PIP \(= 0.434\)) was a notable near-miss; the
IL-6/gp130 axis was discovered from both B cell (PIP \(= 0.554\)) and CD4\(^+\)
T cell (not discovered) senders, and its sub-threshold PIP here suggests NK
cells are not a dominant population-level source of IL-6 in peripheral blood.
\textit{PTGES3}--\textit{PTGER4} (PIP \(= 0.416\)) was the second prostaglandin
near-miss, indicating that the EP4 receptor axis on B cells is weaker than EP2
in this NK sender direction.

\textit{IFNG}--\textit{IFNGR1} (PIP \(= 0.046\)) and
\textit{IFNG}--\textit{IFNGR2} (PIP \(= 0.045\)) both received near-zero
support despite NK cells being a major physiological source of
IFN-\(\gamma\)~\cite{schroder2004interferon}. The very low PIPs likely reflect
that B cells express IFN-\(\gamma\) receptors at comparatively low levels
relative to myeloid and T cell populations in peripheral blood, limiting the
causal footprint of NK cell-derived IFN-\(\gamma\) on B cell Interferon pathway
activity at the population level.

\begin{figure}[htbp]
\centering
\includegraphics[width=\textwidth]{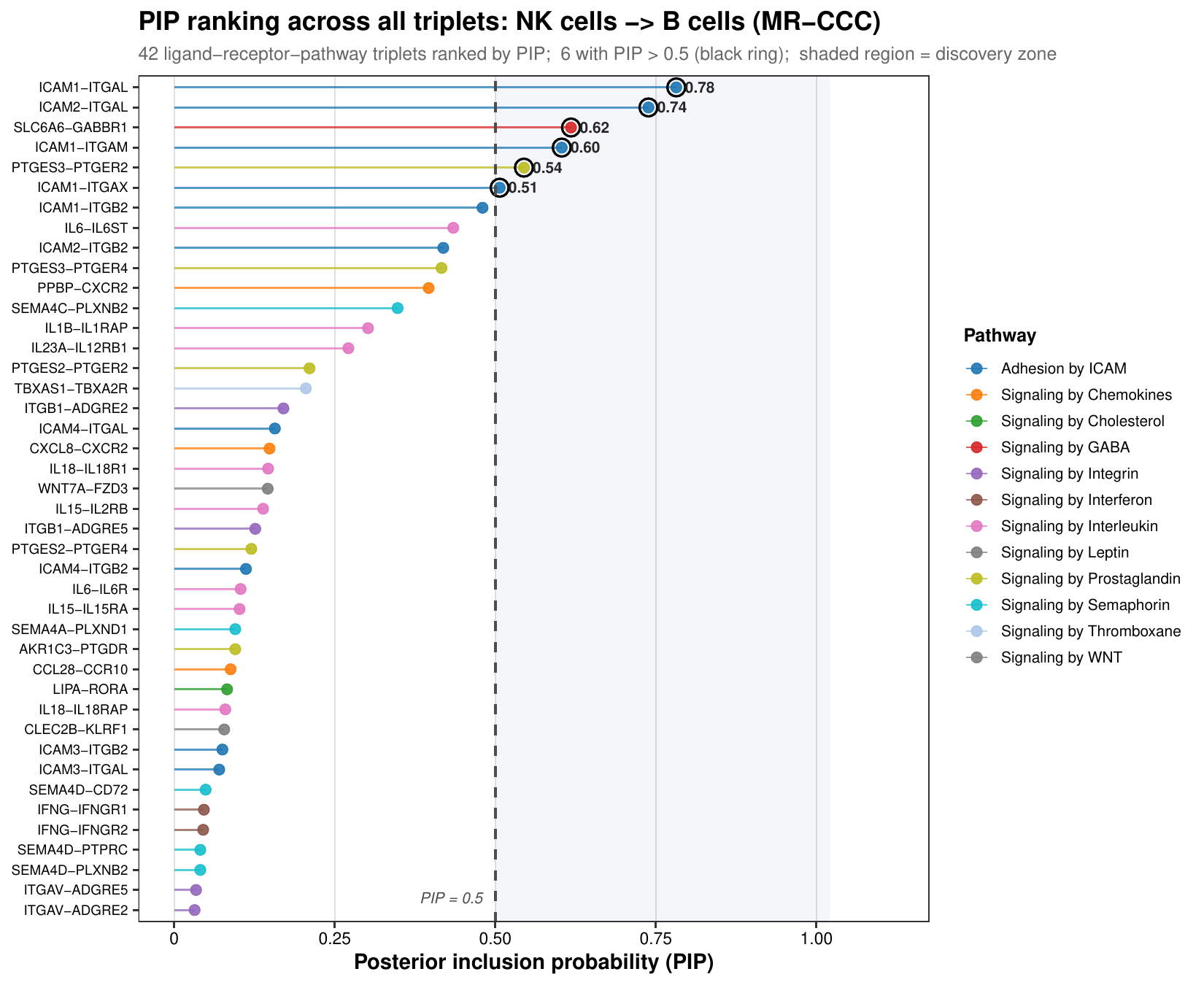}
\caption{\textbf{PIP ranking for the NK cells\,$\rightarrow$\,B cells analysis.} All 42 ligand--receptor--pathway triplets across 818 donors.
  Points are colored by pathway; the dashed vertical line marks the
  discovery threshold of PIP \(= 0.5\); the shaded region to the
  right is the discovery zone. Black rings identify the six discovered
  triplets: \textit{ICAM1}--\textit{ITGAL} (PIP \(= 0.78\)),
  \textit{ICAM2}--\textit{ITGAL} (PIP \(= 0.74\)),
  \textit{SLC6A6}--\textit{GABBR1} (PIP \(= 0.62\)),
  \textit{ICAM1}--\textit{ITGAM} (PIP \(= 0.60\)),
  \textit{PTGES3}--\textit{PTGER2} (PIP \(= 0.54\)), and
  \textit{ICAM1}--\textit{ITGAX} (PIP \(= 0.51\)).}
\label{fig:supp_NK_B_pip}
\end{figure}

\begin{figure}[htbp]
\centering
\includegraphics[width=\textwidth]{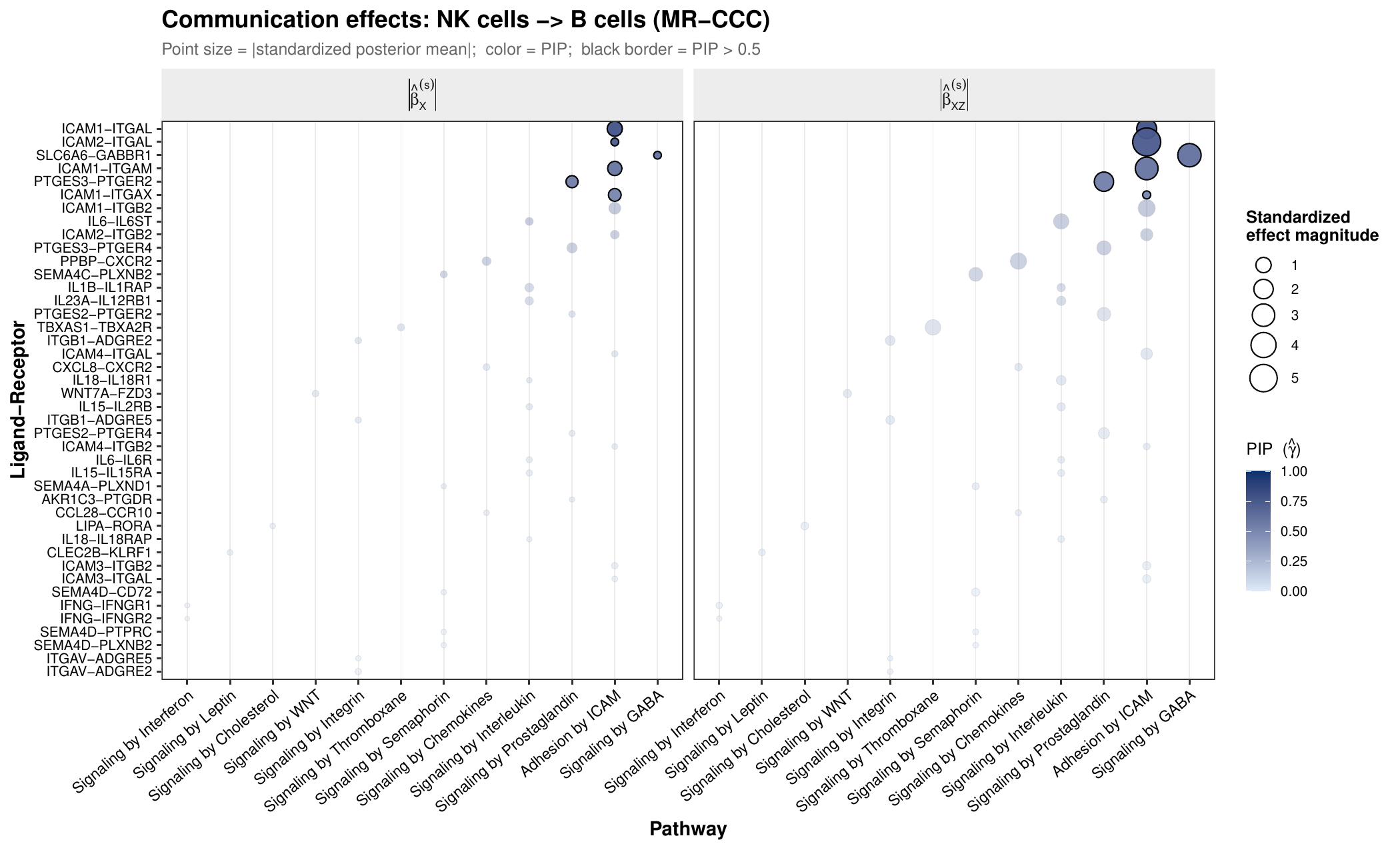}
\caption{\textbf{Standardized posterior effects for the NK cells \(\rightarrow\)
  B cells analysis.} Left panel: absolute main ligand effect
  \(|\hat{\beta}_X^{(s)}|\); right panel: absolute receptor-modulated
  interaction effect \(|\hat{\beta}_{XZ}^{(s)}|\). Point size encodes
  effect magnitude; fill color encodes PIP; black borders identify the
  six discovered triplets. \textit{ICAM2}--\textit{ITGAL} dominates
  the right panel with the largest interaction magnitude in this direction;
  \textit{ICAM1}--\textit{ITGAL} and \textit{ICAM1}--\textit{ITGAM} show
  appreciable effects in both panels; \textit{ICAM1}--\textit{ITGAX}
  shows a large main effect but near-zero interaction, consistent with
  its near-purely main-effect discovery profile.}
\label{fig:supp_NK_B_bubble}
\end{figure}

\begin{figure}[htbp]
\centering
\includegraphics[width=\textwidth]{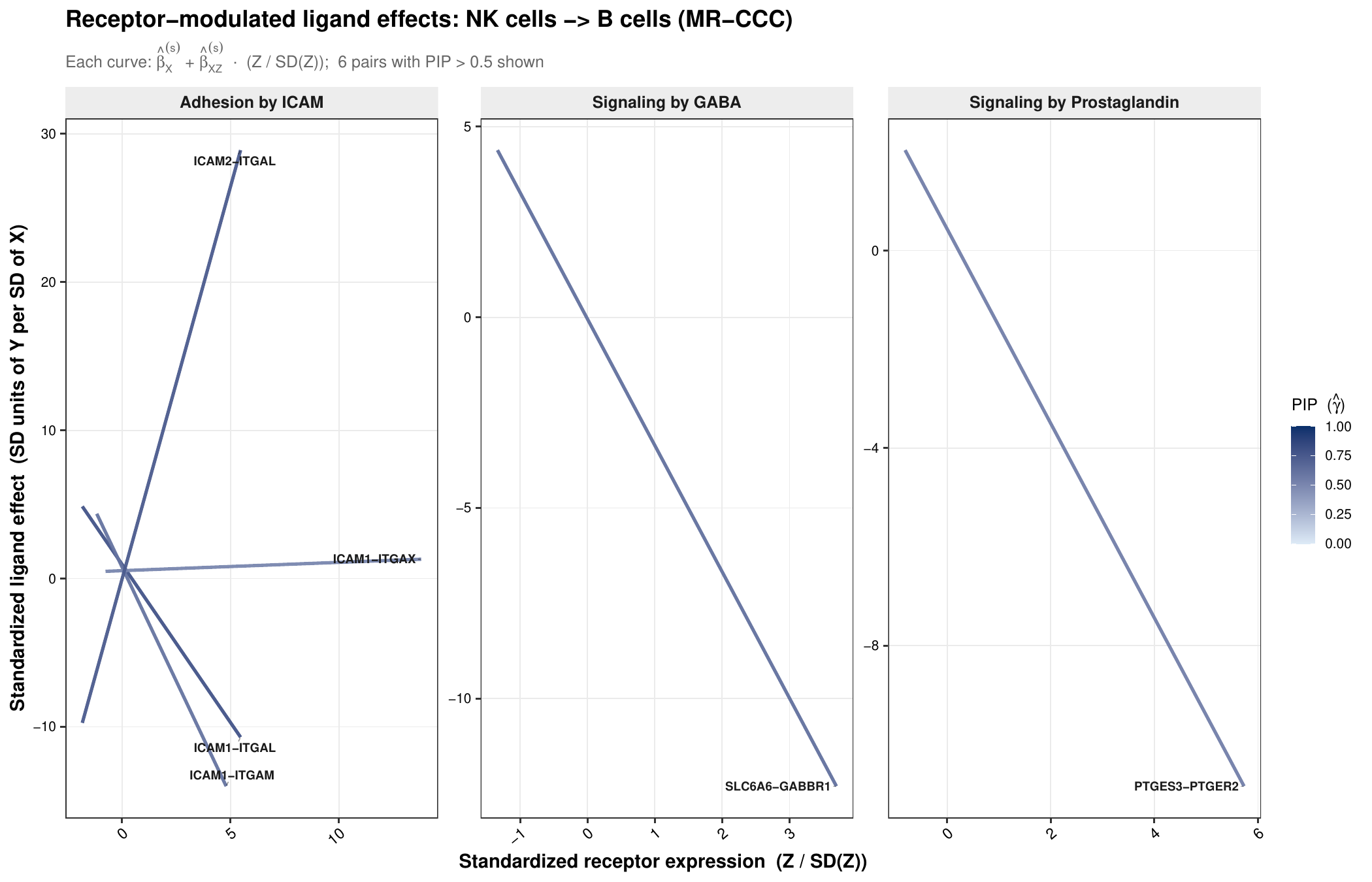}
\caption{\textbf{Receptor-modulated effect curves
  (\(\hat{\beta}_X + \hat{\beta}_{XZ} \cdot Z/\mathrm{SD}(Z)\))
  for the NK cells \(\rightarrow\) B cells analysis.} Only the six
  discovery pairs (PIP \(> 0.5\)) are displayed, grouped into the
  three pathway panels with confirmed discoveries (Adhesion by ICAM,
  Signaling by GABA, and Signaling by Prostaglandin). In the
  Adhesion by ICAM panel, four curves are visible:
  \textit{ICAM2}--\textit{ITGAL} rises most steeply from near zero,
  while \textit{ICAM1}--\textit{ITGAL} and \textit{ICAM1}--\textit{ITGAM}
  descend from positive to negative values at zero crossings near
  \(+0.45\) and \(+0.26\) standard deviations above the mean,
  respectively; \textit{ICAM1}--\textit{ITGAX} is nearly flat and
  positive across the observed receptor range. In the Signaling by
  GABA panel, the \textit{SLC6A6}--\textit{GABBR1} curve crosses zero
  essentially at the population-mean \textit{GABBR1} level and descends
  steeply to strongly negative values at high receptor expression. In
  the Signaling by Prostaglandin panel, the
  \textit{PTGES3}--\textit{PTGER2} curve crosses zero near \(+0.22\)
  standard deviations above the mean and becomes increasingly negative
  at high EP2 density.}
\label{fig:supp_NK_B_curves}
\end{figure}
\FloatBarrier

\subsubsection{NK Cells \(\rightarrow\) CD4\(^+\) T Cells}
\label{supp:NK_CD4}

Across 850 donors, MR-CCC evaluated 42 ligand--receptor--pathway
triplets for the NK cell (sender) to CD4\(^+\) T cell (receiver) direction
and identified one high-confidence causal communication signal with
PIP exceeding 0.5: \textit{ICAM1}--\textit{ITGAL} within the Adhesion by ICAM
pathway (PIP \(= 0.54\))
(Figures~\ref{fig:supp_NK_CD4_pip}--\ref{fig:supp_NK_CD4_curves}).

The sole discovery, \textit{ICAM1}--\textit{ITGAL} (PIP \(= 0.54\),
\(\hat{\beta}_X = -0.064\), \(\hat{\beta}_{XZ} = -0.910\)), implicates NK cell
intercellular adhesion molecule 1 (\textit{ICAM1}) as a causal suppressor of
CD4\(^+\) T cell Adhesion pathway activity through LFA-1 (encoded by
\textit{ITGAL})~\cite{dustin1986icam,springer1990adhesion}.
The near-zero main effect and moderate negative interaction place the zero
crossing at approximately \(0.07\) standard deviations \emph{below} the mean
\textit{ITGAL} level---essentially at the population mean.
The signal is therefore almost entirely interaction-driven: CD4\(^+\) T cells
expressing LFA-1 at or above the population mean experience a progressively
stronger suppression of Adhesion pathway activity upon NK cell ICAM-1
engagement, while those with the lowest LFA-1 expression are effectively
unaffected.
This suppressive NK--CD4 contact through ICAM-1--LFA-1 is consistent with the
established capacity of NK cells to exert regulatory restraint on T cell
activation through direct cell--cell adhesion, limiting excessive CD4\(^+\) T
cell responses in an LFA-1 density-dependent
manner~\cite{springer1990adhesion,dustin1986icam}.
The \textit{ICAM1}--\textit{ITGAL} triplet was also discovered in the NK cells
\(\rightarrow\) B cells direction (PIP \(= 0.782\),
\(\hat{\beta}_{XZ} = -2.13\)), where the same suppressive interaction sign was
observed; the considerably smaller interaction magnitude here
(\(\hat{\beta}_{XZ} = -0.910\) versus \(-2.13\)) suggests that NK cell ICAM-1
contact is a less potent suppressor of Adhesion pathway activity in CD4\(^+\)
T cells than in B cells at the population level.

Two triplets narrowly missed the discovery threshold and merit specific
attention. \textit{ICAM1}--\textit{ITGAX} (PIP \(= 0.491\),
\(\hat{\beta}_X = -0.154\), \(\hat{\beta}_{XZ} = 1.03\)) was the
highest-ranking sub-threshold triplet; the negative main effect and positive
interaction produce a sign-reversing curve that crosses zero approximately
\(0.15\) standard deviations above the mean \textit{ITGAX} level, with CD4\(^+\)
T cells expressing high CD11c showing a net positive Adhesion pathway response
to NK cell ICAM-1 contact.
\textit{ICAM4}--\textit{ITGAL} (PIP \(= 0.487\), \(\hat{\beta}_X = -0.156\),
\(\hat{\beta}_{XZ} = -4.20\)) was the second borderline near-miss; despite
failing to clear the threshold, it carries the largest interaction magnitude of
any triplet in this direction and a zero crossing only \(0.037\) standard
deviations below the mean, placing essentially the entire CD4\(^+\) T cell
population in a regime of steeply increasing suppression of Adhesion pathway
activity as NK cell ICAM4 production rises.
The co-occurrence of near-misses for both \textit{ICAM1}--\textit{ITGAX}
(PIP \(= 0.491\)) and \textit{ICAM4}--\textit{ITGAL} (PIP \(= 0.487\)) just
below the discovery boundary, alongside the confirmed \textit{ICAM1}--\textit{ITGAL}
discovery, indicates that the NK--CD4 adhesion axis involves multiple partially
consistent ICAM--integrin contacts at the population level that collectively
fall near the discovery threshold.

Among other sub-threshold triplets, \textit{SEMA4A}--\textit{PLXND1}
(PIP \(= 0.364\)) represents the highest-ranking Semaphorin pathway near-miss;
\textit{SEMA4C}--\textit{PLXNB2} (PIP \(= 0.261\)) provides parallel
evidence for sub-threshold NK semaphorin signaling to CD4\(^+\) T cells,
consistent with the established role of semaphorin-plexin contacts in
lymphocyte co-stimulation and cytoskeletal
regulation~\cite{suzuki2008sema4d,kumanogoh2013neuropilins}.
\textit{LIPA}--\textit{RORA} (PIP \(= 0.357\)) was a notable Cholesterol
pathway near-miss; this triplet was discovered from B cell senders in the
B cells \(\rightarrow\) CD4\(^+\) T cells direction (PIP \(= 0.686\)), but
here the PIP falls considerably below threshold, consistent with NK cells
having lower lysosomal lipid-catabolic activity than B cells~\cite{yang2008rora}.
\textit{AKR1C3}--\textit{PTGDR} (PIP \(= 0.323\)) indicates a borderline
prostaglandin D\(_2\)--DP receptor axis that does not achieve discovery
strength in this direction.

\textit{IFNG}--\textit{IFNGR1} (PIP \(= 0.041\)) and
\textit{IFNG}--\textit{IFNGR2} (PIP \(= 0.049\)) both received near-zero
support despite NK cells being a major physiological source of
IFN-\(\gamma\)~\cite{schroder2004interferon,bach1997ifngr}.
The low PIPs here likely reflect that CD4\(^+\) T cell Interferon pathway
activity is not well-explained by NK cell cis-eQTL variation for \textit{IFNG}
at the population level, or that the IFN-\(\gamma\) signal from NK cells to
CD4\(^+\) T cells is overshadowed by autocrine or CD4\(^+\) Th1-derived
IFN-\(\gamma\) signaling within the CD4\(^+\) compartment itself.
\textit{SLC6A6}--\textit{GABBR1} (PIP \(= 0.034\)) received a near-zero
PIP despite being discovered in the NK cells \(\rightarrow\) B cells direction
(PIP \(= 0.618\)) in this same analysis, pointing to receiver-cell specificity
in the GABA-B receptor-mediated taurine signaling axis: CD4\(^+\) T cells
are not causal targets of NK cell SLC6A6-mediated GABBR1 signaling at the
population level, whereas B cells are.
This observation echoes the direction-selectivity previously noted for
\textit{SLC6A6}--\textit{GABBR1} in the B cells \(\rightarrow\) CD4\(^+\)
T cells direction (PIP \(= 0.54\)), where B cells rather than NK cells were
the active sender~\cite{tian1999gaba,bjurstrom2008gaba}.

The single-discovery result here contrasts with the four discoveries
identified in the reverse direction (CD4\(^+\) T cells \(\rightarrow\) NK
cells), underscoring a directional asymmetry in which CD4\(^+\) T helper
cells are more potent causal communicators to NK cells than NK cells are to
CD4\(^+\) T cells at the population level in peripheral blood.

\begin{figure}[htbp]
\centering
\includegraphics[width=\textwidth]{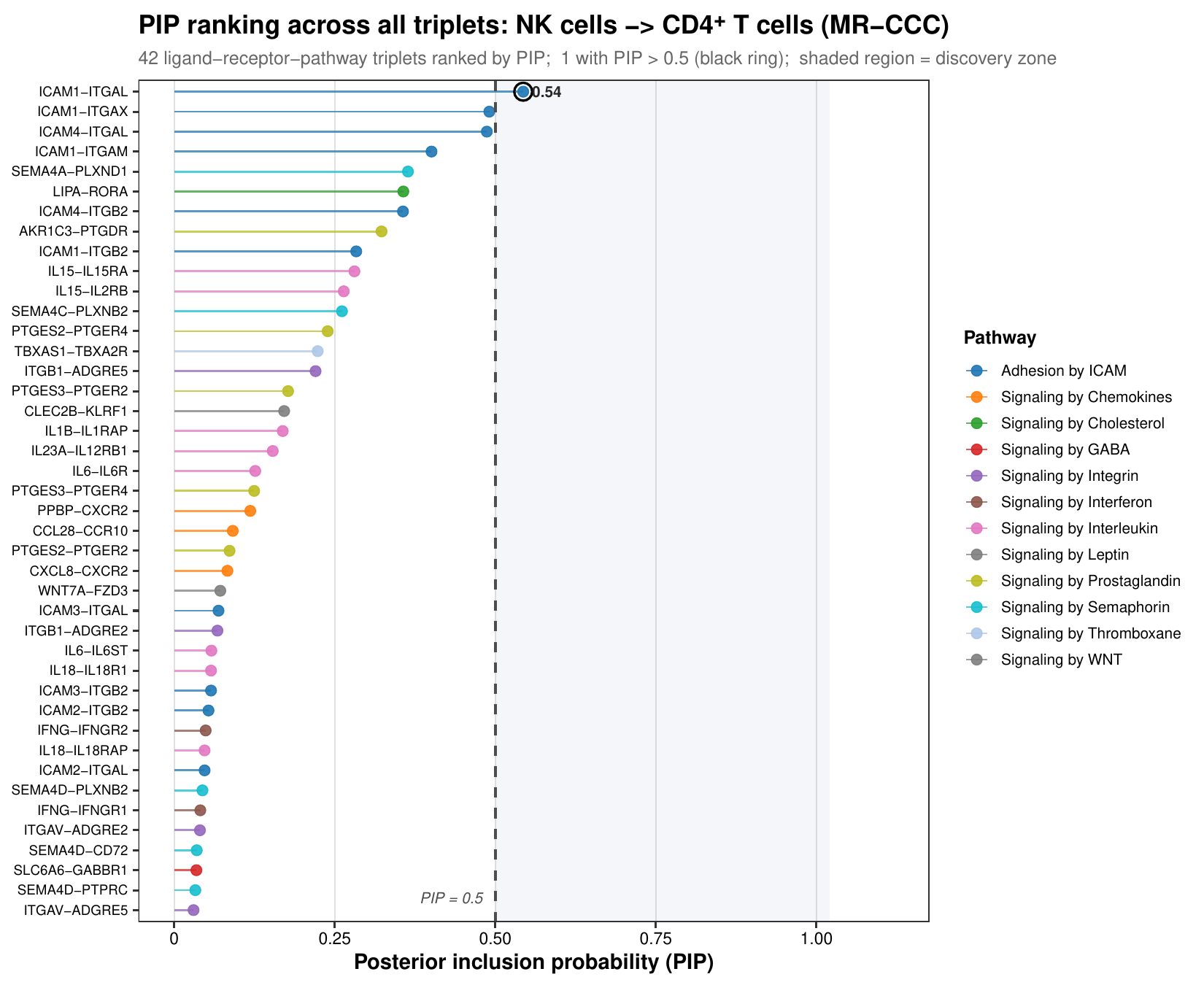}
\caption{\textbf{PIP ranking for the NK cells\,$\rightarrow$\,CD4\(^+\) T cells analysis.} All 42 ligand--receptor--pathway triplets across 850
  donors. Points are colored by pathway; the dashed vertical line
  marks the discovery threshold of PIP \(= 0.5\); the shaded region
  to the right is the discovery zone. The single discovered triplet,
  \textit{ICAM1}--\textit{ITGAL} (PIP \(= 0.54\)), is identified with
  a black ring. The two highest-ranking sub-threshold triplets,
  \textit{ICAM1}--\textit{ITGAX} (PIP \(= 0.49\)) and
  \textit{ICAM4}--\textit{ITGAL} (PIP \(= 0.49\)), fall marginally
  below the threshold.}
\label{fig:supp_NK_CD4_pip}
\end{figure}

\begin{figure}[htbp]
\centering
\includegraphics[width=\textwidth]{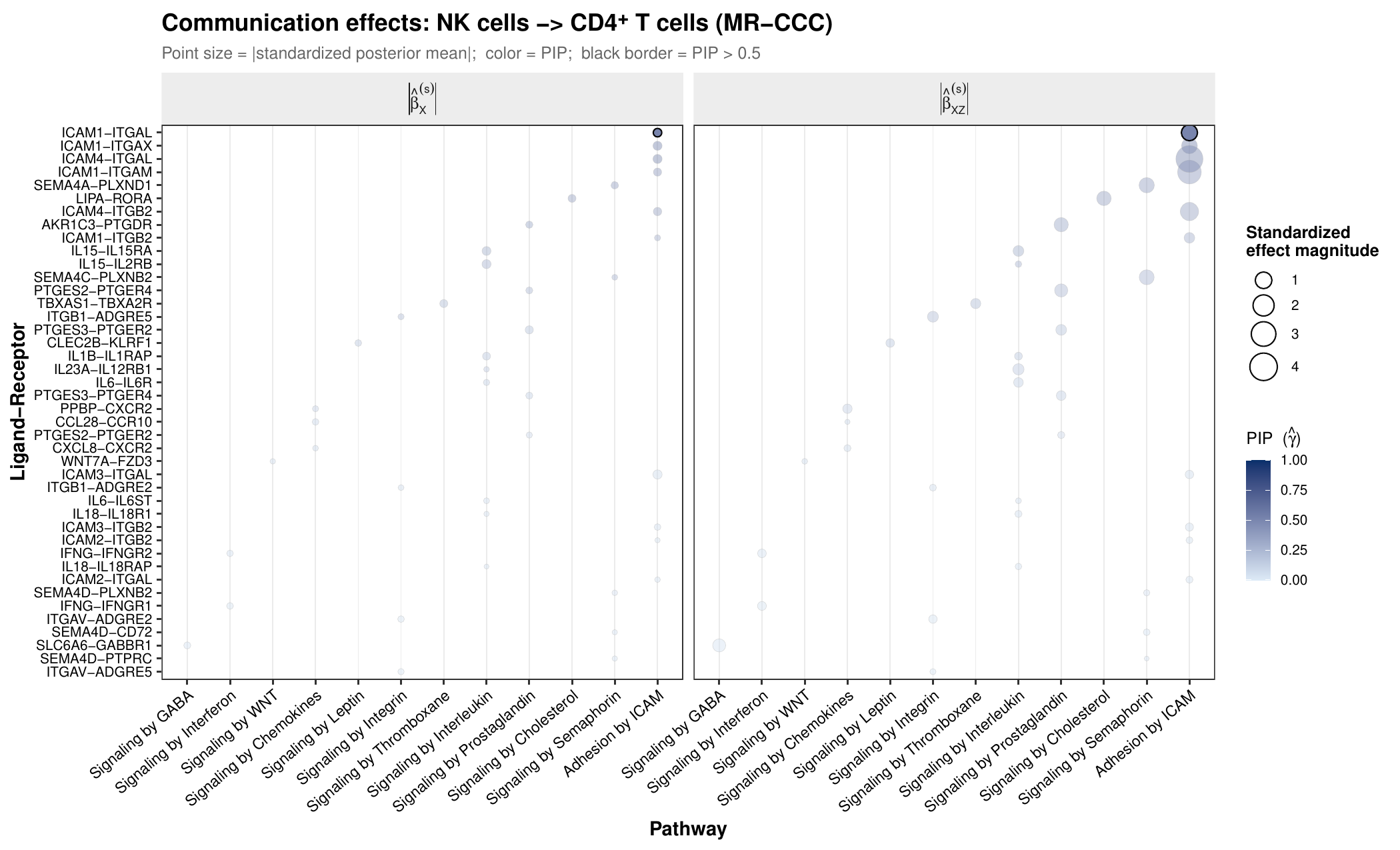}
\caption{\textbf{Standardized posterior effects for the NK cells \(\rightarrow\)
  CD4\(^+\) T cells analysis.} Left panel: absolute main ligand effect
  \(|\hat{\beta}_X^{(s)}|\); right panel: absolute receptor-modulated
  interaction effect \(|\hat{\beta}_{XZ}^{(s)}|\). Point size encodes
  effect magnitude; fill color encodes PIP; the single black-bordered
  point identifies the discovered triplet \textit{ICAM1}--\textit{ITGAL}.
  The near-zero main effect and moderate interaction of the discovered
  triplet are consistent with its near-purely interaction-driven discovery
  profile. \textit{ICAM4}--\textit{ITGAL} (PIP \(= 0.487\)) shows the
  largest interaction magnitude in the right panel despite falling below
  the threshold.}
\label{fig:supp_NK_CD4_bubble}
\end{figure}

\begin{figure}[htbp]
\centering
\includegraphics[width=\textwidth]{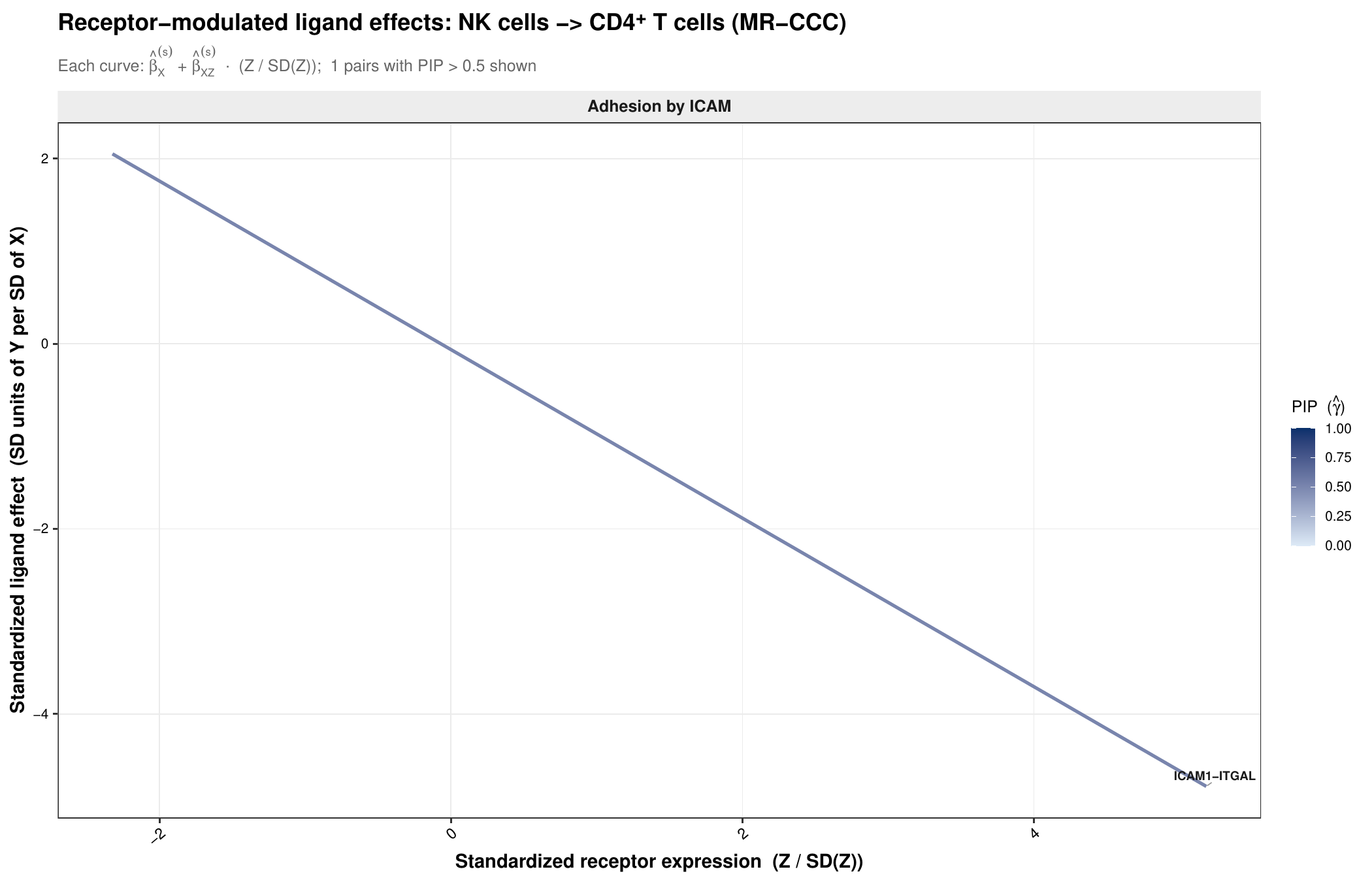}
\caption{\textbf{Receptor-modulated effect curves
  (\(\hat{\beta}_X + \hat{\beta}_{XZ} \cdot Z/\mathrm{SD}(Z)\))
  for the NK cells \(\rightarrow\) CD4\(^+\) T cells analysis.}
  Only the single discovery pair (PIP \(> 0.5\)) is displayed in the
  Adhesion by ICAM panel. The \textit{ICAM1}--\textit{ITGAL} curve
  crosses zero essentially at the population-mean \textit{ITGAL}
  level and descends to negative values at above-average LFA-1
  expression, consistent with the near-purely interaction-driven
  suppressive signal.}
\label{fig:supp_NK_CD4_curves}
\end{figure}
\FloatBarrier

\subsubsection{NK Cells \(\rightarrow\) CD8\(^+\) T Cells}
\label{supp:NK_CD8}

Across 830 donors, MR-CCC evaluated 42 ligand--receptor--pathway
triplets for the NK cell (sender) to CD8\(^+\) T cell (receiver) direction
and identified two high-confidence causal communication signals with
PIP exceeding 0.5: \textit{ICAM1}--\textit{ITGAL} within the Adhesion by ICAM
pathway (PIP \(= 0.862\)) and \textit{CXCL8}--\textit{CXCR2} within the
Chemokine signaling pathway (PIP \(= 0.546\))
(Figures~\ref{fig:supp_NK_CD8_pip}--\ref{fig:supp_NK_CD8_curves}).

The dominant discovery, \textit{ICAM1}--\textit{ITGAL}
(PIP \(= 0.862\), \(\hat{\beta}_X = 1.46\), \(\hat{\beta}_{XZ} = -8.36\)),
implicates NK cell ICAM-1 as a causal regulator of CD8\(^+\) T cell Adhesion
pathway activity through LFA-1 (encoded by \textit{ITGAL}), with an interaction
magnitude that is the second largest of any discovered triplet in the entire
analysis~\cite{dustin1986icam,springer1990adhesion}. The large positive main
effect and strongly negative interaction produce a sign-reversing effect curve
that crosses zero approximately \(0.18\) standard deviations above the mean
\textit{ITGAL} level. CD8\(^+\) T cells with below-average LFA-1 expression
experience a strongly stimulatory response to NK cell ICAM-1 engagement;
above this threshold, the effect reverses sharply, and CD8\(^+\) T cells with
high LFA-1 expression undergo a steeply amplified suppression of Adhesion
pathway activity that grows in magnitude with receptor density. Because
LFA-1 is progressively upregulated during CD8\(^+\) T cell activation and
effector differentiation~\cite{springer1990adhesion}, this pattern suggests that
NK cell ICAM-1 preferentially suppresses activated and effector-like CD8\(^+\)
T cells while leaving naive or resting cells with low LFA-1 expression either
unaffected or mildly stimulated---a receptor-gated regulatory mechanism
consistent with the established capacity of NK cells to selectively eliminate
or restrain over-activated T cells.

The \textit{ICAM1}--\textit{ITGAL} triplet has now been discovered across all
three non-B cell receiver types in the NK sender directions---NK cells
\(\rightarrow\) B cells (PIP \(= 0.782\), \(\hat{\beta}_{XZ} = -2.13\)),
NK cells \(\rightarrow\) CD4\(^+\) T cells (PIP \(= 0.54\),
\(\hat{\beta}_{XZ} = -0.910\)), and here with the largest PIP
and by far the largest interaction magnitude (\(\hat{\beta}_{XZ} = -8.36\)).
The progressive increase in interaction magnitude across receiver cell types
(CD4: \(-0.910\); B: \(-2.13\); CD8: \(-8.36\)) indicates that the NK cell
ICAM-1--LFA-1 suppressive axis is most potently amplified in CD8\(^+\)
T cells, consistent with their higher constitutive and activation-induced LFA-1
expression relative to B and CD4\(^+\) T cells~\cite{springer1990adhesion,dustin1986icam}.

The second discovered triplet, \textit{CXCL8}--\textit{CXCR2}
(PIP \(= 0.546\), \(\hat{\beta}_X = 0.443\), \(\hat{\beta}_{XZ} = 0.530\)),
identifies NK cell-derived CXCL8 (IL-8) as a causal activator of CD8\(^+\)
T cell Chemokine pathway activity through the CXCR2 receptor. Both the main
effect and the interaction are positive, placing the zero crossing at
approximately \(-0.84\) standard deviations below the mean \textit{CXCR2}
expression level, outside the bulk of the observed distribution. The causal
effect is therefore positive across essentially all observed CD8\(^+\) T cells
and grows with receptor density: CD8\(^+\) T cells expressing higher levels of
CXCR2 show a progressively amplified Chemokine pathway response to NK cell
CXCL8 production. CXCL8 is produced by NK cells upon activation and engages
CXCR2 on effector and terminally differentiated CD8\(^+\) T cell subsets, which
express CXCR2 at higher levels than naive or central memory
cells~\cite{ahuja1996cxcr2}. The positive interaction is consistent with a
feed-forward chemokine amplification mechanism in which NK-derived CXCL8
increasingly activates Chemokine pathway gene expression in CD8\(^+\) T cells
that are already in a chemokine-responsive, high-CXCR2 state. Notably,
\textit{CXCL8}--\textit{CXCR2} was the highest-ranking sub-threshold triplet in
the CD8\(^+\) T cells \(\rightarrow\) CD4\(^+\) T cells direction
(PIP \(= 0.367\)), where it did not achieve discovery strength; here the same
triplet is discovered from the NK sender direction, indicating that NK cells
rather than CD8\(^+\) T cells are the dominant population-level source of
CXCL8 driving CXCR2-mediated signaling in CD8\(^+\) T cells in this cohort.

Among sub-threshold triplets, \textit{ICAM4}--\textit{ITGAL}
(PIP \(= 0.476\), \(\hat{\beta}_X = 0.159\), \(\hat{\beta}_{XZ} = -1.55\)) was
the highest-ranking near-miss; the positive main effect and negative interaction
produce a sign-reversing curve crossing zero near \(+0.10\) standard deviations
above the mean \textit{ITGAL} level, with the majority of CD8\(^+\) T cells
falling in the suppressive regime. Its near-discovery alongside the confirmed
\textit{ICAM1}--\textit{ITGAL} signal suggests that NK cell ICAM4 contributes
a partially consistent suppressive adhesion contact that falls just below the
discovery threshold. \textit{ICAM1}--\textit{ITGB2} (PIP \(= 0.398\)) and
\textit{ICAM1}--\textit{ITGAM} (PIP \(= 0.369\)) represent additional
sub-threshold ICAM pairs, both consistent with broad NK--CD8 adhesive contacts
that do not individually achieve the consistency threshold at the population
level~\cite{springer1990adhesion}.

\textit{IL18}--\textit{IL18R1} (PIP \(= 0.376\)) was the highest-ranking
Interleukin near-miss; this triplet was discovered in the CD8\(^+\) T cells
\(\rightarrow\) Monocytes direction (PIP \(= 0.82\)), where CD8\(^+\) T cells
were the sender, but here NK cells as sender to CD8\(^+\) T cells fall below
the threshold, suggesting that the IL-18--IL-18R1 axis from NK cells to CD8
cells is insufficiently consistent at the population level compared with the
CD8-to-monocyte direction~\cite{okamura1995interleukin18}.

\textit{IFNG}--\textit{IFNGR1} (PIP \(= 0.029\)) and
\textit{IFNG}--\textit{IFNGR2} (PIP \(= 0.039\)) both received near-zero
support, consistent with the pattern observed across all NK sender directions:
despite NK cells being a major source of IFN-\(\gamma\), the NK cell
cis-eQTL-driven variation in \textit{IFNG} expression does not produce a
detectable causal signal on CD8\(^+\) T cell Interferon pathway activity at the
population level~\cite{schroder2004interferon,bach1997ifngr}.

\begin{figure}[htbp]
\centering
\includegraphics[width=\textwidth]{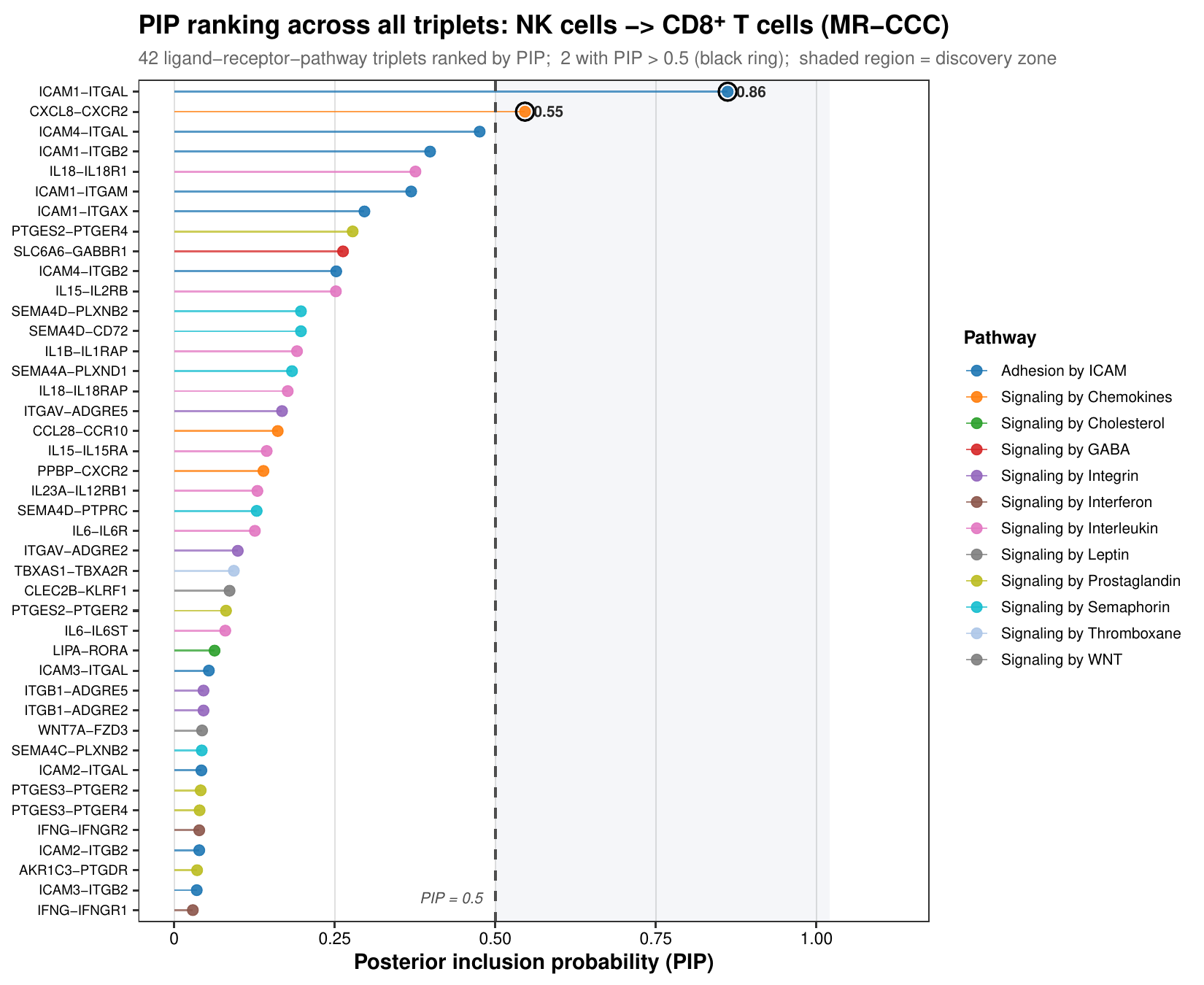}
\caption{\textbf{PIP ranking for the NK cells\,$\rightarrow$\,CD8\(^+\) T cells analysis.} All 42 ligand--receptor--pathway triplets across 830
  donors. Points are colored by pathway; the dashed vertical line
  marks the discovery threshold of PIP \(= 0.5\); the shaded region
  to the right is the discovery zone. Black rings identify the two
  discovered triplets: \textit{ICAM1}--\textit{ITGAL}
  (PIP \(= 0.86\)) and \textit{CXCL8}--\textit{CXCR2}
  (PIP \(= 0.55\)).}
\label{fig:supp_NK_CD8_pip}
\end{figure}

\begin{figure}[htbp]
\centering
\includegraphics[width=\textwidth]{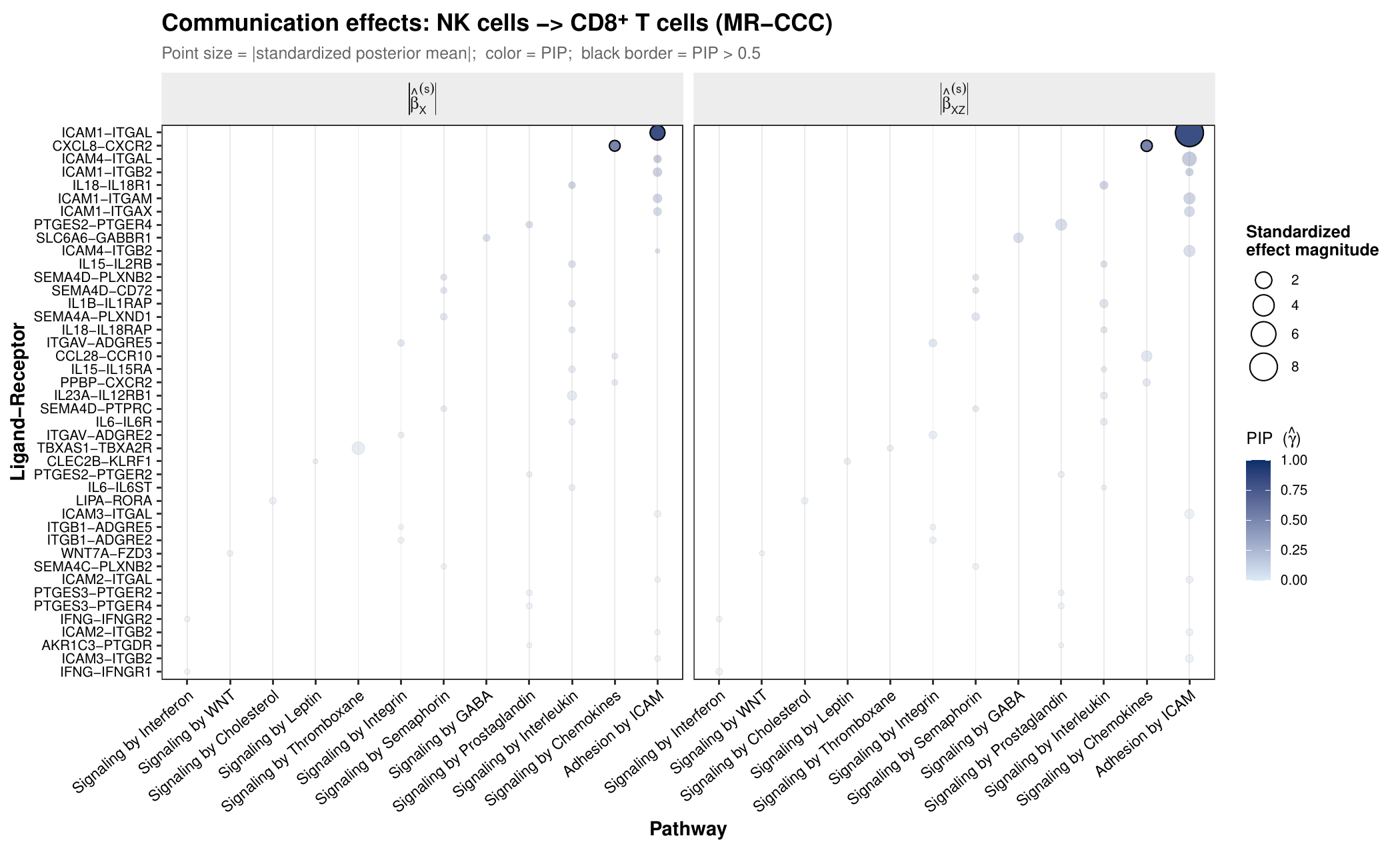}
\caption{\textbf{Standardized posterior effects for the NK cells \(\rightarrow\)
  CD8\(^+\) T cells analysis.} Left panel: absolute main ligand effect
  \(|\hat{\beta}_X^{(s)}|\); right panel: absolute receptor-modulated
  interaction effect \(|\hat{\beta}_{XZ}^{(s)}|\). Point size encodes
  effect magnitude; fill color encodes PIP; black borders identify the
  two discovered triplets. \textit{ICAM1}--\textit{ITGAL} dominates
  the right panel with the largest interaction magnitude in this
  direction; \textit{CXCL8}--\textit{CXCR2} shows appreciable effects
  in both panels, consistent with its combination of a positive main
  effect and a positive interaction.}
\label{fig:supp_NK_CD8_bubble}
\end{figure}

\begin{figure}[htbp]
\centering
\includegraphics[width=\textwidth]{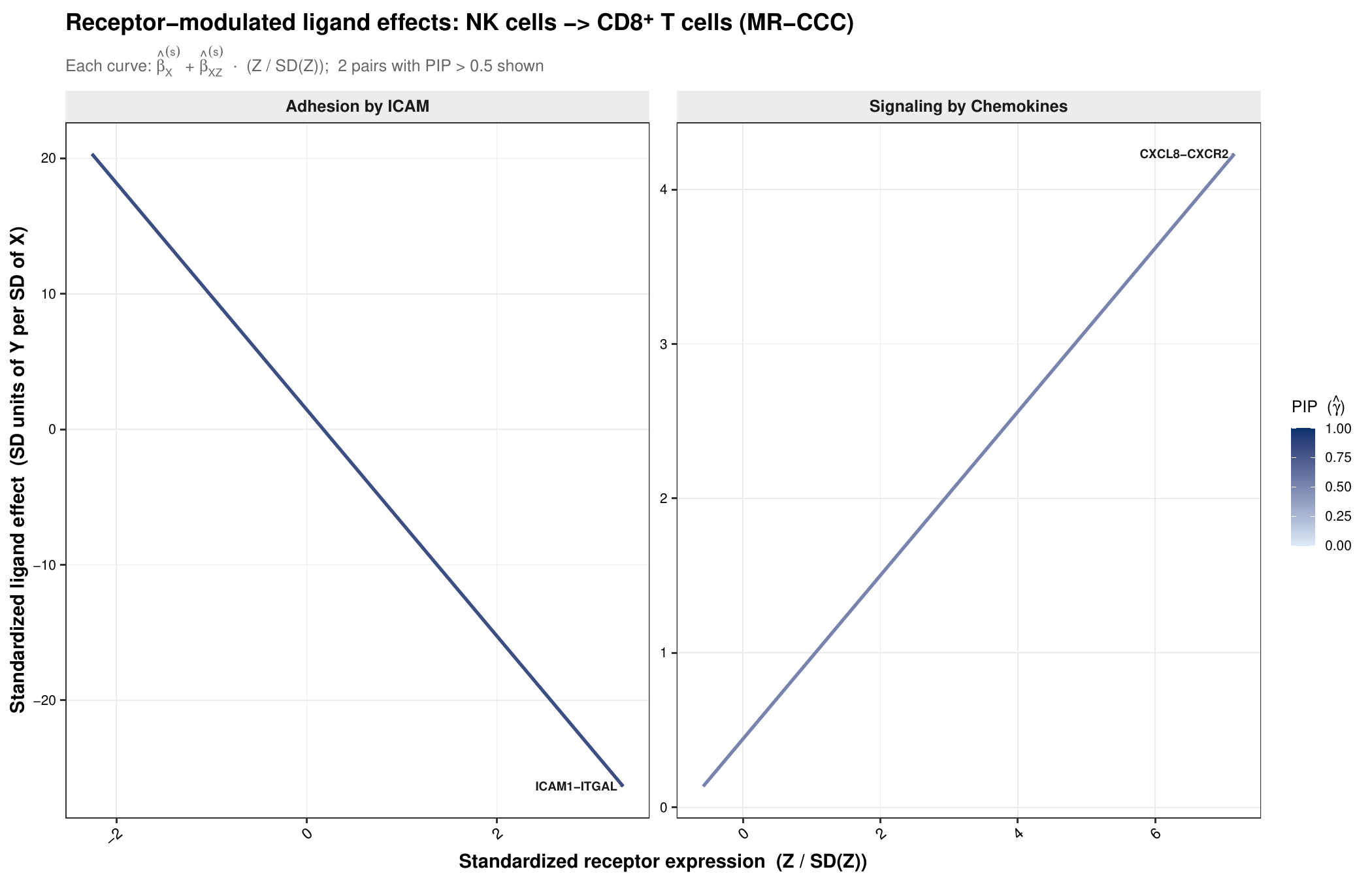}
\caption{\textbf{Receptor-modulated effect curves
  (\(\hat{\beta}_X + \hat{\beta}_{XZ} \cdot Z/\mathrm{SD}(Z)\))
  for the NK cells \(\rightarrow\) CD8\(^+\) T cells analysis.}
  Only the two discovery pairs (PIP \(> 0.5\)) are displayed, grouped
  into the two pathway panels with confirmed discoveries (Adhesion by
  ICAM and Signaling by Chemokines). In the Adhesion by ICAM panel,
  the \textit{ICAM1}--\textit{ITGAL} curve starts at a large positive
  value at low \textit{ITGAL} expression and descends very steeply,
  crossing zero near \(+0.18\) standard deviations above the mean and
  reaching strongly negative values at high LFA-1 density; the steep
  slope reflects the second-largest interaction magnitude in the entire
  analysis. In the Signaling by Chemokines panel, the
  \textit{CXCL8}--\textit{CXCR2} curve rises monotonically, remaining
  positive across the full observed receptor range with a zero crossing
  near \(-0.84\) standard deviations below the mean.}
\label{fig:supp_NK_CD8_curves}
\end{figure}
\FloatBarrier

\subsubsection{NK Cells \(\rightarrow\) Monocytes}
\label{supp:NK_Mono}
The NK cells\,$\rightarrow$\,monocytes direction (8 discoveries
spanning GABA, interferon, interleukin, and prostaglandin
signaling, $n = 651$ donors) is presented in full detail in the
Real Data Analysis section of the main manuscript and is therefore
omitted here.

\subsection{Monocytes as Sender}

\subsubsection{Monocytes \(\rightarrow\) B Cells}
\label{supp:Mono_B}

Across 665 donors, MR-CCC evaluated 41 candidate ligand--receptor--pathway triplets for the Monocyte (sender) to B cell (receiver) direction and identified eight high-confidence causal
communication signals with PIP exceeding 0.5, spanning three pathway classes:
five Adhesion by ICAM pairs (\textit{ICAM1}--\textit{ITGAM},
PIP \(= 0.780\); \textit{ICAM1}--\textit{ITGAL}, PIP \(= 0.732\);
\textit{ICAM1}--\textit{ITGB2}, PIP \(= 0.721\);
\textit{ICAM1}--\textit{ITGAX}, PIP \(= 0.628\);
\textit{ICAM4}--\textit{ITGAL}, PIP \(= 0.548\)), two Signaling by
Interleukin triplets (\textit{IL1B}--\textit{IL1RAP}, PIP \(= 0.685\);
\textit{IL15}--\textit{IL2RB}, PIP \(= 0.548\)), and one Signaling by Prostaglandin triplet (\textit{AKR1C3}--\textit{PTGDR}, PIP \(= 0.687\))
(Figures~\ref{fig:supp_Mono_B_pip}--\ref{fig:supp_Mono_B_curves}).

\noindent\textbf{Adhesion by ICAM signals.}
The five discovered ICAM--integrin triplets constitute the densest single-ligand
adhesion discovery cluster in the Monocyte sender directions: all four evaluated
\textit{ICAM1}--integrin pairs are simultaneously discovered, joined by
\textit{ICAM4}--\textit{ITGAL}~\cite{springer1990adhesion,dustin1986icam}.
Despite sharing the same ligand, the four \textit{ICAM1} pairs divide into two
qualitatively distinct groups based on the sign of the interaction term.

\textit{ICAM1}--\textit{ITGAL} (PIP \(= 0.732\), \(\hat{\beta}_X = 0.613\),
\(\hat{\beta}_{XZ} = 2.36\)) and \textit{ICAM1}--\textit{ITGAX}
(PIP \(= 0.628\), \(\hat{\beta}_X = 0.955\), \(\hat{\beta}_{XZ} = 0.320\))
both show positive main effects and positive interactions, yielding
monotonically stimulatory effect curves across the observed receptor expression
range. For \textit{ICAM1}--\textit{ITGAL} the zero crossing lies at
approximately \(0.26\) standard deviations \emph{below} the mean \textit{ITGAL}
level, meaning the stimulatory effect is active across essentially all B cells
and grows steeply with LFA-1 density. For \textit{ICAM1}--\textit{ITGAX} the
interaction is smaller and the zero crossing lies near \(-3.0\) standard
deviations below the mean, making this effectively a pure positive main effect
that uniformly stimulates the Adhesion pathway of CD11c-expressing B cells.

\textit{ICAM1}--\textit{ITGAM} (PIP \(= 0.780\), \(\hat{\beta}_X = 1.13\),
\(\hat{\beta}_{XZ} = -2.17\)), \textit{ICAM1}--\textit{ITGB2}
(PIP \(= 0.721\), \(\hat{\beta}_X = 1.05\), \(\hat{\beta}_{XZ} = -1.38\)),
and \textit{ICAM4}--\textit{ITGAL} (PIP \(= 0.548\),
\(\hat{\beta}_X = 0.617\), \(\hat{\beta}_{XZ} = -1.15\)) form the second
group: large positive main effects combined with negative interactions yield
sign-reversing curves that cross zero at \(+0.52\), \(+0.76\), and
\(+0.54\) standard deviations above their respective receptor means.
B cells expressing Mac-1 (\textit{ITGAM}), \(\beta_2\) (\textit{ITGB2}), or
LFA-1 (\textit{ITGAL}, via ICAM4) below these thresholds experience net
stimulation; those above them transition to progressively stronger suppression
of Adhesion pathway activity~\cite{myones1988cd11c,dustin1986icam}.

The most striking cross-direction contrast concerns \textit{ICAM1}--\textit{ITGAL}:
in every NK cell sender direction in which this triplet was discovered, the
interaction was large and \emph{negative} (NK cells \(\rightarrow\) B cells:
\(\hat{\beta}_{XZ} = -2.13\); NK cells \(\rightarrow\) CD4\(^+\) T cells:
\(\hat{\beta}_{XZ} = -0.910\); NK cells \(\rightarrow\) CD8\(^+\) T cells:
\(\hat{\beta}_{XZ} = -8.36\)), yielding receptor-gated suppression of
Adhesion pathway activity at high LFA-1. From the monocyte sender, the same
triplet yields a large \emph{positive} interaction
(\(\hat{\beta}_{XZ} = +2.36\)), converting the curve from
receptor-amplified suppression to receptor-amplified stimulation. This
sender-specific sign reversal of the interaction term for the identical
ligand--receptor--pathway triplet illustrates the capacity of MR-CCC to detect
qualitatively opposite causal consequences of the same molecular contact
depending on the cellular source of the ligand.

\noindent\textbf{Signaling by Interleukin.}
\textit{IL1B}--\textit{IL1RAP} (PIP \(= 0.685\), \(\hat{\beta}_X = 0.462\),
\(\hat{\beta}_{XZ} = -0.098\)) identifies monocyte-derived IL-1\(\beta\) as a
causal activator of B cell Interleukin pathway activity through the IL-1
receptor accessory protein~\cite{dinarello2009interleukin}. Monocytes are the primary
cellular source of IL-1\(\beta\) in peripheral blood, and B cells express
functional IL-1 receptor complexes that regulate B cell activation and antibody
production. The negligibly small interaction term places the zero crossing at
approximately \(+4.7\) standard deviations above the mean \textit{IL1RAP}
level, far outside the observed data range; this triplet is therefore a
near-purely main-effect discovery in which monocyte IL-1\(\beta\) exerts a
uniformly positive stimulatory effect on B cell Interleukin pathway activity
that is essentially independent of IL-1RAP expression level.

\textit{IL15}--\textit{IL2RB} (PIP \(= 0.548\), \(\hat{\beta}_X = 0.453\),
\(\hat{\beta}_{XZ} = -2.11\)) implicates monocyte-derived IL-15 as a causal
regulator of B cell Interleukin pathway activity through the shared IL-2R\(\beta\)
chain (encoded by \textit{IL2RB}; CD122)~\cite{waldmann2006il15,fehniger2001il15}.
Monocytes constitutively express IL-15R\(\alpha\) (\textit{IL15RA}) and
transpresent IL-15 to neighboring cells, making them a biologically plausible
cellular source for this signal. The positive main effect and large negative
interaction yield a sign-reversing curve that crosses zero at approximately
\(+0.22\) standard deviations above the mean \textit{IL2RB} level: B cells
with below-average IL-2R\(\beta\) expression experience net stimulation, while
those above the threshold---associated with a more activated or memory B cell
phenotype---undergo progressively stronger suppression, consistent with
receptor-mediated negative feedback under high IL-2R\(\beta\) expression.
This triplet was previously discovered in the B cells \(\rightarrow\) CD8\(^+\)
T cells direction (PIP \(= 0.518\)), where B cells were the sender; its
discovery here from the monocyte sender direction indicates that the
IL-15--IL-2R\(\beta\) axis to B cells is activated by multiple cellular
sources in peripheral blood.

\noindent\textbf{Signaling by Prostaglandin.}
\textit{AKR1C3}--\textit{PTGDR} (PIP \(= 0.687\), \(\hat{\beta}_X = -0.225\),
\(\hat{\beta}_{XZ} = -1.85\)) identifies monocyte aldo-keto reductase 1C3
(\textit{AKR1C3}) as a causal suppressor of B cell Prostaglandin pathway
activity through the prostaglandin D\(_2\) receptor (DP1, encoded by
\textit{PTGDR}). AKR1C3 is a multifunctional lipid-metabolizing enzyme
expressed at high levels in monocytes that catalyses the synthesis of
prostaglandin F\(_2\alpha\) and related eicosanoids from prostaglandin H\(_2\),
and can interconvert prostaglandin D\(_2\) metabolites that engage the DP1
receptor on lymphocytes~\cite{penning2006akr1c3,hata2004pgrecept}. Both the
main effect and the interaction are negative, with the zero crossing at
approximately \(0.12\) standard deviations \emph{below} the mean \textit{PTGDR}
expression level---essentially at the population mean. Monocytes with higher
AKR1C3-mediated eicosanoid synthesis activity therefore exert an increasingly
strong suppression of B cell Prostaglandin DP1 pathway activity, with the
suppression uniformly experienced by essentially all B cells expressing
\textit{PTGDR} at or above the mean. This is the first direction in the
analysis in which \textit{AKR1C3}--\textit{PTGDR} is discovered; it was a sub-threshold near-miss (PIP \(= 0.365\)) in the NK cells \(\rightarrow\)
Monocytes direction.

Several triplets narrowly missed the discovery threshold.
\textit{ITGB1}--\textit{ADGRE5} (PIP \(= 0.492\)) was the highest-ranking
sub-threshold triplet and a borderline near-miss; this triplet was among the
top discoveries in the CD4\(^+\) T cells \(\rightarrow\) B cells direction
(PIP \(= 0.992\)) and was also discovered in the CD4\(^+\) T cells
\(\rightarrow\) Monocytes direction (PIP \(= 0.597\)), but monocyte \(\beta_1\)
integrin-to-B cell CD97 signaling falls just below the discovery threshold here.
\textit{ICAM4}--\textit{ITGB2} (PIP \(= 0.482\)) was also a close
near-miss within the Adhesion by ICAM cluster.
\textit{IFNG}--\textit{IFNGR2} (PIP \(= 0.468\)) was the highest-ranking
non-ICAM, non-prostaglandin sub-threshold triplet; monocytes can produce
low levels of IFN-\(\gamma\) under activation conditions, but the
population-averaged signal does not achieve discovery strength, consistent
with IFN-\(\gamma\) being a more characteristic product of T and NK
cells~\cite{schroder2004interferon}.

\begin{figure}[htbp]
\centering
\includegraphics[width=\textwidth]{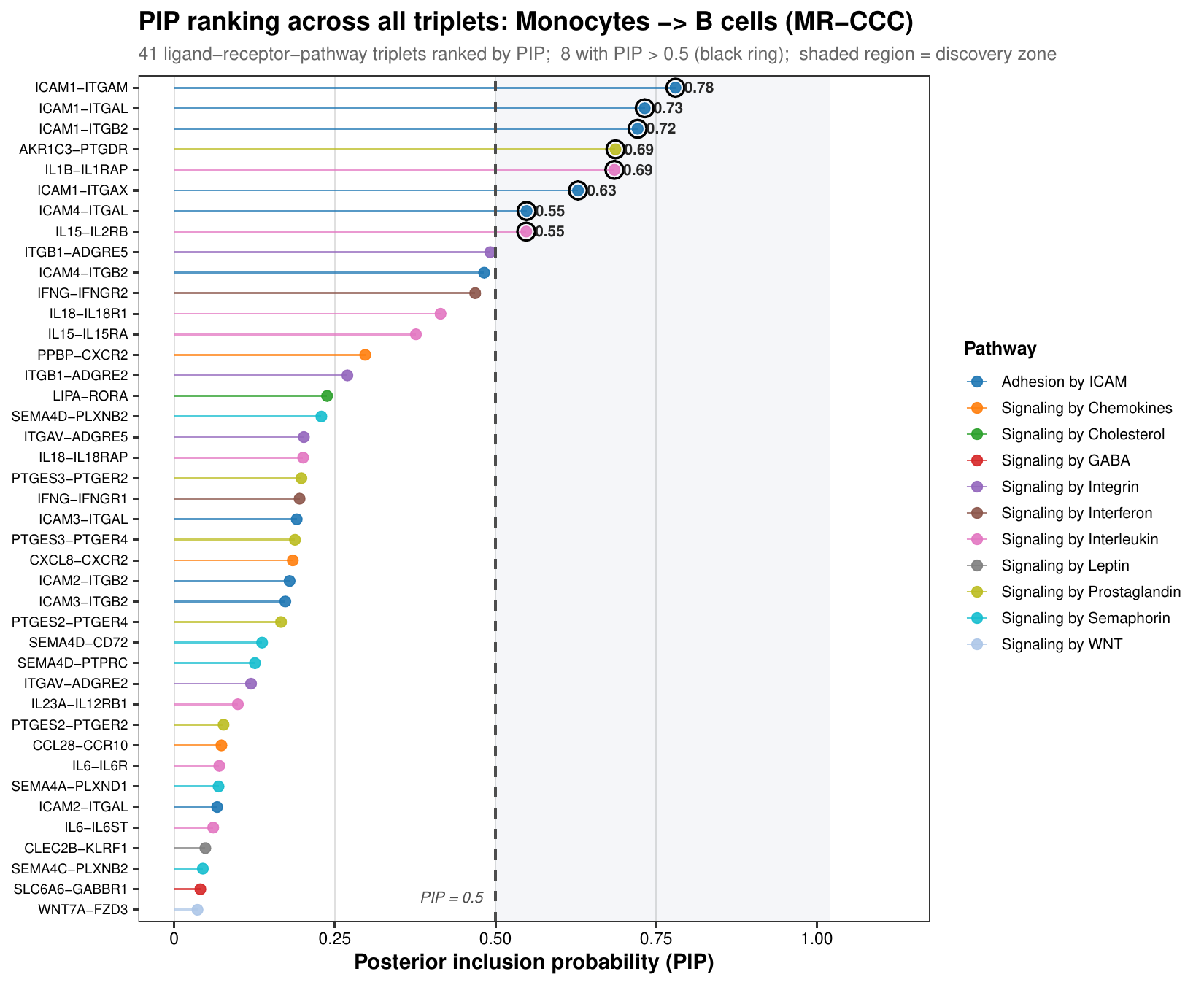}
\caption{\textbf{PIP ranking for the Monocytes\,$\rightarrow$\,B cells analysis.} All 41 ligand--receptor--pathway triplets across 665 donors.
  Points are colored by pathway; the dashed vertical line marks the
  discovery threshold of PIP \(= 0.5\); the shaded region to the
  right is the discovery zone. Black rings identify the eight
  discovered triplets: \textit{ICAM1}--\textit{ITGAM}
  (PIP \(= 0.78\)), \textit{ICAM1}--\textit{ITGAL}
  (PIP \(= 0.73\)), \textit{ICAM1}--\textit{ITGB2}
  (PIP \(= 0.72\)), \textit{AKR1C3}--\textit{PTGDR}
  (PIP \(= 0.687\)), \textit{IL1B}--\textit{IL1RAP}
  (PIP \(= 0.686\)), \textit{ICAM1}--\textit{ITGAX}
  (PIP \(= 0.63\)), \textit{ICAM4}--\textit{ITGAL}
  (PIP \(= 0.55\)), and \textit{IL15}--\textit{IL2RB}
  (PIP \(= 0.55\)).}
\label{fig:supp_Mono_B_pip}
\end{figure}

\begin{figure}[htbp]
\centering
\includegraphics[width=\textwidth]{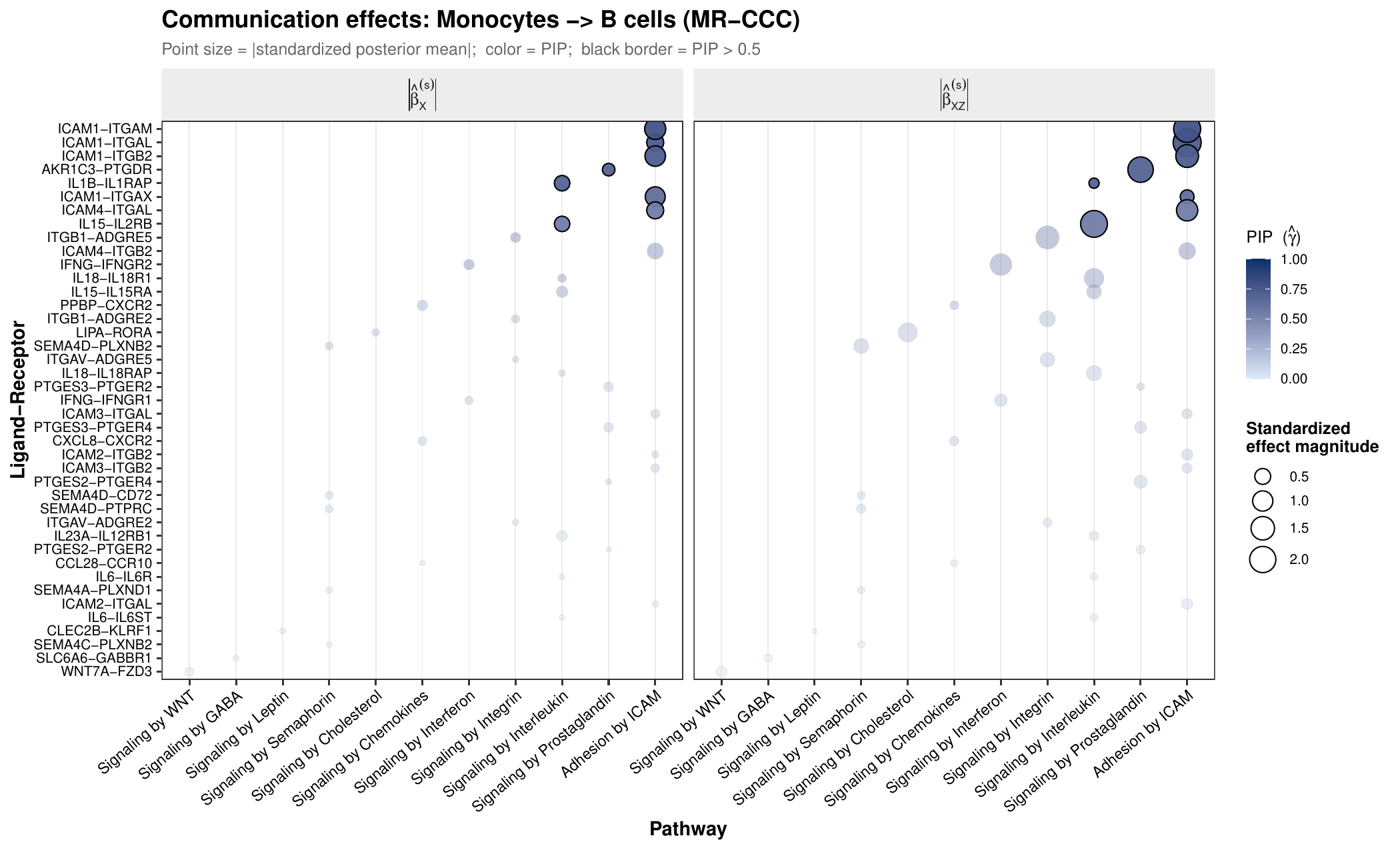}
\caption{\textbf{Standardized posterior effects for the Monocytes \(\rightarrow\)
  B cells analysis.} Left panel: absolute main ligand effect
  \(|\hat{\beta}_X^{(s)}|\); right panel: absolute receptor-modulated
  interaction effect \(|\hat{\beta}_{XZ}^{(s)}|\). Point size encodes
  effect magnitude; fill color encodes PIP; black borders identify the
  eight discovered triplets. The Adhesion by ICAM cluster dominates
  both panels: \textit{ICAM1}--\textit{ITGAL} shows the largest
  positive interaction magnitude while \textit{ICAM1}--\textit{ITGAM}
  and \textit{IL15}--\textit{IL2RB} show the largest negative
  interaction magnitudes. \textit{IL1B}--\textit{IL1RAP} has a large
  main effect but near-zero interaction, consistent with its
  main-effect-driven discovery profile.}
\label{fig:supp_Mono_B_bubble}
\end{figure}

\begin{figure}[htbp]
\centering
\includegraphics[width=\textwidth]{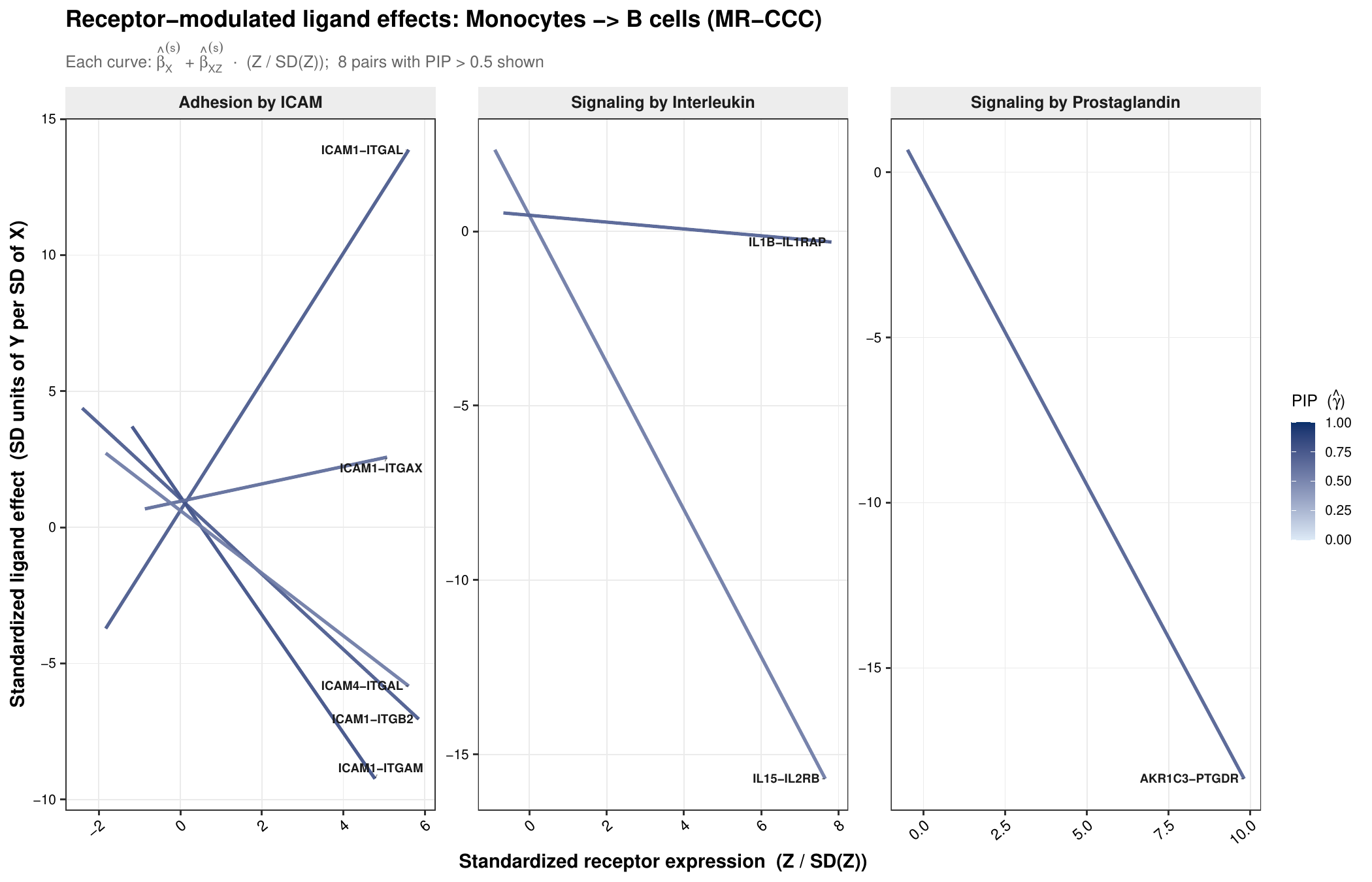}
\caption{\textbf{Receptor-modulated effect curves
  (\(\hat{\beta}_X + \hat{\beta}_{XZ} \cdot Z/\mathrm{SD}(Z)\))
  for the Monocytes \(\rightarrow\) B cells analysis.} Only the eight
  discovery pairs (PIP \(> 0.5\)) are displayed, grouped into the
  three pathway panels with confirmed discoveries (Adhesion by ICAM,
  Signaling by Interleukin, and Signaling by Prostaglandin). In the
  Adhesion by ICAM panel, five curves are visible: the
  \textit{ICAM1}--\textit{ITGAL} curve rises steeply and is positive
  across the full observed LFA-1 range; \textit{ICAM1}--\textit{ITGAX}
  is nearly flat and positive; \textit{ICAM1}--\textit{ITGAM},
  \textit{ICAM1}--\textit{ITGB2}, and \textit{ICAM4}--\textit{ITGAL}
  each start positive and cross zero at \(+0.52\), \(+0.76\), and
  \(+0.54\) standard deviations above the respective receptor mean. In
  the Signaling by Interleukin panel, \textit{IL15}--\textit{IL2RB}
  crosses zero near \(+0.22\) SD above the mean and becomes
  increasingly negative at high \textit{IL2RB} expression, while
  \textit{IL1B}--\textit{IL1RAP} is nearly flat and uniformly positive.
  In the Signaling by Prostaglandin panel,
  \textit{AKR1C3}--\textit{PTGDR} crosses zero just below the
  population mean and descends steeply to strongly negative values at
  high \textit{PTGDR} expression.}
\label{fig:supp_Mono_B_curves}
\end{figure}
\FloatBarrier

\subsubsection{Monocytes \(\rightarrow\) CD4\(^+\) T Cells}
\label{supp:Mono_CD4}

Across 691 donors, MR-CCC evaluated 42 ligand--receptor--pathway
triplets for the Monocyte (sender) to CD4\(^+\) T cell (receiver) direction
and identified three high-confidence causal communication signals
with PIP exceeding 0.5: \textit{IL15}--\textit{IL15RA} and
\textit{IL15}--\textit{IL2RB} within the Signaling by Interleukin pathway
(PIP \(= 0.711\) and PIP \(= 0.694\), respectively) and
\textit{SEMA4D}--\textit{CD72} within the Signaling by Semaphorin pathway
(PIP \(= 0.674\))
(Figures~\ref{fig:supp_Mono_CD4_pip}--\ref{fig:supp_Mono_CD4_curves}). The two IL-15 discoveries jointly implicate monocyte-derived IL-15 as a
causal communicator to CD4\(^+\) T cells through two distinct receptor chains
of the IL-15/IL-2 receptor complex. Monocytes constitutively express
\textit{IL15RA} and transpresent IL-15 to neighboring lymphocytes in
\emph{trans}~\cite{dubois2002il15ra,waldmann2006il15}, making them a primary
cellular vehicle for IL-15 delivery in peripheral blood.

\textit{IL15}--\textit{IL15RA} (PIP \(= 0.711\), \(\hat{\beta}_X = 0.381\),
\(\hat{\beta}_{XZ} = -0.898\)) has a positive main effect and a negative
interaction, yielding a sign-reversing curve that crosses zero at approximately
\(0.42\) standard deviations above the mean \textit{IL15RA} level. CD4\(^+\)
T cells with below-average IL-15R\(\alpha\) expression experience net
stimulation of Interleukin pathway activity by monocyte IL-15, while those
with high \textit{IL15RA} expression undergo progressively stronger
attenuation and eventual suppression, consistent with receptor saturation or
negative feedback at high transpresentation receptor densities. This same
triplet was the dominant discovery in the B cells \(\rightarrow\) CD4\(^+\)
T cells direction (PIP \(= 0.867\)), but with the opposite interaction sign
(\(\hat{\beta}_{XZ} = +2.28\)), producing a monotonically stimulatory curve.
The sender-specific sign reversal of the \textit{IL15RA} interaction
term---positive from B cell sender, negative from monocyte sender---indicates
that the same IL-15R\(\alpha\)-mediated transpresentation interface produces
qualitatively different receptor-modulated outcomes in CD4\(^+\) T cells
depending on the cellular source, consistent with structural differences in
the B cell versus monocyte IL-15 transpresentation complex affecting the
downstream signaling outcome~\cite{waldmann2006il15,fehniger2001il15}.

\textit{IL15}--\textit{IL2RB} (PIP \(= 0.694\), \(\hat{\beta}_X = -0.348\),
\(\hat{\beta}_{XZ} = 1.33\)) exhibits the opposite sign structure: a negative
main effect and a large positive interaction. The zero crossing lies at
approximately \(0.26\) standard deviations above the mean \textit{IL2RB}
expression level. CD4\(^+\) T cells with below-average IL-2R\(\beta\)
expression experience net suppression of Interleukin pathway activity in
response to monocyte IL-15, while those expressing \textit{IL2RB}
above this threshold experience progressively amplified stimulation. The
opposing main-effect signs of the two co-discovered IL-15 triplets---positive
main effect through \textit{IL15RA}, negative main effect through
\textit{IL2RB}---indicate that the two receptor chains impose qualitatively
different baseline responses to monocyte IL-15, with the IL-2R\(\beta\)
chain acting as a gatekeeper that converts a baseline-suppressive signal into
a stimulatory one only in CD4\(^+\) T cells that express it at sufficient
levels.

\textit{SEMA4D}--\textit{CD72} (PIP \(= 0.674\), \(\hat{\beta}_X = 0.192\),
\(\hat{\beta}_{XZ} = 2.61\)) identifies monocyte-expressed SEMA4D (CD100)
as a causal activator of CD4\(^+\) T cell Semaphorin pathway activity through
CD72. The near-zero main effect and large positive interaction place the zero
crossing at approximately \(0.07\) standard deviations \emph{below} the mean
\textit{CD72} level---essentially at the population mean. The causal signal is
therefore almost entirely interaction-driven and operates across the full
observed CD4\(^+\) T cell population: CD4\(^+\) T cells expressing
\textit{CD72} at or above the mean experience a steeply increasing stimulation
of Semaphorin pathway activity as monocyte SEMA4D production rises. CD72 was
first characterized as a counter-receptor for SEMA4D (CD100) on B cells, where
CD100--CD72 signaling modulates B cell activation thresholds and
humoral immunity~\cite{kumanogoh2000cd100}; its discovery here as a receptor
on CD4\(^+\) T cells for monocyte-derived CD100 signals extends this
semaphorin axis to a myeloid--T helper cell communication context, consistent
with the broader role of SEMA4D in lymphocyte co-stimulation and cytoskeletal
regulation~\cite{suzuki2008sema4d}.

Several triplets narrowly missed the discovery threshold. Both
\textit{ITGAV}--\textit{ADGRE2} (PIP \(= 0.469\),
\(\hat{\beta}_X = -0.240\), \(\hat{\beta}_{XZ} = -0.577\)) and
\textit{ITGAV}--\textit{ADGRE5} (PIP \(= 0.436\),
\(\hat{\beta}_X = -0.200\), \(\hat{\beta}_{XZ} = -0.369\)) were notable
Integrin pathway near-misses; both exhibit negative main effects and negative
interactions, producing monotonically suppressive curves across the observed
receptor expression range. The ADGRE2 and ADGRE5 receptors (EMR2 and CD97)
were discovered from CD4\(^+\) T cell senders in the
CD4\(^+\) T cells \(\rightarrow\) NK cells direction
(PIP \(= 0.910\) and PIP \(= 0.858\), respectively) but fall below the
discovery threshold here, suggesting that the integrin--adhesion-GPCR
axis between these cell types is directionally asymmetric: more consistently
detected from CD4\(^+\) T cell sender to NK receiver than from monocyte sender
to CD4\(^+\) T cell receiver at the population level~\cite{hamann1996cd97,stacey2003emr2}.
\textit{SEMA4D}--\textit{PLXNB2} (PIP \(= 0.358\)) was the highest-ranking
sub-threshold Semaphorin triplet beyond the discovered \textit{SEMA4D}--\textit{CD72},
consistent with monocyte SEMA4D engaging multiple receptor systems on CD4\(^+\)
T cells but only the CD72-mediated axis achieving discovery strength.
\textit{CLEC2B}--\textit{KLRF1} (PIP \(= 0.304\)) and
\textit{IL18}--\textit{IL18RAP} (PIP \(= 0.336\)) both fell below threshold;
the latter was discovered from NK cell senders to monocytes (PIP \(= 0.976\)),
again reflecting the directional specificity of IL-18 receptor signaling in
this cohort~\cite{okamura1995interleukin18,born1998il18rap}.

\textit{IFNG}--\textit{IFNGR1} (PIP \(= 0.039\)) and
\textit{IFNG}--\textit{IFNGR2} (PIP \(= 0.031\)) both received near-zero
PIPs, consistent with monocytes not being a dominant physiological source of
IFN-\(\gamma\) in peripheral blood~\cite{schroder2004interferon}.

\begin{figure}[htbp]
\centering
\includegraphics[width=\textwidth]{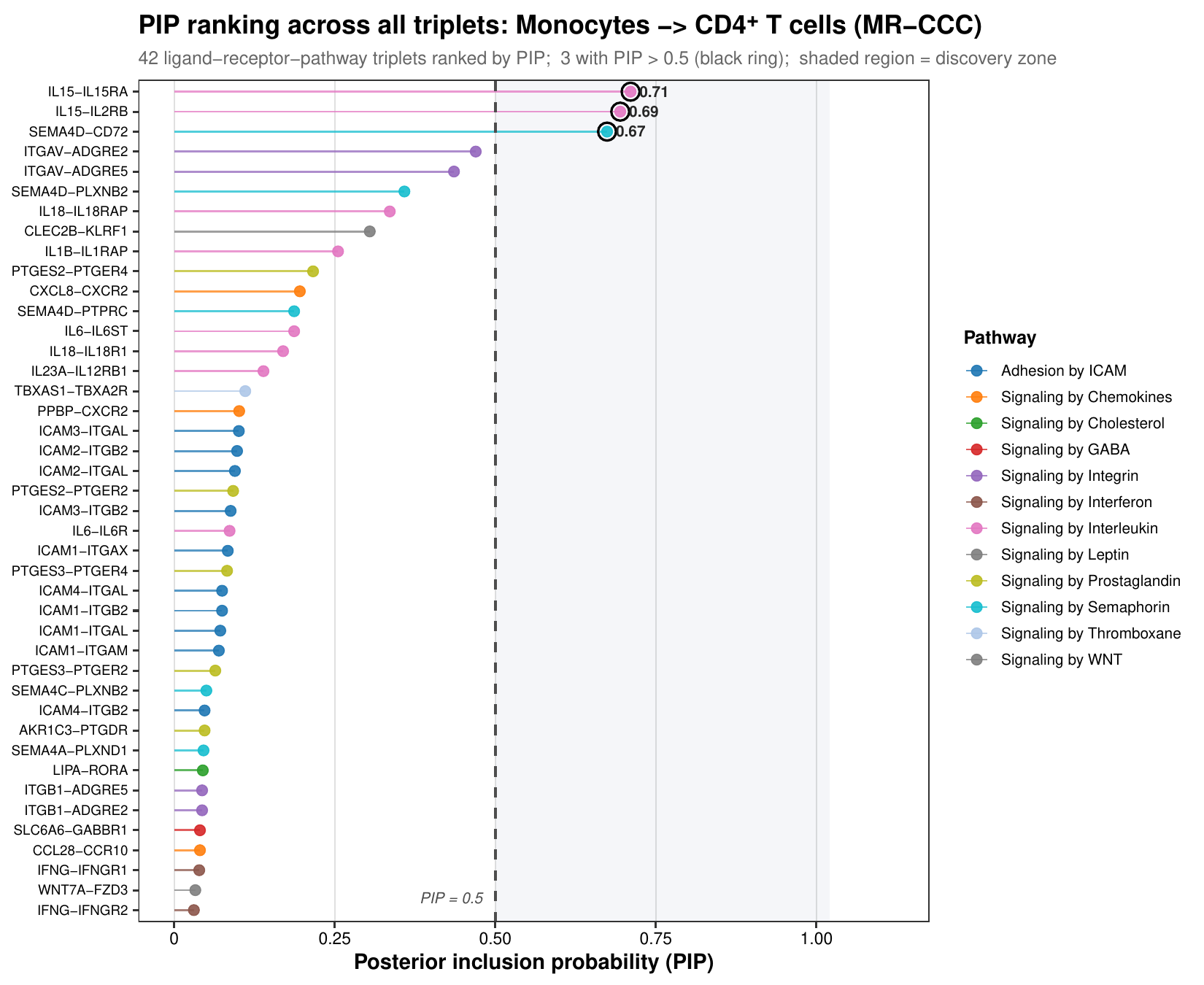}
\caption{\textbf{PIP ranking for the Monocytes\,$\rightarrow$\,CD4\(^+\) T cells analysis.} All 42 ligand--receptor--pathway triplets across 691
  donors. Points are colored by pathway; the dashed vertical line
  marks the discovery threshold of PIP \(= 0.5\); the shaded region
  to the right is the discovery zone. Black rings identify the three
  discovered triplets: \textit{IL15}--\textit{IL15RA}
  (PIP \(= 0.71\)), \textit{IL15}--\textit{IL2RB}
  (PIP \(= 0.69\)), and \textit{SEMA4D}--\textit{CD72}
  (PIP \(= 0.67\)).}
\label{fig:supp_Mono_CD4_pip}
\end{figure}

\begin{figure}[htbp]
\centering
\includegraphics[width=\textwidth]{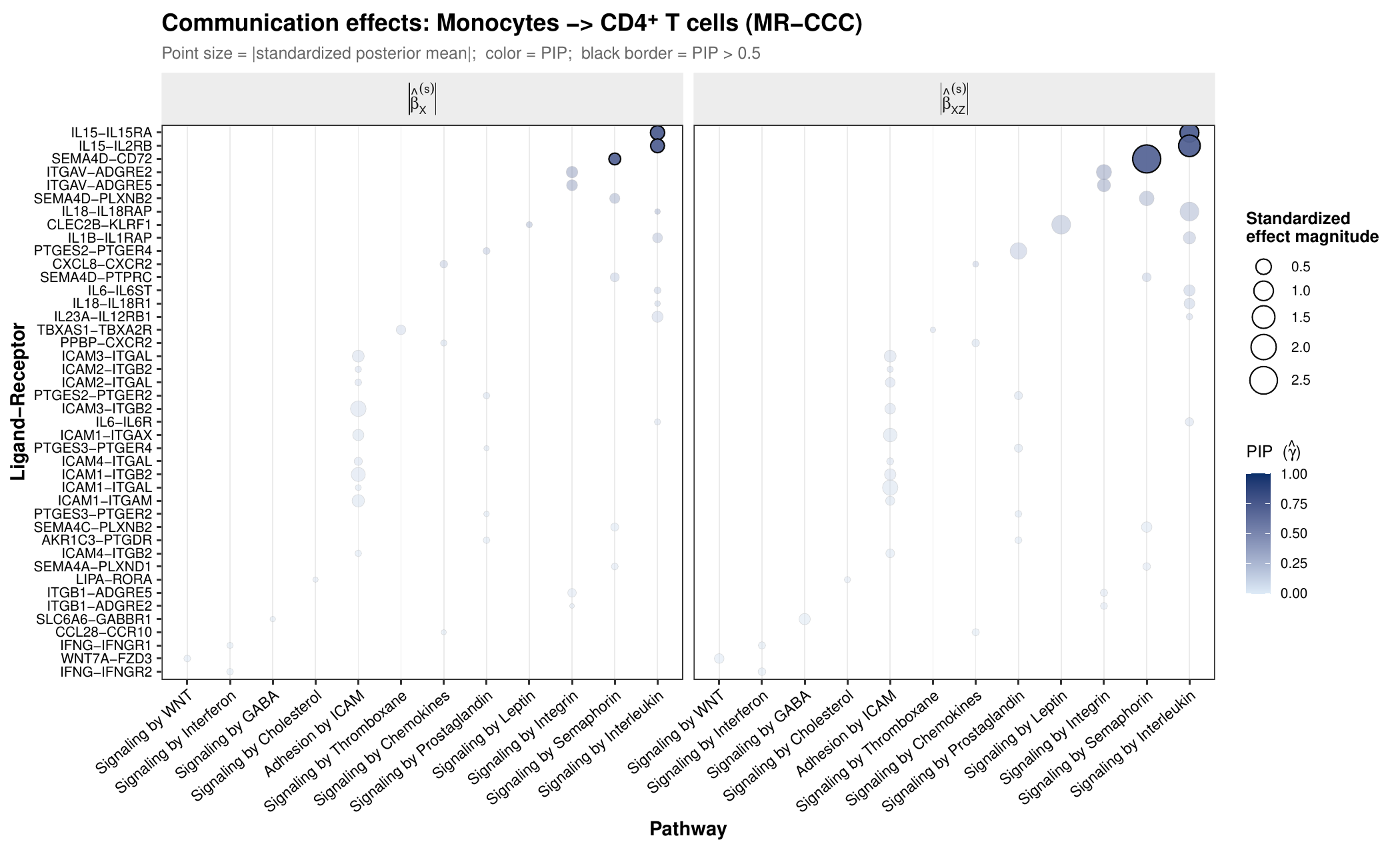}
\caption{\textbf{Standardized posterior effects for the Monocytes \(\rightarrow\)
  CD4\(^+\) T cells analysis.} Left panel: absolute main ligand effect
  \(|\hat{\beta}_X^{(s)}|\); right panel: absolute receptor-modulated
  interaction effect \(|\hat{\beta}_{XZ}^{(s)}|\). Point size encodes
  effect magnitude; fill color encodes PIP; black borders identify the
  three discovered triplets. \textit{SEMA4D}--\textit{CD72} dominates
  the right panel with the largest interaction magnitude in this
  direction; \textit{IL15}--\textit{IL15RA} and
  \textit{IL15}--\textit{IL2RB} show appreciable effects in both
  panels with opposing interaction signs, consistent with their
  divergent receptor-modulated effect profiles.}
\label{fig:supp_Mono_CD4_bubble}
\end{figure}

\begin{figure}[htbp]
\centering
\includegraphics[width=\textwidth]{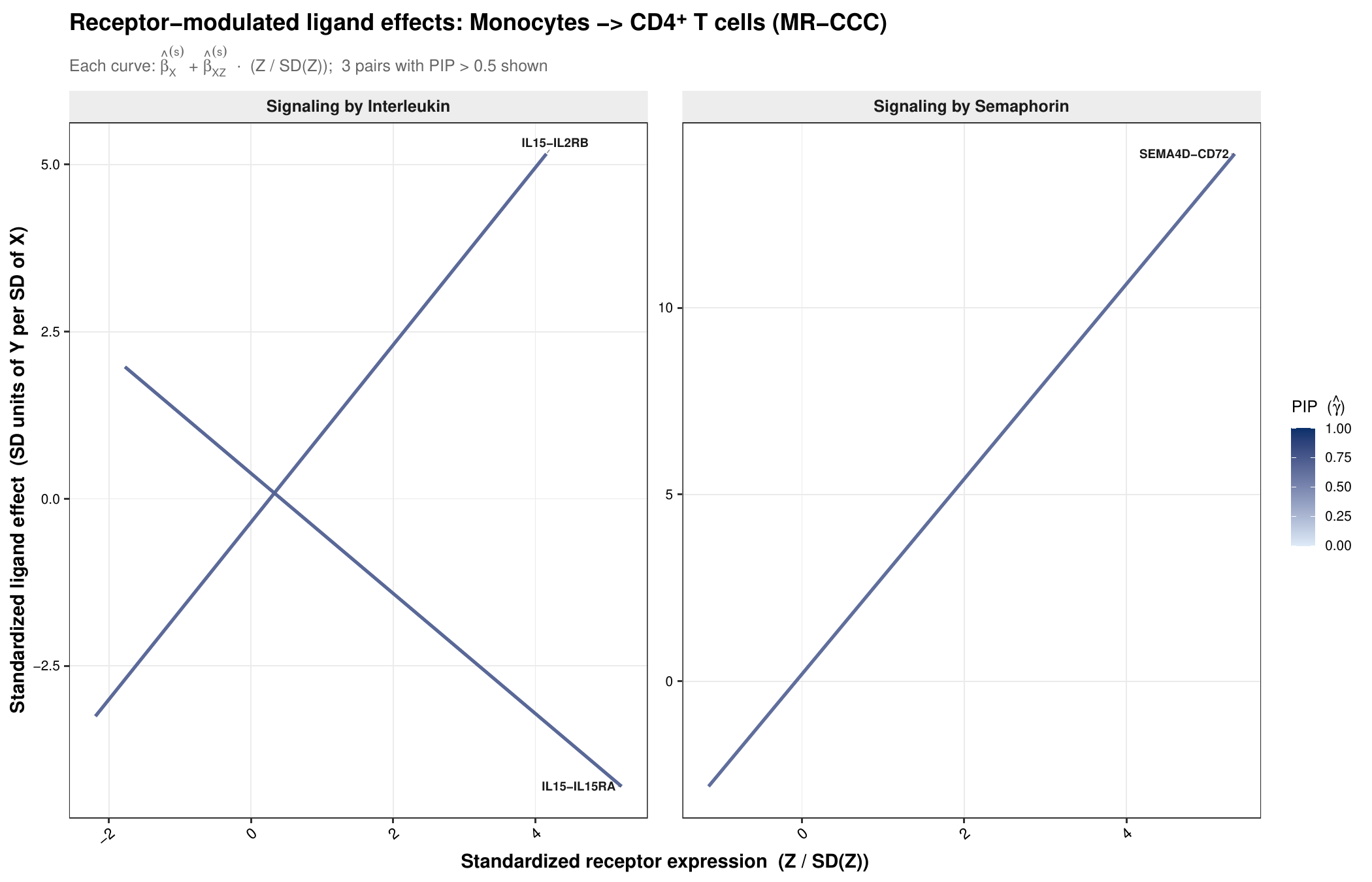}
\caption{\textbf{Receptor-modulated effect curves
  (\(\hat{\beta}_X + \hat{\beta}_{XZ} \cdot Z/\mathrm{SD}(Z)\))
  for the Monocytes \(\rightarrow\) CD4\(^+\) T cells analysis.}
  Only the three discovery pairs (PIP \(> 0.5\)) are displayed,
  grouped into the two pathway panels with confirmed discoveries
  (Signaling by Interleukin and Signaling by Semaphorin). In the
  Signaling by Interleukin panel,
  \textit{IL15}--\textit{IL15RA} descends from positive to negative,
  crossing zero at \(+0.42\) SD above the mean, while
  \textit{IL15}--\textit{IL2RB} ascends from negative to positive,
  crossing zero at \(+0.26\) SD above the mean; the two curves form
  an X-shaped crossing pair with opposite slopes despite sharing the
  same ligand. In the Signaling by Semaphorin panel,
  \textit{SEMA4D}--\textit{CD72} passes near zero just below the
  population mean and rises steeply to strongly positive values at
  high \textit{CD72} expression.}
\label{fig:supp_Mono_CD4_curves}
\end{figure}
\FloatBarrier

\subsubsection{Monocytes \(\rightarrow\) CD8\(^+\) T Cells}
\label{supp:Mono_CD8}

Across 678 donors, MR-CCC evaluated 42 ligand--receptor--pathway
triplets for the Monocyte (sender) to CD8\(^+\) T cell (receiver)
direction and identified three high-confidence causal communication
signals with PIP exceeding 0.5: \textit{ICAM1}--\textit{ITGAL} and
\textit{ICAM1}--\textit{ITGAX} within the Adhesion by ICAM pathway
(PIP \(= 0.654\) and PIP \(= 0.598\), respectively) and
\textit{IL1B}--\textit{IL1RAP} within the Signaling by Interleukin pathway
(PIP \(= 0.549\))
(Figures~\ref{fig:supp_Mono_CD8_pip}--\ref{fig:supp_Mono_CD8_curves}).

The two discovered ICAM pairs both engage monocyte ICAM-1 as ligand but
through different integrin receptors on CD8\(^+\) T cells, producing
qualitatively opposite interaction-modulated
outcomes~\cite{dustin1986icam,springer1990adhesion}.
\textit{ICAM1}--\textit{ITGAL} (PIP \(= 0.654\), \(\hat{\beta}_X = 0.541\),
\(\hat{\beta}_{XZ} = -3.02\)) has a positive main effect and a large negative
interaction, yielding a sign-reversing curve that crosses zero at approximately
\(0.18\) standard deviations above the mean \textit{ITGAL} level. CD8\(^+\)
T cells with LFA-1 expression below this threshold experience net stimulation
of Adhesion pathway activity by monocyte ICAM-1 engagement; those above it
transition to progressively stronger suppression, with the majority of the
observed CD8\(^+\) T cell population falling in the suppressive regime given
the near-mean zero crossing.
This suppressive pattern from monocyte sender to CD8\(^+\) T cell receiver
is consistent with the NK cell sender directions, where
\textit{ICAM1}--\textit{ITGAL} carried uniformly negative interactions
(NK cells \(\rightarrow\) CD8\(^+\) T cells: \(\hat{\beta}_{XZ} = -8.36\)).
Strikingly, however, the same triplet from the same monocyte sender to B cell
receiver showed the opposite sign (\(\hat{\beta}_{XZ} = +2.36\), Monocytes
\(\rightarrow\) B cells), converting receptor-amplified suppression in
CD8\(^+\) T cells into receptor-amplified stimulation in B cells. The
receiver-specific sign reversal of the ICAM1--LFA-1 interaction term between
CD8\(^+\) T cell and B cell targets from the identical monocyte sender
underscores the cell-type specificity of adhesion-mediated causal
communication detected by MR-CCC.

\textit{ICAM1}--\textit{ITGAX} (PIP \(= 0.598\), \(\hat{\beta}_X = 0.468\),
\(\hat{\beta}_{XZ} = 1.49\)) has a positive main effect and a positive
interaction, yielding a monotonically increasing effect curve with a zero
crossing at approximately \(0.31\) standard deviations \emph{below} the mean
\textit{ITGAX} expression level. Monocyte ICAM-1 therefore exerts a broadly
stimulatory causal effect on the Adhesion pathway activity of CD8\(^+\) T cells
expressing CD11c (encoded by \textit{ITGAX}), with the stimulation growing in
magnitude at higher CD11c densities~\cite{myones1988cd11c}. This positive
interaction from monocyte ICAM-1 through CD11c is consistent with the same
triplet in the Monocytes \(\rightarrow\) B cells direction
(\(\hat{\beta}_{XZ} = +0.32\)), where CD11c-expressing B cells were similarly
stimulated. The shared positive interaction sign for \textit{ICAM1}--\textit{ITGAX}
across B cell and CD8\(^+\) T cell receivers---in contrast to the opposing
\textit{ICAM1}--\textit{ITGAL} signs---indicates that the integrin
\(\alpha\) chain identity (\(\alpha\)X vs \(\alpha\)L) is the key determinant
of whether monocyte ICAM-1 contact is stimulatory or suppressive in the
receiver cell.

\textit{IL1B}--\textit{IL1RAP} (PIP \(= 0.549\), \(\hat{\beta}_X = 0.231\),
\(\hat{\beta}_{XZ} = -0.392\)) identifies monocyte-derived IL-1\(\beta\) as
a causal activator of CD8\(^+\) T cell Interleukin pathway activity through
the IL-1 receptor accessory protein~\cite{dinarello2009interleukin}. The positive
main effect and moderate negative interaction yield a sign-reversing curve
that crosses zero at approximately \(0.59\) standard deviations above the
mean \textit{IL1RAP} level. CD8\(^+\) T cells with below-average or
moderately above-average IL-1RAP expression experience net stimulation by
monocyte IL-1\(\beta\); those with the highest IL-1RAP expression undergo
progressive attenuation and reversal, consistent with receptor-mediated
negative feedback under high accessory protein density. This triplet was
also discovered from monocyte sender to B cell receiver (PIP \(= 0.685\),
\(\hat{\beta}_{XZ} = -0.098\)), but there the interaction was negligibly
small and the effect was essentially a uniform positive main effect; the
more pronounced negative interaction here (\(\hat{\beta}_{XZ} = -0.392\))
indicates that IL-1RAP-mediated attenuation of monocyte IL-1\(\beta\)
signaling is a more receptor-density-dependent process in CD8\(^+\) T cells
than in B cells.

Among sub-threshold triplets, \textit{ICAM3}--\textit{ITGAL}
(PIP \(= 0.429\), \(\hat{\beta}_X = -0.192\), \(\hat{\beta}_{XZ} = 1.83\))
was the highest-ranking near-miss; the negative main effect and large positive
interaction produce a sign-reversing curve that crosses zero at approximately
\(0.10\) standard deviations above the mean \textit{ITGAL} level, with
CD8\(^+\) T cells expressing LFA-1 above the mean experiencing increasing
stimulation of Adhesion pathway activity by monocyte ICAM3. Its near-discovery
alongside the confirmed \textit{ICAM1}--\textit{ITGAL} signal highlights the
complex heterogeneity of monocyte ICAM-family contacts at the
monocyte--CD8\(^+\) T cell interface.
\textit{IL15}--\textit{IL2RB} (PIP \(= 0.482\)) and
\textit{IL15}--\textit{IL15RA} (PIP \(= 0.431\)) were the highest-ranking
Interleukin near-misses; both IL-15 pairs were discovered from monocyte sender
to CD4\(^+\) T cell receiver (PIPs \(= 0.694\) and \(0.711\), respectively),
but fall below the discovery threshold here, suggesting that the monocyte
IL-15 transpresentation axis targets CD4\(^+\) T cells more consistently than
CD8\(^+\) T cells at the population level in this
cohort~\cite{waldmann2006il15,fehniger2001il15}.

\textit{IFNG}--\textit{IFNGR1} (PIP \(= 0.166\)) received a higher PIP than
in most sender directions but still fell well below the discovery threshold;
monocytes are not a dominant physiological source of IFN-\(\gamma\), and
the moderate PIP here likely reflects partial overlap of the monocyte
\textit{IFNG} cis-eQTL with IFN-\(\gamma\)-responsive CD8\(^+\) T cell
pathway activity at the population
level~\cite{schroder2004interferon,bach1997ifngr}.

\begin{figure}[htbp]
\centering
\includegraphics[width=\textwidth]{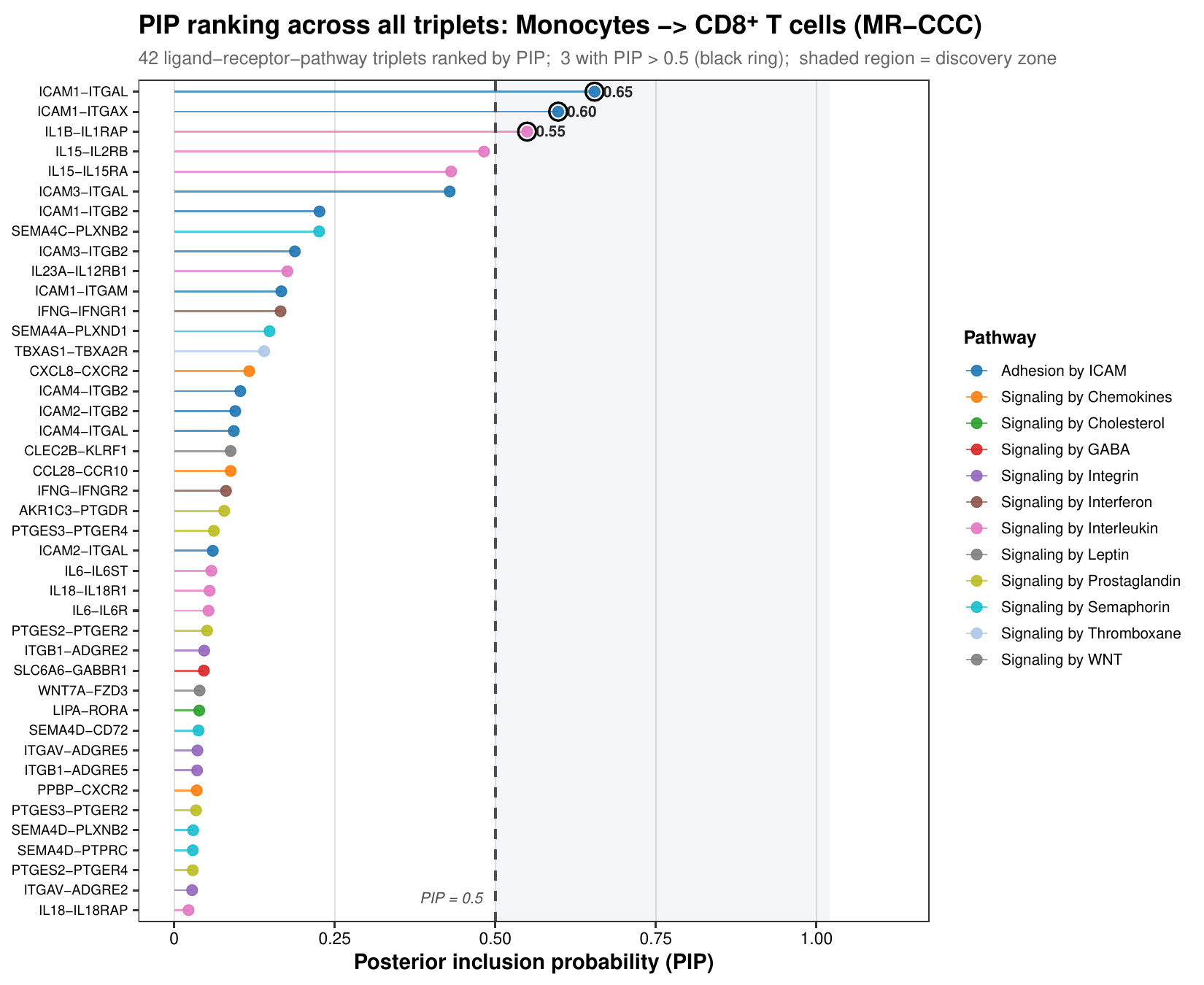}
\caption{\textbf{PIP ranking for the Monocytes\,$\rightarrow$\,CD8\(^+\) T cells analysis.} All 42 ligand--receptor--pathway triplets across 678
  donors. Points are colored by pathway; the dashed vertical line
  marks the discovery threshold of PIP \(= 0.5\); the shaded region
  to the right is the discovery zone. Black rings identify the three
  discovered triplets: \textit{ICAM1}--\textit{ITGAL}
  (PIP \(= 0.65\)), \textit{ICAM1}--\textit{ITGAX}
  (PIP \(= 0.60\)), and \textit{IL1B}--\textit{IL1RAP}
  (PIP \(= 0.55\)).}
\label{fig:supp_Mono_CD8_pip}
\end{figure}

\begin{figure}[htbp]
\centering
\includegraphics[width=\textwidth]{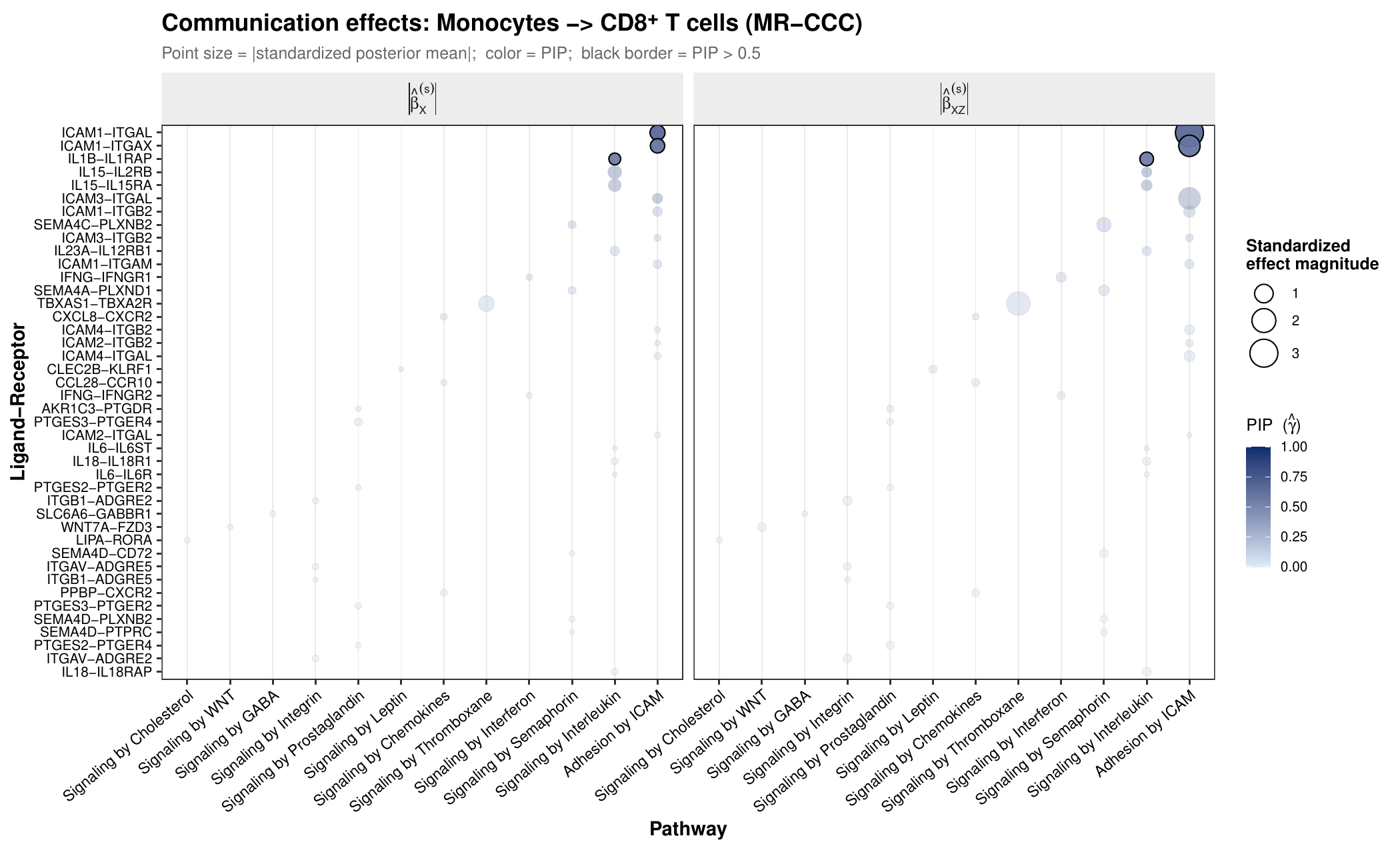}
\caption{\textbf{Standardized posterior effects for the Monocytes \(\rightarrow\)
  CD8\(^+\) T cells analysis.} Left panel: absolute main ligand effect
  \(|\hat{\beta}_X^{(s)}|\); right panel: absolute receptor-modulated
  interaction effect \(|\hat{\beta}_{XZ}^{(s)}|\). Point size encodes
  effect magnitude; fill color encodes PIP; black borders identify the
  three discovered triplets. \textit{ICAM1}--\textit{ITGAL} dominates
  the right panel with the largest negative interaction magnitude in
  this direction; \textit{ICAM1}--\textit{ITGAX} shows the largest
  positive interaction; \textit{IL1B}--\textit{IL1RAP} shows a
  moderate main effect with a smaller negative interaction.}
\label{fig:supp_Mono_CD8_bubble}
\end{figure}

\begin{figure}[htbp]
\centering
\includegraphics[width=\textwidth]{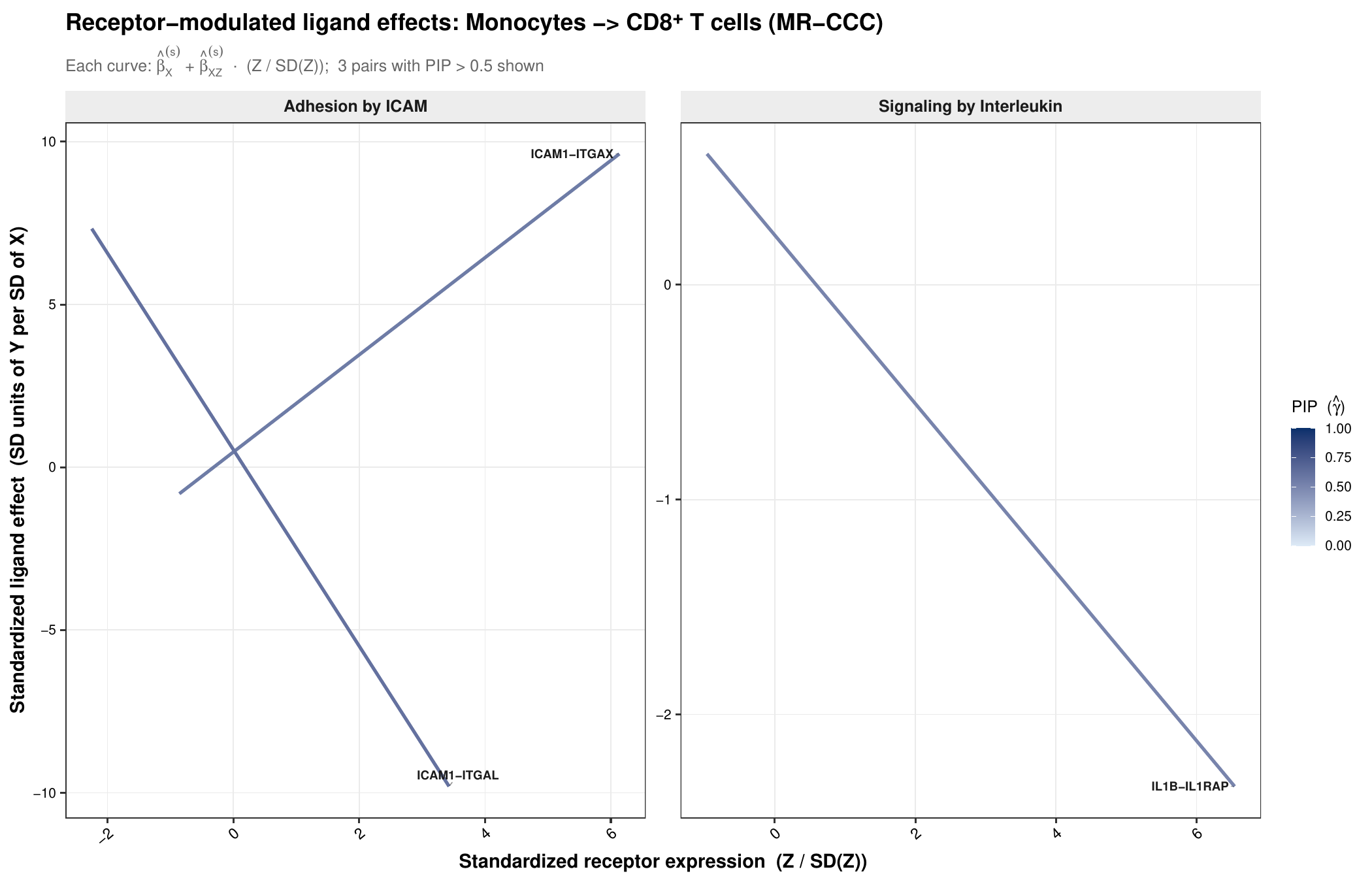}
\caption{\textbf{Receptor-modulated effect curves
  (\(\hat{\beta}_X + \hat{\beta}_{XZ} \cdot Z/\mathrm{SD}(Z)\))
  for the Monocytes \(\rightarrow\) CD8\(^+\) T cells analysis.}
  Only the three discovery pairs (PIP \(> 0.5\)) are displayed,
  grouped into the two pathway panels with confirmed discoveries
  (Adhesion by ICAM and Signaling by Interleukin). In the Adhesion
  by ICAM panel, \textit{ICAM1}--\textit{ITGAL} descends steeply
  from positive to negative, crossing zero near \(+0.18\) SD above
  the mean, while \textit{ICAM1}--\textit{ITGAX} rises monotonically,
  remaining positive across the majority of the observed CD11c
  expression range with a zero crossing near \(-0.31\) SD below the
  mean; the two curves form a crossing pair reflecting the divergent
  interaction signs for the \(\alpha\)L versus \(\alpha\)X integrin
  chains. In the Signaling by Interleukin panel, the
  \textit{IL1B}--\textit{IL1RAP} curve is positive at low to moderate
  \textit{IL1RAP} expression and crosses zero near \(+0.59\) SD above
  the mean.}
\label{fig:supp_Mono_CD8_curves}
\end{figure}
\FloatBarrier

\subsubsection{Monocytes \(\rightarrow\) NK Cells}
\label{supp:Mono_NK}

Across 651 donors, MR-CCC evaluated 41 ligand--receptor--pathway triplets for the Monocyte (sender) to NK cell
(receiver) direction and identified one high-confidence causal
communication signal with PIP exceeding 0.5: \textit{CXCL8}--\textit{CXCR2}
within the Chemokine signaling pathway (PIP \(= 0.518\))
(Figures~\ref{fig:supp_Mono_NK_pip}--\ref{fig:supp_Mono_NK_curves}).

The sole discovery, \textit{CXCL8}--\textit{CXCR2}
(PIP \(= 0.518\), \(\hat{\beta}_X = 0.375\), \(\hat{\beta}_{XZ} = -0.180\)),
identifies monocyte-derived CXCL8 (IL-8) as a causal activator of NK cell
Chemokine pathway activity through the CXCR2 receptor. Monocytes are among
the most prolific cellular sources of CXCL8 in peripheral blood, and NK cells
express CXCR2 on cytotoxic and terminally differentiated
subsets~\cite{ahuja1996cxcr2}. The positive main effect and negligibly small
negative interaction term place the zero crossing at approximately \(+2.08\)
standard deviations above the mean \textit{CXCR2} expression level---far
above the bulk of the observed NK cell distribution. The causal effect is
therefore positive across essentially the full observed CXCR2 expression
range, with only the most extreme high-CXCR2 NK cells approaching neutral
or mildly negative territory. This triplet is a near-pure main-effect
discovery: monocyte CXCL8 production exerts a broadly uniform stimulatory

effect on NK cell Chemokine pathway activity that is largely independent of
receptor expression level. \textit{CXCL8}--\textit{CXCR2} was also discovered
in the NK cells \(\rightarrow\) CD8\(^+\) T cells direction (PIP \(= 0.546\)),
where NK cells were the sender; its appearance here from the monocyte sender
direction confirms that CXCL8--CXCR2 signaling is an active chemokine
communication axis in both directions involving NK cells, with monocytes and
NK cells each capable of driving CXCR2-mediated pathway activity in the other
cell type.

The overall sparsity of discoveries in this direction---a single borderline
triplet, in contrast to the eight discoveries identified in the reverse
direction (NK cells \(\rightarrow\) Monocytes)---represents the strongest
directional asymmetry in the entire analysis. NK cells act as potent causal
communicators to monocytes via IFN-\(\gamma\), IL-18, IL-1\(\beta\),
prostaglandin, and GABA receptor axes, while monocytes return only a minimal
consistent causal footprint on NK cell pathway activity at the population level.
This asymmetry is consistent with the established immunological relationship
between NK cells and monocytes in peripheral blood: NK cells provide
activation, licensing, and cytokine signals to monocytes, while monocyte-to-NK
communication in resting peripheral blood is more limited and dominated by
contact-independent chemokine cues~\cite{schroder2004interferon,kalinski2012pge2}.

Several triplets narrowly missed the discovery threshold.
\textit{IL1B}--\textit{IL1RAP} (PIP \(= 0.471\),
\(\hat{\beta}_X = 0.328\), \(\hat{\beta}_{XZ} = 0.389\)) was the
highest-ranking near-miss; both coefficients are positive, producing a
monotonically stimulatory curve with a zero crossing near \(-0.84\) standard
deviations below the mean \textit{IL1RAP} level. This near-discovery is
noteworthy because \textit{IL1B}--\textit{IL1RAP} was confirmed from monocyte
senders to both B cells (PIP \(= 0.685\)) and CD8\(^+\) T cells
(PIP \(= 0.549\)), and was also discovered from NK cell senders to monocytes (PIP \(= 0.716\)); its sub-threshold PIP here suggests that monocyte IL-1\(\beta\)
signaling to NK cells through IL-1RAP is less consistent at the population
level than to lymphocyte targets~\cite{dinarello2009interleukin}.
\textit{SLC6A6}--\textit{GABBR1} (PIP \(= 0.442\),
\(\hat{\beta}_X = -0.235\), \(\hat{\beta}_{XZ} = 0.422\)) was also a
notable near-miss; the negative main effect and positive interaction produce
a sign-reversing curve crossing zero near \(+0.56\) standard deviations above
the mean \textit{GABBR1} level, with NK cells expressing high GABBR1 showing
a net stimulatory response to monocyte SLC6A6 taurine-transport activity.
This contrasts with the NK cell sender directions (NK cells \(\rightarrow\)
B cells and NK cells \(\rightarrow\) Monocytes), where \textit{SLC6A6}--\textit{GABBR1}
carried large negative interactions producing receptor-gated suppression; the
positive interaction here from monocyte sender to NK receiver is opposite in
sign, but falls below the discovery threshold~\cite{tian1999gaba,bjurstrom2008gaba}.
\textit{ICAM1}--\textit{ITGAX} (PIP \(= 0.393\)) was the highest-ranking
Adhesion by ICAM near-miss, consistent with the monocyte ICAM1--CD11c axis
that was discovered targeting both B cells and CD8\(^+\) T cells but does not
achieve consistency to NK cells at the population level.

\textit{IFNG}--\textit{IFNGR1} (PIP \(= 0.076\)) and
\textit{IFNG}--\textit{IFNGR2} (PIP \(= 0.162\)) both received low support,
consistent with monocytes not being a dominant physiological source of
IFN-\(\gamma\)~\cite{schroder2004interferon}. All Semaphorin, Integrin,
Cholesterol, WNT, and Thromboxane pathway triplets received PIPs below 0.35.

\begin{figure}[htbp]
\centering
\includegraphics[width=\textwidth]{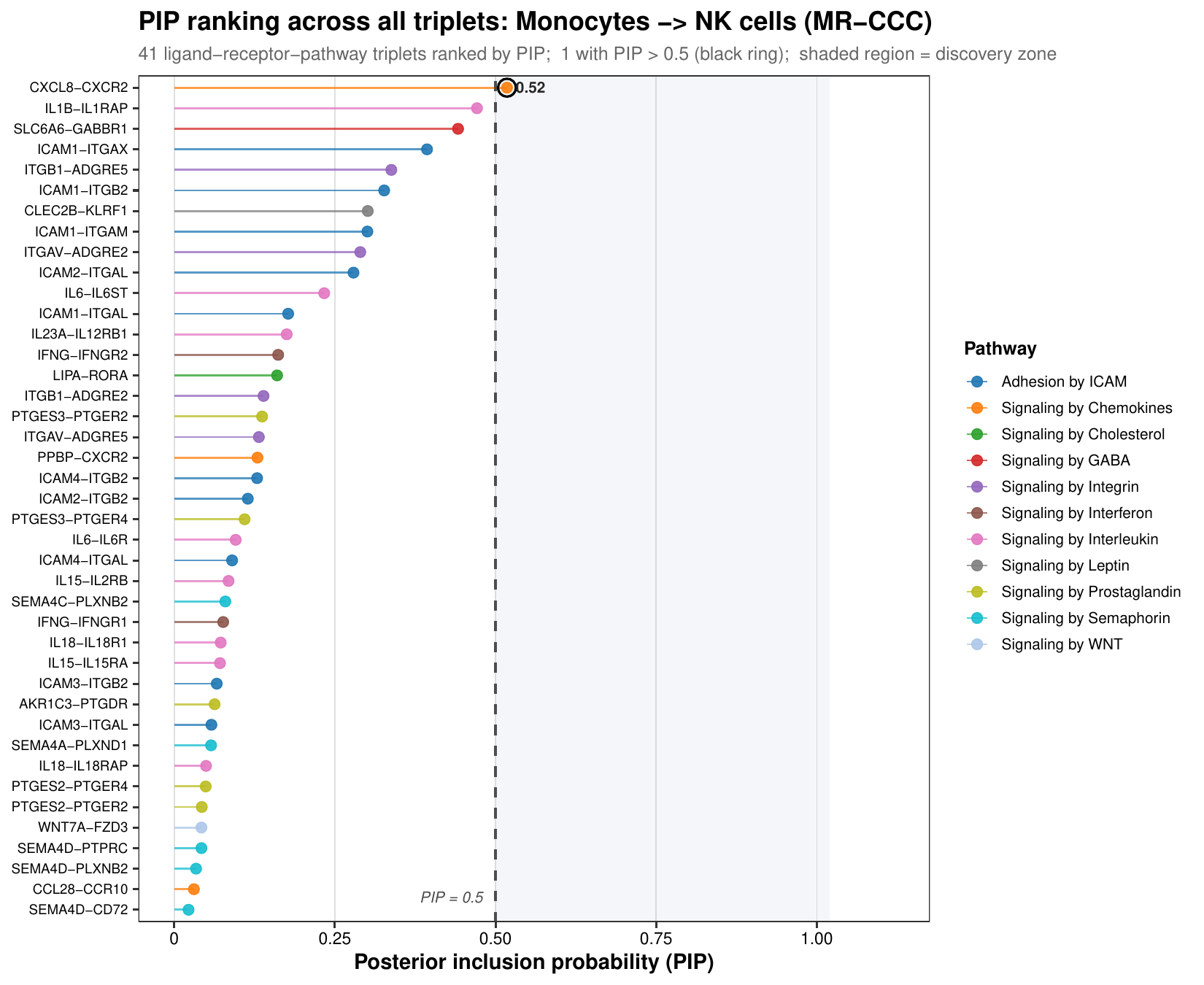}
\caption{\textbf{PIP ranking for the Monocytes\,$\rightarrow$\,NK cells analysis.} All 41 ligand--receptor--pathway triplets across 651 donors.
  Points are colored by pathway; the dashed vertical line marks the
  discovery threshold of PIP \(= 0.5\); the shaded region to the
  right is the discovery zone. The single discovered triplet,
  \textit{CXCL8}--\textit{CXCR2} (PIP \(= 0.52\)), is identified
  with a black ring. The two highest-ranking sub-threshold triplets,
  \textit{IL1B}--\textit{IL1RAP} (PIP \(= 0.47\)) and
  \textit{SLC6A6}--\textit{GABBR1} (PIP \(= 0.44\)), fall just below
  the threshold.}
\label{fig:supp_Mono_NK_pip}
\end{figure}

\begin{figure}[htbp]
\centering
\includegraphics[width=\textwidth]{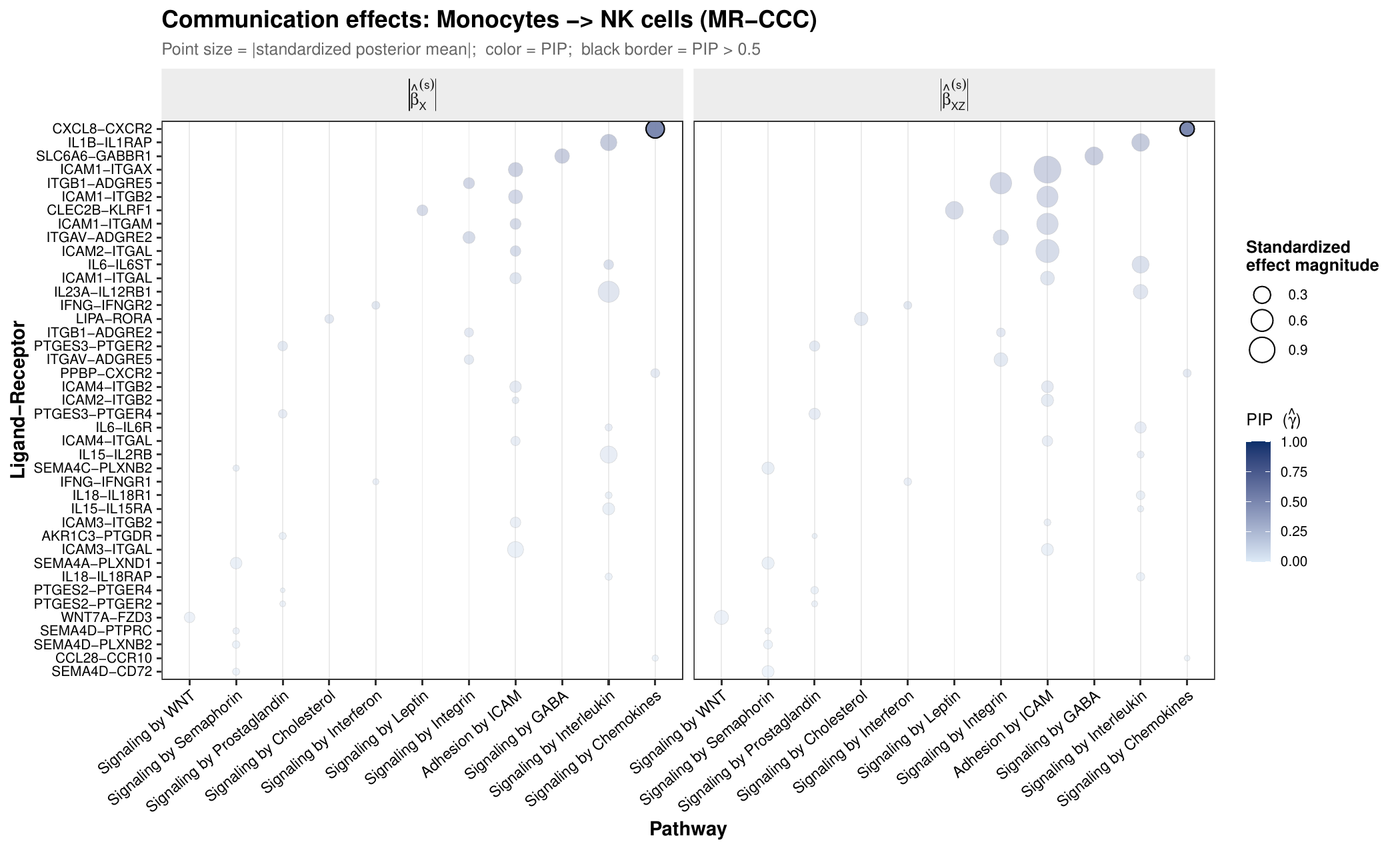}
\caption{\textbf{Standardized posterior effects for the Monocytes \(\rightarrow\)
  NK cells analysis.} Left panel: absolute main ligand effect
  \(|\hat{\beta}_X^{(s)}|\); right panel: absolute receptor-modulated
  interaction effect \(|\hat{\beta}_{XZ}^{(s)}|\). Point size encodes
  effect magnitude; fill color encodes PIP; the single black-bordered
  point identifies the discovered triplet \textit{CXCL8}--\textit{CXCR2}.
  Effect magnitudes are small throughout both panels and no pathway
  cluster shows a concentration of large high-PIP effects, consistent
  with the near-absence of discoveries in this direction.}
\label{fig:supp_Mono_NK_bubble}
\end{figure}

\begin{figure}[htbp]
\centering
\includegraphics[width=\textwidth]{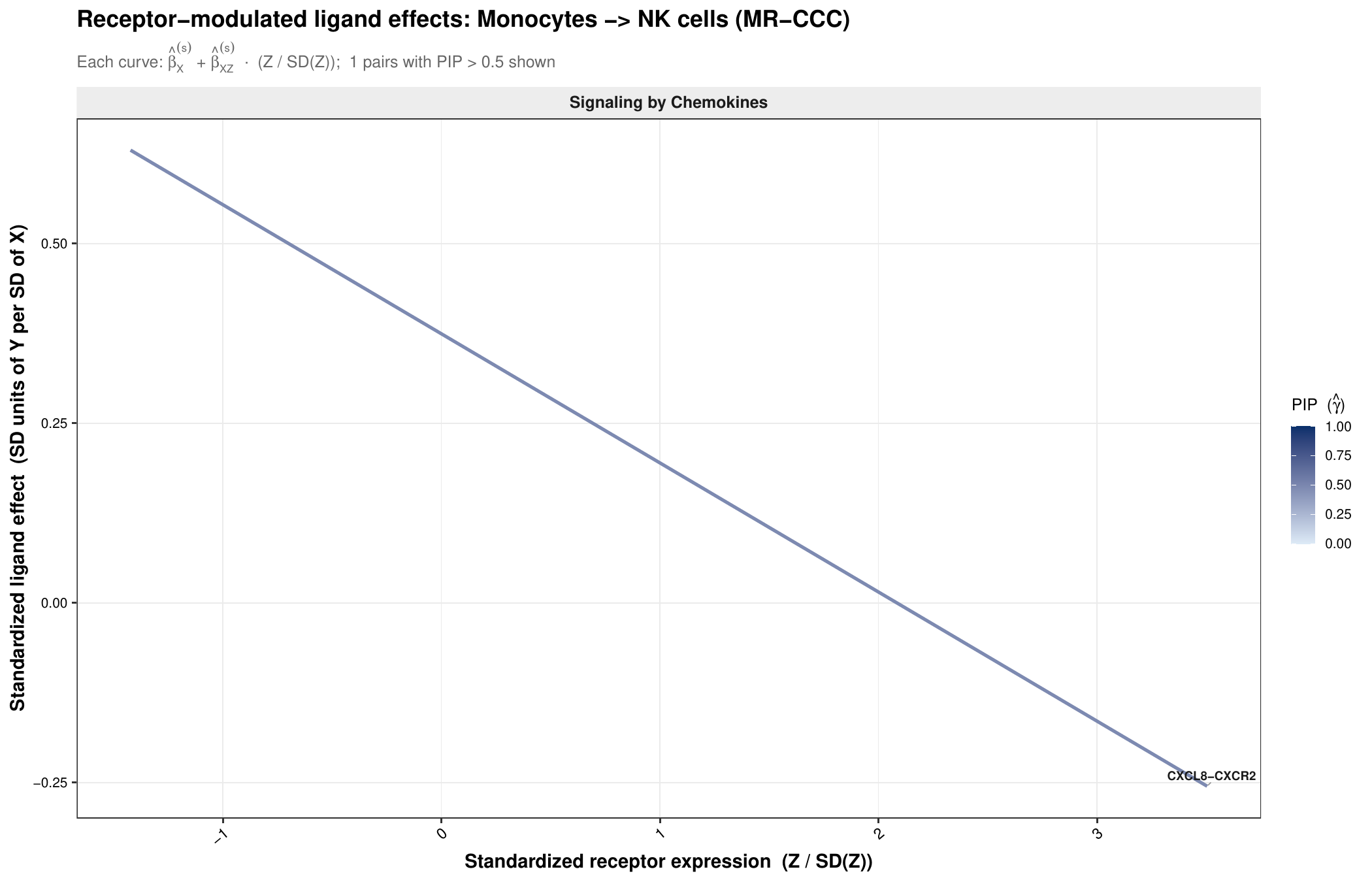}
\caption{\textbf{Receptor-modulated effect curves
  (\(\hat{\beta}_X + \hat{\beta}_{XZ} \cdot Z/\mathrm{SD}(Z)\))
  for the Monocytes \(\rightarrow\) NK cells analysis.} Only the
  single discovery pair (PIP \(> 0.5\)) is displayed in the Signaling
  by Chemokines panel. The \textit{CXCL8}--\textit{CXCR2} curve is
  nearly flat and positive across the bulk of the observed
  \textit{CXCR2} expression range, with a zero crossing near
  \(+2.08\) standard deviations above the mean, consistent with a
  near-pure main-effect discovery.}
\label{fig:supp_Mono_NK_curves}
\end{figure}

\FloatBarrier

Across all 19 supplementary ordered pairs, MR-CCC consistently
assigns high posterior probability to biologically established signaling
axes while returning near-zero PIPs for mechanistically implausible
directions.  Recurrent triplets such as \textit{ICAM1}--\textit{ITGAL},
\textit{PTGES3}--\textit{PTGER2}, and \textit{SLC6A6}--\textit{GABBR1}
appear across multiple sender--receiver combinations with interaction
signs that are specific to the sender cell type and receiver cell type,
illustrating that MR-CCC detects directionally and cellularly resolved
causal signals rather than co-expressed pairs that would be flagged
indiscriminately under association-based methods.
The directional asymmetries observed---most strikingly the eight discoveries from NK cells to monocytes versus the single borderline
discovery in the reverse direction---are consistent with known
immunological communication hierarchies in peripheral blood and
reinforce the biological validity of the causal inference framework
established in the main manuscript.


\FloatBarrier
\bibliographystyle{unsrtnat}  
\bibliography{Bibliography-MM-MC}  